\documentclass{lmcs}
\pdfoutput=1
\usepackage[utf8]{inputenc}

\usepackage{lastpage}
\lmcsdoi{21}{4}{24}
\lmcsheading{}{\pageref{LastPage}}{}{}%
{Feb.~20,~2024}{Nov.~21,~2025}{}

\usepackage{hyperref}
\usepackage{amssymb}
\usepackage{tikz} % Graphs
\usepackage{leftindex} % Left indices
\usepackage{cmll} % LL symbols
\usepackage{bussproofs} % Proof tree
\usepackage{thm-restate} % Have some theorems appear twice
\usepackage{multirow} % For cells spanning multiple rows in a table
\usepackage{afterpage} % Not to have floats flushed to the end of the paper
\usepackage{enumitem}
\usepackage{booktabs}
\usepackage{leftindex} % Left indices
\usepackage{adjustbox}
\adjustboxset{max width = \linewidth, center = \linewidth}

\newcommand*{\orth}{^\perp} % Dual
\newcommand*{\biorth}{^{\perp\perp}} % Bidual
\newcommand*{\tens}{\otimes} % Tensor
\newcommand*{\lolipop}{\multimap} % Linear Implication
\newcommand*{\pw}{\parr\backslash\with} % Parr or With
\newcommand*{\tp}{\tens\backslash\oplus} % Tensor or Plus

\renewcommand*{\fCenter}{\vdash} % Tile symbol

\newcommand*{\Raxocc}[1]{
    \AxiomC{}
    \RightLabel{$ax$}
    \UnaryInfC{$\fCenter$ #1}
}
\newcommand*{\Rax}[1]{\Raxocc{#1$\orth,$ #1}}

\newcommand*{\Rcut}[1]{
    \RightLabel{$cut$}
    \BinaryInfC{$\fCenter$ #1}
}

\newcommand*{\Rtens}[1]{
    \RightLabel{$\tens$}
    \BinaryInfC{$\fCenter$ #1}
}
\newcommand*{\Rparr}[1]{
    \RightLabel{$\parr$}
    \UnaryInfC{$\fCenter$ #1}
}
\newcommand*{\Rwith}[1]{
    \RightLabel{$\with$}
    \BinaryInfC{$\fCenter$ #1}
}
\newcommand*{\Rwiths}[2]{
    \RightLabel{$\with_{#1}$}
    \UnaryInfC{$\fCenter$ #2}
}
\newcommand*{\Rplus}[2]{
    \RightLabel{$\oplus_{#1}$}
    \UnaryInfC{$\fCenter$ #2}
}
\newcommand*{\Rplusone}[1]{\Rplus{1}{#1}}
\newcommand*{\Rplustwo}[1]{\Rplus{2}{#1}}
\newcommand*{\Rplusi}[1]{\Rplus{i}{#1}}
\newcommand*{\Rex}[1]{
\RightLabel{$ex$}
\UnaryInfC{$\fCenter$ #1}
}
\newcommand*{\Rone}{
    \AxiomC{}
    \RightLabel{$1$}
    \UnaryInfC{$\fCenter 1$}
}
\newcommand*{\Rbot}[1]{
    \RightLabel{$\bot$}
    \UnaryInfC{$\fCenter$ #1}
}
\newcommand*{\Rtop}[1]{
    \AxiomC{}
    \RightLabel{$\top$}
    \UnaryInfC{$\fCenter$ #1}
}
\newcommand*{\Rsub}[2]{
    \AxiomC{#1}
    \noLine
    \UnaryInfC{$\fCenter$ #2}
}

\newcommand{\emptyforproof}{{\scriptsize\phantom{x}}}

\usetikzlibrary{decorations.pathmorphing} % For snake arrows
\newenvironment{myscope}{
\begin{scope}[every node/.style={circle,thick,draw,minimum size=8mm,inner sep=0}, every edge/.style={->,draw=black,thick}, every path/.style={-,draw=black,thick}]
}{
\end{scope}
}

\newenvironment{myscopec}[1]{
\begin{scope}[every node/.style={circle,thick,draw=#1,minimum size=8mm,inner sep=0}, every edge/.style={->,draw=#1,thick}, every path/.style={draw=#1, thick}]
\color{#1}
}{
\end{scope}
}
\newenvironment{myscope_sur}{
\begin{scope}[every edge/.style={-,draw=blue,thick,dashed,line width=2.5pt}, every path/.style={-,draw=blue,thick,dashed,line width=2.5pt}]
}{
\end{scope}
}
\newenvironment{myscope_i}{
\begin{scope}[every node/.style={circle,thick,draw=none,minimum size=8mm,inner sep=0}, every edge/.style={->,draw=none,thick}, every path/.style={-,draw=none,thick}]
\color{white}
}{
\end{scope}
}
\newcommand{\edgelabel}[1]{edge node[draw=none, sloped, align=center]{#1 \\}}
\newcommand{\edgelabeld}[1]{edge[-] node[draw=none, sloped, align=center]{#1 \\}}
\tikzset{
    arete nommee/.style={draw=none,inner sep=1pt,outer sep=2pt,rectangle,minimum size=0pt,auto},
    chemin nomme/.style={arete nommee,sloped},
}

\newcommand*{\ie}{\emph{i.e.}}
\newcommand*{\wolog}{\emph{w.l.o.g.}}
\newcommand*{\eg}{\emph{e.g.}}

\newcommand*{\resp}{resp.}

\newcommand*{\G}{\mathcal{G}}

\newcommand*{\Slices}{\mathcal{S}}

\newcommand*{\Rst}{\mathcal{R}} % Standard proof net associated to a proof tree
\newcommand*{\iso}{\mathrel{\simeq}} % Isomorphism
\newcommand*{\Eu}{\mathcal{L}} % With units
\newcommand*{\E}{{\Eu^\dagger}} % Without units
\newcommand*{\AC}{\mathcal{AC}} % Only associativity and commutativity
\newcommand*{\eqeq}{=} % Equality in an equational theory
\newcommand*{\eqE}{\eqeq_\E} % Equality in \E
\newcommand*{\eqEu}[1][]{\eqeq_{\Eu_{#1}}} % Equality in \Eu
\newcommand*{\eqAC}{\eqeq_\AC} % Equality in \AC
\newcommand*{\isom}{\eqEu} % Isomorphisms in MALL
\newcommand*{\thD}{\mathcal{D}}
\newcommand*{\isoc}{\eqeq_\thD} % Isomorphisms in *-autonomous categories
\newcommand*{\thS}{\mathcal{S}}
\newcommand*{\isoo}{\eqeq_\thS} % Isomorphisms in SMCC
\newcommand*{\thCP}{\mathcal{C}}

\newcommand*{\distsystem}{\mathfrak{D}}

\newcommand{\symboltext}[2]{\mathrel{\overset{\makebox[0pt]{\mbox{\normalfont\tiny\sffamily \ensuremath{#1}}}}{#2}}} % General pattern for text above symbol
\newcommand{\textsymbol}[2]{\mathrel{\underset{\makebox[0pt]{\mbox{\normalfont\tiny\sffamily \ensuremath{#1}}}}{#2}}} % General pattern for symbol above text
\newcommand{\eqtext}[1]{\symboltext{#1}{=}} % Equal with text above
\newcommand{\cadratin}{\mathrel{\text{\textemdash}}} % Tiret cadratin
\newcommand{\edgelab}[1]{\symboltext{#1}{\cadratin}} % Cadratin with text above
\newcommand{\jump}{\symboltext{j}{\longrightarrow}}

\newcommand{\betabar}{\overline{\beta}}
\newcommand{\etato}{\symboltext{\eta}{\longrightarrow}}
\newcommand{\etafrom}{\symboltext{\eta}{\longleftarrow}}
\newcommand{\etatostar}{\symboltext{\eta^*}{\longrightarrow}}
\newcommand{\etafromstar}{\symboltext{\eta^*}{\longleftarrow}}
\newcommand{\betato}{\symboltext{\beta}{\longrightarrow}}
\newcommand{\betafrom}{\symboltext{\beta}{\longleftarrow}}
\newcommand{\betatostar}{\symboltext{\beta^*}{\longrightarrow}}
\newcommand{\betafromstar}{\symboltext{\beta^*}{\longleftarrow}}
\newcommand{\betabarto}{\symboltext{\betabar}{\longrightarrow}}
\newcommand{\betabarfrom}{\symboltext{\betabar}{\longleftarrow}}
\newcommand{\betabartostar}{\symboltext{\betabar^*}{\longrightarrow}}
\newcommand{\betabartoplus}{\symboltext{\betabar^+}{\longrightarrow}}
\newcommand{\betabarfromstar}{\symboltext{\betabar^*}{\longleftarrow}}
\newcommand{\betabarfromplus}{\symboltext{\betabar^+}{\longleftarrow}}

\newcommand{\vdashv}{\mathrel{\lower .22ex \hbox{$\vdash\!\dashv$}}}
\DeclareRobustCommand{\astsymbol}[1]{\mathrel{\ooalign{{$\vdashv$}\crcr{$\stackrel{#1}{\hphantom{\vdashv}\vphantom{\to}}$}}}}
\DeclareRobustCommand{\astvdashv}{\astsymbol{\ast}}
\DeclareRobustCommand{\vdashveq}{\astsymbol{=}}

\newcommand{\eqcone}{\astsymbol{r}} % Cut-free rule commutation
\newcommand{\eqc}{\astsymbol{r^\ast}} % Cut-free rule commutation with transitivity and reflexivity
\newcommand{\eqceq}{\astsymbol{r^=}} % Cut-free rule commutation with transitivity and reflexivity
\newcommand{\eqcc}{\astsymbol{c}} % Cut-cut commutation
\newcommand{\eqccast}{\astsymbol{c^\ast}} % Cut-cut commutation

\newcommand{\comm}[2]{C_{#1}^{#2}} % Name of a commutation
\newcommand{\commup}[2]{\textsymbol{\comm{#1}{#2}}{\longleftarrow}} % Arrow above
\newcommand{\commdown}[2]{\symboltext{\comm{#1}{#2}}{\longrightarrow}} % Arrow below
\newcommand{\commbi}[2]{\underset{\commup{#1}{#2}}{\overset{\commdown{#2}{#1}}{}}} % Arrow in both ways
\newcommand*{\pcontext}[1]{\resizebox{!}{7pt}{$\bigoplus$}[#1]}

\newcommand{\dupcut}{\mathrel{\triangleleft}}

\newcommand{\eqb}{\mathrel{\eqeq_{\beta}}}

\newcommand{\eqbn}{\mathrel{\eqeq_{\beta\eta}}}
\newcommand{\eqbb}{\mathrel{\eqeq_{\betabar}}}

\newcommand*{\id}[1]{\ensuremath{\textnormal{id}_{#1}}}
\newcommand*{\ax}[1]{\ensuremath{\textnormal{ax}_{#1}}}

\newcommand*{\cutf}[3]{\ensuremath{#1\symboltext{#3}{\bowtie} #2}}
\newcommand{\isoproofs}[4]{\overset{\ensuremath{#3,#4}}{#1\iso #2}}

\newcommand*{\unique}{\textit{ax}-unique}
\newcommand*{\Unique}{\textit{Ax}-unique}

\newcommand*{\MALL}{\ensuremath{\mathbb{MALL}}}

\newcommand*{\iform}{\iota}
\newcommand*{\tradm}[1]{\ell(#1)}
\newcommand*{\tradc}[1]{\partial(#1)}
\newcommand*{\tradi}[1]{{#1}^\bullet}
\newcommand*{\bigtradi}[1]{\tradi{(#1)}}
\newcommand*{\trado}[1]{{#1}^\circ}
\newcommand*{\bigtrado}[1]{\trado{(#1)}}

\newcommand*{\proofpar}[1]{\textit{#1}}

\newenvironment{tableproof}{
\aboverulesep=.8ex
\belowrulesep=1.2ex
\begin{table}
}{
\end{table}
\aboverulesep=.4ex
\belowrulesep=.65ex
}

\newcommand{\floathere}{\afterpage{\clearpage}}

\newcommand{\coloncoloneqq}{\mathrel{::=}}

\newcommand*{\mass}[1]{m(#1)}
\newcommand*{\density}[1]{d(#1)}
\newcommand*{\ordercut}{\preccurlyeq}


\begin{document}

\title[Type Isomorphisms for Multiplicative-Additive Linear Logic]{Type Isomorphisms for\texorpdfstring{\\}{ }Multiplicative-Additive Linear Logic}
\titlecomment{{\lsuper*}
The present paper is a revised and extended version of~\cite{DiGuardiaLaurent23}, supplying all proofs and with a different organization.}

\author[R.~{Di~Guardia}]{R\'emi {Di~Guardia}\lmcsorcid{0009-0004-8632-108X}}
\author[O.~Laurent]{Olivier Laurent\lmcsorcid{0009-0007-1306-8994}}
\address{Univ Lyon, EnsL, UCBL, CNRS, LIP, F-69342, LYON Cedex 07, France}
\email{remi.di.guardia@ens-lyon.org, olivier.laurent@ens-lyon.fr}

%==================================================

\begin{abstract}
\noindent
We characterize type isomorphisms in the multiplicative-additive fragment of linear logic (MALL), and thus in $\star$-autonomous categories with finite products, extending a result for the multiplicative fragment by Balat and Di~Cosmo~[BDC99].
This yields a much richer equational theory involving distributivity and cancellation laws.
The unit-free case is obtained by relying on the proof-net syntax introduced by Hughes and Van~Glabbeek~[HvG05].
We use the sequent calculus to extend our results to full MALL, including all units, thanks to a study of cut-elimination and rule commutations.
% Must be without the \cite command:
% \cite{mllisos} -> [BDC99]
% \cite{mallpnlong} -> [HvG05]
\end{abstract}

\maketitle

%==================================================

\section{Introduction}
\label{sec:intro}

The question of type isomorphisms consists in trying to understand when two types in a type system, or two formulas in a logic, are ``the same''.
The general question can be described in category theory: two objects $A$ and $B$ are isomorphic, denoted $A \iso B$, if there exist morphisms $A\overset{f}{\longrightarrow} B$ and $B\overset{g}{\longrightarrow} A$ such that $f\circ g=\id{B}$ and $g\circ f=\id{A}$; \ie\ if the following diagram commutes:

\begin{adjustbox}{}
\begin{tikzpicture}
	\node[style={circle,draw}] (A) at (0,0) {$A$};
	\node[style={circle,draw}] (B) at (3,0) {$B$};
	\path[bend left] (A) edge[style=->] node[chemin nomme,below]{$f$} (B);
	\path[bend left] (B) edge[style=->] node[chemin nomme]{$g$} (A);
	\path (A) edge[loop left,style=->] node[arete nommee]{$\id{A}$} (A);
	\path (B) edge[loop right,style=->] node[arete nommee]{$\id{B}$} (B);
\end{tikzpicture}
\end{adjustbox}
The arrows $f$ and $g$ are the underlying isomorphisms.
Given a (class of) category, the question is then to find equations characterizing when two objects $A$ and $B$ are isomorphic (in all instances of the class).
The focus here is on pairs of isomorphic objects rather than on the isomorphisms themselves.
For example, in the class of cartesian categories, one finds the following isomorphic objects: $A\times B\iso B\times A$, $(A\times B)\times C\iso A\times(B\times C)$ and $A\times\top\iso A$.
Regarding type systems and logics, one can instantiate the categorical notion.
For instance in typed $\lambda$-calculi: two types $A$ and $B$ are isomorphic if there exist two $\lambda$-terms $M:A\rightarrow B$ and $N:B\rightarrow A$ such that $\lambda x:B.(M~(N~x))\eqbn \lambda x:B.x$ and $\lambda x:A.(N~(M~x))\eqbn \lambda x:A.x$ where $\eqbn$ is $\beta\eta$-equality.
This corresponds to isomorphic objects in the syntactic category generated by terms up to $\eqbn$.
Similarly, type isomorphisms can also be considered in logic, following what happens in the $\lambda$-calculus through the Curry-Howard correspondence: simply replace $\lambda$-terms with proofs, types with formulas, $\beta$-reduction with cut-elimination and $\eta$-expansion with axiom-expansion.
In this way, type isomorphisms are studied in a wide range of theories, such as category theory~\cite{isosoloviev,isossmcc}, $\lambda$-calculus~\cite{isotypes} and proof theory~\cite{mllisos}.
Knowing exactly the isomorphisms of a theory matters for (at least) two reasons.
First is a semantic motivation: isomorphic objects are those that cannot be distinguished by the theory, making isomorphisms central in category theory -- \eg\ unicity of a limit holds only up to isomorphism.
A natural question is then what is the quotient implicitly done by all these results with an ``up to isomorphism'', which then allows one for instance to look only at representatives of these equivalence classes.
Second, isomorphisms in a theory are often, but not always, those expected: \eg\ in a compact close category with finite products, these products are in fact biproducts~\cite{prodccc}.
Furthermore, type isomorphisms have been used to develop practical tools, such as search in a library of a functional programming language~\cite{isosearch,isosearchradanne}.

Following the definition, it is usually easy to prove that the type isomorphism relation is a congruence.
It is then natural to look for an equational theory generating this congruence.
Testing whether or not two types are isomorphic is then much easier.
An equational theory $\mathcal{T}$ is called \emph{sound} with respect to type isomorphisms if types equal up to $\mathcal{T}$ are isomorphic.
It is said \emph{complete} if it equates any pair of isomorphic types.
Given a (class of) category, a type system or a logic, our goal is to find an associated sound and complete equational theory for type isomorphisms.
This is not always possible as the induced theory may not be finitely axiomatisable -- see for instance~\cite{isosum}.

Soundness is usually the easy direction as it is sufficient to exhibit pairs of terms corresponding to each equation.
The completeness part is often harder, and there are in the literature two main approaches to solve this problem.
The first one is a semantic method, relying on the fact that if two types are isomorphic then they are isomorphic in all (denotational) models.
One thus looks for a model in which isomorphisms can be computed -- more easily than in the syntactic model -- and are all included in the equational theory under consideration; this is the approach used in~\cite{isosoloviev,classisos} for example.
Finding such a model simple enough for its isomorphisms to be computed, but still complex enough not to contain isomorphisms absent in the syntax is the difficulty.
The second method is the syntactic one, which consists in studying isomorphisms directly in the syntax.
The analysis of pairs of terms composing to the identity should provide information on their structure and then on their type so as to deduce the completeness of the equational theory; see for example~\cite{isotypes,mllisos}.
The more easily the equality ($\eqbn$ for example) between proof objects can be computed, the easier the analysis of isomorphisms will be.

We place ourselves in the framework of linear logic (LL)~\cite{ll},
the underlying question being ``is there an equational theory corresponding to the isomorphisms between formulas in this logic?''.
LL is a very rich logic containing three classes of propositional connectives: multiplicative, additive and exponential ones.
The multiplicative and additive families provide two copies of each classical propositional connective: two copies of conjunction ($\tens$ and $\with$), of disjunction ($\parr$ and $\oplus$), of true ($1$ and $\top$) and of false ($\bot$ and $0$).
The exponential family is constituted of two modalities $\oc$ and $\wn$ bridging the gap between multiplicatives and additives through four isomorphisms $\oc(A\with B)\iso \oc A\tens\oc B$, $\wn(A\oplus B)\iso \wn A\parr\wn B$, $\oc\top\iso 1$ and $\wn 0\iso\bot$.
In the multiplicative fragment (MLL) of LL (using only $\tens$, $\parr$, $1$ and $\bot$, and corresponding to $\star$-autonomous categories), the question of type isomorphisms was answered positively using a syntactic method based on proof-nets by Balat and Di~Cosmo~\cite{mllisos}: isomorphisms emerge from associativity and commutativity of the multiplicative connectives $\tens$ and $\parr$, as well as unitality of the multiplicative units $1$ and $\bot$.
The question was also solved for the polarized fragment of LL by one of the authors using game semantics~\cite{classisos}.
It is conjectured that isomorphisms in full LL correspond to those in its polarized fragment (constituted of the equations on \autoref{tab:eqisos} together with the four exponential equations above).
As a step towards solving this conjecture, we prove the type isomorphisms in the multiplicative-additive fragment (MALL) of LL are generated by the equational theory of \autoref{tab:eqisos}.
This applies at the same time to the class of $\star$-autonomous categories with finite products (whose corresponding equational theory in its usual syntax is in \autoref{tab:eqisos_cat} on \autopageref{tab:eqisos_cat}).
This situation is much richer than in the multiplicative fragment since isomorphisms include not only associativity, commutativity and unitality, but also the distributivity of the multiplicative connective $\tens$ (\resp\ $\parr$) over the additive $\oplus$ (\resp\ $\with$) as well as the associated cancellation law for the additive unit $0$ (\resp\ $\top$) over the multiplicative connective $\tens$ (\resp\ $\parr$).
Remark also that the $\with$ and the $\oplus$ (corresponding to the categorical product and sum) do not distribute, contrary to what happens for instance in cartesian closed categories.

\begin{table}
\begin{adjustbox}{}
\begin{tabular}{@{}ccccc@{}}
\multicolumn{2}{@{}c}{$\AC$} & & \\
\multicolumn{2}{@{}c}{\downbracefill} & & \\
\toprule
Associativity & Commutativity & Distributivity & Unitality & Cancellation \\
\midrule
$A\tens(B\tens C)=(A\tens B)\tens C$ & $A\tens B=B\tens A$ & $A\tens(B\oplus C)=(A\tens B)\oplus(A\tens C)$ & $A\tens 1=A$ & $A\tens0=0$ \\
$A\parr(B\parr C)=(A\parr B)\parr C$ & $A\parr B=B\parr A$ & $A\parr(B\with C)=(A\parr B)\with(A\parr C)$ & $A\parr \bot=A$ & $A\parr\top=\top$ \\
$A\oplus(B\oplus C)=(A\oplus B)\oplus C$ & $A\oplus B=B\oplus A$ & & $A\oplus 0=A$ & \\
$A\with(B\with C)=(A\with B)\with C$ & $A\with B=B\with A$ & & $A\with \top=A$ & \\
\bottomrule
\multicolumn{5}{@{}c@{}}{\upbracefill} \\
\multicolumn{5}{@{}c@{}}{$\Eu$} \\
\end{tabular}
\end{adjustbox}
\caption{Type isomorphisms in multiplicative-additive linear logic}
\label{tab:eqisos}
\end{table}

Using a semantic approach for completeness looks difficult here.
In particular, most of the known ``concrete'' models of MALL (with the meaning that isomorphisms are more easily computed inside) immediately come with unwanted isomorphisms not valid in the syntax.
For example, one can check that $\top \tens A \iso \top$ in coherent spaces~\cite{ll}, while it is plainly obvious that $\top \tens A \not\iso \top$ in the syntax as the second formula is provable but generally the first one is not.
An idea would then be to consider less ``concrete'' models of MALL or of $\star$-autonomous categories with finite products, typically categories built in a natural way over MLL or $\star$-autonomous categories.
For instance, one could consider the free completion by products and coproducts of $\star$-autonomous categories, hoping the resulting category corresponds to MALL -- in the same spirit as what has been done for bicompletion of $\star$-autonomous categories~\cite{freebicompletion1,freebicompletion2}, except MALL should not be the free bicompletion of MLL as it is not expected to have all limits.
Or one could consider models based on coherence spaces such as~\cite{coherencecompletion1,coherencecompletion2,hypercoherencecompletenessmall}.
We see two main difficulties with this approach.
The first is that isomorphisms may not be that easy to compute in these categories, and we lack results giving isomorphisms of a completion from isomorphisms of the base category.
Especially in our case, where we have distributivity isomorphisms involving both multiplicative and additive connectives, meaning isomorphisms of multiplicative-additive linear logic are not directly deducible from those of multiplicative linear logic and of additive linear logic.
The second obstacle is that these categories may not correspond exactly to MALL, and can have unwanted isomorphisms.
In particular, models built following coherence spaces contain the same unwanted $\top \tens A \iso \top$.

For this reason we prefer to use a syntactic method, and follow the approach from Balat and Di~Cosmo~\cite{mllisos} based on proof-nets~\cite{pn}.
Indeed, proof-nets provide a very good syntax for linear logic where it is very natural to study properties such as composition of proofs by cut, cut-elimination and identity of proofs.
All these highly simplify the problem of isomorphisms.
However, already in~\cite{mllisos} some trick had to be used to deal with units as proof-nets are working perfectly only in the unit-free multiplicative fragment of linear logic.
Indeed, the main result from~\cite{mllpnpspace} rules out ``simple'' canonical proof-nets with multiplicative units, as using such graphs one can solve a PSPACE-complete problem.
Nonetheless, if one forgets the canonical requirement for the units, some nice definitions for MLL exist, \eg~\cite{cohweakdistrcat,hughes*cat}, but one then has to consider such proof-nets up to an equivalence relation known as rewiring.
Looking at additives, there are canonical proof-nets for additive linear logic with units~\cite{allpn}, but these proof-nets hardly seem expandable to multiplicative-additive linear logic with additive units.
If one puts units aside, there is a notion of proof-nets incorporating both multiplicative and additive connectives in such a way that cut-free proofs are represented in a \emph{canonical} way, and cut-elimination can be dealt with in a parallel manner.
This is the syntax of proof-nets introduced by Hughes \& Van~Glabbeek in~\cite{mallpnlong}.
As a complement, proof-nets are not the only tool for considering canonical proofs.
For instance, an approach based on focusing in sequent calculus is developed in~\cite{canonicmultifoc}; however, while it reduces the size of the equivalence class of a proof, there still are different but equivalent proofs.

\paragraph{Sketch of the proof}
Our proof of completeness can be sketched as follows.
\begin{enumerate}[label=(\arabic*)]
\item
A simple but key idea is to use the distributivity, unitality and cancellation equations of \autoref{tab:eqisos} from left to right so as to rewrite formulas into a \emph{distributed} form, where none of these equations can be applied anymore.
Between such distributed formulas, the only isomorphisms left should be associativity and commutativity ones.
Furthermore, units should not play any special role.
All the problem now is to prove this is the case.
\item
Working in sequent calculus, we prove that in any isomorphism between distributed formulas one can indeed replace units by fresh atoms and the resulting formulas are still isomorphic.
The main difficulty here is that this only holds for distributed formulas, as generally units do not behave like atoms at all.
This leads us to identify the following \emph{patterns} in proofs of isomorphisms:
\begin{equation*}
\Rtop{$\top,0$}
\DisplayProof
\qquad
\Rone{}
\Rbot{$\bot, 1$}
\DisplayProof
\end{equation*}
Once proven that unit rules are constrained to these patterns, it is easy to replace each of their uses by an $ax$-rule.
A problem here is that we have to consider proofs up to cut-elimination, which is confluent only up to \emph{rule commutation}~\cite[Theorem~5.1]{calculusformall}.
In particular, instead of the second pattern above, one must consider the following more general one, with a sequence of $\oplus_1$- and $\oplus_2$-rules between the $1$- and $\bot$-rules, that turns the $1$-formula into a more complex one $F$:
\begin{prooftree}
\Rone{}
\doubleLine\dashedLine\RightLabel{$\oplus_i$}\UnaryInfC{$\fCenter F$}
\Rbot{$\bot,F$}
\end{prooftree}
Moreover, we need two main steps to get these patterns.
\begin{enumerate}[label=(2\alph*)]
\item
First, we prove these patterns are in the identity proofs and are preserved by all rule commutations, so as to know the equivalence class of the identity has them.
This is done by giving properties preserved by rule commutations, with the tedious work of checking each rule commutation.
\item
Then, we transpose these patterns backward through cut-elimination to get them in the proofs of isomorphisms.
This needs a detailed study to prove that cutting proofs of an isomorphism cannot ``completely erase'' a unit rule, \ie\ that the behavior of erasing cases are quite constrained.
We do so by a thorough consideration of the evolution of \emph{slices} when eliminating a $cut$-rule, as well as a meticulous monitoring of $\top$-rules during the reduction.
\end{enumerate}
\item
Once the problem reduced to the unit-free case, we can transpose isomorphisms to proof-nets, so as to not have to bother with rule commutations and equivalence classes anymore.
There, we prove that proof-nets associated to isomorphisms have very particular shapes, with local constraints on their axioms.
More precisely, we prove that the three following configurations are forbidden: an axiom between atoms of a same formula, an atom with no axiom on it, and an atom with axioms to several distinct atoms.

\begin{adjustbox}{}
\begin{tikzpicture}
\begin{myscope}
	\coordinate (Tlc) at (0,0);
	\coordinate (Tll) at (-1.25,2.5);
	\coordinate (Tlr) at (1.25,2.5);
	\path (Tlc) edge[-] (Tll);
	\path (Tlc) edge[-] (Tlr);
	\path (Tll) edge[-] (Tlr);

	\node (X) at (-.5,2) {$X$};
	\node (X^) at (.15,1.3) {$X\orth$};
\end{myscope}
\begin{myscopec}{blue}
	\path (X) -- ++(0,.7) -| (X^);
\end{myscopec}
\end{tikzpicture}
\quad\vrule\quad
\begin{tikzpicture}
\begin{myscope}
	\coordinate (Tlc) at (0,0);
	\coordinate (Tll) at (-1.25,2.5);
	\coordinate (Tlr) at (1.25,2.5);
	\path (Tlc) edge[-] (Tll);
	\path (Tlc) edge[-] (Tlr);
	\path (Tll) edge[-] (Tlr);

	\node (X) at (-.5,2) {$X\orth$};
\end{myscope}
\begin{myscopec}{red}
	\path (X) edge[-] (-.5,2.75);
	\path (-.6,2.75) edge[-] (-.4,2.55);
	\path (-.4,2.75) edge[-] (-.6,2.55);
\end{myscopec}
\end{tikzpicture}
\quad\vrule\quad
\begin{tikzpicture}
\begin{myscope}
	\coordinate (Tlc) at (0,0);
	\coordinate (Tll) at (-1.25,2.5);
	\coordinate (Tlr) at (1.25,2.5);
	\path (Tlc) edge[-] (Tll);
	\path (Tlc) edge[-] (Tlr);
	\path (Tll) edge[-] (Tlr);
	
	\coordinate (Trc) at (3,0);
	\coordinate (Trl) at (1.75,2.5);
	\coordinate (Trr) at (4.25,2.5);
	\path (Trc) edge[-] (Trl);
	\path (Trc) edge[-] (Trr);
	\path (Trl) edge[-] (Trr);

	\node (X1) at (-.5,2) {$X$};
	\node (X2) at (0,1) {$X$};
	\node (X^) at (3,2) {$X\orth$};
\end{myscope}
\begin{myscopec}{blue}
	\path (X^) -- ++(0,.7) -| (X1);
\end{myscopec}
\begin{myscopec}{red}
	\path (X^) -- ++(0,-.5) -| (X2);
\end{myscopec}
\end{tikzpicture}
\end{adjustbox}
This means that proof-nets associated to an isomorphism $A \iso B$ have the following general shape, with each atom of $A$ linked to a unique atom of $B$ by means of one or several axioms:

\begin{adjustbox}{}
\begin{tikzpicture}
\begin{myscope}
	\coordinate (Tlc) at (0,1);
	\coordinate (Tll) at (-2.5,4);
	\coordinate (Tlr) at (2.5,4);
	\path (Tlc) edge[-] (Tll);
	\path (Tlc) edge[-] (Tlr);
	\path (Tll) edge[-] (Tlr);

	\coordinate (Trc) at (6.5,1);
	\coordinate (Trl) at (4,4);
	\coordinate (Trr) at (9,4);
	\path (Trc) edge[-] (Trl);
	\path (Trc) edge[-] (Trr);
	\path (Trl) edge[-] (Trr);

	\node (X) at (1,3.4) {$X$};
	\node (X^) at (7,3) {$X\orth$};
	
	\node (X0) at (0,2.2) {$X\orth$};
	\node (X0^) at (5,3.5) {$X$};
	
	\node (Y) at (-.65,3.25) {$Y$};
	\node (Y^) at (6.25,2.25) {$Y\orth$};
	
	\node (X1) at (-1.5,3.5) {$X$};
	\node (X1^) at (8,3.5) {$X\orth$};
\end{myscope}
\begin{myscopec}{blue}
	\path (X) -- ++(0,.8) -| (X^);
	\path (X1) -- ++(0,1.1) -| (X1^);
	\path (X0) -- ++(0,.5) -| (X0^);
\end{myscopec}
\begin{myscopec}{red}
	\path (X) -- ++(0,1) -| (X^);
\end{myscopec}
\begin{myscopec}{olive}
	\path (X0) -- ++(0,.7) -| (X0^);
	\path (Y) -- ++(0,-1.8) -| (Y^);
\end{myscopec}
\end{tikzpicture}
\end{adjustbox}
This is the main challenge, because having such a shape is again only true for isomorphisms between distributed formulas, so we have to use this global property on distributed formulas in order to deduce a local property on the shape of axioms.
We employ here the correctness criterion of Hughes \& Van~Glabbeek~\cite{mallpnlong} in a direct way.
The idea is to prove that a forbidden configuration, typically the third one, implies having in one formula a $\with$ being an ancestor of a $\parr$.
The distributivity hypothesis then gives a $\tens$- or $\oplus$-vertex between these two, thanks to which a cycle contradicting the correctness criterion can be built.
This crucial part of the proof seems very hard to transpose in sequent calculus as it hinges on a geometric reasoning, and even if it were possible we would expect some heavy work on rule commutations, turning this already complicated proof into an utterly bewildering one.
\item
Knowing that proof-nets have very specific shapes, we recognize they correspond only to a rearrangement of formulas because axioms bind an atom of a formula with one of the isomorphic formula in a bijective fashion.
This is enough to finally conclude that the only isomorphisms are associativity and commutativity ones.
\end{enumerate}

\paragraph{Outline}
Once the necessary definitions given (\autoref{sec:def}), our proof of the completeness of the equational theory of \autoref{tab:eqisos} goes in two main steps.
First, we work in the sequent calculus to simplify the problem and reduce it to the unit-free fragment, by lack of a good-enough notion of proof-nets for MALL including units.
This goes through a characterization of the equality of proofs up to axiom-expansion and cut-elimination by means of rule commutations, using a result from~\cite{calculusformall} (Sections~\ref{subsec:no_eta} and~\ref{subsec:gen_eqc}).
This allows us to analyze the behavior of units inside isomorphisms so as to conclude that they can be replaced with fresh atoms, once formulas are simplified appropriately (\autoref{sec:red_unit-free}, and points $(1)$ and $(2)$ of the sketch above).
Secondly, being in a setting admitting proof-nets, we adapt the proof of Balat \& Di~Cosmo~\cite{mllisos} to the framework of Hughes \& Van~Glabbeek's proof-nets~\cite{mallpnlong}.
We transpose the definition of isomorphisms to this new syntax (Sections~\ref{subsec:mall-pn} and~\ref{sec:red_pn}), then prove completeness (\autoref{sec:complet_unit-free}, and points $(3)$ and $(4)$ of the sketch above).
This last step is the core of the work and requires a precise analysis of the structure of proof-nets because of the richer structure induced by the presence of the additive connectives.
The situation is much more complex than in the multiplicative setting since for example sub-formulas can be duplicated through distributivity equations, breaking a linearity property crucial in~\cite{mllisos}.

Finally, seeing MALL as a category, we extend our result to conclude that \autoref{tab:eqisos} (or more precisely its adaptation to the language of categories, \autoref{tab:eqisos_cat} on \autopageref{tab:eqisos_cat}) provides the equational theory of isomorphisms valid in all $\star$-autonomous categories with finite products (\autoref{sec:cat}).
We solve the situation of symmetric monoidal closed categories with finite products as well.

We are quite exhaustive on notions in sequent calculus.
An informed reader could skip these basic definitions and results, and start reading directly from \autoref{sec:red_unit-free} on \autopageref{sec:red_unit-free}.

%==================================================

\section{Definitions}
\label{sec:def}

In this section are first defined formulas and proofs of multiplicative-additive linear logic (\autoref{subsec:mall}).
Then are introduced the standard operations on these sequent calculus proofs: axiom-expansion, cut-elimination and rule commutations (\autoref{subsec:transfo_seq}).
This finally leads to the definition of type isomorphisms for this logic (\autoref{subsec:iso}).

\subsection{Multiplicative-Additive Linear Logic}
\label{subsec:mall}

The multiplicative-additive fragment of linear logic~\cite{ll}, denoted by MALL, has formulas given by the following grammar, where $X$ belongs to a given enumerable set of \emph{atoms}:
\[A,B \coloncoloneqq X \mid X\orth \mid A\tens B \mid A\parr B \mid 1 \mid \bot \mid A \with B \mid A \oplus B \mid \top \mid 0\]

Orthogonality $(\_)\orth$ expands into an involution on arbitrary formulas through ${X\biorth=X}$, ${1\orth=\bot}$, ${\bot\orth=1}$, ${\top\orth=0}$, ${0\orth=\top}$ and De Morgan's laws ${(A\tens B)\orth=B\orth\parr A\orth}$, $(A\parr B)\orth=B\orth\tens A\orth$, ${(A\with B)\orth=B\orth\oplus A\orth}$, ${(A\oplus B)\orth=B\orth\with A\orth}$.
We choose the \emph{non-commutative} De Morgan's laws for duality because this will often result in planar graphs on our illustrations, with axiom links not crossing each others (\eg\ in the identity proof-net, see \autoref{fig:id_bi_u} on \autopageref{fig:id_bi_u}).
These non-commutative De Morgan's laws have been used for instance in the context of cyclic linear logic where this leads to planar proof-nets~\cite{cccypn}.

The \emph{size} $s(A)$ of a formula $A$ is defined as its number of connectives $\tens$, $\parr$, $\with$ and $\oplus$, plus its number of units $1$, $\bot$, $\top$ and $0$ and of atoms and negated atoms (counting occurrences).

Sequents are lists of formulas of the form $\fCenter A_1, \ldots, A_n$.
Sequent calculus rules are, with $A$ and $B$ arbitrary formulas, $\Gamma$ and $\Delta$ contexts (\ie\ lists of formulas) and $\sigma$ a permutation:
\begin{center}
\bottomAlignProof
\Rax{$A$}
\DisplayProof
\hskip 2em
\bottomAlignProof
\AxiomC{$\fCenter \Gamma$}
\Rex{$\sigma(\Gamma)$}
\DisplayProof
\hskip 2em
\bottomAlignProof
\AxiomC{$\fCenter A,\Gamma$}
\AxiomC{$\fCenter A\orth,\Delta$}
\Rcut{$\Gamma,\Delta$}
\DisplayProof
\vskip .5em
\bottomAlignProof
\AxiomC{$\fCenter A,\Gamma$}
\AxiomC{$\fCenter B,\Delta$}
\Rtens{$A\tens B,\Gamma,\Delta$}
\DisplayProof
\hskip 2em
\bottomAlignProof
\AxiomC{$\fCenter A,B,\Gamma$}
\Rparr{$A\parr B,\Gamma$}
\DisplayProof
\hskip 2em
\bottomAlignProof
\Rone
\DisplayProof
\hskip 2em
\bottomAlignProof
\AxiomC{$\fCenter \Gamma$}
\Rbot{$\bot,\Gamma$}
\DisplayProof
\vskip .5em
\bottomAlignProof
\AxiomC{$\fCenter A,\Gamma$}
\AxiomC{$\fCenter B,\Gamma$}
\Rwith{$A\with B,\Gamma$}
\DisplayProof
\hskip 2em
\bottomAlignProof
\AxiomC{$\fCenter A,\Gamma$}
\Rplusone{$A\oplus B,\Gamma$}
\DisplayProof
\hskip 2em
\bottomAlignProof
\AxiomC{$\fCenter B,\Gamma$}
\Rplustwo{$A\oplus B,\Gamma$}
\DisplayProof
\hskip 2em
\bottomAlignProof
\Rtop{$\top,\Gamma$}
\DisplayProof
\end{center}
In practice, we consider exchange rules ($ex$) as incorporated in the conclusion of the rule above, thus dealing with rules like:
\AxiomC{$\fCenter A,B,\Gamma, \Delta$}
\Rparr{$\Gamma, A\parr B,\Delta$}
\DisplayProof
.
In this spirit, by
\AxiomC{$\fCenter \Gamma, A,B,\Delta$}
\Rparr{$\Gamma, A\parr B,\Delta$}
\DisplayProof
we mean that the appropriate permutation is also incorporated in the rule above.
Equivalently, we only consider proofs with exactly one exchange rule below each non-exchange rule.

The main difference with the multiplicative fragment of linear logic (MLL) is the $\with$-rule, which introduces some sharing of the context $\Gamma$.
From this comes the notion of a \emph{slice}~\cite{ll,pn} which is a partial proof missing some additive component.

\begin{defi}[Slice]
\label{def:slice_pt}
For $\pi$ a proof, consider the (non-correct) proof tree obtained by deleting one of the two sub-trees of each $\with$-rule of $\pi$ (thus, in the new proof tree, $\with$-rules are unary):
\begin{center}
\bottomAlignProof
\AxiomC{$\fCenter A,\Gamma$}
\Rwiths{1}{$A\with B,\Gamma$}
\DisplayProof
\hskip 6em
\bottomAlignProof
\AxiomC{$\fCenter B,\Gamma$}
\Rwiths{2}{$A\with B,\Gamma$}
\DisplayProof
\end{center}
The remaining rules form a \emph{slice} of $\pi$.
We denote by $\Slices(\pi)$ the set of slices of $\pi$.
\end{defi}

Slices satisfy a linearity property (validated by proofs of MLL as well): any connective in the conclusion is introduced by at most one rule in a slice.

By \emph{unit-free} MALL, we mean the restriction of MALL to formulas not involving the units $1$, $\bot$, $\top$ and $0$, and as such without the $1$, $\bot$ and $\top$-rules.
When speaking of a \emph{positive} formula, we mean a formula with main connective $\tens$ or $\oplus$, a unit $1$ or $0$, or an atom $X$. 
A \emph{negative} formula is one with main connective $\parr$ or $\with$, a unit $\bot$ or $\top$, or a negated atom $X\orth$.

\subsection{Transformations of proofs}
\label{subsec:transfo_seq}

There are two well-known rewriting relations on proofs, \emph{axiom-expansion} and \emph{cut-elimination}.
These transformations correspond respectively to $\eta$-expansion and $\beta$-reduction in the $\lambda$-calculus, hence the notations in the following definitions.

\begin{defi}
\label{def:eta_pt}
We call \emph{axiom-expansion} the rewriting system $\etato$ described in \autoref{tab:ax_exp_pt}.
\end{defi}

\begin{defi}
\label{def:beta_pt}
We call \emph{cut-elimination} the rewriting system $\betato$ described in Tables~\ref{tab:cut_elim_pt_key} and~\ref{tab:cut_elim_pt} (up to permuting the two branches of any $cut$-rule).
\end{defi}

\begin{rem}
\label{rem:tens_parr_key_cut}
Another possible $\parr-\tens$ key case for cut-elimination would be the following:\\
{
\Rsub{$\pi_1$}{$A,\Gamma$}
\Rsub{$\pi_2$}{$B,\Delta$}
\Rtens{$A\tens B,\Gamma,\Delta$}
\Rsub{$\pi_3$}{$B\orth,A\orth,\Sigma$}
\Rparr{$B\orth\parr A\orth,\Sigma$}
\Rcut{$\Gamma,\Delta,\Sigma$}
\DisplayProof
$\betato$
\Rsub{$\pi_2$}{$B,\Delta$}
\Rsub{$\pi_1$}{$A,\Gamma$}
\Rsub{$\pi_3$}{$B\orth,A\orth,\Sigma$}
\Rcut{$B\orth,\Gamma,\Sigma$}
\Rcut{$\Gamma,\Delta,\Sigma$}
\DisplayProof
}\\
This case can be simulated with the given $\parr-\tens$ key case and a $cut-cut$ commutative case.
\end{rem}

Another relation on proofs is \emph{rule commutation}, which is closely related to cut-elimination as proved in~\cite{calculusformall}.
These commutations associate proofs which differ only by the order in which their rules are applied.

\begin{defi}
\label{def:rule_comm}
We call \emph{rule commutation} the symmetric relation $\eqcone$ described in Tables~\ref{tab:rule_comm} and~\ref{tab:rule_comm_u}.
This corresponds to rule commutations in \emph{cut-free} MALL, \ie\ with no $cut$-rules above the commutation (there may be $cut$-rules below).
In particular, in a $\top-\tens$ commutation we assume the created or erased sub-proof to be cut-free, and in a $\with-\tens$ commutation the duplicated or superimposed sub-proof to be cut-free.
\end{defi}

The restriction of rule commutations in a cut-free setting is important.
First, this choice is more appropriate for our purposes as it has fewer cases and corresponds to the equivalence relation between normal forms, namely cut-free proofs.
Second, it is crucial for the $\with-\tens$ commutation: while there is no such restriction in~\cite{calculusformall}, their proof contains a mistake that can be patched up thanks to this restriction (more details are given in \autoref{subsec:gen_eqc} and a corrected proof is written in \autoref{subsec:proofs_add_1}).

\begin{rem}
\label{rem:rule_comm_cut}
A more general theory of rule commutations exists, without the restriction on cut-free proofs and with commutations involving the $cut$-rule, the latter being the symmetric closure of the commutative cases of cut-elimination in \autoref{tab:cut_elim_pt}.
See~\cite{mallpncom} for these commutations in unit-free MALL (and also with the $mix$-rule).
\end{rem}

\begin{rem}
\label{rem:comm_lacking}
Looking at the commutative cut-elimination cases, as well as rule commutations, there is no commutation with $ax$- and $1$-rules because these rules have no context.
For there is no rule for $0$, there is no commutation with $0$ nor a $\top-0$ key cut-elimination case.
\end{rem}

We denote the reflexive transitive closure of $\etato$ (\resp\ $\betato$, $\eqcone$) by $\etatostar$ (\resp\ $\betatostar$, $\eqc$).
By $\overset{\beta^n}{\longrightarrow}$ we mean a sequence of $n$ $\betato$ steps, and similarly for $\etato$.
Because of the analogy with the $\lambda$-calculus and since there will be no ambiguity, we use the notation $\eqbn$ for equality of proofs up to cut-elimination ($\beta$) and axiom-expansion ($\eta$).
Similarly, $\eqb$ is equality up to cut-elimination only.

\begin{tableproof}[p]
\centering
\begin{tabular}{@{}c@{\qquad}rcl@{}}
\toprule
$\parr-\tens$ &
\Raxocc{$A\tens B,B\orth\parr A\orth$}
\DisplayProof
& $\etato$ &
\Rax{$A$}
\Rax{$B$}
\Rtens{$A\tens B,B\orth,A\orth$}
\Rparr{$A\tens B,B\orth\parr A\orth$}
\DisplayProof
\\\midrule
$\with-\oplus$ &
\Raxocc{$A\oplus B,B\orth\with A\orth$}
\DisplayProof
& $\etato$ &
\Raxocc{$B,B\orth$}
\Rplustwo{$A\oplus B,B\orth$}
\Raxocc{$A,A\orth$}
\Rplusone{$A\oplus B,A\orth$}
\Rwith{$A\oplus B,B\orth\with A\orth$}
\DisplayProof
\\\midrule
$\bot-1$ &
\Raxocc{$1,\bot$}
\DisplayProof
& $\etato$ &
\Rone
\Rbot{$1,\bot$}
\DisplayProof
\\\midrule
$\top-0$ &
\Raxocc{$0,\top$}
\DisplayProof
& $\etato$ &
\Rtop{$0,\top$}
\DisplayProof
\\\bottomrule
\end{tabular}
\caption{Axiom-expansion in sequent calculus (up to a permutation of the conclusion)}
\label{tab:ax_exp_pt}
\end{tableproof}

\begin{tableproof}[p]
\centering
\resizebox{\textwidth}{!}{
\begin{tabular}{@{}c@{\qquad}rcl@{}}
\toprule
$ax$ &
\Rax{$A$}
\Rsub{$\pi$}{$A,\Gamma$}
\Rcut{$A,\Gamma$}
\DisplayProof
& $\betato$ &
\Rsub{$\pi$}{$A,\Gamma$}
\DisplayProof
\\\midrule
$\parr-\tens$ &
\Rsub{$\pi_1$}{$A,\Gamma$}
\Rsub{$\pi_2$}{$B,\Delta$}
\Rtens{$A\tens B,\Gamma,\Delta$}
\Rsub{$\pi_3$}{$B\orth,A\orth,\Sigma$}
\Rparr{$B\orth\parr A\orth,\Sigma$}
\Rcut{$\Gamma,\Delta,\Sigma$}
\DisplayProof
& $\betato$ &
\Rsub{$\pi_1$}{$A,\Gamma$}
\Rsub{$\pi_2$}{$B,\Delta$}
\Rsub{$\pi_3$}{$B\orth,A\orth,\Sigma$}
\Rcut{$A\orth,\Delta,\Sigma$}
\Rcut{$\Gamma,\Delta,\Sigma$}
\DisplayProof
\\\midrule
$\with-\oplus_1$ &
\Rsub{$\pi_1$}{$A_1,\Gamma$}
\Rsub{$\pi_2$}{$A_2,\Gamma$}
\Rwith{$A_1\with A_2,\Gamma$}
\Rsub{$\pi_3$}{$A_2\orth,\Delta$}
\Rplusone{$A_2\orth\oplus A_1\orth,\Delta$}
\Rcut{$\Gamma,\Delta$}
\DisplayProof
& $\betato$ &
\Rsub{$\pi_2$}{$A_2,\Gamma$}
\Rsub{$\pi_3$}{$A_2\orth,\Delta$}
\Rcut{$\Gamma,\Delta$}
\DisplayProof
\\\midrule
$\with-\oplus_2$ &
\Rsub{$\pi_1$}{$A_1,\Gamma$}
\Rsub{$\pi_2$}{$A_2,\Gamma$}
\Rwith{$A_1\with A_2,\Gamma$}
\Rsub{$\pi_3$}{$A_1\orth,\Delta$}
\Rplustwo{$A_2\orth\oplus A_1\orth,\Delta$}
\Rcut{$\Gamma,\Delta$}
\DisplayProof
& $\betato$ &
\Rsub{$\pi_1$}{$A_1,\Gamma$}
\Rsub{$\pi_3$}{$A_1\orth,\Delta$}
\Rcut{$\Gamma,\Delta$}
\DisplayProof
\\\midrule
$\bot-1$ &
\Rone
\Rsub{$\pi$}{$\Gamma$}
\Rbot{$\Gamma,\bot$}
\Rcut{$\Gamma$}
\DisplayProof
& $\betato$ &
\Rsub{$\pi$}{$\Gamma$}
\DisplayProof
\\\bottomrule
\end{tabular}
}
\caption{Cut-elimination in sequent calculus -- key cases (up to a permutation of the conclusion)}
\label{tab:cut_elim_pt_key}
\end{tableproof}

\begin{tableproof}[p]
\resizebox{\textwidth}{!}{
\begin{tabular}{@{}c@{\qquad}rcl@{}}
\toprule
$\parr-cut$ &
\Rsub{$\pi_1$}{$A,B,C,\Gamma$}
\Rparr{$A,B\parr C,\Gamma$}
\Rsub{$\pi_2$}{$A\orth,\Delta$}
\Rcut{$B\parr C,\Gamma,\Delta$}
\DisplayProof
& $\betato$ &
\Rsub{$\pi_1$}{$A,B,C,\Gamma$}
\Rsub{$\pi_2$}{$A\orth,\Delta$}
\Rcut{$B,C,\Gamma,\Delta$}
\Rparr{$B\parr C,\Gamma,\Delta$}
\DisplayProof
\\\midrule
$\tens-cut-1$ &
\Rsub{$\pi_1$}{$A,B,\Gamma$}
\Rsub{$\pi_2$}{$C,\Delta$}
\Rtens{$A,B\tens C,\Gamma,\Delta$}
\Rsub{$\pi_3$}{$A\orth,\Sigma$}
\Rcut{$B\tens C,\Gamma,\Delta,\Sigma$}
\DisplayProof
& $\betato$ &
\Rsub{$\pi_1$}{$A,B,\Gamma$}
\Rsub{$\pi_3$}{$A\orth,\Sigma$}
\Rcut{$B,\Gamma,\Sigma$}
\Rsub{$\pi_2$}{$C,\Delta$}
\Rtens{$B\tens C,\Gamma,\Delta,\Sigma$}
\DisplayProof
\\\midrule
$\tens-cut-2$ &
\Rsub{$\pi_1$}{$B,\Gamma$}
\Rsub{$\pi_2$}{$A,C,\Delta$}
\Rtens{$A,B\tens C,\Gamma,\Delta$}
\Rsub{$\pi_3$}{$A\orth,\Sigma$}
\Rcut{$B\tens C,\Gamma,\Delta,\Sigma$}
\DisplayProof
& $\betato$ &
\Rsub{$\pi_1$}{$B,\Gamma$}
\Rsub{$\pi_2$}{$A,C,\Delta$}
\Rsub{$\pi_3$}{$A\orth,\Sigma$}
\Rcut{$C,\Delta,\Sigma$}
\Rtens{$B\tens C,\Gamma,\Delta,\Sigma$}
\DisplayProof
\\\midrule
$\with-cut$ &
\Rsub{$\pi_1$}{$A,B,\Gamma$}
\Rsub{$\pi_2$}{$A,C,\Gamma$}
\Rwith{$A,B\with C,\Gamma$}
\Rsub{$\pi_3$}{$A\orth,\Delta$}
\Rcut{$B\with C,\Gamma,\Delta$}
\DisplayProof
& $\betato$ &
\Rsub{$\pi_1$}{$A,B,\Gamma$}
\Rsub{$\pi_3$}{$A\orth,\Delta$}
\Rcut{$B,\Gamma,\Delta$}
\Rsub{$\pi_2$}{$A,C,\Gamma$}
\Rsub{$\pi_3$}{$A\orth,\Delta$}
\Rcut{$C,\Gamma,\Delta$}
\Rwith{$B\with C,\Gamma,\Delta$}
\DisplayProof
\\\midrule
$\oplus_1-cut$ &
\Rsub{$\pi_1$}{$A,B_1,\Gamma$}
\Rplusone{$A,B_1\oplus B_2,\Gamma$}
\Rsub{$\pi_2$}{$A\orth,\Delta$}
\Rcut{$B_1\oplus B_2,\Gamma,\Delta$}
\DisplayProof
& $\betato$ &
\Rsub{$\pi_1$}{$A,B_1,\Gamma$}
\Rsub{$\pi_2$}{$A\orth,\Delta$}
\Rcut{$B_1,\Gamma,\Delta$}
\Rplusone{$B_1\oplus B_2,\Gamma,\Delta$}
\DisplayProof
\\\midrule
$\oplus_2-cut$ &
\Rsub{$\pi_1$}{$A,B_2,\Gamma$}
\Rplustwo{$A,B_1\oplus B_2,\Gamma$}
\Rsub{$\pi_2$}{$A\orth,\Delta$}
\Rcut{$B_1\oplus B_2,\Gamma,\Delta$}
\DisplayProof
& $\betato$ &
\Rsub{$\pi_1$}{$A,B_2,\Gamma$}
\Rsub{$\pi_2$}{$A\orth,\Delta$}
\Rcut{$B_2,\Gamma,\Delta$}
\Rplustwo{$B_1\oplus B_2,\Gamma,\Delta$}
\DisplayProof
\\\midrule
$\bot-cut$ &
\Rsub{$\pi_1$}{$A,\Gamma$}
\Rbot{$A,\bot,\Gamma$}
\Rsub{$\pi_2$}{$A\orth,\Delta$}
\Rcut{$\bot,\Gamma,\Delta$}
\DisplayProof
& $\betato$ &
\Rsub{$\pi_1$}{$A,\Gamma$}
\Rsub{$\pi_2$}{$A\orth,\Delta$}
\Rcut{$\Gamma,\Delta$}
\Rbot{$\bot,\Gamma,\Delta$}
\DisplayProof
\\\midrule
$\top-cut$ &
\Rtop{$A,\top,\Gamma$}
\Rsub{$\pi$}{$A\orth,\Delta$}
\Rcut{$\top,\Gamma,\Delta$}
\DisplayProof
& $\betato$ &
\Rtop{$\top,\Gamma,\Delta$}
\DisplayProof
\\\midrule
$cut-cut$ &
\Rsub{$\pi_1$}{$A,B,\Gamma$}
\Rsub{$\pi_2$}{$B\orth,\Delta$}
\Rcut{$A,\Gamma,\Delta$}
\Rsub{$\pi_3$}{$A\orth,\Sigma$}
\Rcut{$\Gamma,\Delta,\Sigma$}
\DisplayProof
& $\betato$ &
\Rsub{$\pi_1$}{$A,B,\Gamma$}
\Rsub{$\pi_3$}{$A\orth,\Sigma$}
\Rcut{$B,\Gamma,\Sigma$}
\Rsub{$\pi_2$}{$B\orth,\Delta$}
\Rcut{$\Gamma,\Delta,\Sigma$}
\DisplayProof
\\\bottomrule
\end{tabular}
}
\caption{Cut-elimination in sequent calculus -- commutative cases (up to a permutation of the conclusion)}
\label{tab:cut_elim_pt}
\end{tableproof}

\begin{tableproof}[p]
\centering
\resizebox{!}{.48\textheight}{
\begin{tabular}{@{}rcl@{}}
\toprule
\Rsub{$\pi$}{$A_1,A_2,B_1,B_2,\Gamma$}
\Rparr{$A_1\parr A_2,B_1,B_2,\Gamma$}
\Rparr{$A_1\parr A_2,B_1\parr B_2,\Gamma$}
\DisplayProof
& $\commdown{\parr}{\parr}$ &
\Rsub{$\pi$}{$A_1,A_2,B_1,B_2,\Gamma$}
\Rparr{$A_1,A_2,B_1\parr B_2,\Gamma$}
\Rparr{$A_1\parr A_2,B_1\parr B_2,\Gamma$}
\DisplayProof
\\\midrule
\Rsub{$\pi_1$}{$A_1,\Gamma$}
\Rsub{$\pi_2$}{$A_2,B_1,\Delta$}
\Rsub{$\pi_3$}{$B_2,\Sigma$}
\Rtens{$A_2,B_1\tens B_2,\Delta,\Sigma$}
\Rtens{$A_1\tens A_2,B_1\tens B_2,\Gamma,\Delta,\Sigma$}
\DisplayProof
& $\commdown{\tens}{\tens}$ &
\Rsub{$\pi_1$}{$A_1,\Gamma$}
\Rsub{$\pi_2$}{$A_2,B_1,\Delta$}
\Rtens{$A_1\tens A_2,B_1, \Gamma,\Delta$}
\Rsub{$\pi_3$}{$B_2,\Sigma$}
\Rtens{$A_1\tens A_2,B_1\tens B_2,\Gamma,\Delta,\Sigma$}
\DisplayProof
\\\midrule
\Rsub{$\pi_1$}{$A_1,\Gamma$}
\Rsub{$\pi_2$}{$B_1,\Delta$}
\Rsub{$\pi_3$}{$A_2,B_2,\Sigma$}
\Rtens{$A_2,B_1\tens B_2,\Delta,\Sigma$}
\Rtens{$A_1\tens A_2,B_1\tens B_2,\Gamma,\Delta,\Sigma$}
\DisplayProof
& $\commdown{\tens}{\tens}$ &
\Rsub{$\pi_2$}{$B_1,\Delta$}
\Rsub{$\pi_1$}{$A_1,\Gamma$}
\Rsub{$\pi_3$}{$A_2,B_2,\Sigma$}
\Rtens{$A_1\tens A_2,B_2, \Gamma,\Sigma$}
\Rtens{$A_1\tens A_2,B_1\tens B_2,\Gamma,\Delta,\Sigma$}
\DisplayProof
\\\midrule
\Rsub{$\pi_1$}{$A_1,B_1,\Gamma$}
\Rsub{$\pi_2$}{$B_2,\Delta$}
\Rtens{$A_1,B_1\tens B_2,\Gamma,\Delta$}
\Rsub{$\pi_3$}{$A_2,\Sigma$}
\Rtens{$A_1\tens A_2,B_1\tens B_2,\Gamma,\Delta,\Sigma$}
\DisplayProof
& $\commdown{\tens}{\tens}$ &
\Rsub{$\pi_1$}{$A_1,B_1,\Gamma$}
\Rsub{$\pi_3$}{$A_2,\Sigma$}
\Rtens{$A_1\tens A_2,B_1,\Gamma,\Sigma$}
\Rsub{$\pi_2$}{$B_2,\Delta$}
\Rtens{$A_1\tens A_2,B_1\tens B_2,\Gamma,\Delta,\Sigma$}
\DisplayProof
\\\midrule
\Rsub{$\pi_1$}{$B_1,\Gamma$}
\Rsub{$\pi_2$}{$A_1,B_2,\Delta$}
\Rtens{$A_1,B_1\tens B_2,\Gamma,\Delta$}
\Rsub{$\pi_3$}{$A_2,\Sigma$}
\Rtens{$A_1\tens A_2,B_1\tens B_2,\Gamma,\Delta,\Sigma$}
\DisplayProof
& $\commdown{\tens}{\tens}$ &
\Rsub{$\pi_1$}{$B_1,\Gamma$}
\Rsub{$\pi_2$}{$A_1,B_2,\Delta$}
\Rsub{$\pi_3$}{$A_2,\Sigma$}
\Rtens{$A_1\tens A_2,B_2,\Delta,\Sigma$}
\Rtens{$A_1\tens A_2,B_1\tens B_2,\Gamma,\Delta,\Sigma$}
\DisplayProof
\\\midrule
\Rsub{$\pi_1$}{$A_1,B_1,\Gamma$}
\Rsub{$\pi_2$}{$A_2,B_1,\Gamma$}
\Rwith{$A_1\with A_2,B_1,\Gamma$}
\Rsub{$\pi_3$}{$A_1,B_2,\Gamma$}
\Rsub{$\pi_4$}{$A_2,B_2,\Gamma$}
\Rwith{$A_1\with A_2,B_2,\Gamma$}
\Rwith{$A_1\with A_2,B_1\with B_2,\Gamma$}
\DisplayProof
& $\commdown{\with}{\with}$ &
\Rsub{$\pi_1$}{$A_1,B_1,\Gamma$}
\Rsub{$\pi_3$}{$A_1,B_2,\Gamma$}
\Rwith{$A_1,B_1\with B_2,\Gamma$}
\Rsub{$\pi_2$}{$A_2,B_1,\Gamma$}
\Rsub{$\pi_4$}{$A_2,B_2,\Gamma$}
\Rwith{$A_2, B_1\with B_2,\Gamma$}
\Rwith{$A_1\with A_2,B_1\with B_2,\Gamma$}
\DisplayProof
\\\midrule
\Rsub{$\pi$}{$A_i,B_j,\Gamma$}
\Rplus{i}{$A_1\oplus A_2,B_j,\Gamma$}
\Rplus{j}{$A_1\oplus A_2,B_1\oplus B_2,\Gamma$}
\DisplayProof
& $\commdown{\oplus_j}{\oplus_i}$ &
\Rsub{$\pi$}{$A_i,B_j,\Gamma$}
\Rplus{j}{$A_i,B_1\oplus B_2,\Gamma$}
\Rplus{i}{$A_1\oplus A_2,B_1\oplus B_2,\Gamma$}
\DisplayProof
\\\midrule
\Rsub{$\pi_1$}{$A_1,A_2,B_1,\Gamma$}
\Rparr{$A_1\parr A_2,B_1,\Gamma$}
\Rsub{$\pi_2$}{$B_2,\Delta$}
\Rtens{$A_1\parr A_2,B_1\tens B_2,\Gamma,\Delta$}
\DisplayProof
& $\commbi{\parr}{\tens}$ &
\Rsub{$\pi_1$}{$A_1,A_2,B_1,\Gamma$}
\Rsub{$\pi_2$}{$B_2,\Delta$}
\Rtens{$A_1,A_2,B_1\tens B_2,\Gamma,\Delta$}
\Rparr{$A_1\parr A_2,B_1\tens B_2,\Gamma,\Delta$}
\DisplayProof
\\\midrule
\Rsub{$\pi_1$}{$B_1,\Gamma$}
\Rsub{$\pi_2$}{$A_1,A_2,B_2,\Delta$}
\Rparr{$A_1\parr A_2,B_2,\Delta$}
\Rtens{$A_1\parr A_2,B_1\tens B_2,\Gamma,\Delta$}
\DisplayProof
& $\commbi{\parr}{\tens}$ &
\Rsub{$\pi_1$}{$B_1,\Gamma$}
\Rsub{$\pi_2$}{$A_1,A_2,B_2,\Delta$}
\Rtens{$A_1,A_2,B_1\tens B_2,\Gamma,\Delta$}
\Rparr{$A_1\parr A_2,B_1\tens B_2,\Gamma,\Delta$}
\DisplayProof
\\\midrule
\Rsub{$\pi_1$}{$A_1,A_2,B_1,\Gamma$}
\Rparr{$A_1\parr A_2,B_1,\Gamma$}
\Rsub{$\pi_2$}{$A_1,A_2,B_2,\Gamma$}
\Rparr{$A_1\parr A_2,B_2,\Gamma$}
\Rwith{$A_1\parr A_2,B_1\with B_2,\Gamma$}
\DisplayProof
& $\commbi{\parr}{\with}$&
\Rsub{$\pi_1$}{$A_1,A_2,B_1,\Gamma$}
\Rsub{$\pi_2$}{$A_1,A_2,B_2,\Gamma$}
\Rwith{$A_1,A_2,B_1\with B_2,\Gamma$}
\Rparr{$A_1\parr A_2,B_1\with B_2,\Gamma$}
\DisplayProof
\\\midrule
\Rsub{$\pi_1$}{$A_1,A_2,B_i,\Gamma$}
\Rparr{$A_1\parr A_2,B_i,\Gamma$}
\Rplusi{$A_1\parr A_2,B_1\oplus B_2,\Gamma$}
\DisplayProof
& $\commbi{\parr}{\oplus_i}$ &
\Rsub{$\pi_1$}{$A_1,A_2,B_i,\Gamma$}
\Rplusi{$A_1,A_2,B_1\oplus B_2,\Gamma$}
\Rparr{$A_1\parr A_2,B_1\oplus B_2,\Gamma$}
\DisplayProof
\\\midrule
\Rsub{$\pi_1$}{$A_1,\Gamma$}
\Rsub{$\pi_2$}{$A_2,B_1,\Delta$}
\Rtens{$A_1\tens A_2,B_1,\Gamma,\Delta$}
\Rsub{$\pi_1$}{$A_1,\Gamma$}
\Rsub{$\pi_3$}{$A_2,B_2,\Delta$}
\Rtens{$A_1\tens A_2,B_2,\Gamma,\Delta$}
\Rwith{$A_1\tens A_2,B_1\with B_2,\Gamma,\Delta$}
\DisplayProof
& $\commbi{\tens}{\with}$ &
\Rsub{$\pi_1$}{$A_1,\Gamma$}
\Rsub{$\pi_2$}{$A_2,B_1,\Delta$}
\Rsub{$\pi_3$}{$A_2,B_2,\Delta$}
\Rwith{$A_2,B_1\with B_2,\Delta$}
\Rtens{$A_1\tens A_2,B_1\with B_2,\Gamma,\Delta$}
\DisplayProof
\\\midrule
\Rsub{$\pi_1$}{$A_1,B_1,\Gamma$}
\Rsub{$\pi_2$}{$A_2,\Delta$}
\Rtens{$A_1\tens A_2,B_1,\Gamma,\Delta$}
\Rsub{$\pi_3$}{$A_1,B_2,\Gamma$}
\Rsub{$\pi_2$}{$A_2,\Delta$}
\Rtens{$A_1\tens A_2,B_2,\Gamma,\Delta$}
\Rwith{$A_1\tens A_2,B_1\with B_2,\Gamma,\Delta$}
\DisplayProof
& $\commbi{\tens}{\with}$ &
\Rsub{$\pi_1$}{$A_1,B_1,\Gamma$}
\Rsub{$\pi_3$}{$A_1,B_2,\Gamma$}
\Rwith{$A_1,B_1\with B_2,\Gamma$}
\Rsub{$\pi_2$}{$A_2,\Delta$}
\Rtens{$A_1\tens A_2,B_1\with B_2,\Gamma,\Delta$}
\DisplayProof
\\\midrule
\Rsub{$\pi_1$}{$A_1,B_i,\Gamma$}
\Rsub{$\pi_2$}{$A_2,\Delta$}
\Rtens{$A_1\tens A_2,B_i,\Gamma,\Delta$}
\Rplusi{$A_1\tens A_2,B_1\oplus B_2,\Gamma,\Delta$}
\DisplayProof
& $\commbi{\tens}{\oplus_i}$ &
\Rsub{$\pi_1$}{$A_1,B_i,\Gamma$}
\Rplusi{$A_1,B_1\oplus B_2,\Gamma$}
\Rsub{$\pi_2$}{$A_2,\Delta$}
\Rtens{$A_1\tens A_2,B_1\oplus B_2,\Gamma,\Delta$}
\DisplayProof
\\\midrule
\Rsub{$\pi_1$}{$A_1,\Gamma$}
\Rsub{$\pi_2$}{$A_2,B_i,\Delta$}
\Rtens{$A_1\tens A_2,B_i,\Gamma,\Delta$}
\Rplusi{$A_1\tens A_2,B_1\oplus B_2,\Gamma,\Delta$}
\DisplayProof
& $\commbi{\tens}{\oplus_i}$ &
\Rsub{$\pi_1$}{$A_1,B_i,\Gamma$}
\Rsub{$\pi_2$}{$A_2,\Delta$}
\Rplusi{$A_2,B_1\oplus B_2,\Delta$}
\Rtens{$A_1\tens A_2,B_1\oplus B_2,\Gamma,\Delta$}
\DisplayProof
\\\midrule
\Rsub{$\pi_1$}{$A_1,B_i,\Gamma$}
\Rsub{$\pi_2$}{$A_2,B_i,\Gamma$}
\Rwith{$A_1\with A_2,B_i,\Gamma$}
\Rplusi{$A_1\with A_2,B_1\oplus B_2,\Gamma$}
\DisplayProof
& $\commbi{\with}{\oplus_i}$ &
\Rsub{$\pi_1$}{$A_1,B_i,\Gamma$}
\Rplusi{$A_1,B_1\oplus B_2,\Gamma$}
\Rsub{$\pi_2$}{$A_2,B_i,\Gamma$}
\Rplusi{$A_2,B_1\oplus B_2,\Gamma$}
\Rwith{$A_1\with A_2,B_1\oplus B_2,\Gamma$}
\DisplayProof
\\\bottomrule
\end{tabular}
}
\caption{Rule commutations not involving a unit rule (up to a permutation of the conclusion)}
\label{tab:rule_comm}
\end{tableproof}

\begin{tableproof}[p]
\centering
\resizebox{!}{.46\textheight}{
\begin{tabular}{@{}rcl@{}}
\toprule
\Rtop{$A_1\parr A_2,\top,\Gamma$}
\DisplayProof
& $\commbi{\parr}{\top}$ &
\Rtop{$A_1, A_2,\top,\Gamma$}
\Rparr{$A_1\parr A_2,\top,\Gamma$}
\DisplayProof
\\\midrule
\Rtop{$A_1\tens A_2,\top,\Gamma,\Delta$}
\DisplayProof
& $\commbi{\tens}{\top}$ &
\Rtop{$A_1,\top,\Gamma$}
\Rsub{$\pi$}{$A_2,\Delta$}
\Rtens{$A_1\tens A_2,\top,\Gamma,\Delta$}
\DisplayProof
\\\midrule
\Rtop{$A_1\tens A_2,\top,\Gamma,\Delta$}
\DisplayProof
& $\commbi{\tens}{\top}$ &
\Rsub{$\pi$}{$A_1,\Gamma$}
\Rtop{$A_2,\top,\Delta$}
\Rtens{$A_1\tens A_2,\top,\Gamma,\Delta$}
\DisplayProof
\\\midrule
\Rtop{$A_1\with A_2,\top,\Gamma$}
\DisplayProof
& $\commbi{\with}{\top}$ &
\Rtop{$A_1,\top,\Gamma$}
\Rtop{$A_2,\top,\Gamma$}
\Rwith{$A_1\with A_2,\top,\Gamma$}
\DisplayProof
\\\midrule
\Rtop{$A_1\oplus A_2,\top,\Gamma$}
\DisplayProof
& $\commbi{\oplus_i}{\top}$ &
\Rtop{$A_i,\top,\Gamma$}
\Rplusi{$A_1\oplus A_2,\top,\Gamma$}
\DisplayProof
\\\midrule
\AxiomC{}
\RightLabel{$\top_0$}
\UnaryInfC{$\fCenter$ $\top_0,\top_1,\Gamma$}
\DisplayProof
& $\commdown{\top}{\top}$ &
\AxiomC{}
\RightLabel{$\top_1$}
\UnaryInfC{$\fCenter$ $\top_0,\top_1,\Gamma$}
\DisplayProof
\\\midrule
\Rtop{$\top,\bot,\Gamma$}
\DisplayProof
& $\commbi{\bot}{\top}$ &
\Rtop{$\top,\Gamma$}
\Rbot{$\top,\bot,\Gamma$}
\DisplayProof
\\\midrule
\Rsub{$\pi_1$}{$A_1,A_2,\Gamma$}
\Rparr{$A_1\parr A_2,\Gamma$}
\Rbot{$A_1\parr A_2,\bot,\Gamma$}
\DisplayProof
& $\commbi{\parr}{\bot}$ &
\Rsub{$\pi_1$}{$A_1,A_2,\Gamma$}
\Rbot{$A_1, A_2,\bot,\Gamma$}
\Rparr{$A_1\parr A_2,\bot,\Gamma$}
\DisplayProof
\\\midrule
\Rsub{$\pi_1$}{$A_1,\Gamma$}
\Rsub{$\pi_2$}{$A_2,\Delta$}
\Rtens{$A_1\tens A_2,\Gamma,\Delta$}
\Rbot{$A_1\tens A_2,\bot,\Gamma,\Delta$}
\DisplayProof
& $\commbi{\tens}{\bot}$ &
\Rsub{$\pi_1$}{$A_1,\Gamma$}
\Rbot{$A_1,\bot,\Gamma$}
\Rsub{$\pi_2$}{$A_2,\Delta$}
\Rtens{$A_1\tens A_2,\bot,\Gamma,\Delta$}
\DisplayProof
\\\midrule
\Rsub{$\pi_1$}{$A_1,\Gamma$}
\Rsub{$\pi_2$}{$A_2,\Delta$}
\Rtens{$A_1\tens A_2,\Gamma,\Delta$}
\Rbot{$A_1\tens A_2,\bot,\Gamma,\Delta$}
\DisplayProof
& $\commbi{\tens}{\bot}$ &
\Rsub{$\pi_1$}{$A_1,\Gamma$}
\Rsub{$\pi_2$}{$A_2,\Delta$}
\Rbot{$A_2,\bot,\Delta$}
\Rtens{$A_1\tens A_2,\bot,\Gamma,\Delta$}
\DisplayProof
\\\midrule
\Rsub{$\pi_1$}{$A_1,\Gamma$}
\Rsub{$\pi_2$}{$A_2,\Gamma$}
\Rwith{$A_1\with A_2,\Gamma$}
\Rbot{$A_1\with A_2,\bot,\Gamma$}
\DisplayProof
& $\commbi{\with}{\bot}$ &
\Rsub{$\pi_1$}{$A_1,\Gamma$}
\Rbot{$A_1,\bot,\Gamma$}
\Rsub{$\pi_2$}{$A_2,\Gamma$}
\Rbot{$A_2,\bot,\Gamma$}
\Rwith{$A_1\with A_2,\bot,\Gamma$}
\DisplayProof
\\\midrule
\Rsub{$\pi$}{$A_i,\Gamma$}
\Rplusi{$A_1\oplus A_2,\Gamma$}
\Rbot{$A_1\oplus A_2,\bot,\Gamma$}
\DisplayProof
& $\commbi{\oplus_i}{\bot}$ &
\Rsub{$\pi$}{$A_i,\Gamma$}
\Rbot{$A_i,\bot,\Gamma$}
\Rplusi{$A_1\oplus A_2,\bot,\Gamma$}
\DisplayProof
\\\midrule
\Rsub{$\pi$}{$\Gamma$}
\RightLabel{$\bot_0$}
\UnaryInfC{$\fCenter$ $\bot_0,\Gamma$}
\RightLabel{$\bot_1$}
\UnaryInfC{$\fCenter$ $\bot_0,\bot_1,\Gamma$}
\DisplayProof
& $\commdown{\bot}{\bot}$ &
\Rsub{$\pi$}{$\Gamma$}
\RightLabel{$\bot_1$}
\UnaryInfC{$\fCenter$ $\bot_1,\Gamma$}
\RightLabel{$\bot_0$}
\UnaryInfC{$\fCenter$ $\bot_0,\bot_1,\Gamma$}
\DisplayProof
\\\bottomrule
\end{tabular}
}
\vskip.5em
{\small In the $\comm{\top}{\top}$ and $\comm{\bot}{\bot}$ commutations, indices serve to identify occurrences of $\top$ and $\bot$, and the index of the rules to identify the distinguished occurrence associated with the rule.}
\caption{Rule commutations involving a unit rule (up to a permutation of the conclusion)}
\label{tab:rule_comm_u}
\end{tableproof}

\subsection{Linear isomorphisms}
\label{subsec:iso}

We denote by $\cutf{\pi}{\pi'}{B}$ the proof obtained by adding a cut on $B$ between proofs $\pi$ with conclusion $\vdash\Gamma, B$ and $\pi'$ with conclusion $\vdash B\orth,\Delta$, and by $\ax{A}$ the proof of $\fCenter A\orth, A$ containing just an $ax$-rule.

\begin{defi}[Isomorphism]
\label{def:iso}
Two formulas $A$ and $B$ are \emph{isomorphic}, denoted $A\iso B$, if there exist proofs $\pi$ of $\fCenter A\orth,B$ and $\pi'$ of $\fCenter B\orth,A$ such that $\cutf{\pi}{\pi'}{B} \eqbn \ax{A}$ and $\cutf{\pi'}{\pi}{A} \eqbn \ax{B}$:
\begin{gather*}
\cutf{\pi}{\pi'}{B} \;=\;
\Rsub{$\pi$}{$A\orth,B$}
\Rsub{$\pi'$}{$B\orth,A$}
\Rcut{$A\orth,A$}
\DisplayProof
\;\eqbn\;
\Rax{$A$}
\DisplayProof
\;=\; \ax{A}
\\
\text{and}\\
\cutf{\pi'}{\pi}{A} \;=\;
\Rsub{$\pi'$}{$B\orth,A$}
\Rsub{$\pi$}{$A\orth,B$}
\Rcut{$B\orth,B$}
\DisplayProof
\;\eqbn\;
\Rax{$B$}
\DisplayProof
\;=\; \ax{B}
\end{gather*}
\end{defi}

By $\isoproofs{A}{B}{\pi}{\pi'}$ we mean the \emph{cut-free} proofs $\pi$ and $\pi'$ define an isomorphism between formulas $A$ and $B$, that is $\pi$ is a proof of $\fCenter A\orth, B$ and $\pi'$ one of $\fCenter B\orth, A$ such that $\cutf{\pi}{\pi'}{B}\eqbn\ax{A}$ and $\cutf{\pi}{\pi'}{A}\eqbn\ax{B}$.
It is quite easy to link these two notations of isomorphisms.

\begin{lem}
\label{lem:iso_isoproofs}
Given two formulas $A$ and $B$, $A\iso B$ if and only if there exist proofs $\pi$ and $\pi'$ such that $\isoproofs{A}{B}{\pi}{\pi'}$.
\end{lem}
\begin{proof}
The converse way follows by definition of an isomorphism.
For the direct way, take proofs $\pi$ and $\pi'$ given by the definition of an isomorphism, not necessarily cut-free.
One  may eliminate all $cut$-rules inside, using weak normalization of cut-elimination $\betato$ (\autoref{cor:beta_wn}).
\end{proof}

\floathere

We aim to prove that two MALL formulas are isomorphic if and only if they are equal in the equational theory $\Eu$ defined as follows, with $\AC$ a smaller equational theory to which we will reduce the problem.

\begin{defi}[Equational theories $\Eu$ and $\AC$]
\label{def:E}
We denote by $\Eu$ the equational theory given on \autoref{tab:eqisos} on \autopageref{tab:eqisos}, while $\AC$ denotes the part with associativity and commutativity equations only.
\end{defi}

Given an equational theory $\mathcal{T}$, the notation $A\eqeq_\mathcal{T}B$ means that formulas $A$ and $B$ are equal in the theory $\mathcal{T}$.
As often, the soundness part is easy (but tedious) to prove.

\begin{thm}[Isomorphisms soundness, see~{\cite[Lemma~3]{classisos}}]
\label{th:iso_sound}
If $A\eqEu B$ then $A\iso B$.
\end{thm}
\begin{proof}
It suffices to give the proofs for each equation, then check their compositions can be reduced by cut-elimination to an axiom-expansion of an $ax$-rule.
For instance, looking at the commutativity of $\parr$, \ie\ $A\parr B \iso B \parr A$, we set $\pi$ and $\pi'$ the following proofs:
\begin{center}
$\pi=$
\Rax{$B$}
\Rax{$A$}
\Rtens{$B\orth\tens A\orth, B, A$}
\Rparr{$B\orth\tens A\orth, B\parr A$}
\DisplayProof
\quad and\quad $\pi'=$
\Rax{$A$}
\Rax{$B$}
\Rtens{$A\orth\tens B\orth, A, B$}
\Rparr{$A\orth\tens B\orth, A\parr B$}
\DisplayProof
\end{center}
One can check
\begin{center}
\begin{tikzpicture}
\node (1) at (0,0) {
\Rsub{$\pi$}{$B\orth\tens A\orth, B\parr A$}
\Rsub{$\pi'$}{$A\orth\tens B\orth, A\parr B$}
\Rcut{$B\orth\tens A\orth, A\parr B$}
\DisplayProof};
\node (2) at (7,-2) {
\Rax{$B$}
\Rax{$A$}
\Rtens{$B\orth\tens A\orth, A, B$}
\Rparr{$B\orth\tens A\orth, A\parr B$}
\DisplayProof};
\node (3) at (0,-3) {
\Raxocc{$B\orth\tens A\orth, A\parr B$}
\DisplayProof};

\path (1) edge[draw=none] node[draw=none, sloped, align=center]{$\betatostar$} (2);
\path (3) edge[draw=none] node[draw=none, sloped, align=center]{$\etato$} (2);
\end{tikzpicture}
\end{center}
and
\begin{center}
\begin{tikzpicture}
\node (1) at (0,0) {
\Rsub{$\pi'$}{$A\orth\tens B\orth, A\parr B$}
\Rsub{$\pi$}{$B\orth\tens A\orth, B\parr A$}
\Rcut{$A\orth\tens B\orth, B\parr A$}
\DisplayProof};
\node (2) at (7,-2) {
\Rax{$A$}
\Rax{$B$}
\Rtens{$A\orth\tens B\orth, B, A$}
\Rparr{$A\orth\tens B\orth, B\parr A$}
\DisplayProof};
\node (3) at (0,-3) {
\Raxocc{$A\orth\tens B\orth, B\parr A$}
\DisplayProof};

\path (1) edge[draw=none] node[draw=none, sloped, align=center]{$\betatostar$} (2);
\path (3) edge[draw=none] node[draw=none, sloped, align=center]{$\etato$} (2);
\end{tikzpicture} \qedhere
\end{center}
\end{proof}

All the difficulty lies in the proof of the other implication, completeness, on which the rest of this work focuses.

%==================================================

\section{Axiom-expansion}
\label{subsec:no_eta}

The goal of this section is to reduce the study of isomorphisms to proofs with \emph{atomic axioms} only, \ie\ whose $ax$-rules are on a formula $A=X$ or $A=X\orth$ (for $X$ an atom), no more considering axiom-expansion.
Said in another manner, we restrict our study to proofs in normal form for the $\eta$ relation, so as to have only $\eqb$ instead of $\eqbn$ when considering isomorphisms.

To begin with, axiom-expansion is convergent.

\begin{prop}
\label{prop:eta_canonic}
Axiom-expansion $\etato$ is strongly normalizing and confluent.
\end{prop}
\begin{proof}
Strong normalization follows from the fact that a $\etato$ step strictly decreases the sum of the sizes of the formulas on which an $ax$-rule is applied.

Confluence can be deduced from the diamond property.
Observe that two distinct steps $\tau\etafrom\pi\etato\phi$ always commute, \ie\ there exists $\mu$ such that $\tau\etato\mu\etafrom\phi$ (there is no critical pair).
Hence the diamond property of $\etato$, thus its confluence.
\end{proof}

\begin{rem}
\label{rem:eta_ax-exp}
As long as there is an $ax$-rule not on an atom, a $\etato$ axiom-expansion step can be applied.
Thus, atomic-axiom proofs correspond to proofs in normal form for $\etato$.
\end{rem}

Thanks to \autoref{prop:eta_canonic}, we denote by $\eta(\pi)$ the unique $\eta$-normal form of a proof $\pi$, \ie\ the proof obtained by expanding iteratively all $ax$-rules in $\pi$ (in any order thanks to confluence).
Thus, we can define $\id{A} = \eta(\ax{A})$, the axiom-expansion of the proof consisting of only one $ax$-rule.
The goal of this section is to prove the following (proof on \autopageref{proof:eta_nom_eqbn}).

\begin{prop}[Reduction to atomic-axiom proofs]
\label{prop:eta_nom_eqbn}
Let $\pi$ and $\tau$ be proofs such that $\pi\eqbn\tau$.
Then $\eta(\pi)\eqb\eta(\tau)$ with, in this sequence, only proofs in $\eta$-normal form.
\end{prop}

This allows us to then consider proofs with atomic axioms only, manipulated through composition by cut and cut-elimination, never to speak again of $\eta$ rewriting.

We will need some intermediate results to reach our goal.
We set $a(\pi\betatostar\tau)$ the multiset of the sizes of the formulas in the $ax$ key cases of these $\betato$ reductions.

\begin{lem}
\label{lem:betafrom_etato}
Let $\pi$, $\tau$ and $\phi$ be proofs such that $\tau\betafrom\pi\etato\phi$.
Then there exists $\mu$ such that $\tau\etatostar\mu\betafrom\phi$ or there exist $\mu_1$ and $\mu_2$ such that $\tau\etatostar\mu_2\betafromstar\mu_1\etafromstar\phi$.
Furthermore, $a(\phi\betato\mu)=a(\pi\betato\tau)$ in the first case and $a(\mu_1\betatostar\mu_2) < a(\pi\betato\tau)$ in the second one.
Diagrammatically:
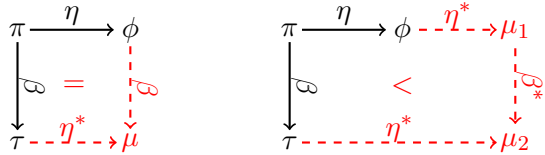
\begin{center}
\begin{tikzpicture}
\begin{myscope}
	\node[draw=none,minimum size=3mm] (pi) at (0,1.5) {$\pi$};
	\node[draw=none,minimum size=3mm] (phi) at (1.5,1.5) {$\phi$};
	\node[draw=none,minimum size=3mm] (varpi) at (0,0) {$\tau$};

	\path (pi) \edgelabel{$\eta$} (phi);
	\path (pi) \edgelabel{$\beta$} (varpi);
\end{myscope}
\begin{myscopec}{red}
	\node[draw=none,minimum size=3mm] (varphi) at (1.5,0) {$\mu$};
	
	\path (phi) edge[dashed] node[draw=none, sloped, align=center]{$\beta$ \\} (varphi);
	\path (varpi) edge[dashed] node[draw=none, sloped, align=center]{$\eta^*$ \\} (varphi);
	\node[draw=none] at (.75,.75) {$=$};
\end{myscopec}
\end{tikzpicture}
\hskip3em
\begin{tikzpicture}
\begin{myscope}
	\node[draw=none,minimum size=3mm] (pi) at (0,1.5) {$\pi$};
	\node[draw=none,minimum size=3mm] (phi) at (1.5,1.5) {$\phi$};
	\node[draw=none,minimum size=3mm] (varpi) at (0,0) {$\tau$};

	\path (pi) \edgelabel{$\eta$} (phi);
	\path (pi) \edgelabel{$\beta$} (varpi);
\end{myscope}
\begin{myscopec}{red}
	\node[draw=none,minimum size=3mm] (varphi1) at (3,1.5) {$\mu_1$};
	\node[draw=none,minimum size=3mm] (varphi2) at (3,0) {$\mu_2$};
	
	\path (phi) edge[dashed] node[draw=none, sloped, align=center]{$\eta^*$ \\} (varphi1);
	\path (varphi1) edge[dashed] node[draw=none, sloped, align=center]{$\beta^*$ \\} (varphi2);
	\path (varpi) edge[dashed] node[draw=none, sloped, align=center]{$\eta^*$ \\} (varphi2);
	\node[draw=none] at (1.5,.75) {$>$};
\end{myscopec}
\end{tikzpicture}
\end{center}
\end{lem}
\begin{proof}
Call $r$ the $ax$-rule that $\pi\etato\phi$ expands, and $A$ its formula.
If the cut-elimination step is not an $ax$ key case using $r$, then the two steps commute and there exists $\mu$ such that $\phi\betato\mu$ and $\tau\etato\mu$ (or $\tau\etato\cdot\etato\mu$ if $r$ belongs to a sub-proof duplicated by the $\betato$ step, or $\tau=\mu$ if it belongs to a sub-proof erased by the $\betato$ step).
In particular, $a(\phi\betato\mu)=a(\pi\betato\tau)$ for they use the same rules.

Otherwise, the cut-elimination step is an $ax$ key case on $r$, with a $cut$-rule we call $c$ and a sub-proof $\rho$ in the other branch of $c$ than the one leading to $r$; $A$ is the formula cut by $c$.
The reasoning we apply is depicted on \autoref{fig:proof_betafrom_etato}.
Starting from $\phi$, consider the rules introducing $A\orth$ in (all slices of) $\rho$.
If any of them are $ax$-rules, then these are necessarily on the formula $A$; expand those $ax$-rules, in both $\phi$ and $\tau$ (keeping the same names for proofs $\pi$, $\tau$ and $\rho$ by abuse).
Then, in $\phi$, commute the $cut$-rule $c$ with rules of $\rho$ until reaching the rules introducing $A\orth$ in all slices (which are rules of the main connective of $A\orth$ or $\top$-rules).
Applying the corresponding key cases or $\top-cut$ commutative case (first commuting with a rule of the expanded axiom $r$ if $A$ is a positive formula, and doing it after if $A$ is a negative formula), then the units or $ax$ key cases on strict sub-formulas of $A$ yield $\tau$.
During these $ax$ key cases, we cut on sub-formulas of $A$, so on formulas of a strictly smaller size.
Therefore $\phi\etatostar\cdot\betatostar\cdot\etafromstar\tau$, with $a(\cdot\betatostar\cdot) < a(\pi\betato\tau)$.
\end{proof}

\begin{figure}
\centering
\begin{tikzpicture}
	\node (pi) at (0,0) {$\pi=${\AxiomC{}\RightLabel{$r$}\UnaryInfC{$\fCenter A\orth, A$}\AxiomC{$\rho$}\RightLabel{$c$}\BinaryInfC{\vdots}\DisplayProof}};
	\node (phi) at (8,0) {$\phi=${\AxiomC{}\RightLabel{$ax$}\UnaryInfC{\emptyforproof}\AxiomC{}\RightLabel{$ax$}\UnaryInfC{\emptyforproof}\RightLabel{pos}\BinaryInfC{\emptyforproof}\RightLabel{neg}\UnaryInfC{$\fCenter A\orth, A$}\AxiomC{$\rho$}\RightLabel{$c$}\BinaryInfC{\vdots}\DisplayProof}};
	\node (tau) at (0,-3) {$\tau=${\AxiomC{$\rho$}\noLine\UnaryInfC{\vdots}\DisplayProof}};

	\node (phi') at (8,-6) {\AxiomC{}\RightLabel{$ax$}\UnaryInfC{\emptyforproof}\AxiomC{}\RightLabel{$ax$}\UnaryInfC{\emptyforproof}\RightLabel{pos}\BinaryInfC{\emptyforproof}\RightLabel{neg}\UnaryInfC{$\fCenter A\orth, A$}\AxiomC{\vdots}\RightLabel{$A\orth$}\UnaryInfC{$\gamma$}\RightLabel{$c$}\BinaryInfC{\vdots}\DisplayProof};
	\node (tau') at (0,-9) {\AxiomC{\vdots}\RightLabel{$A\orth$}\UnaryInfC{$\gamma$}\noLine\UnaryInfC{\vdots}\DisplayProof};

	\node (phi'') at (8,-9) {\AxiomC{}\RightLabel{$ax$}\UnaryInfC{\emptyforproof}\AxiomC{}\RightLabel{$ax$}\UnaryInfC{\emptyforproof}\RightLabel{pos}\BinaryInfC{\emptyforproof}\RightLabel{neg}\UnaryInfC{$\fCenter A\orth, A$}\AxiomC{\vdots}\RightLabel{$A\orth$}\UnaryInfC{\emptyforproof}\RightLabel{$c$}\BinaryInfC{$\gamma$}\noLine\UnaryInfC{\vdots}\DisplayProof};
\begin{myscope}
	\path (pi) \edgelabel{$\eta$} (phi);
	\path (pi) \edgelabel{$\beta$} (tau);
\end{myscope}
	\node at (4,-4.5) {\color{red}expand all $ax$-rules on $A\orth$ in $\rho$};
	\node at (4,-7.5) {\color{red}commute $c$ until reaching $A\orth$};
\begin{myscopec}{red}
	\path (phi) \edgelabel{$\eta^*$} (phi');
	\path (tau) \edgelabel{$\eta^*$} (tau');
	\path (phi') \edgelabel{$\beta^*$} (phi'');
	\path (phi'') \edgelabel{$\beta^*$} (tau');
\end{myscopec}
\end{tikzpicture}
\caption{Schematic representation of the second case of the proof of \autoref{lem:betafrom_etato}}
\label{fig:proof_betafrom_etato}
\end{figure}

\begin{fact}
\label{fact:beta_in_eta}
Cut-elimination $\betato$ preserves being in $\eta$-normal form.
\end{fact}

\begin{lem}
\label{lem:etafromstar_betatostar}
Let $\pi$, $\tau$ and $\phi$ be proofs such that $\tau\etafromstar\pi\betatostar\phi$, with $\tau$ an $\eta$-normal proof.
There exists an $\eta$-normal proof $\mu$ such that $\tau\betatostar\mu\etafromstar\phi$.
Diagrammatically:
\begin{center}
\begin{tikzpicture}
\begin{myscope}
	\node[draw=none,minimum size=3mm] (pi) at (0,1.5) {$\pi$};
	\node[draw=none,minimum size=3mm] (phi) at (0,0) {$\phi$};
	\node[draw=none,minimum size=3mm] (varpi) at (1.5,1.5) {$\tau$};
	\node[draw=none] at (2.75,1.5) {$\eta$-normal};

	\path (pi) \edgelabel{$\beta^*$} (phi);
	\path (pi) \edgelabel{$\eta^*$} (varpi);
\end{myscope}
\begin{myscopec}{red}
	\node[draw=none,minimum size=3mm] (varphi) at (1.5,0) {$\mu$};
	\node[draw=none] at (2.75,0) {$\eta$-normal};

	\path (phi) edge[dashed] node[draw=none, sloped, align=center]{$\eta^*$ \\} (varphi);
	\path (varpi) edge[dashed] node[draw=none, sloped, align=center]{$\beta^*$ \\} (varphi);
\end{myscopec}
\end{tikzpicture}
\end{center}
\end{lem}
\begin{proof}
Assume $\tau\overset{\eta^m}{\longleftarrow}\pi\overset{\beta^n}{\longrightarrow}\phi$.
If we find a $\mu$ respecting the diagram, then it is an $\eta$-normal form thanks to \autoref{fact:beta_in_eta}.
We reason by induction on the lexicographic order of the triple $\left(a(\pi\betatostar\phi), n, m\right)$.
If $n=0$ or $m=0$, then the result trivially holds ($\mu = \tau$ or $\mu = \phi$).

Consider the case $n+1$ and $m+1$.
Therefore, $\tau\overset{\eta^m}{\longleftarrow}\iota\etafrom\pi\betato\kappa\overset{\beta^n}{\longrightarrow}\phi$.
We apply \autoref{lem:betafrom_etato} on $\iota\etafrom\pi\betato\kappa$, yielding $\rho$ such that $\iota\betato\rho\etafromstar\phi$ with $a(\iota\betato\rho)=a(\pi\betato\kappa)$ or $\rho_1$ and $\rho_2$ such that $\iota\etatostar\rho_1\betatostar\rho_2\etafromstar\kappa$ with $a(\rho_1\betatostar\rho_2) < a(\pi\betato\kappa)$.
Both of these cases, and the reasoning we will apply, are illustrated by diagrams preceding them, with the following color convention: in blue are uses of \autoref{lem:betafrom_etato}, in green of confluence of $\etato$ and in red of the induction hypothesis.

\begin{center}
\begin{tikzpicture}
\begin{myscope}
	\node[draw=none,minimum size=3mm] (pi) at (0,1.5) {$\pi$};
	\node[draw=none,minimum size=3mm] (iota) at (1.5,1.5) {$\iota$};
	\node[draw=none,minimum size=3mm] (varpi) at (3,1.5) {$\tau$};
	\node[draw=none] at (4.25,1.5) {$\eta$-normal};
	\node[draw=none,minimum size=3mm] (kappa) at (0,-1.5) {$\kappa$};
	\node[draw=none,minimum size=3mm] (phi) at (0,-3) {$\phi$};

	\path (pi) \edgelabel{$\eta$} (iota);
	\path (iota) \edgelabel{$\eta^m$} (varpi);
	\path (pi) \edgelabel{$\beta$} (kappa);
	\path (kappa) \edgelabel{$\beta^n$} (phi);
\end{myscope}
\begin{myscopec}{blue}
	\node[draw=none,minimum size=3mm] (varphi) at (1.5,-1.5) {$\rho$};

	\node[draw=none] at (.75,0) {$=$};
	\path (kappa) \edgelabel{$\eta^*$} (varphi);
	\path (iota) \edgelabel{$\beta$} (varphi);
\end{myscopec}
\begin{myscopec}{red}
	\node[draw=none] at (2.25,0) {\small $IH$};
	\node[draw=none,minimum size=3mm] (nu) at (3,-1.5) {$\nu$};
	\node[draw=none] at (4.25,-1.5) {$\eta$-normal};

	\path (varpi) \edgelabel{$\beta^*$} (nu);
	\path (varphi) \edgelabel{$\eta^*$} (nu);

	\node[draw=none] at (1.5,-2.25) {\small $IH$};
	\node[draw=none,minimum size=3mm] (mu) at (3,-3) {$\mu$};
	\node[draw=none] at (4.25,-3) {$\eta$-normal};

	\path (nu) \edgelabel{$\beta^*$} (mu);
	\path (phi) \edgelabel{$\eta^*$} (mu);
\end{myscopec}
\end{tikzpicture}
\end{center}
Assume to be in the first case.
Applying the induction hypothesis on $\tau\overset{\eta^m}{\longleftarrow}\iota\betato\rho$, with $\tau$ in $\eta$-normal form, $a(\iota\betato\rho)=a(\pi\betato\kappa)\leq a(\pi\betatostar\phi)$, $1\leq n+1$ and $m<m+1$, there exists an $\eta$-normal proof $\nu$ such that $\tau\betatostar\nu\etafromstar\rho$.
We now apply the induction hypothesis on $\nu\etafromstar\rho\etafromstar\kappa\overset{\beta^n}{\longrightarrow}\phi$, with $\nu$ in $\eta$-normal form, $a(\kappa\overset{\beta^n}{\longrightarrow}\phi)\leq a(\pi\betatostar\phi)$ and $n< n+1$.
We obtain an $\eta$-normal proof $\mu$ such that $\nu\betatostar\mu\etafromstar\phi$.
This concludes the first case.

\begin{center}
\begin{tikzpicture}
\begin{myscope}
	\node[draw=none,minimum size=3mm] (pi) at (0,1.5) {$\pi$};
	\node[draw=none,minimum size=3mm] (iota) at (1.5,1.5) {$\iota$};
	\node[draw=none,minimum size=3mm] (varpi) at (3,1.5) {$\tau$};
	\node[draw=none] at (4.25,1.5) {$\eta$-normal};
	\node[draw=none,minimum size=3mm] (kappa) at (0,-1.5) {$\kappa$};
	\node[draw=none,minimum size=3mm] (phi) at (0,-3) {$\phi$};

	\path (pi) \edgelabel{$\eta$} (iota);
	\path (iota) \edgelabel{$\eta^m$} (varpi);
	\path (pi) \edgelabel{$\beta$} (kappa);
	\path (kappa) \edgelabel{$\beta^n$} (phi);
\end{myscope}
\begin{myscopec}{blue}
	\node[draw=none,minimum size=3mm] (varphi1) at (1.5,0) {$\rho_1$};
	\node[draw=none,minimum size=3mm] (varphi2) at (1.5,-1.5) {$\rho_2$};

	\node[draw=none] at (.75,-.75) {$>$};
	\path (iota) \edgelabel{$\eta^*$} (varphi1);
	\path (varphi1) \edgelabel{$\beta^*$} (varphi2);
	\path (kappa) \edgelabel{$\eta^*$} (varphi2);
\end{myscopec}
\begin{myscopec}{olive}
	\path (varphi1) \edgelabel{$\eta^*$} (varpi);
\end{myscopec}
\begin{myscopec}{red}
	\node[draw=none] at (2.25,-.75) {\small $IH$};
	\node[draw=none,minimum size=3mm] (nu) at (3,-1.5) {$\nu$};
	\node[draw=none] at (4.25,-1.5) {$\eta$-normal};

	\path (varpi) \edgelabel{$\beta^*$} (nu);
	\path (varphi2) \edgelabel{$\eta^*$} (nu);

	\node[draw=none] at (1.5,-2.25) {\small $IH$};
	\node[draw=none,minimum size=3mm] (mu) at (3,-3) {$\mu$};
	\node[draw=none] at (4.25,-3) {$\eta$-normal};

	\path (nu) \edgelabel{$\beta^*$} (mu);
	\path (phi) \edgelabel{$\eta^*$} (mu);
\end{myscopec}
\end{tikzpicture}
\end{center}
Consider the second case.
Using the confluence of $\etato$ on $\tau\overset{\eta^m}{\longleftarrow}\iota\etatostar\rho_1$ (\autoref{prop:eta_canonic}), with $\tau$ in $\eta$-normal form, yields $\tau\etafromstar\rho_1$.
We then apply the induction hypothesis on $\tau\etafromstar\rho_1\betatostar\rho_2$, with $\tau$ in $\eta$-normal form and $a(\rho_1\betatostar\rho_2)< a(\pi\betatostar\phi)$.
This yields an $\eta$-normal proof $\nu$ such that $\tau\betatostar\nu\etafromstar\rho_2$.
We use the induction hypothesis again, this time on $\nu\etafromstar\rho_2\etafromstar\kappa\overset{\beta^n}{\longrightarrow}\phi$, with $\nu$ in $\eta$-normal form, $a(\kappa\overset{\beta^n}{\longrightarrow}\phi)\leq a(\pi\betatostar\phi)$ and $n<n+1$.
We obtain an $\eta$-normal proof $\mu$ with $\nu\betatostar\mu\etafromstar\phi$, solving the second case.
\end{proof}

We can now prove the main result of this section, \autoref{prop:eta_nom_eqbn}.

\begin{proof}[Proof of \autoref{prop:eta_nom_eqbn}]
\label{proof:eta_nom_eqbn}
We reason by induction on the length of the sequence $\pi\eqbn\tau$.
If it is of null length, then $\pi=\tau$ hence $\eta(\pi)=\eta(\tau)$.
Otherwise, there is a proof $\phi$ such that $\pi\cadratin\phi\eqbn\tau$ with $\cadratin~\in\{\betato;\betafrom;\etato;\etafrom\}$.
By induction hypothesis, $\eta(\phi)\eqb\eta(\tau)$, with only axiom-expanded proofs in this sequence.
We distinguish cases according to $\pi\cadratin\phi$.

If $\pi\betato\phi$, then as $\eta(\pi)\etafromstar\pi$ we can apply \autoref{lem:etafromstar_betatostar} to obtain an $\eta$-normal proof $\mu$ such that $\eta(\pi)\betatostar\mu\etafromstar\phi$.
Thus, $\mu=\eta(\phi)$, so $\eta(\pi)\betatostar\eta(\phi)\eqb\eta(\tau)$ and the result holds (with only $\eta$-normal proofs in $\eta(\pi)\betatostar\eta(\phi)$ thanks to \autoref{fact:beta_in_eta}).

Similarly, if $\pi\betafrom\phi$ then, as $\phi\etatostar\eta(\phi)$, there exists an $\eta$-normal proof $\mu$ such that $\pi\etatostar\mu\betafromstar\eta(\phi)$ (\autoref{lem:etafromstar_betatostar}).
Thus $\mu=\eta(\pi)$, and $\eta(\pi)\betafromstar\eta(\phi)\eqb\eta(\tau)$.

Finally, if $\pi\etato\phi$ or $\pi\etafrom\phi$, then $\eta(\pi)=\eta(\phi)\eqb\eta(\tau)$ and the conclusion follows.
\end{proof}

%==================================================

\section{Cut-elimination and rule commutations}
\label{subsec:gen_eqc}

It is not possible to apply the reasoning from the previous section (about axiom-expansion) to cut-elimination, because cut-elimination is confluent only up to rule commutations.
In the literature of linear logic, there already are known results about (strong) normalization and confluence of cut-elimination but most are for the proof-nets syntax and not for the sequent calculus syntax: \eg~\cite{mallpnlong,allpnbox1,allpnbox2,cutelimmallpn} about multiplicative-additive linear logic, \cite{llsnaccattoli,stnormpnmodstrcong} about multiplicative-exponential linear logic, and notably~\cite{snll} for all connectives and second order quantifiers.
The confluence of cut-elimination up to rule commutations in sequent calculus is expected to hold in linear logic (with all connectives and units), but it has not been proved yet.
Fortunately, in the restricted case of MALL, this key result has been proved~\cite[Theorem~5.1]{calculusformall}, which allows us to reduce the problem of isomorphisms from looking at cut-elimination to considering rule commutations.

A first, easy but important, result is that rule commutation is included in equality up to cut-elimination.

\begin{prop}[$\eqcone~\subseteq~\eqb$]
\label{prop:eqc_in_eqbn}
Given proofs $\pi$ and $\tau$, if $\pi\eqcone\tau$ then $\pi\eqb\tau$.
\end{prop}
\begin{proof}
It suffices, for each commutation, to give a proof that can be reduced by cut elimination to both sides.
We give here only a few representative cases.
\begin{itemize}
\item \proofpar{$\comm{\with}{\top}$ and $\comm{\top}{\with}$ commutations.}
The proof
\begin{prooftree}
\Rtop{$0,\top$}

\Rtop{$A_1,\top,\Gamma$}
\Rtop{$A_2,\top,\Gamma$}
\Rwith{$A_1\with A_2,\top,\Gamma$}

\Rcut{$A_1\with A_2,\top,\Gamma$}
\end{prooftree}
reduces to both
\begin{center}
\Rtop{$A_1\with A_2,\top,\Gamma$}
\DisplayProof
\hskip 2em and \hskip 2em
\Rtop{$A_1,\top,\Gamma$}
\Rtop{$A_2,\top,\Gamma$}
\Rwith{$A_1\with A_2,\top,\Gamma$}
\DisplayProof
\end{center}
according respectively to whether the $cut$-rule is first commuted with the $\top$-rule on its left or with the $\with$-rule on its right and then with the $\top$-rule on its left in both occurrences.

\item \proofpar{$\comm{\oplus_i}{\bot}$ and $\comm{\bot}{\oplus_i}$ commutations.}
The proof
\begin{prooftree}
\Rsub{$\pi$}{$A_i,\Gamma$}
\Rbot{$A_i,\bot,\Gamma$}
\Rax{$A_i$}
\Rplusi{$A_i\orth, A_1\oplus A_2$}

\Rcut{$A_1\oplus A_2,\bot,\Gamma$}
\end{prooftree}
reduces to both
\begin{center}
\Rsub{$\pi$}{$A_i,\Gamma$}
\Rplusi{$A_1\oplus A_2,\Gamma$}
\Rbot{$A_1\oplus A_2,\bot,\Gamma$}
\DisplayProof
\hskip 2em and \hskip 2em
\Rsub{$\pi$}{$A_i,\Gamma$}
\Rbot{$A_i,\bot,\Gamma$}
\Rplusi{$A_1\oplus A_2,\bot,\Gamma$}
\DisplayProof
\end{center}
according respectively to whether the $cut$-rule is first commuted with the $\bot$-rule on its left or with the $\oplus_i$-rule on its right.

\item \proofpar{$\comm{\tens}{\tens}$ commutations.}
The proof
\begin{prooftree}
\Rsub{$\pi_1$}{$A_1,\Gamma$}
\Rax{$A_2$}
\Rtens{$A_2\orth,A_1\tens A_2,\Gamma$}

\Rsub{$\pi_2$}{$A_2,B_1,\Delta$}
\Rsub{$\pi_3$}{$B_2,\Sigma$}
\Rtens{$B_1\tens B_2,\Delta,\Sigma,A_2$}

\Rcut{$A_1\tens A_2,B_1\tens B_2,\Gamma,\Delta,\Sigma$}
\end{prooftree}
reduces to both
\begin{center}
\resizebox{.975\textwidth}{!}{
\Rsub{$\pi_1$}{$A_1,\Gamma$}
\Rsub{$\pi_2$}{$A_2,B_1,\Delta$}
\Rsub{$\pi_3$}{$B_2,\Sigma$}
\Rtens{$A_2,B_1\tens B_2,\Delta,\Sigma$}
\Rtens{$A_1\tens A_2,B_1\tens B_2,\Gamma,\Delta,\Sigma$}
\DisplayProof
\hskip.5em and\hskip.5em
\Rsub{$\pi_1$}{$A_1,\Gamma$}
\Rsub{$\pi_2$}{$A_2,B_1,\Delta$}
\Rtens{$A_1\tens A_2,B_1,\Gamma,\Delta$}
\Rsub{$\pi_3$}{$B_2,\Sigma$}
\Rtens{$A_1\tens A_2,B_1\tens B_2,\Gamma,\Delta,\Sigma$}
\DisplayProof
}
\end{center}
according respectively to whether the $cut$-rule is first commuted with the $\tens$-rule on its left or on its right.

Other $\comm{\tens}{\tens}$ commutations are similar.

\item \proofpar{$\comm{\parr}{\with}$ and $\comm{\with}{\parr}$ commutations.}
The proof
\begin{prooftree}
\Rsub{$\pi_1$}{$A_1,A_2,B_1,\Gamma$}
\Rparr{$A_1\parr A_2,B_1,\Gamma$}
\Rsub{$\pi_2$}{$A_1,A_2,B_2,\Gamma$}
\Rparr{$A_1\parr A_2,B_2,\Gamma$}
\Rwith{$A_1\parr A_2,B_1\with B_2,\Gamma$}

\Rax{$A_2$}
\Rax{$A_1$}
\Rtens{$A_1,A_2,A_2\orth\tens A_1\orth$}
\Rparr{$A_1\parr A_2,A_2\orth\tens A_1\orth$}

\Rcut{$A_1\parr A_2,B_1\with B_2,\Gamma$}
\end{prooftree}
reduces to both
\begin{center}
\resizebox{.975\textwidth}{!}{
\Rsub{$\pi_1$}{$A_1,A_2,B_1,\Gamma$}
\Rparr{$A_1\parr A_2,B_1,\Gamma$}
\Rsub{$\pi_2$}{$A_1,A_2,B_2,\Gamma$}
\Rparr{$A_1\parr A_2,B_2,\Gamma$}
\Rwith{$A_1\parr A_2,B_1\with B_2,\Gamma$}
\DisplayProof
\hskip.5em and\hskip.5em
\Rsub{$\pi_1$}{$A_1,A_2,B_1,\Gamma$}
\Rsub{$\pi_2$}{$A_1,A_2,B_2,\Gamma$}
\Rwith{$A_1,A_2,B_1\with B_2,\Gamma$}
\Rparr{$A_1\parr A_2,B_1\with B_2,\Gamma$}
\DisplayProof
}
\end{center}
according respectively to whether the $cut$-rule is first commuted with the $\with$-rule on its left or with the $\parr$-rule on its right.\qedhere
\end{itemize}
\end{proof}

Cut-elimination $\betato$ is weakly normalizing (see \autoref{cor:beta_wn}).
Furthermore, rule commutation is its ``core'', in the following meaning.

\begin{restatable}[Rule commutation is the core of cut-elimination~{\cite[Theorem~5.1]{calculusformall}}]{thm}{theqbeqc}
\label{cor:CReqc_simpl}
If two proofs $\pi$ and $\tau$ are $\beta$-equal, then any of their normal forms by $\betato$ are related by $\eqc$.
\end{restatable}

We can extend the previous result to $\beta\eta$-equality.

\begin{thm}
\label{th:eqbn_eqc}
Let $\pi_1$ and $\pi_2$ be $\beta\eta$-equal proofs.
Then, letting $\pi_1'$ (\resp\ $\pi_2'$) be a result of expanding all axioms and then eliminating all cuts in $\pi_1$ (\resp\ $\pi_2$), $\pi_1' \eqc \pi_2'$.
\end{thm}
\begin{proof}
We have $\pi_1\eqbn\pi_2$, so $\eta(\pi_1)\eqb\eta(\pi_2)$ by \autoref{prop:eta_nom_eqbn}.
By \autoref{cor:CReqc_simpl}, it follows $\eta(\pi_1)\betatostar\pi_1'\eqc\pi_2'\betafromstar\eta(\pi_2)$.
\end{proof}

Note there seems to be a mistake, as well as a forgotten case, in the proof of~\cite[Theorem~5.1]{calculusformall}, as pointed out previously (\autoref{subsec:transfo_seq}).
A main tool of this proof is a measure that decreases when applying a cut-elimination step and is stable by ($\top$-free) rule commutations~\cite[Proposition~B.2]{calculusformall}.
This measure is a multiset of some value associated to each cut-rule.
For rule commutations are not defined in a cut-free setting in~\cite{calculusformall}, a $\with-\tens$ commutation that duplicates a sub-proof can in particular increase the number of cut-rules.
This means that a $\with-\tens$ commutation can increase the measure, as it does not change the value of each cut-rule (as claimed in~\cite{calculusformall}) but can duplicate a value in the multiset.
We give a corrected proof of \autoref{cor:CReqc_simpl} in \autoref{subsec:proofs_add_1}.

%==================================================

\section{Reduction to unit-free distributed formulas}
\label{sec:red_unit-free}

As written in the introduction, our plan is to use proof-nets, which currently exist only for axiom-expanded proofs in unit-free MALL.
For we already reduced the problem to axiom-expanded proofs in \autoref{subsec:no_eta}, it remains to take care of the units.
The main idea is that the multiplicative and additive units can be replaced by fresh atoms for the study of isomorphisms.
However, this is not true in general, for instance we have $(1\oplus A)\tens B \iso B\oplus(A\tens B)$ (using the soundness theorem, \autoref{th:iso_sound}), and this isomorphism uses that $1$ is unital for $\tens$.
Hence, we begin by reducing the problem to so-called \emph{distributed} formulas, using the distributivity, unitality and cancellation equations of \autoref{tab:eqisos} on \autopageref{tab:eqisos} (\autoref{subsec:red_distr}).
Then, for theses special formulas, we identify patterns containing units in proofs equal to an identity up to rule commutations, patterns which can be lifted to proofs of isomorphisms (\autoref{subsec:pattern_dist_seqisos}).
Finally, we prove that, in an isomorphism between distributed formulas, replacing units by fresh atoms preserves being an isomorphism, reducing the study of isomorphisms to distributed formulas in the unit-free fragment (\autoref{subsec:compl_units}).

\subsection{Reduction to distributed formulas}
\label{subsec:red_distr}

\begin{defi}[Distributed formula]
\label{def:dist}
A formula is \emph{distributed} if it does not have any sub-formula of the form
\begin{equation*}
A\tens(B\oplus C)
\quad
(A\oplus B)\tens C
\quad
A\tens 1
\quad
1\tens A
\quad
A\oplus 0
\quad
0\oplus A
\quad
A\tens 0
\quad
0\tens A
\end{equation*}
or their duals
\begin{equation*}
(C\with B)\parr A
\quad
C\parr(B\with A)
\quad
\bot\parr A
\quad
A\parr\bot
\quad
\top\with A
\quad
A\with\top
\quad
\top\parr A
\quad
A\parr\top
\end{equation*}
(where $A$, $B$ and $C$ are any formulas).
\end{defi}

\begin{rem}
\label{rem:dist}
This notion is stable under duality: if $A$ is distributed, so is $A\orth$.
\end{rem}

\begin{defi}
\label{def:distsystem}
We call $\distsystem$ the following rewriting system on MALL formulas.
\begin{center}
\begin{tabular}{cccc|cccc|cccc|cccc}
\multicolumn{3}{c}{$A\tens(B\oplus C)$} & \multicolumn{2}{c}{$\rightarrow$} & \multicolumn{3}{c|}{$(A\tens B)\oplus(A\tens C)$} & \multicolumn{3}{c}{$(C\with B)\parr A$} & \multicolumn{2}{c}{$\rightarrow$} & \multicolumn{3}{c}{$(C\parr A)\with(B\parr A)$} \\
\multicolumn{3}{c}{$(A\oplus B)\tens C$} & \multicolumn{2}{c}{$\rightarrow$} & \multicolumn{3}{c|}{$(A\tens C)\oplus(B\tens C)$} & \multicolumn{3}{c}{$C\parr(B\with A)$} & \multicolumn{2}{c}{$\rightarrow$} & \multicolumn{3}{c}{$(C\parr B)\with(C\parr A)$} \\\hline
$A\tens 1$ & $\rightarrow$ & $A$ & & & $1\tens A$ & $\rightarrow$ & $A$ & $A\parr\bot$ & $\rightarrow$ & $A$ & & & $\bot\parr A$ & $\rightarrow$ & $A$ \\
$A\oplus 0$ & $\rightarrow$ & $A$ & & & $0\oplus A$ & $\rightarrow$ & $A$ & $A\with\top$ & $\rightarrow$ & $A$ & & & $\top\with A$ & $\rightarrow$ & $A$ \\
$A\tens 0$ & $\rightarrow$ & $0$ & & & $0\tens A$ & $\rightarrow$ & $0$ & $A\parr \top$ & $\rightarrow$ & $\top$ & & & $\top\parr A$ & $\rightarrow$ & $\top$ \\
\end{tabular}
\end{center}
\end{defi}

\begin{fact}
\label{lem:distsystem_sn}
The rewriting system $\distsystem$ is strongly normalizing, with as normal forms distributed formulas.
\end{fact}

\begin{prop}
\label{prop:red_to_dist}
If the equational theory denoted $\AC$ -- recall \autoref{tab:eqisos} on \autopageref{tab:eqisos} -- is complete for isomorphisms between distributed formulas of MALL, then $\Eu$ is complete for isomorphisms between arbitrary formulas of MALL.
In other words, if for all $A_d$ and $B_d$ distributed formulas of MALL, it stands that $A_d \iso B_d \implies A_d \eqAC B_d$, then for all $A$ and $B$ arbitrary formulas of MALL, $A \iso B \implies A \eqEu B$.
\end{prop}
\begin{proof}
Consider an isomorphism $A\iso B$ between two arbitrary formulas $A$ and $B$ of MALL.
Let $A_d$ and $B_d$ be associated distributed formulas, obtained as normal forms of the rewriting system $\distsystem{}$ (using \autoref{lem:distsystem_sn}).
As each rule of this rewriting system corresponds to a valid equality in the theory $\Eu$, we have $A \eqEu A_d$ and $B \eqEu B_d$.

By soundness of $\Eu$ (\autoref{th:iso_sound}) and as linear isomorphism is a congruence, we deduce $A_d \iso A \iso B \iso B_d$.
The completeness hypothesis on $\AC$ for distributed formulas yields $A_d \eqAC B_d$ from $A_d \iso B_d$.
Thus $A \eqEu A_d \eqEu B_d \eqEu B$, using that $\AC$ is included in $\Eu$.
\end{proof}

Therefore, we can now consider only distributed formulas.
We will use this to study units, as well as when solving the unit-free case to prove there are only commutativity and associativity isomorphisms left.

\subsection{Patterns in distributed isomorphisms}
\label{subsec:pattern_dist_seqisos}

In this section, we prove units in distributed isomorphisms are in very specific sets of rules, which are the following patterns.

\begin{defi}[Patterns]
\label{def:patterns_units_isos}
\hfill
\begin{itemize}
\item
We call a \emph{$\top/0$-pattern} the following sub-proof:
\begin{prooftree}
\Rtop{$\top,0$}
\end{prooftree}
\item
Likewise, we set \emph{$1/\bot$-pattern} the sub-proof:
\begin{prooftree}
\Rone{}
\Rbot{$\bot,1$}
\end{prooftree}
\item
A \emph{$1/\oplus/\bot$-pattern} is a $1$-rule followed by a (possibly empty) sequence $\rho$ of $\oplus_i$-rules and finally a $\bot$-rule:
\begin{prooftree}
\Rone{}
\doubleLine\dashedLine\RightLabel{$\rho$}\UnaryInfC{$\fCenter F$}
\Rbot{$\bot,F$}
\end{prooftree}
\end{itemize}
\end{defi}

The main goal of this part is proving $\top$- $1$- and $\bot$-rules in proofs of distributed isomorphisms belong to these patterns.
The two $\top/0$- and $1/\bot$-patterns are those we truly wish for.
However, we will not find directly $1/\bot$-patterns, but the more generic $1/\oplus/\bot$-patterns; this is not a problem, for using $\oplus_i-\bot$ commutations turns a $1/\oplus/\bot$-pattern into a $1/\bot$-pattern.
We will need at different points the following grammar.

\begin{defi}
\label{def:oplus_grammar}
Given a formula $B$, we define the grammar
\begin{equation*}
\pcontext{B} \coloncoloneqq B \mid \pcontext{B}\oplus C \mid C\oplus \pcontext{B}
\end{equation*}
where $C$ ranges over all formulas.
The occurrence of $B$ which is the base case of this grammar is called its distinguished occurrence.
\end{defi}

We first analyze the behavior of units in proofs equal to $\id{A}$ up to rule commutation (recall $\id{A}$ is the axiom-expansion of $\ax{A}$, the proof composed of an $ax$-rule on $A$), and prove they belong to these patterns (\autoref{sec:pattern_dist_seqisos_id}).
We only do so for a \emph{distributed} formula $A$ as we have already seen it is enough in \autoref{subsec:red_distr}, and as the above patterns hold only for these formulas.
We can then obtain our result, for this property is preserved by cut anti-reduction.
However, it is not so easy to prove it, and we will consider slices to do so (\autoref{sec:pattern_dist_seqisos_slices}).
This finally allows us to remove the units (\autoref{sec:pattern_dist_seqisos_proofs}).

\subsubsection{Patterns in identities up to rule commutation}
\label{sec:pattern_dist_seqisos_id}

We study here patterns containing units in proofs equal to an identity up to rule commutation.
The tedious case study on rule commutations and the properties to find in this proof are a great example on why we do not want to do the full proof of completeness in sequent calculus, but wish instead to use proof-nets.

\begin{prop}
\label{prop:eqc_id}
Let $\pi$ be a proof equal, up to rule commutation, to $\id{A}$ with $A$ \emph{distributed}.
Then:
\begin{itemize}
\item the $\top$-rules of $\pi$ are in a $\top/0$-pattern {\Rtop{$\top,0$}\DisplayProof} with $\top$ in $A$ being the dual of $0$ in $A\orth$ (or vice-versa $\top$ in $A\orth$ being the dual of $0$ in $A$);
\item $\bot$-rules and $1$-rules come by pairs in a $1/\oplus/\bot$-pattern
{\Rone{}\doubleLine\dashedLine\RightLabel{$\rho$}\UnaryInfC{$\fCenter F$}\Rbot{$\bot,F$}\DisplayProof}
with $\bot$ in $A$ being the dual of $1$ in $A\orth$ (or vice-versa);
\item there is no sequent in $\pi$ of the shape $\fCenter B\with C$.
\end{itemize}
\end{prop}
\begin{proof}
The key idea is to find properties of $\id{A}$ preserved by all rule commutations and ensuring the properties described in the statement.
Hence, we prove a stronger property: any sequent $S$ of a proof $\pi$ obtained through a sequence of rule commutations from $\id{A}$ for a distributed formula $A$ respects:
\begin{enumerate}[(1)]
\item\label{item:eqc_id:1}
the formulas of $S$ are distributed;
\item\label{item:eqc_id:2}
if $\top$ is a formula of $S$, then $S=~\fCenter \top,0$, with $0$ in $A\orth$ the dual of $\top$ if $\top$ is a sub-formula of $A$ (or vice-versa);
\item\label{item:eqc_id:3}
if $\bot$ is a formula of $S$, then $S=~\fCenter \bot,\pcontext{1}$ with $\pcontext{1}$ the grammar defined in \autoref{def:oplus_grammar}, where the distinguished occurrence of $1$ is the dual of $\bot$ in $A\orth$ if $\bot$ is a sub-formula of $A$ (or vice-versa), and the sub-proof of $\pi$ above $S$ is a sequence of $\oplus_i$ rules leading to the distinguished $1$, with in addition a $\bot$-rule inside this sequence;
\item\label{item:eqc_id:4}
if $B\with C$ is a formula of $S$, then $S=~\fCenter B\with C,\pcontext{C\orth\oplus B\orth}$ with $\pcontext{C\orth\oplus B\orth}$ the grammar defined in \autoref{def:oplus_grammar}, where the distinguished occurrence of $C\orth\oplus B\orth$ is the dual of $B\with C$ in $A\orth$ if $B\with C$ is a sub-formula of $A$ (or vice-versa), and in the sub-proof of $\pi$ above $S$ the $\oplus$-rules of the distinguished $C\orth\oplus B\orth$ are a $\oplus_2$-rule in the left branch of the $\with$-rule of $B\with C$, and a $\oplus_1$-rule in its right branch;
\item\label{item:eqc_id:5}
if $S$ contains several negative formulas or several positive formulas, then its negative formulas are all $\parr$-formulas or negated atoms.
\end{enumerate}
Remark that \ref{item:eqc_id:5} is a corollary of properties \ref{item:eqc_id:2}, \ref{item:eqc_id:3} and \ref{item:eqc_id:4}.
As in $\eqcone$ there is no commutation with a $cut$-rule (in particular no $cut-\top$ commutation) and no $\tens-\top$ commutation creating a sub-proof with a $cut$-rule, it follows that $\pi$ is cut-free and has the sub-formula property, making \ref{item:eqc_id:1} trivially true.
We prove that the fully expanded axiom respects properties \ref{item:eqc_id:2}, \ref{item:eqc_id:3} and \ref{item:eqc_id:4}, and that these properties are preserved by any rule commutation of $\eqcone$.

\proofpar{The fully expanded axiom respects the properties.}
We prove by induction on the distributed formula $A$ properties \ref{item:eqc_id:2}, \ref{item:eqc_id:3} and \ref{item:eqc_id:4}.
Notice that sub-formulas of $A$ are also distributed.
By symmetry, assume $A$ is positive.

If $A\in\{X,1,0\}$ where $X$ is an atom, then:
\[\id{A}\in\left\{\text{\Rax{$X$}\DisplayProof};\text{\Rone\Rbot{$\bot,1$}\DisplayProof};\text{\Rtop{$\top,0$}\DisplayProof}\right\}\]
Each of these proofs respects \ref{item:eqc_id:2}, \ref{item:eqc_id:3} and \ref{item:eqc_id:4}.

Assume the result holds for $B$ and $C$, and that $A=B\tens C$.
The proof $\id{A}$ is:
\begin{prooftree}
\Rsub{$\id{B}$}{$B\orth,B$}
\Rsub{$\id{C}$}{$C\orth,C$}
\Rtens{$C\orth,B\orth,B\tens C$}
\Rparr{$C\orth\parr B\orth,B\tens C$}
\end{prooftree}
We have to prove the sequents $\fCenter C\orth,B\orth,B\tens C$ and $\fCenter C\orth\parr B\orth,B\tens C$ respect the properties.
The latter respects \ref{item:eqc_id:2}, \ref{item:eqc_id:3} and \ref{item:eqc_id:4} trivially for it has neither a $\top$, $\bot$ nor $\with$ formula.
As $C\orth\parr B\orth$ is distributed, it follows that neither $C\orth$ nor $B\orth$ can be a $\top$, $\bot$ or $\with$ formula, and as such the former sequent also respects the properties.

Suppose $A=B\oplus C$ with sequents of $\id{B}$ and $\id{C}$ respecting the properties.
Now, $\id{A}$ is:
\begin{prooftree}
\Rsub{$\id{C}$}{$C\orth,C$}
\Rplustwo{$C\orth,B\oplus C$}
\Rsub{$\id{B}$}{$B\orth,B$}
\Rplusone{$B\orth,B\oplus C$}
\Rwith{$C\orth\with B\orth, B\oplus C$}
\end{prooftree}
The sequent $\fCenter C\orth\with B\orth, B\oplus C$ respects \ref{item:eqc_id:2}, \ref{item:eqc_id:3} and \ref{item:eqc_id:4}, as the $\oplus$ is the dual of the $\with$.
By symmetry, we show the properties are also fulfilled by $\fCenter B\orth,B\oplus C$, and they will be respected by $\fCenter C\orth,B\oplus C$ with a similar proof.
As the formulas are distributed, $B\orth$ cannot be a $\top$ formula, hence the sequent respects \ref{item:eqc_id:2}.
If $B\orth$ is not a $\bot$ nor $\with$ formula, then \ref{item:eqc_id:3} and \ref{item:eqc_id:4} hold for $\fCenter B\orth,B\oplus C$.
If it is, then using that $\fCenter B\orth,B$ respects \ref{item:eqc_id:3} and \ref{item:eqc_id:4}, it follows that $B\oplus C$ is also of the required shape, for $B$ was.

\proofpar{Every possible rule commutation preserves the properties.}
We show that each rule commutation preserves properties \ref{item:eqc_id:2}, \ref{item:eqc_id:3} and \ref{item:eqc_id:4}, using every time the notations from Tables~\ref{tab:rule_comm} and~\ref{tab:rule_comm_u} in \autoref{def:rule_comm}, on Pages~\pageref{tab:rule_comm} and~\pageref{tab:rule_comm_u}.
By symmetry, we treat only one case for $\tens-\tens$, $\parr-\tens$, $\with-\tens$ and $\oplus_i-\tens$ commutations.
\begin{description}
\item[$\top$-commutations]
Using property~\ref{item:eqc_id:2}, we cannot do any commutation between a $\top$-rule and a $\parr$, $\tens$, $\with$, $\oplus_i$, $\bot$ or $\top$-rule, so no commutations at all involving a $\top$-rule.
\item[$\bot$-commutations]
Using property~\ref{item:eqc_id:3}, we cannot do any commutation between a $\bot$-rule and a $\parr$, $\tens$, $\with$ or $\bot$-rule.
A commutation between a $\bot$ and a $\oplus_i$-rule preserves property \ref{item:eqc_id:3}: we have by hypothesis $\Gamma$ empty and $A_1\oplus A_2$ of the right shape.
It also respects \ref{item:eqc_id:2} and \ref{item:eqc_id:4} trivially.
\item[$\comm{\parr}{\parr}$ commutation]
We have to show the properties for $\fCenter A_1,A_2,B_1\parr B_2,\Gamma$.
Because $\fCenter A_1\parr A_2,B_1\parr B_2,\Gamma$ respects them, negative formulas of $\Gamma$ are $\parr$-formulas or negated atoms by \ref{item:eqc_id:5}.
By distributivity, if $A_1$ (or $A_2$) is a negative formula, then it must be a $\parr$ one or a negated atom.
Thus, $\fCenter A_1,A_2,B_1\parr B_2,\Gamma$ fulfills \ref{item:eqc_id:2}, \ref{item:eqc_id:3} and \ref{item:eqc_id:4}.
\item[$\comm{\oplus_j}{\oplus_i}$ commutation]
We show the properties for $\fCenter A_i,B_1\oplus B_2,\Gamma$.
As $\fCenter A_1\oplus A_2,B_1\oplus B_2,\Gamma$ respects them, negative formulas of $\Gamma$ are $\parr$-formulas or negated atoms by \ref{item:eqc_id:5}.
If $A_i$ is positive, a $\parr$ or a negated atom, then we are done.
Otherwise, as $\fCenter A_i,B_j,\Gamma$ fulfills the properties, it follows $\Gamma$ is empty and $B_j$ of the desired shape.
By \ref{item:eqc_id:1}, $B_j$ is not $0$, thus $A_i$ is not $\top$.
Whether $A_i$ is $\bot$ or $\with$, the sequent $\fCenter A_i, B_1\oplus B_2$ respects the properties.
\item[$\comm{\tens}{\tens}$ commutation]
We have to show the properties for $\fCenter A_1\tens A_2,B_1,\Gamma,\Delta$.
Because ${\fCenter A_1\tens A_2,B_1\tens B_2,\Gamma,\Delta,\Sigma}$ respects them, negative formulas of $\Gamma$ and $\Delta$ are $\parr$-formulas or negated atoms by \ref{item:eqc_id:5}.
If $B_1$ is positive, a $\parr$ or a negated atom, then we are done.
Otherwise, as $\fCenter A_2,B_1,\Delta$ fulfills the properties, it follows $\Delta$ is empty and $B_1$ of the desired shape, so $B_1$ is a $0$, $1$ or $\oplus$-formula.
This is impossible as $B_1\tens B_2$ is distributed by~\ref{item:eqc_id:1}.
\item[$\comm{\with}{\with}$, $\comm{\parr}{\with}$, $\comm{\with}{\parr}$, $\comm{\tens}{\with}$ and $\comm{\with}{\tens}$ commutations]
These cases are impossible by property~\ref{item:eqc_id:4}.
\item[$\comm{\oplus_i}{\with}$ and $\comm{\with}{\oplus_i}$ commutations]
In these cases, \ref{item:eqc_id:4} for $\fCenter A_1\with A_2,B_1\oplus B_2,\Gamma$ implies $\Gamma$ empty and $B_1\oplus B_2$ of the desired shape.
Thus $B_i$ of the desired shape ($B_1\oplus B_2$ is not the distinguished formula as it has the same rule $\oplus_i$ in both branches of the $\with$-rule), proving the result for $\fCenter A_1\with A_2,B_i$.
For $\fCenter A_1, B_1\oplus B_2$ (and similarly $\fCenter A_2, B_1\oplus B_2$), $A_1$ cannot be a $\top$ by \ref{item:eqc_id:1}, and if it is a $\bot$ or a $\with$, then the hypothesis on $\fCenter A_1, B_i$ implies that the properties are also respected in $\fCenter A_1, B_1\oplus B_2$.
\item[$\comm{\oplus_i}{\parr}$ and $\comm{\parr}{\oplus_i}$ commutations]
Let us show the properties for $\fCenter A_1,A_2,B_1\oplus B_2,\Gamma$ in the first commutation and $\fCenter A_1\parr A_2,B_i,\Gamma$ in the second.
As they hold for $\fCenter A_1, A_2, B_i,\Gamma$, negative formulas in $A_1, A_2, B_i,\Gamma$ are $\parr$-formulas or negated atoms by \ref{item:eqc_id:5} and the result follows.
\item[$\comm{\oplus_i}{\tens}$ commutation]
We have to prove $\fCenter A_1,B_1\oplus B_2,\Gamma$ respects the properties.
Because $\fCenter A_1\tens A_2,B_1\oplus B_2,\Gamma,\Delta$ fulfills them, negative formulas of $\Gamma$ are $\parr$ or negated atoms by \ref{item:eqc_id:5}.
If $A_1$ is a negative formula other than a $\parr$ or an atom, then for $\fCenter A_1,B_i,\Gamma$ respects the properties we have that $\Gamma$ is empty and $B_i$ of the desired shape.
By \ref{item:eqc_id:1}, $B_i$ is not a $0$, so $A_1$ is not a $\top$.
But then $B_1\oplus B_2$ also has the wished shape for $A_1$, and $\fCenter A_1,B_1\oplus B_2$ fulfills the properties.
\item[$\comm{\tens}{\oplus_i}$ commutation]
We prove $\fCenter A_1\tens A_2,B_i,\Gamma,\Delta$ respects the properties.
As they are fulfilled by $\fCenter A_1\tens A_2,B_1\oplus B_2,\Gamma,\Delta$, negative formulas of $\Gamma$ and $\Delta$ are $\parr$ or atoms by \ref{item:eqc_id:5}.
As $A_1\tens A_2$ is distributed by \ref{item:eqc_id:1}, $A_1$ cannot be a $0$, $1$ nor $\oplus$ formula, so by $\fCenter A_1,B_i,\Gamma$ fulfilling the properties it follows that $B_i$ cannot be a negative formula other than a $\parr$ or an atom.
The conclusion follows.
\item[$\comm{\tens}{\parr}$ commutation]
We prove the properties for $\fCenter A_1, A_2,B_1\tens B_2,\Gamma,\Delta$.
As they are respected by $\fCenter  A_1\parr A_2,B_1\tens B_2,\Gamma,\Delta$, according to \ref{item:eqc_id:5} negative formulas of $\Gamma$ and $\Delta$ can only be $\parr$-formulas or atoms.
As $A_1\parr A_2$ is distributed by \ref{item:eqc_id:1}, $A_1$ and $A_2$ are positive or $\parr$-formulas or atoms.
The conclusion follows.
\item[$\comm{\parr}{\tens}$ commutation]
We prove the properties for $\fCenter A_1\parr A_2, B_1,\Gamma$.
As $\fCenter A_1, A_2,B_1,\Gamma$ respects them, by \ref{item:eqc_id:5} negative formulas of $\Delta$ and $B_1$ can only be $\parr$-formulas or atoms, proving the result.
\end{description}
Therefore, we proved the expanded axiom respects these properties, and they are preserved by all rule commutations.
The conclusion follows.
\end{proof}

\subsubsection{Surgery on slices}
\label{sec:pattern_dist_seqisos_slices}

In order to prove that $\top$ and $\bot$-rules in proofs of distributed isomorphisms belong to $\top/0$- and $1/\oplus/\bot$-patterns, we show these patterns are preserved by cut anti-reduction and conclude as the identity has these patterns.
However, proving this preservation is not easy and uses the notion of slices (recall \autoref{def:slice_pt} on \autopageref{def:slice_pt}).

Cut-elimination can be extended from proofs to slices except that some reduction steps produce failures for slices: when a $\with_i$ faces a $\oplus_i$ and conversely.
The reduction of the slice
\begin{prooftree}
  \AxiomC{$\fCenter A,\Gamma$}
  \Rplus{1}{$A\oplus B,\Gamma$}
  \AxiomC{$\fCenter B\orth,\Delta$}
  \Rwiths{1}{$B\orth\with A\orth,\Delta$}
  \Rcut{$\Gamma,\Delta$}
\end{prooftree}
is a failure since the selected sub-formulas of $A\oplus B$ and its dual do not match.
More precisely, we have a failure exactly when reducing a key case between $\with_i$ and $\oplus_j$ for $i = j$ -- when $i \neq j$ (\ie\ $j = 1-i$) the cut reduces correctly as follows:
\begin{equation*}
  \AxiomC{$\fCenter A_i,\Gamma$}
  \Rplus{i}{$A_1\oplus A_2,\Gamma$}
  \AxiomC{$\fCenter A\orth_i,\Delta$}
  \Rwiths{j}{$A\orth_2\with A\orth_1,\Delta$}
  \Rcut{$\Gamma,\Delta$}
\DisplayProof
\quad
\betato
\quad
  \AxiomC{$\fCenter A_i,\Gamma$}
  \AxiomC{$\fCenter A\orth_i,\Delta$}
  \Rcut{$\Gamma,\Delta$}
\DisplayProof
\end{equation*}

Given two slices $s\in\Slices(\pi)$ and $r\in\Slices(\rho)$ with respective conclusions $\fCenter A,\Gamma$ and $\fCenter A\orth,\Delta$, their composition by cut $\cutf{s}{r}{A}$ reduces either to a slice of a normal form of $\cutf{\pi}{\rho}{A}$ or to a failure.

\begin{lem}
\label{lem:cut_slice1}
Let $\pi_1$ and $\pi_2$ be proofs such that $\pi_1\betato\pi_2$.
For each $s_2\in\Slices(\pi_2)$, there exists $s_1\in\Slices(\pi_1)$ such that $s_1\betato s_2$ or $s_1 = s_2$.
Reciprocally, for each $s_1\in\Slices(\pi_1)$, $s_1$ reduces to a failure or there exists $s_2\in\Slices(\pi_2)$ such that $s_1\betato s_2$ or $s_1 = s_2$.
\end{lem}
\begin{proof}
We can check that each cut-elimination step respects this property, with the equality case coming from a reduction in $\pi_1\betato\pi_2$ on rules not in the considered slice $s_2$ (\resp\ $s_1$).
\end{proof}

\begin{lem}
\label{lem:cut_slice}
Let $\pi_1$ and $\pi_2$ be proofs whose composition over $A$ reduces to a proof $\tau$.
For each $s\in\Slices(\tau)$, there exist $s_1\in\Slices(\pi_1)$ and $s_2\in\Slices(\pi_2)$ such that $\cutf{s_1}{s_2}{A}$ reduces to $s$.
\end{lem}
\begin{proof}
By induction on the sequence $\cutf{\pi_1}{\pi_2}{A}\betatostar\tau$, using \autoref{lem:cut_slice1}.
\end{proof}

\begin{lem}
\label{lem:com_slice}
Let $\pi_1$ and $\pi_2$ be cut-free proofs of $\fCenter\Gamma$ with $\top$-rules all in $\top/0$-patterns.
Assume that $\pi_1\eqc\pi_2$, where in this sequence there is no rule commutation involving a $\top$-rule.
Then for each slice $s_1\in\Slices(\pi_1)$, there exists a unique $s_2\in\Slices(\pi_2)$ such that $s_1$ and $s_2$ make the same choices for additive connectives in $\Gamma$.
\end{lem}
\begin{proof}
This can be easily checked for each equation in $\eqcone$, save for those involving $\top$.
\end{proof}

Given a choice $\mathcal{C}$ of premises for some additive connectives of a sequent, we say a slice is \emph{on $\mathcal{C}$} if each $\with_i$- and $\oplus_i$-rule in this slice takes premise $i$ for a connective in $\mathcal{C}$ whose chosen premise is $i$.
This concept will be essential when speaking about proof-nets (in \autoref{sec:def_pn} for instance), as it is related to additive resolutions.

\begin{lem}
\label{lem:id_slice}
Given a choice $\mathcal{C}$ of premises for additive connectives of $A$ (but not $A\orth$), there exists a unique slice of $\Slices(\id{A})$ on it, which furthermore makes on $A\orth$ the dual choices of $\mathcal{C}$.
\end{lem}
\begin{proof}
Direct induction on $A$, following the definition of $\id{A}$ on \autoref{tab:ax_exp_pt}.
\end{proof}

We now prove a partial reciprocal to \autoref{lem:cut_slice}.
Notice in the following statement that the assumption is on the composition over $A$, and the conclusion on the existence of a slice for the composition over $B$, the other formula.

\begin{lem}
\label{lem:id_match_pt}
Let $\pi$ and $\pi'$ be cut-free proofs respectively of $\fCenter A\orth,B$ and $\fCenter B\orth,A$, whose composition over $A$ reduces to $\id{B}$ up to rule commutation.
Set $\rho$ a normal form of $\cutf{\pi}{\pi'}{B}$, \ie\ any result of eliminating all $cut$-rules in this proof.
Then, for any slice $s$ of $\pi$, there exists a slice $s'$ of $\pi'$ such that $\cutf{s}{s'}{B}$ reduces to a slice of $\rho$.
\end{lem}
\begin{proof}
Take $s\in\Slices(\pi)$, and denote by $\mathcal{C}$ the choices made in $s$ on $\with$ and $\oplus$ connectives of the formula $B$.
We will use that the composition over the formula $A$ reduces to $\id{B}$ up to $\eqc$ to find $s'\in\Slices(\pi')$ that makes the dual choices of $\mathcal{C}$ on additive connectives of $B\orth$.
This ensures no failure happens during the reduction of $\cutf{s}{s'}{B}$, which thus reduces to a slice of $\rho$ (\autoref{lem:cut_slice1}).

By hypothesis, call $\tau$ a cut-free proof resulting from cut-elimination of $\cutf{\pi}{\pi'}{A}$, with $\tau\eqc\id{B}$.
By \autoref{lem:id_slice}, there is a (unique) slice of $\id{B}$ with choices $\mathcal{C}$ on $B$ and dual choices $\mathcal{C}\orth$ on $B\orth$.
Applying \autoref{prop:eqc_id}, all proofs in the sequence $\tau\eqc\id{B}$ have $\top$-rules with only $0$ in their context; in particular, there cannot be any commutation involving a $\top$-rule in this sequence.
Using \autoref{lem:com_slice}, there is a slice $t$ of $\tau$ with choices $\mathcal{C}$ and $\mathcal{C}\orth$.
According to \autoref{lem:cut_slice}, we have slices $r\in\Slices(\pi)$ and $r'\in\Slices(\pi')$ whose composition on $A$ reduces to $t$.
In particular, $r$ makes choices $\mathcal{C}$ on $B$ and $r'$ choices $\mathcal{C}\orth$ on $B\orth$, as these choices are those in the resulting slice $t$, and no reversed cut-elimination step can modify them (a $\top-cut$ commutative case can erase such choices, but taking it in the other direction can only create choices).
Therefore, $r'$ makes on $B\orth$ the dual choices of $s$ on $B$, and we can take such a slice as $s'$.
\end{proof}

\autoref{lem:id_match_pt} will be used to prove that any non-$ax$-rule $r$ in $\pi$ is not erased in all slices by $\with-\oplus_i$ key cases during normalization.
More precisely, taking a slice $s$ containing $r$, of principal connective in, say $A\orth$, the lemma gives a slice $s'\in\Slices(\pi')$ that has no failure for a composition over $B$.
As $r$ does not introduce a cut formula, the only way for it to be erased during the reduction is by a $\top-cut$ commutative case, in which case its principal formula becomes a sub-formula of the sequent on which a $\top$-rule is applied.
This will be enough to conclude in our cases, as we know that the resulting normal form only has $\top$-rules of the shape {\Rtop{$\top,0$}\DisplayProof} (using \autoref{th:eqbn_eqc} and \autoref{prop:eqc_id}).

\begin{lem}
\label{lem:top_eat_visible}
Let $r$ be a non-$ax$-rule in a cut-free slice $s$ of conclusion $\fCenter \Gamma, A$, with the conclusion sequent of $r$ of the shape $\Sigma_\Gamma,\Sigma_A$ where $\Sigma_\Gamma$ is a sub-sequent of $\Gamma$ and $\Sigma_A$ of $A$.
Also assume the principal formula of $r$ belongs to $\Sigma_\Gamma$, \ie\ is not a sub-formula of $A$.
Take $s'$ a cut-free slice of conclusion $\fCenter A\orth, \Delta$ such that $\cutf{s}{s'}{A}$ reduces to a cut-free slice $s''$.
Then either there is in $s''$ a rule of the same kind as $r$, applied on a sequent $\fCenter \Sigma_\Gamma,\Theta$ with the same principal formula as $r$, or $\Sigma_\Gamma$ is a sub-sequent of a $\top$-rule in $s''$, whose main $\top$-formula is not in $\Sigma_\Gamma$.
\end{lem}
\begin{proof}
By hypothesis, $r$ does not introduce a $cut$-formula, for these formulas are sub-formulas of $A$ or $A\orth$.
Therefore, the only reductions in $\cutf{s}{s'}{A}\betatostar s''$ that can erase $r$ are $\top-cut$ commutative cases, for the $\with_i-\oplus_j$ key cases in the reduction do not lead to a failure.
If no such erasure happens, then we are done: other cut commutative cases involving $r$ may modify its conclusion sequent, but only in the part coming from $A$, thus $\Sigma_\Gamma$ is preserved.

If a $\top-cut$ reduction erases $r$, then $\Sigma_\Gamma$ is in the context of the resulting $\top$-rule $t$, possibly as a sub-sequent.
Furthermore, the principal $\top$-formula of $t$ is not in $\Sigma_\Gamma$.
This is enough to conclude, using an induction on the number of steps in $\cutf{s}{s'}{A}\betatostar s''$.
Notice that $t$ may be erased too during the reduction but, like $r$, this would lead to $\Sigma_\Gamma$ being a sub-sequent of the context of another $\top$-rule.
\end{proof}

\begin{rem}
\label{rem:sub-sequent_property}
Cut-free MALL proofs have the \emph{sub-sequent property}: every sequent in such a proof is a sub-sequent of the conclusion sequent, as every rule respects this property, except the $cut$-rule.
\end{rem}

\subsubsection{Patterns in proofs of isomorphisms}
\label{sec:pattern_dist_seqisos_proofs}

Thanks to the technical results on slices, we can finally deduce that $\top$ and $\bot$-rules in proofs of distributed isomorphisms belong to $\top/0$- and $1/\oplus/\bot$-patterns.

\begin{lem}
\label{lem:eqct_id_top}
If $\isoproofs{A}{B}{\pi}{\pi'}$ with $A$ and $B$ distributed, then all $\top$-rules in $\pi$ and $\pi'$ are in a $\top/0$-pattern.
\end{lem}
\begin{proof}
Consider $t$ a $\top$-rule {\Rtop{$\Gamma,\Delta$}\DisplayProof} in $\pi$, with $\Gamma$ occurrences of sub-formulas of $A\orth$ and $\Delta$ of $B$ (as the cut-free $\pi$ has for conclusion sequent $\fCenter A\orth, B$, see \autoref{rem:sub-sequent_property}).
By symmetry, say the main $\top$-formula of $t$ belongs to $\Gamma$, \ie\ is a sub-formula of $A\orth$.
Call $s$ a slice the $\top$-rule $t$ belongs to.

Set $\rho$ a normal form of $\cutf{\pi}{\pi'}{B}$.
By \autoref{lem:id_match_pt}, there exists a slice $s'\in\Slices(\pi')$ such that $\cutf{s}{s'}{B}$ reduces to a slice $s''\in\Slices(\rho)$.
By \autoref{lem:top_eat_visible} applied to $t$, $t$ is either preserved and $\Gamma$ as well, or $t$ is absorbed by another $\top$-rule (during a $\top-cut$ commutative case) and $\Gamma$ stays in the context of a $\top$-rule, possibly as a sub-sequent.
But by \autoref{th:eqbn_eqc}, $\rho\eqc\id{A}$ and, by \autoref{prop:eqc_id}, the only $\top$-rules of $\rho$, and so of $s''$, are {\Rtop{$\top,0$}\DisplayProof} rules, with $\top$ being the dual occurrence of $0$.
Thus, $\top$ and $0$ are not both sub-formula of $A\orth$, and it follows $\Gamma$ is either a sub-sequent of $\top$ or one of $0$.
For $\Gamma$ contains $\top$, we conclude that $\Gamma$ is $\top$.
Moreover, it follows that $t$ had not been erased during a $\top-cut$ commutative case, using \autoref{lem:top_eat_visible} (otherwise there would be at least two $\top$-formulas in the resulting $\top$-rule).
This implies that $\Delta$ cannot be empty: if it were, $t$ could not commute with any $cut$-rule, as it could not do a $\top-cut$ commutative case for it does not have a cut formula, which is a sub-formula of $B$, in its context, and no other cut-elimination step can change this.
Thus, if $\Delta$ were empty then $t$ would be a rule of $s''$, impossible as it is a {\Rtop{$\top$}\DisplayProof} rule.

Similarly, set $\tau$ a normal form of $\cutf{\pi}{\pi'}{A}$.
By \autoref{lem:id_match_pt} and \autoref{lem:top_eat_visible}, $t$ is either preserved during the reduction (in a slice) and $\Delta$ as well, or $t$ is absorbed by another $\top$-rule and $\Delta$ stays in the context of a $\top$-rule, possibly as a sub-sequent.
But $t$ cannot be preserved: its main formula is a sub-formula of $A\orth$, and in $\tau$ there are only occurrences of sub-formulas of $B$ and $B\orth$.
Therefore, $\Delta$ in $\tau$ is in the context of a $\top$-rule, and does not contain its principal $\top$-formula.
But by \autoref{th:eqbn_eqc} and \autoref{prop:eqc_id}, the only $\top$-rules of $\tau$ are {\Rtop{$\top,0$}\DisplayProof} rules, with $\top$ being the dual occurrence of $0$.
Hence, $\Delta$ must be a sub-sequent of $0$.
As $\Delta$ cannot be empty, $\Delta$ is $0$.

Thus, any $\top$-rule $t$ in $\pi$ (and $\pi'$ by symmetry) is of the shape {\Rtop{$\top,0$}\DisplayProof}.
\end{proof}

\begin{lem}
\label{lem:top_0_strategy}
Let $\pi$ be a proof whose $\top$-rules are all in a $\top/0$-pattern.
There is a cut-elimination strategy in $\pi$ whose $\top-cut$ commutative cases are all of the form:
\begin{center}
\Rtop{$\top,0$}
\Rtop{$\top,0$}
\Rcut{$\top,0$}
\DisplayProof
$\betato$
\Rtop{$\top,0$}
\DisplayProof
\end{center}
\end{lem}
\begin{proof}
Our strategy is the following.
First, while we can apply a cut-elimination step which is not a $\top-cut$ nor a $cut-cut$ commutative case, we do such a reduction step.
These operations preserve that all $\top$-rules are in $\top/0$ patterns.

If no such reduction is possible, consider a highest $cut$-rule, \ie\ one with no $cut$-rule above it.
The only possible cases that can be applied using this $cut$-rule and rules above it are by hypothesis $\top-cut$ commutative cases.
Thus, the $cut$-rule is below a $\top$-rule, so necessarily one of its premises is {\Rtop{$\top,0$}\DisplayProof}, with $\top$ not the formula we cut on.
But then $0$ is the formula we cut on, so there is a $\top$-formula on the other premise; we are in the following situation:
\begin{prooftree}
\Rtop{$\top,0$}
\AxiomC{$\phi$}
\RightLabel{$r$}
\UnaryInfC{$\fCenter\top,\Gamma$}
\Rcut{$\top, \Gamma$}
\end{prooftree}
We prove the rule $r$ above the premise $\fCenter \top, \Gamma$ of the $cut$-rule is the $\top$-rule corresponding to the cut formula $\top$.
If $r$ were not a $\top$-rule, then one could apply a commutative or $ax$ key case, which cannot be.
Plus, $r$ cannot be a $\top$-rule corresponding to another $\top$-formula, because such a rule would have two $\top$-formulas in its conclusion.
Thence, $r$ is the $\top$-rule introducing the formula we cut on, and our sub-proof is:
\begin{prooftree}
\Rtop{$\top,0$}
\Rtop{$\top,0$}
\Rcut{$\top,0$}
\end{prooftree}
A $\top-cut$ commutative case yields
{\Rtop{$\top,0$}\DisplayProof}
which is the allowed $\top-cut$ reduction step.

This reduction strategy terminates (\autoref{lem:beta_wn}) and reaches a cut-free proof.
\end{proof}

\begin{lem}
\label{lem:eqct_id_with}
If $\isoproofs{A}{B}{\pi}{\pi'}$ with $A$ and $B$ distributed, then there is no sequent of the shape $\fCenter D\with E$ in $\pi$ (and $\pi'$).
\end{lem}
\begin{proof}
Assume \wolog\ $D\with E$ is a sub-formula of $A\orth$, and let $s$ be a slice containing the sequent $\fCenter D\with E$.
Pose $\rho$ a normal form of $\cutf{\pi}{\pi'}{B}$ obtained by following the strategy given by \autoref{lem:top_0_strategy}.
By \autoref{lem:id_match_pt}, there exists a slice $s'\in\Slices(\pi')$ such that $\cutf{s}{s'}{B}$ reduces to a slice $s''$ of $\rho$.
Since $D\with E$ is a sub-formula of $A\orth$, it is not a cut formula during the reduction, and the sequent $\fCenter D\with E$ remains in $s''$, so in $\rho$.
Indeed, reducing cuts using these steps preserves having the sequent $\fCenter D\with E$ (as there is no failure in the reduction).
This is trivial for all steps except $\top-cut$.
In the case a $\top-cut$ step, by hypothesis on the reduction strategy, it cannot erase the sequent $\fCenter D\with E$ from the proof.
Thus, the sequent $\fCenter D\with E$ belongs to $\rho$, which is equal to the identity up to rule commutations (\autoref{th:eqbn_eqc}).
This is impossible by \autoref{prop:eqc_id}.
\end{proof}

\begin{lem}
\label{lem:eqct_id_bot}
If $\isoproofs{A}{B}{\pi}{\pi'}$ with $A$ and $B$ distributed, then all $\bot$-rules and $1$-rules belong to $1/\oplus/\bot$-patterns.
\end{lem}
\begin{proof}
In $\pi$, we look at a possible rule $r$ below a sequent $\fCenter \pcontext{1}$ (\autoref{def:oplus_grammar}).
It cannot be a $\tens$-rule by distributivity, nor a $\parr$-rule for the sequent has a unique formula, nor a $\with$-rule due to \autoref{lem:eqct_id_with}.
If $r$ is a $\oplus_i$-rule, then we keep a sequent $\fCenter \pcontext{1}$, and if it is a $\bot$-rule then it is one of the required shape.

As a consequence, each $1$-rule is followed by some $\oplus_i$-rules and possibly a $\bot$-rule; let us call a \emph{$1/\oplus$-pattern} a $1$-rule followed by a maximal such sequence of $\oplus_i$-rules.
If a $1/\oplus$-pattern stops without a $\bot$-rule below it, we have only one formula in the conclusion sequent of the proof: impossible as $\pi$ is a proof of $\fCenter A\orth,B$.
Thus, the $\bot$-rule exists and to each $1$-rule we can associate a $\bot$-rule leading to a $1/\oplus/\bot$-pattern.
Henceforth, there are at least as many $\bot$-rules as $1$-rules, and as the patterns they belong to have no $\with$-rule, this also holds in any slice.

Consider a slice $s$ of $\pi$.
By \autoref{lem:id_match_pt}, there exists a slice $s'\in\Slices(\pi')$ such that $\cutf{s}{s'}{B}$ reduces to a slice $s''$ of $\rho$, the latter being a normal form of $\cutf{\pi}{\pi'}{B}$ obtained by following the strategy given by \autoref{lem:top_0_strategy}.
Moreover, $s''$ contains as many $\bot$-rules as $1$-rules, as in $\rho$ they belong to a $1/\oplus/\bot$-pattern (\autoref{th:eqbn_eqc} and \autoref{prop:eqc_id}), so are in the same slices.
Furthermore, in $s''$, each $1$ from $A$ (\resp\ $A\orth$) corresponds to a $\bot$ from $A\orth$ (\resp\ $A$).

Remark that, in the reduction $\cutf{s}{s'}{B}\betatostar s''$, the only steps that may erase a $\bot$ or $1$-rule are $\bot-1$ and $\top-cut$ cases.
But a $\top-cut$ commutative case cannot erase non-$\top$-rules by definition of our cut-elimination strategy.
Furthermore, a $\bot-1$ key case erases one $1$-rule and one $\bot$-rule.
Therefore, with $r_s$ the number of $r$-rules of a slice $s$, we have $\bot_s + \bot_{s'} = 1_s + 1_{s'}$ as $\bot_{s''} = 1_{s''}$ and any reduction step in $\cutf{s}{s'}{B}\betatostar s''$ preserves this equality.
But, by our analysis at the beginning of this proof, $1_s \leq \bot_s$ and $1_{s'} \leq \bot_{s'}$.
We conclude $1_s = \bot_s$, \ie\ that $s$ has as many $\bot$-rules than $1$-rules.

However, each $1$-rule, being in a $1/\oplus/\bot$-pattern, belongs to exactly the same slices as the corresponding $\bot$-rule of the pattern.
Hence, a $\bot$-rule not in a $1/\oplus/\bot$-pattern would yield a slice $s$ with strictly more $\bot$-rules than $1$-rules (taking $s$ any slice containing this $\bot$-rule).
Thus, every $\bot$-rule belongs to a $1/\oplus/\bot$-pattern.
\end{proof}

\subsection{Completeness with units from unit-free completeness}
\label{subsec:compl_units}

In the $1/\oplus/\bot$-patterns, moving each $\bot$-rule up to the associated $1$-rule (which can be done up to $\beta$-equality by \autoref{prop:eqc_in_eqbn}) allows us to consider units as fresh atoms introduced by $ax$-rules, whence we reduce the problem to the unit-free fragment.

\begin{thm}[Isomorphisms completeness from unit-free completeness]
\label{th:iso_complet_mall}
If $\AC$ is complete for isomorphisms in unit-free distributed MALL, then $\AC$ is complete for isomorphisms in distributed MALL (\ie\ $A\iso B \implies A\eqAC B$ for all distributed $A$ and $B$).
\end{thm}
\begin{proof}
Take $A$ and $B$ distributed formulas in MALL (possibly with units) such that $A\iso B$.
We want $A\eqAC B$.
Using \autoref{lem:iso_isoproofs}, we have $\isoproofs{A}{B}{\pi}{\tau}$.
By Lemmas~\ref{lem:eqct_id_top} and~\ref{lem:eqct_id_bot}, $\pi$ and $\tau$ have $\top$-rules only of the shape {\Rtop{$\top,0$}\DisplayProof} and $\bot$- and $1$-rules in $1/\oplus/\bot$-patterns.
Using $\bot$-commutations to move each $\bot$-rule just below the $1$-rule above it, we build $\pi'$ and $\tau'$ such that $\pi'$ and $\tau'$ have $\top$-rules only of the shape {\Rtop{$\top,0$}\DisplayProof}, $\bot$ and $1$-rules of the form {\Rone\Rbot{$\bot,1$}\DisplayProof}, $\pi\eqc\pi'$ and $\tau\eqc\tau'$.
Whence, $\cutf{\pi'}{\tau'}{B}\eqc\cutf{\pi}{\tau}{B}$ and $\cutf{\pi'}{\tau'}{A}\eqc\cutf{\pi}{\tau}{A}$.
By \autoref{prop:eqc_in_eqbn} and \autoref{th:eqbn_eqc}, for any normal form $\rho$ of $\cutf{\pi'}{\tau'}{B}$ (\resp\ $\cutf{\pi'}{\tau'}{A}$), $\rho\eqc\id{A}$ (\resp\ $\rho\eqc\id{B}$).

We reduce cuts in $\cutf{\pi'}{\tau'}{B}$ (and similarly in $\cutf{\pi'}{\tau'}{A}$) following a particular strategy, ensuring that the proofs obtained during the reduction have $\top$-rules only of the shape {\Rtop{$\top,0$}\DisplayProof}, and $\bot$ and $1$-rules of the form {\Rone\Rbot{$\bot,1$}\DisplayProof}.

First, while we can apply a step of cut-elimination which is not a $\top-cut$, $\bot-cut$, $\bot-1$ or $cut-cut$ case, we do such a reduction step.
These operations preserve that all $\top$-rules are applied on sequents $\fCenter \top, 0$ (up to exchange) and $\bot$- and $1$-rules are in the wished pattern.

If no such reduction is possible, consider a highest $cut$-rule, \ie\ one with no $cut$-rule above it.
The only possible cases that can be applied are $\top-cut$, $\bot-cut$ or $\bot-1$.
\begin{itemize}
\item If a $\top-cut$ commutative case can be applied, then the $cut$-rule is below a $\top$-rule, so necessarily one of its premises is {\Rtop{$\top,0$}\DisplayProof}, with $\top$ not the formula we cut on.
But then $0$ is the formula we cut on, so there is a $\top$-formula on the other premise; we are in the following situation:
\begin{prooftree}
\Rtop{$\top,0$}
\AxiomC{$\phi$}
\RightLabel{$r$}
\UnaryInfC{$\top, \Gamma$}
\Rcut{$\top,\Gamma$}
\end{prooftree}
We prove the rule $r$ above the premise $\fCenter \top,\Gamma$ of the $cut$-rule is the $\top$-rule corresponding to the cut formula $\top$.
If it were not the case, then $r$ commutes with the $cut$-rule.
But $r$ cannot be a $\bot$-rule (which cannot have a $\top$-formula in its context), nor a $\top$-rule corresponding to another $\top$-formula (because such a rule would have two $\top$-formulas in its conclusion).
Thence, $r$ is the $\top$-rule introducing the formula we cut on.
Thus, our sub-proof is:
\begin{prooftree}
\Rtop{$\top,0$}
\Rtop{$\top,0$}
\Rcut{$\top,0$}
\end{prooftree}
A $\top-cut$ commutative case yields
{\Rtop{$\top,0$}\DisplayProof}
as if we had done an $ax$ key case.
\item If a $\bot-cut$ commutative case can be applied, then the $cut$-rule is below a $\bot$-rule, so necessarily one of its premises is \raisebox{0pt}[7mm][0mm]{} {\Rone\Rbot{$\bot,1$}\DisplayProof}, with $\bot$ not the formula we cut on.
But then $1$ is the formula we cut on, so there is a $\bot$-formula on the other premise; we are in the following situation:
\begin{prooftree}
\Rone
\Rbot{$\bot,1$}
\AxiomC{$\phi$}
\RightLabel{$r$}
\UnaryInfC{$\bot, \Gamma$}
\Rcut{$\bot,\Gamma$}
\end{prooftree}
We prove the rule $r$ above the premise $\fCenter \bot,\Gamma$ of the $cut$-rule is the $\bot$-rule corresponding to the cut formula $\bot$.
If it were not the case, then $r$ commutes with the $cut$-rule.
But $r$ cannot be a $\top$-rule (which cannot have a $\bot$-formula in its context), nor a $\bot$-rule corresponding to another $\bot$-formula (because such a rule would have two $\bot$-formulas in its conclusion).
Thence, $r$ is the $\bot$-rule introducing the formula we cut on.
Thus, our sub-proof is:
\begin{prooftree}
\Rone
\Rbot{$\bot,1$}
\Rone
\Rbot{$\bot,1$}
\Rcut{$\bot,1$}
\end{prooftree}
We apply a $\bot-cut$ commutative case, followed by a $\bot-1$ key case, obtaining
{\Rone\Rbot{$\bot,1$}\DisplayProof}
as if we had done an $ax$ key case.
\item No $\bot-1$ key case can be applied, for $1$-rules have below them a $\bot$-rule, so not a $cut$-rule.
\end{itemize}

This strategy of reduction allows reaching a normal form $\rho$, with $\top$-rules only of the shape {\Rtop{$\top,0$}\DisplayProof}, and $\bot$ and $1$-rules of the form {\Rone\Rbot{$\bot,1$}\DisplayProof} (\autoref{lem:beta_wn}).

Furthermore, call $\sigma$ the substitution replacing $\top$-, $0$-, $\bot$- and $1$-formulas respectively by $X\orth$, $X$, $Y\orth$ and $Y$, for $X$ and $Y$ fresh atoms.
We can reach $\sigma(\rho)$ by cut-elimination from $\cutf{\sigma(\pi')}{\sigma(\tau')}{\sigma(B)}$, for the reductions we did on units could as well have been done by $ax$ key cases: no $\bot-cut$, nor $\bot-1$, nor $\top-cut$ case was used, except for cases that could be simulated using $ax$ key cases.
Moreover, $\sigma(\id{A})=\id{\sigma(A)}$, and in $\rho\eqc\id{A}$ we can assume not to commute any $\bot$-rule because we start and end with $1$-rules and $\bot$-rules in {\Rone\Rbot{$\bot,1$}\DisplayProof} shapes only, and such commutations could only move the $\bot$-rule below or above some $\oplus_i$-rules according to \autoref{prop:eqc_id}.
Thus, $\sigma(\rho)\eqc\id{\sigma(A)}$.
Using \autoref{prop:eqc_in_eqbn}, it follows $\sigma(\rho)\eqb\id{\sigma(A)}$, and therefore $\cutf{\sigma(\pi')}{\sigma(\tau')}{\sigma(B)}\eqbn\ax{\sigma(A)}$ with $\ax{\sigma(A)}$ the $ax$-rule on $\sigma(A)$.
A similar result holding for a cut over $A$, we have $\sigma(A)\iso\sigma(B)$, these formulas being unit-free and distributed.
For we assume $\AC$ to be complete for unit-free distributed isomorphisms, this yields $\sigma(A)\eqAC\sigma(B)$.
We conclude $A\eqAC B$ by substituting $X$ by $0$ and $Y$ by $1$, as $X$ and $Y$ were fresh.
\end{proof}

%==================================================

\section{Proof-nets for unit-free MALL}
\label{subsec:mall-pn}

We will at present change the syntax, no longer considering sequent calculus proofs but \emph{proof-nets}, choosing the syntax from Hughes \& Van~Glabbeek for unit-free MALL~\cite{mallpnlong}.
A key property of proof-nets is to be a more canonical representation of proofs, as they define a quotient of sequent calculus proofs up to rule commutations~\cite{mallpncom} (we recall Tables~\ref{tab:rule_comm} and~\ref{tab:rule_comm_u} on Pages~\pageref{tab:rule_comm} and~\pageref{tab:rule_comm_u} give rule commutations).\footnote{While the definition of rule commutations in~\cite{mallpncom} differs a little from ours, for they also consider commutations involving a $cut$-rule, both definitions coincide on cut-free proof-nets, which are the objects we want canonical.}
This yields better properties concerning the study of isomorphisms.
As a key example, cut-elimination in proof-nets is confluent and leads to a unique normal form.
This spares us tedious case studies on rule commutations -- like the one in the proof of \autoref{prop:eqc_id}, which was due to the need to relate the different possible cut-free proofs obtained by cut-elimination.
Thus, studying isomorphisms using proof-nets will be less complex.
Nonetheless, we leave an inductive definition for a graphical syntax.
It would have been ideal to use this syntax from the very beginning.
Unfortunately, this was not possible as no notion of proof-nets exists with units and exact quotient on normal forms.

Other definitions of proof-nets exist, see the original one from Girard~\cite{pn}, or others such as~\cite{jumpadd,conflictnet}. Still, the definition we take is one of the most satisfactory, from the point of view of canonicity and cut-elimination for instance (see~\cite{mallpnlong,mallpncom}, or the introduction of~\cite{conflictnet} for a comparison of alternative definitions).
We recall here this definition of proof-nets, and define composition by cut and cut-elimination for this syntax.
Please refer to~\cite{mallpnlong} for more details, as well as the intuitions behind the definition.
In all that follows, as we use proof-nets, we consider only the unit-free fragment of MALL, unless stated otherwise.

\subsection{Proof-net}
\label{sec:def_pn}

A sequent is seen as its syntactic forest, with as internal vertices its connectives and as leaves the atoms of its formulas.
We always identify a formula $A$ with its syntactic tree $T(A)$.
A \emph{cut pair} is a formula $A\ast A\orth$, given a formula $A$; the connective $\ast$ is unordered.
A \emph{cut sequent} $[\Sigma]~\Gamma$ is composed of a list $\Sigma$ of cut pairs and a sequent $\Gamma$. When $\Sigma=\emptyset$ is empty, we denote it simply by $\Gamma$.
For instance, $[X_5\ast X_6\orth]~X_1\with X_2\orth, X_3\oplus X_4\orth$ (where each $X_i$ is an occurrence of the same atom $X$) is a cut sequent, on which we will instantiate the concepts defined in this part.

An \emph{additive resolution} of a cut sequent $[\Sigma]~\Gamma$ is any result of deleting zero or more cut pairs from $\Sigma$ and one argument sub-tree of each additive connective ($\with$ or $\oplus$) of $\Sigma\cup\Gamma$.
A \emph{$\with$-resolution} of a cut sequent $[\Sigma]~\Gamma$ is any result of deleting one argument sub-tree of each $\with$-connective of $\Sigma\cup\Gamma$.
For example, $[~]~X_1\with~, ~\oplus X_4\orth$ is one of the eight additive resolutions of $[X_5\ast X_6\orth]~X_1\with X_2\orth, X_3\oplus X_4\orth$, while $[X_5\ast X_6\orth]~X_1\with~, X_3\oplus X_4\orth$ is one of its two $\with$-resolutions.
Notice the difference on cut pairs: they may be deleted in an additive resolution, but never in a $\with$-resolution.

An \emph{(axiom) link} on $[\Sigma]~\Gamma$ is an unordered pair of complementary leaves in $\Sigma\cup\Gamma$ (labeled with $X$ and $X\orth$ for some atom $X$).
A \emph{linking} $\lambda$ on $[\Sigma]~\Gamma$ is a set of disjoint links on $[\Sigma]~\Gamma$ respecting the following property:
the set made of the leaves of the axiom links of $\lambda$ is the set of leaves of an additive resolution of $[\Sigma]~\Gamma$;
this (unique) additive resolution is denoted $[\Sigma]~\Gamma\upharpoonright\lambda$.
For instance, on the left-most graph of \autoref{fig:ex_pn}, the red axiom links form a linking $\lambda_1 = \{(X_1,X_6\orth);(X_4\orth,X_5)\}$, whose additive resolution is $[X_5\ast X_6\orth]~X_1\with X_2\orth, X_3\oplus X_4\orth\upharpoonright\lambda_1 = [X_5\ast X_6\orth]~X_1\with~, ~\oplus X_4\orth$ (which is the unique additive resolution associated to the set of leaves $\{X_1;X_4\orth;X_5;X_6\orth\}$).

Note the relation between slices in the sequent calculus and linkings in proof-nets: a slice ``belongs'' to an additive resolution, and a $\with$-resolution ``selects'' a slice from a proof.

A set of linkings $\Lambda$ on $[\Sigma]~\Gamma$ \emph{toggles} a $\with$-vertex $W$ if both arguments (called \emph{premises}) of $W$ are in ${[\Sigma]~\Gamma\upharpoonright\Lambda\coloneqq\bigcup_{\lambda\in\Lambda}[\Sigma]~\Gamma\upharpoonright\lambda}$.
We say a link $a$ \emph{depends} on a $\with$-vertex $W$ in $\Lambda$ if there exist $\lambda,\lambda'\in\Lambda$ such that $a\in\lambda\backslash\lambda'$ and $W$ is the only $\with$-vertex toggled by $\{\lambda;\lambda'\}$.
Looking at our running example, and taking $\lambda_1=\{(X_1, X_6\orth), (X_4\orth, X_5)\}$ and $\lambda_2=\{(X_2\orth, X_3)\}$, the $\with$-vertex is toggled by $\{\lambda_1;\lambda_2\}$.
Furthermore, all links depend on this $\with$-vertex for $\lambda_1$ and $\lambda_2$ contain only different pairs.

The graph $\G_\Lambda$ is defined as $[\Sigma]~\Gamma\upharpoonright\Lambda$ with the edges from $\bigcup\Lambda$ and enriched with jump edges $l\rightarrow W$ for each leaf $l$ and each $\with$-vertex $W$ such that there exists $a\in\lambda\in\Lambda$, between $l$ and some $l'$, with $a$ depending on $W$ in $\Lambda$.
When $\Lambda=\{\lambda\}$ is composed of a single linking, we shall simply denote $\G_\lambda=\G_{\{\lambda\}}$ (which is the graph $[\Sigma]~\Gamma\upharpoonright\lambda$ with the edges from $\lambda$ and no jump edge).
For our example, the graphs $\G_{\lambda_1}$, $\G_{\lambda_2}$ and $\G_{\{\lambda_1;\lambda_2\}}$ are illustrated on \autoref{fig:ex_pn}.

In the text of this paper (but not on the graphs), we write $l\jump W$ for a jump edge from a leaf $l$ to a $\with$-vertex $W$.
When drawing proof-nets, we will denote membership of a linking by means of colors.

When we write a $\pw$-vertex, we mean a $\parr$- or $\with$-vertex (a \emph{negative} vertex); similarly a $\tp$-vertex is a $\tens$- or $\oplus$-vertex (a \emph{positive} vertex).
A \emph{switch edge} of a $\pw$-vertex $N$ is an in-edge of $N$, \ie\ an edge between $N$ and one of its premises or a jump to $N$.
A \emph{switching cycle} is a cycle with at most one switch edge of each $\pw$-vertex.
A \emph{$\parr$-switching} of a linking $\lambda$ is any subgraph of $\G_\lambda$ obtained by deleting a switch edge of each $\parr$-vertex; denoting by $\phi$ this choice of edges, the subgraph it yields is $\G_\phi$.
For example, a cycle in a $\parr$-switching is a switching cycle, as in a graph $\G_\phi$ all $\pw$-vertices have one premise.

\begin{defi}[Proof-net]
\label{def:pn}
A \emph{unit-free MALL proof-net} $\theta$ on a cut sequent $[\Sigma]~\Gamma$ is a set of linkings satisfying:
\begin{itemize}[left=\parindent]
\item[\textbf{(P0)}] \emph{Cut:} Every cut pair of $\Sigma$ has a leaf in $\theta$.
\item[\textbf{(P1)}] \emph{Resolution:} Exactly one linking of $\theta$ is on any given $\with$-resolution of $[\Sigma]~\Gamma$.
\item[\textbf{(P2)}] \emph{MLL:} For every $\parr$-switching $\phi$ of every linking $\lambda\in\theta$, $\G_\phi$ is a tree.
\item[\textbf{(P3)}] \emph{Toggling:} Every set $\Lambda\subseteq\theta$ of two or more linkings toggles a $\with$-vertex that is in no switching cycle of $\G_\Lambda$.
\end{itemize}
\end{defi}

These conditions are called the \emph{correctness criterion}.
Condition (P0) is here to prevent unused $\ast$-vertices.
A \emph{cut-free} proof-net is one without $\ast$-vertices (it respects (P0) trivially).
Condition (P1) is a correctness criterion for additive proof-nets~\cite{mallpnlong} and (P2) is the Danos-Regnier criterion for multiplicative proof-nets~\cite{structmult}.
However, (P1) and (P2) together are insufficient for cut-free MALL proof-nets, hence the last condition (P3) taking into account interactions between the slices (see also~\cite{jumpadd} for a similar condition for example).
Sets composed of a single linking $\lambda$ are not considered in (P3), for by (P2) the graph $\G_\lambda$ has no switching cycle.
%We say that $\theta$ is a \emph{proof structure} if it satisfies (P0) and (P1).
One can check that our example on \autoref{fig:ex_pn}, $\{\lambda_1;\lambda_2\}$, is a proof-net.

\begin{figure}
\centering
\resizebox{.85\width}{!}{
\begin{tikzpicture}
\begin{myscope}
	\node (P0) at (-1,1) {$X_1$};
	\node (P3) at (1.5,1) {$X_4\orth$};
	\node (W) at (-.25,0) {$\with$};
	\node (O) at (.75,0) {$\oplus$};
	\node (C) at (3.25,0) {$\ast$};
	\node (P4) at (2.5,1) {$X_5$};
	\node (P5) at (4,1) {$X_6\orth$};

	\path (P0) edge[-] (W);
	\path (P3) edge[-] (O);
	\path (P4) edge[-] (C);
	\path (P5) edge[-] (C);
\end{myscope}
\begin{myscopec}{red}
	\path (P0) -- ++(0,.8) -| (P5);
	\path (P3) -- ++(0,.6) -| (P4);
\end{myscopec}
\end{tikzpicture}
}
\resizebox{.85\width}{!}{
\begin{tikzpicture}
\begin{myscope}
	\node (P1) at (-.5,1) {$X_2\orth$};
	\node (P2) at (.5,1) {$X_3$};
	\node (W) at (-1.25,0) {$\with$};
	\node (O) at (1.25,0) {$\oplus$};

	\path (P1) edge[-] (W);
	\path (P2) edge[-] (O);
\end{myscope}
\begin{myscopec}{blue}
	\path (P1) -- ++(0,-1.1) -| (P2);
\end{myscopec}
\end{tikzpicture}
}
\resizebox{.85\width}{!}{
\begin{tikzpicture}
\begin{myscope}
	\node (P0) at (-1.75,1) {$X_1$};
	\node (P1) at (-.25,1) {$X_2\orth$};
	\node (P2) at (.75,1) {$X_3$};
	\node (P3) at (2.25,1) {$X_4\orth$};
	\node (W) at (-1,0) {$\with$};
	\node (O) at (1.5,0) {$\oplus$};
	\node (C) at (4,0) {$\ast$};
	\node (P4) at (3.25,1) {$X_5$};
	\node (P5) at (4.75,1) {$X_6\orth$};

	\path (P0) edge[-] (W);
	\path (P1) edge[-] (W);
	\path (P2) edge[-] (O);
	\path (P3) edge[-] (O);
	\path (P4) edge[-] (C);
	\path (P5) edge[-] (C);
	\path[out=225,in=0] (P2) edge (W);
	\coordinate (c3) at (1.5,.65);
	\path[out=220,in=0] (P3) edge[-] (c3);
	\path[out=180,in=0] (c3) edge (W);
	\coordinate (c4) at (1.5,.55);
	\path[out=215,in=0] (P4) edge[-] (c4);
	\path[out=180,in=0] (c4) edge (W);
	\coordinate (c5) at (1.5,.45);
	\path[out=205,in=0] (P5) edge[-] (c5);
	\path[out=180,in=0] (c5) edge (W);
\end{myscope}
\begin{myscopec}{blue}
	\path (P1) -- ++(0,-1.1) -| (P2);
\end{myscopec}
\begin{myscopec}{red}
	\path (P0) -- ++(0,.8) -| (P5);
	\path (P3) -- ++(0,.6) -| (P4);
\end{myscopec}
\end{tikzpicture}
}
\caption{Graphs from an example of a proof-net: from left to right $\G_{\lambda_1}$, $\G_{\lambda_2}$ and $\G_{\{\lambda_1;\lambda_2\}}$}
\label{fig:ex_pn}
\end{figure}

In the particular setting of isomorphisms, we mainly consider proof-nets with two conclusions.
This allows to define a notion of duality on leaves and connectives.
Consider a cut sequent containing both $A$ and $A\orth$.
For $V$ a vertex in (the syntactic tree $T(A)$ of) $A$, we denote by $V\orth$ the corresponding vertex in $A\orth$.
As expected, $V\biorth=V$.
This also respects orthogonality for formulas on leaves: given a leaf $l$ of $A$, labeled by a formula $X$, the label of $l\orth$ is $X\orth$.
We can also define a notion of duality on premises:
given a premise of a vertex $V\in T(A)$, the dual premise of $V\orth$ is the corresponding premise in $T(A\orth)$.
In other words, if in $L-V-R$ we consider the premise $L$ then in $R\orth-V\orth-L\orth$ its dual premise is $L\orth$.

\subsection{Cut-elimination in proof-nets}
\label{sec:ce_pn}

\begin{defi}[Composition]
\label{def:cut}
For proof-nets $\theta$ and $\psi$ of respective conclusions $[\Sigma]~\Gamma,A$ and $[\Xi]~\Delta,A\orth$, the \emph{composition} over $A$ of $\theta$ and $\psi$ is the proof-net ${\cutf{\theta}{\psi}{A}=\{\lambda\cup\mu\ |\ \lambda\in\theta,\mu\in\psi\}}$, with conclusions $[\Sigma,\Xi,A\ast A\orth]~\Gamma,\Delta$.
\end{defi}

For example, see \autoref{fig:dist_not_uniq_comp} with a composition of the proof-nets on \autoref{fig:dist_not_uniq}.

\begin{defi}[Cut-elimination]
\label{def:cut_elim}
Let $\theta$ be a set of linkings on a cut sequent $[\Sigma]~\Gamma$, and $A\ast A\orth$ a cut pair in $\Sigma$.
Define the \emph{elimination} of $A\ast A\orth$ (or of the cut $\ast$ between $A$ and $A\orth$) as:
\begin{enumerate}[(a)]
\item If $A$ is an atom, delete $A\ast A\orth$ from $\Sigma$ and replace any pair of links $(l,A)$, $(A\orth,m)$ ($l$ and $m$ being other occurrences of $A\orth$ and $A$ respectively) with the link $(l,m)$.
\item If $A=A_1\tens A_2$ and $A\orth=A_2\orth\parr A_1\orth$ (or vice-versa), replace $A\ast A\orth$ with two cut pairs $A_1\ast A_1\orth$ and $A_2\ast A_2\orth$. Retain all original linkings.
\item If $A=A_1\with A_2$ and $A\orth=A_2\orth\oplus A_1\orth$ (or vice-versa), replace $A\ast A\orth$ with two cut pairs $A_1\ast A_1\orth$ and $A_2\ast A_2\orth$. Delete all \emph{inconsistent} linkings, namely those $\lambda\in\theta$ such that in $[\Sigma]~\Gamma\upharpoonright\lambda$ the children $\with$ and $\oplus$ of the cut do not take dual premises.
Finally, ``garbage collect'' by deleting any cut pair $B\ast B\orth$ for which no leaf of $B\ast B\orth$ is in any of the remaining linkings.
\end{enumerate}
\end{defi}

See \autoref{fig:dist_not_uniq_comp_red} for a result on applying steps (b) and (c) to the proof-net of \autoref{fig:dist_not_uniq_comp}.
We use for proof-nets the same notations $\betato$ and $\eqb$ as for the sequent calculus.

\begin{propC}[{\cite[Proposition~5.4]{mallpnlong}}]
\label{prop:cut_elim_correct}
Eliminating a cut in a proof-net yields a proof-net.
\end{propC}

\begin{thmC}[{\cite[Theorem~5.5]{mallpnlong}}]
\label{th:cut_sn}
Cut-elimination of proof-nets is strongly normalizing and confluent.
\end{thmC}

A linking $\lambda$ on a cut sequent $[\Sigma]~\Gamma$ \emph{matches} if, for every cut pair $A\ast A\orth$ in $\Sigma$, any given leaf $l$ of $A$ is in $[\Sigma]~\Gamma\upharpoonright\lambda$ if and only if $l\orth$ of $A\orth$ is in $[\Sigma]~\Gamma\upharpoonright\lambda$.
A linking matches if and only if, when cut-elimination is carried out, the linking never becomes inconsistent, and thus is never deleted.
This allows defining \emph{Turbo Cut-elimination}~\cite{mallpnlong}, eliminating a cut in a single step by removing inconsistent linkings.

%==================================================

\section{Reduction to proof-nets}
\label{sec:red_pn}

The goal of this section is to shift the study of isomorphisms to the syntax of proof-nets, never to speak of sequent calculus again.
We do so by first defining \emph{desequentialization}, a function from sequent calculus proofs to proof-nets (\autoref{subsec:deseq}).
We then show that cut-elimination in proof-nets simulates the one from sequent calculus (\autoref{subsec:simulation_th}).
Finally, we define isomorphisms directly in the syntax of proof-nets (\autoref{subsec:iso_in_pn}).
Recall that all sequent calculus proofs we consider have expanded axioms, thanks to \autoref{prop:eta_nom_eqbn}.

\subsection{Desequentialization}
\label{subsec:deseq}

We desequentialize a unit-free MALL proof $\pi$ (with expanded axioms) into a set of linkings $\Rst(\pi)$ by induction on $\pi$:

\begin{center}
% ax rule
\AxiomC{}
\RightLabel{$ax$}
\UnaryInfC{$\{\{(X, X\orth)\}\}\triangleright [\emptyset]~X,X\orth$}
\DisplayProof
\hskip 2em
% ex rule
\AxiomC{$\theta\triangleright [\Sigma]~\Gamma$}
\RightLabel{$ex$}
\UnaryInfC{$\theta\triangleright [\Sigma]~\sigma(\Gamma)$}
\DisplayProof
\vskip 1em
% cut rule
\AxiomC{$\theta\triangleright [\Sigma]~A,\Gamma$}
\AxiomC{$\psi\triangleright [\Xi]~A\orth,\Delta$}
\RightLabel{$cut$}
\BinaryInfC{$\{\lambda\cup\mu~|~\lambda\in\theta,\mu\in\psi\}\triangleright [\Sigma,\Xi,A\ast A\orth]~\Gamma,\Delta$}
\DisplayProof
\vskip 1em
% tens rule
\AxiomC{$\theta\triangleright [\Sigma]~A,\Gamma$}
\AxiomC{$\psi\triangleright [\Xi]~B,\Delta$}
\RightLabel{$\tens$}
\BinaryInfC{$\{\lambda\cup\mu~|~\lambda\in\theta,\mu\in\psi\}\triangleright [\Sigma,\Xi]~A\tens B,\Gamma,\Delta$}
\DisplayProof
\hskip 2em
% parr rule
\AxiomC{$\theta\triangleright [\Sigma]~A,B,\Gamma$}
\RightLabel{$\parr$}
\UnaryInfC{$\theta\triangleright [\Sigma]~A\parr B,\Gamma$}
\DisplayProof
\vskip 1em
% with rule
\AxiomC{$\theta\triangleright [\Xi]~A,\Gamma$}
\AxiomC{$\psi\triangleright [\Phi]~B,\Gamma$}
\RightLabel{$\with$}
\BinaryInfC{$\theta\cup\psi\triangleright [\Xi,\Phi]~A\with B,\Gamma$}
\DisplayProof
\hskip 2em
% oplus_1 rule
\AxiomC{$\theta\triangleright [\Sigma]~A,\Gamma$}
\RightLabel{$\oplus_1$}
\UnaryInfC{$\theta\triangleright [\Sigma]~A\oplus B,\Gamma$}
\DisplayProof
\hskip 2em
% oplus_2 rule
\AxiomC{$\theta\triangleright [\Sigma]~B,\Gamma$}
\RightLabel{$\oplus_2$}
\UnaryInfC{$\theta\triangleright [\Sigma]~A\oplus B,\Gamma$}
\DisplayProof
\end{center}

This definition uses the implicit tracking of formula occurrences downwards through the rules, and follows~\cite{mallpnlong} with the notation $\theta\triangleright [\Sigma]~\Gamma$ for ``$\theta$ is a set of linkings on the cut sequent $[\Sigma]~\Gamma$''.
As identified in~\cite[Section~5.3.4]{mallpnlong}, desequentializing with both $cut$- and $\with$-rules is complex, for cuts can be shared (or not) when translating a $\with$-rule:
{
\AxiomC{$\theta\triangleright [\Sigma,\Xi]~A,\Gamma$}
\AxiomC{$\psi\triangleright [\Sigma,\Phi]~B,\Gamma$}
\RightLabel{$\with$}
\BinaryInfC{$\theta\cup\psi\triangleright [\Sigma,\Xi,\Phi]~A\with B,\Gamma$}
\DisplayProof
}.
We choose to never share cuts (${\Sigma=\emptyset}$), thus desequentialization is a function.
The cost being that the following $\with-cut$ commutation yields different proof-nets (contrary to the other commutations, see~\cite{mallpncom}).
\resizebox{\textwidth}{!}{
\Rsub{$\pi_1$}{$A,B,\Gamma$}
\Rsub{$\pi_2$}{$A,C,\Gamma$}
\Rwith{$A,B\with C,\Gamma$}
\Rsub{$\pi_3$}{$A\orth,\Delta$}
\Rcut{$B\with C,\Gamma,\Delta$}
\DisplayProof
\hskip 1em $\equiv$ \hskip 1em
\Rsub{$\pi_1$}{$A,B,\Gamma$}
\Rsub{$\pi_3$}{$A\orth,\Delta$}
\Rcut{$B,\Gamma,\Delta$}
\Rsub{$\pi_2$}{$A,C,\Gamma$}
\Rsub{$\pi_3$}{$A\orth,\Delta$}
\Rcut{$C,\Gamma,\Delta$}
\Rwith{$B\with C,\Gamma,\Delta$}
\DisplayProof
}

\begin{rem}
An alternative definition of desequentialization in~\cite{mallpnlong} consists in building a linking by slice.
In this spirit, if a proof-net $\theta$ is obtained by desequentializing a proof $\pi$, there is a bijection between linkings in $\theta$ and slices of $\pi$.
\end{rem}

\begin{thm}[Sequentialization~{\cite[Theorem 5.9]{mallpnlong}}]
\label{th:seq}
A set of linkings on a cut sequent is a translation of a unit-free MALL proof if and only if it is a proof-net.
\end{thm}

\begin{defi}[Identity proof-net]
\label{def:id_pn}
We call \emph{identity proof-net} of a unit-free formula $A$, the proof-net corresponding to the proof $\id{A}$ (the axiom-expansion of $\ax{A}=${
\Rax{$A$}
\DisplayProof
}).
\end{defi}

\subsection{Simulation of cut elimination}
\label{subsec:simulation_th}

We show here that cut elimination in proof-nets mimics the one in sequent calculus, which will allow us in the next section to consider isomorphisms on proof-nets only.
As written in \autoref{subsec:deseq}, proof-nets have difficulties with the $\with-cut$ commutation, which corresponds to superimposing $\ast$-vertices.

\begin{defi}[$\dupcut$]
\label{def:dupcut}
Let $\theta$ and $\psi$ be proof-nets.
We denote $\theta\dupcut\psi$ if there exists a $\ast$-vertex $C$ in $\theta$ such that the syntactic forest of $\psi$ is the syntactic forest of $\theta$ where the syntactic tree of $C$ is duplicated into the syntactic trees of $C_0$ and $C_1$ (which are different occurrences of $C$), $\theta=\theta_0\sqcup\theta_1$\footnote{The symbol $\sqcup$ means a union $\cup$ which happens to be between disjoint sets.} and $\psi=\psi_0\sqcup\psi_1$ with, for $i\in\{0;1\}$, $\psi_i=\theta_i$ up to assimilating $C_i$ with $C$.
\end{defi}

See \autoref{fig:dupcut} for a graphical representation of this concept, as well as the link with the $\with-cut$ commutative case of cut-elimination.

\begin{figure}
\begin{adjustbox}{}
\begin{tikzpicture}
	\node at (-3.5,.5) {\phantom{x}};
\begin{myscope}
	\node (wl) at (-2,.5) {$\with$};
	\node (C) at (.7,-.5) {$\ast$};
\end{myscope}
	\draw (-2.5,1.5) edge[draw=red,->] (wl);
	\draw (-1.5,1.5) edge[draw=blue,->] (wl);
	\node at (1.45,-.5) {$C$};
	\coordinate (l) at (-1.4,3);
	\coordinate (r) at (2.8,3);
	\draw (l) edge (r);
	\draw (l) edge (C);
	\draw (r) edge (C);
\begin{myscope}
	\node (X1) at (-1.8,2.55) {$X\orth$};
	\node (X2) at (-0.65,2.55) {$X$};
	\node (X3) at (.25,2.55) {$Y$};
	\node (X4) at (1.15,2.55) {$Y\orth$};
	\node (X5) at (2.05,2.55) {$X$};
	\node (X6) at (3.1,2.55) {$X\orth$};
	\path[draw=red] (X1) -- ++(0,.6) -| (X2);
	\path[draw=blue] (X1) -- ++(0,.8) -| (X2);
	\path[draw=red] (X3) -- ++(0,.6) -| (X4);
	\path[draw=blue] (X5) -- ++(0,.6) -| (X6);
\end{myscope}

	\node at (4.4,1) {$\dupcut$};

\begin{myscope}
	\node (wr) at (5.5,.5) {$\with$};
	\node (C0) at (8.2,-.5) {$\ast$};
	\node (C1) at (12.6,-.5) {$\ast$};
\end{myscope}
	\draw (5,1.5) edge[draw=red,->] (wr);
	\draw (6,1.5) edge[draw=blue,->] (wr);
	\node at (8.95,-.5) {$C_0$};
	\node at (13.35,-.5) {$C_1$};
	\coordinate (l0) at (6.1,3);
	\coordinate (r0) at (10.3,3);
	\draw (l0) edge (r0);
	\draw (l0) edge (C0);
	\draw (r0) edge (C0);
	\coordinate (l1) at (10.5,3);
	\coordinate (r1) at (14.7,3);
	\draw (l1) edge (r1);
	\draw (l1) edge (C1);
	\draw (r1) edge (C1);
\begin{myscope}
	\node (X1') at (5.7,2.55) {$X\orth$};
	\node (X2'0) at (6.85,2.55) {$X$};
	\node (X3'0) at (7.75,2.55) {$Y$};
	\node (X4'0) at (8.65,2.55) {$Y\orth$};
	\node (X5'0) at (9.55,2.55) {$X$};
	\node (X2'1) at (11.25,2.55) {$X$};
	\node (X3'1) at (12.15,2.55) {$Y$};
	\node (X4'1) at (13.05,2.55) {$Y\orth$};
	\node (X5'1) at (13.95,2.55) {$X$};
	\node (X6') at (15,2.55) {$X\orth$};
	\path[draw=red] (X1') -- ++(0,.6) -| (X2'0);
	\path[draw=blue] (X1') -- ++(0,.8) -| (X2'1);
	\path[draw=red] (X3'0) -- ++(0,.6) -| (X4'0);
	\path[draw=blue] (X5'1) -- ++(0,.6) -| (X6');
\end{myscope}
\end{tikzpicture}
\end{adjustbox}
\vskip1em
\begin{adjustbox}{}
\AxiomC{$\fCenter A,B,\Gamma$}
\AxiomC{$\fCenter A,D,\Gamma$}
\Rwith{$A,B\with D,\Gamma$}
\AxiomC{$\fCenter A\orth,\Delta$}
\Rcut{$\fCenter B\with D,\Gamma,\Delta$}
\DisplayProof
$\betato$
\AxiomC{$\fCenter A,B,\Gamma$}
\AxiomC{$\fCenter A\orth,\Delta$}
\Rcut{$\fCenter B,\Gamma,\Delta$}
\AxiomC{$\fCenter A,D,\Gamma$}
\AxiomC{$\fCenter A\orth,\Delta$}
\Rcut{$\fCenter D,\Gamma,\Delta$}
\Rwith{$B\with D,\Gamma,\Delta$}
\DisplayProof
\end{adjustbox}
\caption{Illustration of $\dupcut$ (\autoref{def:dupcut}) and correspondence with $\with-cut$ cut-elimination}
\label{fig:dupcut}
\end{figure}

\begin{lem}[Simulation - $\beta$]
\label{lem:simulation_beta}
Let $\pi$ and $\pi'$ be unit-free MALL proof trees such that $\pi\betato\pi'$.
Then either $\Rst(\pi)=\Rst(\pi')$, $\Rst(\pi)\dupcut\Rst(\pi')$ or $\Rst(\pi)\betato\Rst(\pi')$.
\end{lem}
\begin{proof}
We reason by cases according to the step $\pi\betato\pi'$.
Recall we desequentialize by separating all cuts, and use the notations for steps from \autoref{def:cut_elim}.
If $\pi\betato\pi'$ is an $ax$ (\resp\ $\parr-\tens$, $\with-\oplus$) key case, then using a step (a) (\resp\ (b), (c)), we get $\Rst(\pi)\betato\Rst(\pi')$.
If it is a $\parr-cut$, $\tens-cut-1$, $\tens-cut-2$, $\oplus_1-cut$ or $\oplus_2-cut$ commutative case, then $\Rst(\pi)=\Rst(\pi')$.
Finally, in a $\with-cut$ commutative case, we duplicate the $cut$-rule: $\Rst(\pi)\dupcut\Rst(\pi')$ (see \autoref{fig:dupcut}).
\end{proof}

Nonetheless, the $\dupcut$ relation is not a hard problem since two proofs differing by a $\with-cut$ commutation yield proof-nets equal up to cut-elimination.

\begin{lem}[$\dupcut~\subseteq~\eqb$]
\label{lem:dupcut_eqb}
Let $\theta$ and $\theta'$ be proof-nets such that $\theta\dupcut\theta'$.
Then $\theta\eqb\theta'$.
\end{lem}
\begin{proof}
By \autoref{def:dupcut} of $\dupcut$, there exists a $\ast$-vertex $C$ in $\theta$, with $\theta=\theta_0\sqcup\theta_1$, such that $\theta'$ is $\theta$ where the syntactic tree of $C$ is duplicated into $C_0$ and $C_1$, and linkings in $\theta_0$ (\resp\ $\theta_1$) use $C_0$ (\resp\ $C_1$) as $C$.

We reason by induction on the formula $A$ of $C$ (and also $C_0$ and $C_1$); \wolog\ $A$ is positive.
Applying a step of cut-elimination on $C$ in $\theta$ yields a proof-net $\Theta$.
On the other hand, a corresponding step of cut-elimination on $C_0$ and one on $C_1$ in $\theta'$ yields $\Theta'$.

If $A$ is an atom, then we applied step (a), and we find $\Theta=\Theta'$.

If $A$ is a $\tens$-formula, \ie\ $A=A_0\tens A_1$, then we applied step (b) and produced cuts $A_0\ast A_0\orth$ and $A_1\ast A_1\orth$ in $\Theta$, and two occurrences of these cuts in $\Theta'$.
Thus, $\Theta\dupcut\Xi\dupcut\Theta'$ with $\Xi$ the proof-net $\Theta$ where the cut on $A_0$ is duplicated.
By induction hypothesis, $\Theta\eqb\Xi\eqb\Theta'$.
It follows $\theta\eqb\theta'$ as $\theta\betato\Theta\eqb\Theta'\betafrom\cdot\betafrom\theta'$.

Finally, if $A$ is a $\oplus$-formula with $A=A_0\oplus A_1$, then we used step (c), producing cuts $A_0\ast A_0\orth$ and $A_1\ast A_1\orth$ in $\Theta$, and two occurrences of these cuts in $\Theta'$.
Remark that inconsistent linkings in $\theta'$ for these steps are exactly those of $\theta$, and therefore the same cuts are garbage collected.
Whence, $\Theta\dupcut\cdot\dupcut\Theta'$, $\Theta\dupcut\Theta'$ or $\Theta=\Theta'$ (according to the number of cuts garbage collected).
In all cases, using the induction hypothesis we conclude $\theta\eqb\theta'$.
\end{proof}

\begin{rem}
\label{rem:turbo_cut}
Another proof of \autoref{lem:dupcut_eqb}, using the Turbo Cut-elimination procedure and no induction, is possible.
We use the Turbo Cut-elimination procedure on $C$ in $\theta$, yielding a proof-net $\Theta$; we also use it in $\theta'$ on $C_0$ then $C_1$, yielding $\Theta'$.
Whence, $\theta\betatostar\Theta$ and $\Theta'\betafromstar\theta'$.
It remains to prove that $\Theta=\Theta'$.
Remark that $\Theta$ and $\Theta'$ can only differ by their linkings, for they have the same syntactic forest.
Notice that a linking in $\theta_i$, $i\in\{0;1\}$, matches for $C$ in $\theta$ if and only if it matches for $C_i$ in $\theta'$, because this linking uses $C_i$ as $C$.
Thence, the same linkings stay in $\Theta$ and $\Theta'$, and $\Theta=\Theta'$ follows.
\end{rem}

\begin{thm}[Simulation Theorem]
\label{th:simulation}
Let $\pi$ and $\pi'$ be unit-free MALL proof trees with expanded axioms.
If $\pi\eqb\pi'$, then $\Rst(\pi)\eqb\Rst(\pi')$.
\end{thm}
\begin{proof}
This is a corollary of Lemmas~\ref{lem:simulation_beta} and~\ref{lem:dupcut_eqb}.
\end{proof}

\subsection{Isomorphisms in proof-nets}
\label{subsec:iso_in_pn}

A notion of isomorphism $\isoproofs{A}{B}{\theta}{\psi}$ can be defined directly on proof-nets: $\theta$ and $\psi$ are two cut-free proof-nets of respective conclusions $A\orth,B$ and $B\orth,A$ such that $\cutf{\theta}{\psi}{B}$ and $\cutf{\psi}{\theta}{A}$ reduce by cut-elimination to identity proof-nets.
Thanks to the Simulation Theorem (\autoref{th:simulation}), we obtain:

\begin{thm}[Type isomorphisms in proof-nets]
\label{th:red_to_pn}
Let $A$ and $B$ be two unit-free MALL formulas.
If $A\iso B$ then there exist two proof-nets $\theta$ and $\psi$ such that $\isoproofs{A}{B}{\theta}{\psi}$.
\end{thm}
\begin{proof}
Using \autoref{lem:iso_isoproofs} followed by \autoref{prop:eta_nom_eqbn}, there exist unit-free MALL cut-free proofs with expanded axioms $\pi$ and $\tau$, respectively of $\fCenter A\orth,B$ and $\fCenter B\orth,A$, such that $\cutf{\pi}{\tau}{B}\eqb\id{A}$ and $\cutf{\tau}{\pi}{A}\eqb\id{B}$.
We will prove that $\cutf{\Rst(\pi)}{\Rst(\tau)}{B}$ reduces by cut-elimination to $\Rst(\id{A})$, and by symmetry a similar reasoning entails $\cutf{\Rst(\tau)}{\Rst(\pi)}{A}$ reduces to $\Rst(\id{B})$.
This is enough to conclude: $\theta \coloneqq \Rst(\pi)$ and $\psi \coloneqq \Rst(\tau)$ are cut-free proof-nets whose composition over $B$ (\resp\ $A$) yields after cut-elimination the identity proof-net of $A$ (\resp\ $B$).

By the Simulation Theorem (\autoref{th:simulation}), from $\cutf{\pi}{\tau}{B}\eqb\id{A}$ we deduce $\Rst(\cutf{\pi}{\tau}{A})\eqb\Rst(\id{A})$.
But $\Rst(\cutf{\pi}{\tau}{A})=\cutf{\Rst(\pi)}{\Rst(\tau)}{A}$ by definition of $\Rst$ (\autoref{subsec:deseq}).
Using convergence of cut-elimination (\autoref{th:cut_sn}) and that $\Rst(\id{A})$ is cut-free, one gets $\cutf{\Rst(\pi)}{\Rst(\tau)}{A}$ reduces by cut-elimination to $\Rst(\id{A})$, as wanted.
\end{proof}

\begin{rem}
\label{rem:iso_iso_pn}
The converse of \autoref{th:red_to_pn} holds.
Indeed, the goal of the next section is to prove $\isoproofs{A}{B}{\theta}{\psi} \implies A\eqEu B$.
Therefore, using in addition \autoref{th:iso_sound}:
\[\isoproofs{A}{B}{\theta}{\psi} \implies A\eqEu B \implies A\iso B\]
\end{rem}

%==================================================

\section{Completeness}
\label{sec:complet_unit-free}

Our method relates closely to the one used by Balat and Di~Cosmo in~\cite{mllisos}, with some more work due to the distributivity isomorphisms.
We work on proof-nets, as they highly simplify the problem by representing proofs up to rule commutations~\cite{mallpncom}.
The core of the proof is as follows.
Looking at the expected isomorphisms between distributed formulas, there should be only associativity and commutativity, \ie\ reordering of atoms in formulas.
In particular, a proof-net of an isomorphism $A \iso B$ should simply be a bijection between atoms of $A$ and atoms of $B$.
We prove it is indeed the case: such a proof-net has a very particular shape with each atom of $A$ having a unique axiom link on it, which goes to an atom of $B$ (and vice-versa); we call bipartite \unique\ such a proof-net.
Once this is proved, it is not difficult (albeit a bit long) to conclude only rearrangements of atoms are possible, so only associativity and commutativity of connectives.
However, proving that proof-nets of an isomorphism are bipartite \unique\ is complicated; in particular, this only holds for isomorphisms between distributed formulas (\eg\ the proof-nets of a distributivity isomorphism are not of this shape).
This is why we go through intermediate steps: we first prove that axiom links are between $A$ and $B$ (\ie\ there is no link between two atoms of $A$), then that every atom has a link on it, and finally we prove the unicity of this link.
More formally, the special shapes of proof-nets we consider are the followings.

\begin{defi}[Full, \Unique, Bipartite proof-net]
A proof-net is called \emph{full} if any of its leaves has (at least) one link on it.
Furthermore, if for any leaf there exists a unique link on it (possibly shared among several linkings), then we call this proof-net \emph{\unique}.

\noindent
A cut-free proof-net is \emph{bipartite} if it has two conclusions, $A$ and $B$, and each of its links is between a leaf of $A$ and a leaf of $B$ (no link between leaves of $A$, or between leaves of $B$).
\end{defi}
For instance:
\begin{itemize}
\item the three proof-nets on \autoref{fig:id_bi_u} (\autopageref{fig:id_bi_u}) are bipartite and \unique\
\item the top-left proof-net of \autoref{fig:no_anticut} (\autopageref{fig:no_anticut}) is non-bipartite, full and non-\unique; meanwhile, its top-right proof-net is bipartite and non-full (so non-\unique)
\item both proof-nets on \autoref{fig:dist_not_uniq} (\autopageref{fig:dist_not_uniq}) are bipartite, full and non-\unique\
\end{itemize}

Our proof of completeness can be sketched as follows:
\begin{enumerate}[label=(\arabic*)]
\item
we start by studying identity proof-nets and easily prove they are bipartite \unique\ (\autoref{subsec:id});
\item
then, we show isomorphisms yield bipartite full proof-nets (\autoref{sec:red_bi_full}) -- the difficulty here is that bipartiteness and fullness are not preserved by cut anti-reduction, so we have to use a reasoning specific to isomorphisms;
\item
next is proven that distributed isomorphisms have \unique\ proof-nets (\autoref{sec:dist_bi_u}), which is the core of the argument and where we need the distributivity hypothesis;
\item
we then exploit bipartite \unique ness so as to rename atoms in $A \iso B$ to make $A$ and $B$ having no repetitions of atoms (\ie\ each atom appears at most once in $A$), which is called \emph{non-ambiguousness} (\autoref{sec:red_nonambi});
\item
finally, isomorphisms between non-ambiguous formulas are easily characterized as being exactly compositions of associativity and commutativity (\autoref{sec:compl_unit_free}).
\end{enumerate}
Looking at how isomorphisms of MLL are characterized in~\cite{mllisos}, the key differences are that \unique ness is given for free as there is no slice nor distributivity isomorphism.
In particular, steps $(4)$ and $(5)$ are similar to their MLL counterpart -- but more complex in MALL as there are more connectives.

\subsection{Properties of identity proof-nets}
\label{subsec:id}

Using an induction on the formula $A$, we can prove the following results on the identity proof-net of $A$, and in particular that it is bipartite \unique.
See \autoref{fig:id_bi_u} for a graphical intuition.

\begin{figure}
\centering
\resizebox{.8\width}{!}{
\begin{tikzpicture}
\begin{myscope}
	\node (P^) at (-1,1) {$X\orth$};
	\node (P) at (1,1) {$X$};
\end{myscope}
\begin{myscopec}{red}
	\path (P) -- ++(0,.8) -| (P^);
\end{myscopec}
\end{tikzpicture}
}
\resizebox{.8\width}{!}{
\begin{tikzpicture}
\begin{myscope}
	\node (P^) at (-3,1) {$X\orth$};
	\node (Q^) at (-1,1) {$Y\orth$};
	\node (Q) at (1,1) {$Y$};
	\node (P) at (3,1) {$X$};
	\node (p) at (-2,0) {$\parr$};
	\node (T) at (2,0) {$\tens$};
	\coordinate (c) at (.6,1.25);

	\path (P^) edge (p);
	\path (Q^) edge (p);
	\path (P) edge (T);
	\path (Q) edge (T);
\end{myscope}
\begin{myscopec}{red}
	\path (P) -- ++(0,.8) -| (P^);
	\path (Q) -- ++(0,.6) -| (Q^);
\end{myscopec}
\end{tikzpicture}
}
\resizebox{.8\width}{!}{
\begin{tikzpicture}
\begin{myscope}
	\node (P^) at (-3,1) {$X\orth$};
	\node (Q^) at (-1,1) {$Y\orth$};
	\node (Q) at (1,1) {$Y$};
	\node (P) at (3,1) {$X$};
	\node (W) at (-2,0) {$\with$};
	\node (O) at (2,0) {$\oplus$};
	\coordinate (c) at (.6,1.25);

	\path (P^) edge (W);
	\path (Q^) edge (W);
	\path (Q) edge (O);
	\path (P) edge (O);
	\path[out=225,in=-15] (Q) edge (W);
	\path[out=135,in=45] (P) edge[-] (c);
	\path[out=225,in=5] (c) edge (W);
\end{myscope}
\begin{myscopec}{red}
	\path (P^) -- ++(0,.8) -| (P);
\end{myscopec}
\begin{myscopec}{blue}
	\path (Q^) -- ++(0,-.6) -| (Q);
\end{myscopec}
\end{tikzpicture}
}
\caption{Identity proof-nets (from left to right: atom, $\parr\backslash\tens$ and $\with\backslash\oplus$)}
\label{fig:id_bi_u}
\end{figure}

\begin{lem}
\label{lem:id_ax}
The axiom links of an identity proof-net are exactly the $(l,l\orth)$, for any leaf $l$.
\end{lem}
\begin{proof}
By induction on the formula (see \autoref{fig:id_bi_u}).
\end{proof}

\begin{cor}
\label{cor:id_bi_u}
An identity proof-net is bipartite \unique.
\end{cor}
\begin{proof}
This follows from \autoref{lem:id_ax}.
\end{proof}

\begin{lem}
\label{lem:id_choices}
Let $\lambda$ be a linking of an identity proof-net and $V$ an additive vertex in its additive resolution.
Recall $V\orth$ is the dual occurrence of $V$ in $A\orth$ (\resp\ $A$) when $V$ belongs to $A$ (\resp\ $A\orth$).
Then $V\orth$ is also inside the additive resolution of $\lambda$, and its kept premise in this resolution is the dual premise of the one kept for $V$ -- \ie\ if the left (\resp\ right) premise of $V$ is kept, then the right (\resp\ left) premise of $V\orth$ is kept.
\end{lem}
\begin{proof}
Assume \wolog\ that the left premise of $V$ is kept in $\lambda$.
There is a left-ancestor $l$ of $V$ in the additive resolution of $\lambda$, hence with a link $a\in\lambda$ on it.
By \autoref{lem:id_ax}, $a=(l,l\orth)$.
As $l\orth$ is a right-ancestor of $V\orth$, the conclusion follows.
\end{proof}

The next result allows to go from exactly one linking on any $\with$-resolution of $A\orth, A$ -- \autoref{def:pn} (P1) -- to exactly one linking on any additive resolution of $A$.

\begin{lem}
\label{lem:id_res}
In the identity proof-net of $A$, exactly one linking is on any given additive resolution of the conclusion $A$.
\end{lem}
\begin{proof}
Consider such an additive resolution $R$. There is an associated $\with$-resolution $R'$ of $A\orth,A$ by taking the choices of premise of $R$ on $A$ and, for a $\with$-vertex $W$ of $A\orth$, taking the dual premise chosen in $R$ for $W\orth$.
By \autoref{lem:id_choices}, a linking $\lambda$ is on $R$ if and only if it is on $R'$.
Meanwhile, by \autoref{def:pn} (P1) there is a unique linking $\lambda$ on $R'$; thus the same holds on $R$.
\end{proof}

\subsection{Bipartite full proof-nets}
\label{sec:red_bi_full}

We prove here that proof-nets of isomorphisms are bipartite full.
Neither fullness, \unique ness nor bipartiteness is preserved by cut anti-reduction.
A counter-example is given on \autoref{fig:no_anticut}, with a non bipartite proof-net and a non full one whose composition reduces to the identity proof-net (bipartite \unique\ by \autoref{cor:id_bi_u}).
However, if \emph{both} compositions yield identity proof-nets, 
we get bipartiteness and fullness.

\begin{figure}
\begin{adjustbox}{}
\begin{tikzpicture}
\begin{myscope}
	\node (A00) at (-9,3) {$A$};
	\node (P0) at (-8,2) {$\parr$};
	\node (A01) at (-7,3) {$A\orth$};
	\node (T0) at (-7,1) {$\tens$};
	\node (B0) at (-6,2) {$B$};
	\path (A00) edge (P0);
	\path (A01) edge (P0);
	\path (P0) edge (T0);
	\path (B0) edge (T0);
	
	\node (B10) at (-5,1) {$B\orth$};
	\node (W1) at (-4,0) {$\with$};
	\node (B11) at (-4,2) {$B\orth$};
	\node (P1) at (-3,1) {$\parr$};
	\node (A10) at (-3,3) {$A$};
	\node (T1) at (-2,2) {$\tens$};
	\node (A11) at (-1,3) {$A\orth$};
	\path (B10) edge (W1);
	\path (B11) edge (P1);
	\path (A10) edge (T1);
	\path (A11) edge (T1);
	\path (T1) edge (P1);
	\path (P1) edge (W1);
\end{myscope}
\begin{myscopec}{blue}
	\path (A00) -- ++(0,.8) -| (A11);
	\path (A01) -- ++(0,.6) -| (A10);
	\path (B0) -- ++(0,.6) -| (B11);
\end{myscopec}
\begin{myscopec}{red}
	\path (A00) -- ++(0,-1.5) -| (A01);
	\path (B0) -- ++(0,-1.6) -| (B10);
\end{myscopec}
\begin{myscope_i}%To not touch the ligne below
	\node () at (-4,-.1) {};
\end{myscope_i}
\end{tikzpicture}
\quad\vrule\quad
\begin{tikzpicture}
\begin{myscope}
	\node (B20) at (5,1) {$B$};
	\node (W2) at (4,0) {$\oplus$};
	\node (B21) at (4,2) {$B$};
	\node (P2) at (3,1) {$\tens$};
	\node (A20) at (3,3) {$A\orth$};
	\node (T2) at (2,2) {$\parr$};
	\node (A21) at (1,3) {$A$};
	\path (B20) edge (W2);
	\path (B21) edge (P2);
	\path (A20) edge (T2);
	\path (A21) edge (T2);
	\path (T2) edge (P2);
	\path (P2) edge (W2);
	
	\node (A30) at (9,3) {$A\orth$};
	\node (P3) at (8,2) {$\tens$};
	\node (A31) at (7,3) {$A$};
	\node (T3) at (7,1) {$\parr$};
	\node (B3) at (6,2) {$B\orth$};
	\path (A30) edge (P3);
	\path (A31) edge (P3);
	\path (P3) edge (T3);
	\path (B3) edge (T3);
\end{myscope}
\begin{myscopec}{olive}
	\path (A21) -- ++(0,.8) -| (A30);
	\path (A20) -- ++(0,.6) -| (A31);
	\path (B21) -- ++(0,.6) -| (B3);
\end{myscopec}
\begin{myscope_i}%To not touch the ligne below
	\node () at (4,-.1) {};
\end{myscope_i}
\end{tikzpicture}
\end{adjustbox}
\hrule\vskip.5em
\begin{adjustbox}{}
\begin{tikzpicture}
\begin{myscope}
	\node (A00) at (-9,3) {$A$};
	\node (P0) at (-8,2) {$\parr$};
	\node (A01) at (-7,3) {$A\orth$};
	\node (T0) at (-7,1) {$\tens$};
	\node (B0) at (-6,2) {$B$};
	\path (A00) edge (P0);
	\path (A01) edge (P0);
	\path (P0) edge (T0);
	\path (B0) edge (T0);
	
	\node (B10) at (-5,1) {$B\orth$};
	\node (W1) at (-4,0) {$\with$};
	\node (B11) at (-4,2) {$B\orth$};
	\node (P1) at (-3,1) {$\parr$};
	\node (A10) at (-3,3) {$A$};
	\node (T1) at (-2,2) {$\tens$};
	\node (A11) at (-1,3) {$A\orth$};
	\path (B10) edge (W1);
	\path (B11) edge (P1);
	\path (A10) edge (T1);
	\path (A11) edge (T1);
	\path (T1) edge (P1);
	\path (P1) edge (W1);
	
	\node (B20) at (5,1) {$B$};
	\node (W2) at (4,0) {$\oplus$};
	\node (B21) at (4,2) {$B$};
	\node (P2) at (3,1) {$\tens$};
	\node (A20) at (3,3) {$A\orth$};
	\node (T2) at (2,2) {$\parr$};
	\node (A21) at (1,3) {$A$};
	\path (B20) edge (W2);
	\path (B21) edge (P2);
	\path (A20) edge (T2);
	\path (A21) edge (T2);
	\path (T2) edge (P2);
	\path (P2) edge (W2);
	
	\node (A30) at (9,3) {$A\orth$};
	\node (P3) at (8,2) {$\tens$};
	\node (A31) at (7,3) {$A$};
	\node (T3) at (7,1) {$\parr$};
	\node (B3) at (6,2) {$B\orth$};
	\path (A30) edge (P3);
	\path (A31) edge (P3);
	\path (P3) edge (T3);
	\path (B3) edge (T3);
	
	\node (*) at (0,-.5) {$\ast$};
	\path (W1) edge (*);
	\path (W2) edge (*);
\end{myscope}
\begin{myscopec}{blue}
	\path (A00) -- ++(0,.8) -| (A11);
	\path (A01) -- ++(0,.6) -| (A10);
	\path (B0) -- ++(0,.6) -| (B11);
	
	\path (A21) -- ++(0,.8) -| (A30);
	\path (A20) -- ++(0,.6) -| (A31);
	\path (B21) -- ++(0,.6) -| (B3);
\end{myscopec}
\begin{myscopec}{red}
	\path (A00) -- ++(0,-1.5) -| (A01);
	\path (B0) -- ++(0,-1.6) -| (B10);

	\path (A21) -- ++(0,-3.5) -| (A30);
	\path (A20) -- ++(0,-1.7) -| (A31);
	\path (B21) -- ++(0,-.55) -| (B3);
\end{myscopec}
\end{tikzpicture}
\end{adjustbox}
\caption{Non bipartite proof-net (top-left), non full proof-net (top-right) and one of their compositions yielding the identity proof-net (bottom) (jump edges not represented)}
\label{fig:no_anticut}
\end{figure}

\begin{lem}
\label{lem:id_match}
Let $\theta$ and $\theta'$ be cut-free proof-nets of respective conclusions $A\orth,B$ and $B\orth,A$, such that $\cutf{\theta'}{\theta}{A}$ reduces to the identity proof-net of $B$.
For any linking $\lambda\in\theta$, there exists $\lambda'\in\theta'$ such that $\lambda\cup\lambda'$ matches in the composition $\cutf{\theta}{\theta'}{B}$ of $\theta$ and $\theta'$ over $B$.
\end{lem}
\begin{proof}
Let us consider a linking $\lambda\in\theta$, and call $\mathcal{C}$ the choices of premise on additive connectives of $B$ that $\lambda$ makes.
We search some $\lambda'\in\theta'$ making the dual choices of premise on additive connectives of $B\orth$ compared to $\mathcal{C}$.
Consider the composition of $\theta$ and $\theta'$ over $A$.
It reduces to the identity proof-net of $B$ by hypothesis.
By \autoref{lem:id_res}, there exists a unique linking $\nu$ in the identity proof-net of $B$ corresponding to $\mathcal{C}$.
Furthermore, this linking $\nu$ of the identity proof-net is derived from some $\mu\cup\mu'$ for $\mu$ a linking of $\theta$ and $\mu'$ one of $\theta'$, with $\mu\cup\mu'$ matching for a cut over $A$: a linking in the identity proof-net is a linking of the form $\mu\cup\mu'$ where axiom links $(l,m_1)$, $(m_1\orth, m_2)$, \dots, $(m_n\orth, l\orth)$ in $\mu$ and $\mu'$ are replaced with $(l,l\orth)$, with $l$ a leaf of $B$ and the $m_i$ and $m_i\orth$ of $A\orth$ and $A$ (because an identity proof-net has only links of the form $(l,l\orth)$ by \autoref{lem:id_ax}).
Therefore, $\mu$ makes the choices $\mathcal{C}$ on $B$ and $\mu\cup\mu'$ matches for the composition of $\theta$ and $\theta'$ over both $A$ and $B$.
But $\lambda$ makes the same choices $\mathcal{C}$ on $B$ as $\mu$: $\lambda\cup\mu'$ also matches for a cut over $B$.
\end{proof}

\begin{rem}
\label{rem:id_match_pt}
\autoref{lem:id_match} is the analogue of \autoref{lem:id_match_pt} in proof-nets.
Indeed, \autoref{lem:id_match_pt} states that given two sequent calculus proofs $\pi$ and $\pi'$ composing to the identity on $B$, and $s$ a slice of $\pi$, there exists a slice $s'$ of $\pi'$ such that $\cutf{s}{s'}{B}$  does not fail when reducing cuts -- \ie\ the two slices make the dual choices for additive connectives in $B$.
Seeing a linking as a slice, this corresponds to having two matching linkings.
\end{rem}

\begin{cor}
\label{cor:id_qbi}
Assuming $\isoproofs{A}{B}{\theta}{\theta'}$, $\theta$ and $\theta'$ are bipartite.
\end{cor}
\begin{proof}
We proceed by contradiction: \wolog\ there is a link $a$ in some linking $\lambda\in\theta$ which is between leaves of $A\orth$.
Remember that in the notation $\isoproofs{A}{B}{\theta}{\theta'}$, both proof-nets $\theta$ and $\theta'$ are assumed cut-free.
Thence, one can apply \autoref{lem:id_match}: there exists $\lambda'\in\theta'$ such that $\lambda\cup\lambda'$ matches for a cut over $B$.
For $a$ does not involve leaves of $B$, it stays in the linking of the normal form resulting from $\lambda\cup\lambda'$ (after eliminating all cuts in the composition).
But this normal form is by hypothesis the identity proof-net of $A$, which is bipartite by \autoref{cor:id_bi_u}.
Thus, there cannot be an axiom link between leaves of $A\orth$ inside: contradiction.
\end{proof}

\begin{lem}
\label{lem:ret_half_full}
Assume $\theta$ and $\theta'$ are cut-free proof-nets of respective conclusions $A\orth,B$ and $B\orth,A$, and that their composition over $B$ reduces to the identity proof-net of $A$.
Then any leaf of $A\orth$ (\resp\ $A$) has at least one axiom link on it in $\theta$ (\resp\ $\theta'$).
\end{lem}
\begin{proof}
Towards a contradiction, assume \wolog\ a leaf $l$ of $A\orth$ has no link on it in $\theta$.
Then, the composition over $B$ of $\theta$ and $\theta'$ has no link on $l$ either.
And reducing cuts cannot create links using $l$, for it only takes links $(l,m)$ and $(m,n)$ to merge them into $(l,n)$.
However, the identity proof-net of $A$ is \unique\ by \autoref{cor:id_bi_u}, thence full: contradiction.
\end{proof}

\begin{rem}
\label{rem:ret_half_full}
\autoref{lem:ret_half_full} holds not only for isomorphisms but more generally for retractions, which are formulas $A$ and $B$ such that there exist proofs $\pi$ of $\fCenter A\orth,B$ and $\pi'$ of $\fCenter B\orth,A$ whose composition by cut over $B$ (but not necessarily over $A$) is equal to the axiom on $\fCenter A\orth,A$ up to axiom-expansion and cut-elimination.
An example of retraction is given by the proof-nets on \autoref{fig:no_anticut}, yielding a retraction between $(A\parr A\orth)\tens B$ and $((A\parr A\orth)\tens B)\oplus B$, which is not an isomorphism.
As another example, in MLL there is a well-known retraction between $A$ and $(A\lolipop A)\lolipop A = (A\tens A\orth)\parr A$.
\end{rem}

\begin{thm}
\label{th:iso_bipartite}
Assuming $\isoproofs{A}{B}{\theta}{\theta'}$, $\theta$ and $\theta'$ are bipartite full.
\end{thm}
\begin{proof}
By \autoref{cor:id_qbi}, $\theta$ and $\theta'$ are bipartite, and thanks to \autoref{lem:ret_half_full}, they are full.
\end{proof}

\subsection{\Unique ness}
\label{sec:dist_bi_u}

In general, isomorphisms do not yield \unique\ proof-nets.
A counter-example is distributivity: $A\tens(B\oplus C)\iso (A\tens B)\oplus(A\tens C)$, see \autoref{fig:dist_not_uniq}.
Nonetheless, distributivity equations are the only unit-free equations in $\Eu$ not giving \unique\ proof-nets.
Recall we can consider only distributed formulas (\autoref{prop:red_to_dist}).
On distributed formulas, distributivity isomorphisms can be ignored, and isomorphisms between distributed formulas happen to be bipartite \unique.

\begin{figure}
\begin{adjustbox}{}
\begin{tikzpicture}
\begin{myscope}
	\node (lA) at (-6,0) {$A$};
	\node (lB) at (-4.5,0) {$B$};
	\node (lC) at (-3,0) {$C$};
	\node (lW) at (-3.75,-1) {$\oplus$};
	\node (lT) at (-4.5,-2) {$\tens$};
    \path (lB) edge (lW);
    \path (lC) edge (lW);
    \path (lA) edge (lT);
    \path (lW) edge (lT);
    
	\node (lC^) at (-1.5,0) {$C\orth$};
	\node (lA^1) at (0,0) {$A\orth$};
	\node (lB^) at (1.5,0) {$B\orth$};
	\node (lA^0) at (3,0) {$A\orth$};
	\node (lP1) at (-.75,-1) {$\parr$};
	\node (lP0) at (2.25,-1) {$\parr$};
	\node (lO) at (.75,-2) {$\with$};
    \path (lA^0) edge (lP0);
    \path (lB^) edge (lP0);
    \path (lA^1) edge (lP1);
    \path (lC^) edge (lP1);
    \path (lP0) edge (lO);
    \path (lP1) edge (lO);
\end{myscope}
\begin{myscopec}{blue}
	\path (lA) -- ++(0,.8) -| (lA^0);
	\path (lB) -- ++(0,.6) -| (lB^);
\end{myscopec}
\begin{myscopec}{red}
	\path (lA) -- ++(0,-2.5) -| (lA^1);
	\path (lC) -- ++(0,-.8) -| (lC^);
\end{myscopec}
\end{tikzpicture}
\quad\vrule\quad
\begin{tikzpicture}
\begin{myscope}    
	\node (rA0) at (4.5,0) {$A$};
	\node (rB) at (6,0) {$B$};
	\node (rA1) at (7.5,0) {$A$};
	\node (rC) at (9,0) {$C$};
	\node (rT0) at (5.25,-1) {$\tens$};
	\node (rT1) at (8.25,-1) {$\tens$};
	\node (rW) at (6.75,-2) {$\oplus$};
    \path (rA0) edge (rT0);
    \path (rB) edge (rT0);
    \path (rA1) edge (rT1);
    \path (rC) edge (rT1);
    \path (rT0) edge (rW);
    \path (rT1) edge (rW);
    
	\node (rC^) at (10.5,0) {$C\orth$};
	\node (rB^) at (12,0) {$B\orth$};
	\node (rA^) at (13.5,0) {$A\orth$};
	\node (rO) at (11.25,-1) {$\with$};
	\node (rP) at (12,-2) {$\parr$};
    \path (rB^) edge (rO);
    \path (rC^) edge (rO);
    \path (rA^) edge (rP);
    \path (rO) edge (rP);
\end{myscope}
\begin{myscopec}{blue}
	\path (rA^) -- ++(0,.8) -| (rA0);
	\path (rB^) -- ++(0,.6) -| (rB);
\end{myscopec}
\begin{myscopec}{red}
	\path (rA^) -- ++(0,-2.5) -| (rA1);
	\path (rC^) -- ++(0,-.8) -| (rC);
\end{myscopec}
\end{tikzpicture}
\end{adjustbox}
\caption{Proof-nets for $A\tens(B\oplus C)\iso (A\tens B)\oplus(A\tens C)$}
\label{fig:dist_not_uniq}
\end{figure}

\subsubsection{Preliminary results}
\label{subsubsec:4.32+}

We mostly use the correctness criterion through the fact we can sequentialize, \ie\ recover a proof tree from a proof-net by \autoref{th:seq}.
However, in order to prove \unique ness, we make a direct use of the correctness criterion itself to deduce geometric properties of proof-nets.
This part of the proof benefits from the specifics of this syntax.
We begin with two preliminary results.
For $\Lambda$ a set of linkings and $W$ a $\with$-vertex, let $\Lambda^W$ denote the set of all linkings in $\Lambda$ whose additive resolution does not contain the right argument of $W$.
In this definition -- taken from~\cite{mallpnlong} -- we could have chosen left argument instead of right argument, meaning the asymmetry is irrelevant.
We use $\Lambda^W$ to deduce from the toggling condition -- \autoref{def:pn} (P3) -- a simpler and easier to manipulate (but weaker) condition that a proof-net respects, using that $W$ is not toggled by $\Lambda^W$.
Then, given proof-nets of an isomorphism that are not \unique, we will show they do not respect this simpler condition, hence we will reach a contradiction.

\begin{lem}[Adaptation of Lemma 4.32 in~\cite{mallpnlong}]
\label{lem:4.32+}
Let $\omega$ be a jump-free switching cycle in a proof-net $\theta$.
There exists a subset of linkings $\Lambda\subseteq\theta$ and a $\with$-vertex $W$ toggled by $\Lambda$ such that $\omega\subseteq\G_\Lambda$, $\omega\not\subseteq\G_{\Lambda^W}$ and there exists an axiom link $a\in\omega$ depending on $W$ in $\Lambda$.
\end{lem}
\begin{proof}
The proof of this lemma uses some facts from~\cite{mallpnlong} reproduced verbatim here.
Write $\lambda\eqtext{W}\lambda'$ if linkings $\lambda,\lambda'\in\theta$ are either equal or $W$ is the only $\with$ toggled by $\{\lambda,\lambda'\}$.
A subset $\Lambda$ of a proof-net $\theta$ is \emph{saturated} if any strictly larger subset toggles more $\with$ than $\Lambda$.
It is straightforward to check that:
\begin{itemize}[left=14pt]
\item[(S1)] If $\Lambda$ is saturated and toggles $W$, then $\Lambda^W$ is saturated.
\item[(S2)] If $\Lambda$ is saturated and toggles $W$ and $\lambda\in\Lambda$, then $\lambda\eqtext{W}\lambda_W$ for some $\lambda_W\in\Lambda^W$.
\end{itemize}

Let us now prove our lemma.
Take $\Lambda$ a minimal saturated subset of $\theta$ with $\G_\Lambda$ containing $\omega$, and $W$ a $\with$-vertex it toggles.
Since $\Lambda$ is minimal, $\omega\not\subseteq\G_{\Lambda^W}$ using (S1), so some edge $e$ of $\omega$ is in $\G_\Lambda$ but not in $\G_{\Lambda^W}$.
We claim that, without loss of generality, $e$ is an axiom link.
If it is indeed the case then, by (S2), $e\in\lambda\in\Lambda$ and $e\notin\lambda_W$ for $e\notin\Lambda^W$, so $e$ depends on $W$ in $\Lambda$.
We now prove our claim by eliminating other possibilities step by step.

Without loss of generality, $e$ is an edge from a leaf $l$ to some $X$, because for any other edge $Y\rightarrow X$ in $\omega$ we have $l\rightarrow Z_1\rightarrow\dots\rightarrow Z_n\rightarrow Y\rightarrow X$ in $\omega$ for some leaf $l$, and $Y\rightarrow X$ is in $\G_{\Lambda^W}$ whenever $l\rightarrow Z_1$ is in $\G_{\Lambda^W}$.
This is because a switching cycle must have an edge using a leaf, for all other edges are in the syntactic forest of the sequent, which is acyclic.

Still without loss of generality, $e$ is not an edge in a syntactic tree.
Indeed, in such a case $e\notin\G_{\Lambda^W}$ implies $l\notin\G_{\Lambda^W}$.
As $e$ belongs to the switching cycle $\omega$, let us look at the other edge in this cycle with endpoint $l$, say $e'$. As  $l\notin\G_{\Lambda^W}$, we also have $e'\notin\G_{\Lambda^W}$.
Remark that $e'$ cannot be an edge in a syntactic tree, for only one such edge has for endpoint the leaf $l$, namely $e$.
We can replace $e$ with $e'$ to assume $e$ is not an edge in a syntactic tree.

As $\omega$ is jump-free, $e$ cannot be a jump edge.
The sole possibility is $e$ being a link.
\end{proof}

For $U$ and $V$ vertices in a tree, their \emph{first common descendant} is the vertex of the tree which is a descendant of both $U$ and $V$ and which has no ancestor respecting this property, with a tree represented with its root at the bottom, which is a \emph{descendant} of the leaves.
Or equivalently, looking at a tree as a partial order of minimal element the root, the first common descendant is the infimum.

\begin{lem}
\label{lem:jump_par}
Let $\theta$ be a proof-net of conclusions $\Gamma,A$, with a jump edge $l\jump W$ between $l,W\in T(A)$.
If $W$ is not a descendant of $l$, then their first common descendant $C$ is a $\parr$.
\end{lem}
\begin{proof}
As there is a jump $l\jump W$, there exist linkings $\lambda,\lambda'\in\theta$ such that $W$ is the only $\with$ toggled by $\{\lambda;\lambda'\}$, and a link $a\in\lambda\backslash\lambda'$ using the leaf $l$. In particular, the jump $l\jump W$ is in $\G_{\{\lambda;\lambda'\}}$.
For $l$ and $W$ are both in the additive resolution of $\lambda$, both premises of $C$ are also in this additive resolution, thus $C$ cannot be an additive connective, so not a $\with$ nor a $\oplus$-vertex.

Assume by contradiction that $C$ is a $\tens$.
Call $\delta$ the path in $T(A)$ from $W$ to $C$, and $\mu$ the one from $C$ to $l$ (see \autoref{fig:jump_par}).
Then, $(l\jump W)\delta\mu$ is a switching cycle in $\G_{\{\lambda;\lambda'\}}$.
According to \autoref{def:pn} (P3), there exists a $\with$ toggled by $\{\lambda;\lambda'\}$ not in any switching cycle of $\G_{\{\lambda;\lambda'\}}$.
A contradiction, for $W$ is the only $\with$ toggled by $\{\lambda;\lambda'\}$.
Whence, $C$ can only be a $\parr$-vertex.
\end{proof}

\begin{figure}
\begin{adjustbox}{}
\begin{tikzpicture}
\begin{myscope}
	\node[draw=none] () at (0,.75) {$A$};
	\node[draw=none] () at (-1.5,2.25) {$T(A)$};
	\coordinate (Ac) at (0,1);
	\coordinate (Al) at (-2,4.5);
	\coordinate (Ar) at (2,4.5);
	\path (Ac) edge[-] (Al);
	\path (Ac) edge[-] (Ar);
	\path (Al) edge[-] (Ar);
	
	\node[rectangle,draw=none,minimum size=6mm] (l) at (-1,4.25) {$l$};
	\node[rectangle,draw=none,minimum size=6mm] (W) at (.75,3.3) {$W\with$};
	\node[rectangle,draw=none,minimum size=6mm] (C) at (-.2,2) {$C$};
	\node[draw=none] () at (-1,3.25) {$\mu$};
	\path (C) edge[decorate, decoration={snake}] (l);
	\node[draw=none] () at (.5,2.5) {$\delta$};
	\path (W) edge[decorate, decoration={snake}] (C);
    \path[bend left] (l) \edgelabel{$j$} (W);
\end{myscope}
\end{tikzpicture}
\end{adjustbox}
\caption{Illustration of the proof of \autoref{lem:jump_par}}
\label{fig:jump_par}
\end{figure}

\subsubsection{Isomorphisms of distributed formulas}
\label{subsubsec:proof_iso_bi_u}

Now, let us prove that isomorphisms of distributed formulas are bipartite \unique.
We will consider proof-nets corresponding to an isomorphism that we cut and where we eliminate all cuts not involving atoms.
To give some intuition, let us consider the non-\unique\ proof-nets of \autoref{fig:dist_not_uniq} (on \autopageref{fig:dist_not_uniq}).
Composing them together by cut on $(A\tens B)\oplus(A\tens C)$ gives the proof-net illustrated on \autoref{fig:dist_not_uniq_comp}.
Reducing all cuts not involving atoms yields the proof-net on \autoref{fig:dist_not_uniq_comp_red}, that we call an \emph{almost reduced composition}.
We observe here a switching cycle produced by the two links on $A$ (dashed in blue on \autoref{fig:dist_not_uniq_comp_red}), less visible in the non-reduced composition of \autoref{fig:dist_not_uniq_comp} -- remember that in MALL proof-nets some switching cycles are allowed, and this is one of those.
If we continue to reduce cuts, also eliminating the atomic ones, we obtain the identity proof-net, which has no switching cycle: during these reductions, both links on $A$ are merged.
The idea of this section is to proceed by contradiction, considering a non-\unique\ proof-net of an isomorphism.
We argue that, taking the almost reduced composition of the considered proof-nets, links preventing \unique ness yield switching cycles.
Then, we prove that for these switching cycles not to contradict the toggling condition of proof-nets (\autoref{def:pn} (P3)), there must be a non-distributed formula, reaching a contradiction.

\begin{figure}
\begin{adjustbox}{}
\begin{tikzpicture}
\begin{myscope}
	\node (lA) at (-6,0) {$A$};
	\node (lB) at (-4.5,0) {$B$};
	\node (lC) at (-3,0) {$C$};
	\node (lW) at (-3.75,-1) {$\oplus$};
	\node (lT) at (-4.5,-2) {$\tens$};
    \path (lB) edge (lW);
    \path (lC) edge (lW);
    \path (lA) edge (lT);
    \path (lW) edge (lT);
    
	\node (lC^) at (-1.5,0) {$C\orth$};
	\node (lA^1) at (0,0) {$A\orth$};
	\node (lB^) at (1.5,0) {$B\orth$};
	\node (lA^0) at (3,0) {$A\orth$};
	\node (lP1) at (-.75,-1) {$\parr$};
	\node (lP0) at (2.25,-1) {$\parr$};
	\node (lO) at (.75,-2) {$\with$};
    \path (lA^0) edge (lP0);
    \path (lB^) edge (lP0);
    \path (lA^1) edge (lP1);
    \path (lC^) edge (lP1);
    \path (lP0) edge (lO);
    \path (lP1) edge (lO);
    
	\node (rA0) at (4.5,0) {$A$};
	\node (rB) at (6,0) {$B$};
	\node (rA1) at (7.5,0) {$A$};
	\node (rC) at (9,0) {$C$};
	\node (rT0) at (5.25,-1) {$\tens$};
	\node (rT1) at (8.25,-1) {$\tens$};
	\node (rW) at (6.75,-2) {$\oplus$};
    \path (rA0) edge (rT0);
    \path (rB) edge (rT0);
    \path (rA1) edge (rT1);
    \path (rC) edge (rT1);
    \path (rT0) edge (rW);
    \path (rT1) edge (rW);
    
	\node (rC^) at (10.5,0) {$C\orth$};
	\node (rB^) at (12,0) {$B\orth$};
	\node (rA^) at (13.5,0) {$A\orth$};
	\node (rO) at (11.25,-1) {$\with$};
	\node (rP) at (12,-2) {$\parr$};
    \path (rB^) edge (rO);
    \path (rC^) edge (rO);
    \path (rA^) edge (rP);
    \path (rO) edge (rP);
    
    \node (*) at (3.75, -2.5) {$\ast$};
    \path (lO) edge (*);
    \path (rW) edge (*);
\end{myscope}
\begin{myscopec}{olive}
	\path (lA) -- ++(0,1) -| (lA^0);
	\path (lB) -- ++(0,.6) -| (lB^);
    
	\path (rA^) -- ++(0,1) -| (rA0);
	\path (rB^) -- ++(0,.6) -| (rB);
\end{myscopec}
\begin{myscopec}{orange}
	\path (lA) -- ++(0,-2.7) -| (lA^1);
	\path (lC) -- ++(0,-1) -| (lC^);
    
	\path (rA^) -- ++(0,1.2) -| (rA0);
	\path (rB^) -- ++(0,.8) -| (rB);
\end{myscopec}
\begin{myscopec}{blue}
	\path (lA) -- ++(0,1.2) -| (lA^0);
	\path (lB) -- ++(0,.8) -| (lB^);
    
	\path (rA^) -- ++(0,-2.7) -| (rA1);
	\path (rC^) -- ++(0,-1) -| (rC);
\end{myscopec}
\begin{myscopec}{red}
	\path (lA) -- ++(0,-2.5) -| (lA^1);
	\path (lC) -- ++(0,-.8) -| (lC^);
    
	\path (rA^) -- ++(0,-2.5) -| (rA1);
	\path (rC^) -- ++(0,-.8) -| (rC);
\end{myscopec}
\end{tikzpicture}
\end{adjustbox}
\caption{Proof-nets from \autoref{fig:dist_not_uniq} composed by cut on $(A\tens B)\oplus(A\tens C)$}
\label{fig:dist_not_uniq_comp}
\end{figure}

\begin{figure}
\begin{adjustbox}{}
\begin{tikzpicture}
\begin{myscope}
	\node (lA) at (-6,0) {$A$};
	\node (lB) at (-4.5,0) {$B$};
	\node (lC) at (-3,0) {$C$};
	\node (lW) at (-3.75,-1) {$\oplus$};
	\node (lT) at (-4.5,-2) {$\tens$};
    \path (lB) edge (lW);
    \path (lC) edge (lW);
    \path (lA) edge (lT);
    \path (lW) edge (lT);
    
	\node (lC^) at (-1.5,0) {$C\orth$};
	\node (lA^1) at (0,0) {$A\orth$};
	\node (lB^) at (1.5,0) {$B\orth$};
	\node (lA^0) at (3,0) {$A\orth$};
    
	\node (rA0) at (4.5,0) {$A$};
	\node (rB) at (6,0) {$B$};
	\node (rA1) at (7.5,0) {$A$};
	\node (rC) at (9,0) {$C$};
    
	\node (rC^) at (10.5,0) {$C\orth$};
	\node (rB^) at (12,0) {$B\orth$};
	\node (rA^) at (13.5,0) {$A\orth$};
	\node (rO) at (11.25,-1) {$\with$};
	\node (rP) at (12,-2) {$\parr$};
    \path (rB^) edge (rO);
    \path (rC^) edge (rO);
    \path (rA^) edge (rP);
    \path (rO) edge (rP);
    
    \node[draw=none,minimum size=4mm] (*C) at (3.75, -2.5) {$\ast$};
    \path (lC^) edge (*C);
    \path (rC) edge (*C);
    \node[draw=none,minimum size=4mm] (*A1) at (3.75, -2) {$\ast$};
    \path (lA^1) edge (*A1);
    \path (rA1) edge (*A1);
    \node[draw=none,minimum size=4mm] (*B) at (3.75, -1.5) {$\ast$};
    \path (lB^) edge (*B);
    \path (rB) edge (*B);
    \node[draw=none,minimum size=4mm] (*A0) at (3.75, -1) {$\ast$};
    \path (lA^0) edge (*A0);
    \path (rA0) edge (*A0);
\end{myscope}
\begin{myscopec}{olive}
	\path (lA) -- ++(0,.8) -| (lA^0);
	\path (lB) -- ++(0,.6) -| (lB^);
    
	\path (rA^) -- ++(0,.8) -| (rA0);
	\path (rB^) -- ++(0,.6) -| (rB);
\end{myscopec}
\begin{myscopec}{red}
	\path (lA) -- ++(0,-2.5) -| (lA^1);
	\path (lC) -- ++(0,-.8) -| (lC^);
    
	\path (rA^) -- ++(0,-2.5) -| (rA1);
	\path (rC^) -- ++(0,-.8) -| (rC);
\end{myscopec}
\begin{myscope_sur}
	\path (lA) -- ++(0,.8) -| (lA^0);
    \path (lA^0) edge (*A0);
    \path (rA0) edge (*A0);
	\path (rA^) -- ++(0,.8) -| (rA0);
	\path (rA^) -- ++(0,-2.5) -| (rA1);
    \path (rA1) edge (*A1);
    \path (lA^1) edge (*A1);
	\path (lA) -- ++(0,-2.5) -| (lA^1);
\end{myscope_sur}
\end{tikzpicture}
\end{adjustbox}
\caption{An almost reduced composition of the proof-nets on \autoref{fig:dist_not_uniq}}
\label{fig:dist_not_uniq_comp_red}
\end{figure}

\begin{defi}[Almost reduced composition]
\label{def:almost_reduced}
Take $\theta$ and $\theta'$ cut-free proof-nets of respective conclusions $\Gamma,B$ and $B\orth,\Delta$.
The \emph{almost reduced composition} over $B$ of $\theta$ and $\theta'$ is the proof-net resulting from the composition over $B$ of $\theta$ and $\theta'$ where we repeatedly reduce all cuts not involving atoms (\ie\ not applying step (a) of \autoref{def:cut_elim}).
\end{defi}

Let us fix $A$ and $B$ two unit-free MALL (not necessarily distributed yet) formulas as well as $\theta$ and $\theta'$ such that $\isoproofs{A}{B}{\theta}{\theta'}$.
By \autoref{th:iso_bipartite}, $\theta$ and $\theta'$ are bipartite full.
We denote by $\psi$ the almost reduced composition over $B$ of $\theta$ and $\theta'$.
Here, we can extend our duality on vertices and premises (defined in \autoref{subsec:mall-pn}) to links.

\begin{lem}
\label{lem:iso_bi_u_link}
An axiom link $a=(l,m)$ belongs to some linking $\lambda\in\psi$ if and only if, up to swapping $l$ and $m$, $l$ is a leaf of $A\orth$ (\resp\ $A$), $m$ is in the leaves of $B$ (\resp\ $B\orth$) and there is an axiom link $(l\orth,m\orth)$ in the same linking $\lambda$, that we will denote $a\orth=(l\orth,m\orth)$ (see \autoref{fig:iso_bi_u_link}).
\end{lem}
\begin{proof}
By symmetry, we only need to prove the ``if'' statement.
Linkings of $\psi$ are disjoint union of linkings in $\theta$ and $\theta'$.
By symmetry, assume $(l,m)\in\lambda\in\psi$ comes from a linking in $\theta$.
As $\theta$ is bipartite, one of the leaves, say $l$, is in $A\orth$ and the other, $m$, is a leaf of $B$.
Since the cut $m\ast m\orth$ belongs to the additive resolution of $\lambda$ ($m$ is inside), $m\orth$ is a leaf in this resolution. Thus, there is a link $(m\orth,l')\in\lambda$ for some leaf $l'$, which necessarily belongs to $A$ by bipartiteness of $\theta'$.
It remains to prove $l'=l\orth$.
If we were to eliminate all cuts in $\psi$, we would get the identity proof-net on $A$ by hypothesis.
But eliminating the cut $m\ast m\orth$ yields a link $(l,l')$, which is not modified by the elimination of the other atomic cuts.
By \autoref{lem:id_ax}, $l'=l\orth$ follows.
\end{proof}

\begin{figure}
\centering
\begin{tikzpicture}
\begin{myscope}
	\node[draw=none] at (-5,1.25) {$A\orth$};
	\node[draw=none] at (5,1.25) {$A$};
	\node[draw=none] at (-2,1.25) {$B$};
	\node[draw=none] at (2,1.25) {$B\orth$};
	
	\node[draw=none] (*0-) at (0,1.6) {$\vdots$};
	\node[draw=none,minimum size=6mm] (*0) at (0,2) {$\ast$};
	\node[draw=none] (*1+) at (0,2.65) {$\vdots$};
	
	\node[draw=none] () at (-5,2.3) {$T(A\orth)$};
	\node[draw=none] () at (5,2.3) {$T(A)$};
	\coordinate (A^c) at (-5,1.5);
	\coordinate (A^l) at (-6.25,3);
	\coordinate (A^r) at (-3.75,3);
	\coordinate (Ac) at (5,1.5);
	\coordinate (Al) at (6.25,3);
	\coordinate (Ar) at (3.75,3);
	\node[draw=none,minimum size=6mm] (l^) at (-5.5,2.75) {$l$};
	\node[draw=none,minimum size=6mm] (l0-) at (-3,2.75) {$\dots$};
	\node[draw=none,minimum size=6mm] (l0) at (-2.3,2.75) {$m$};
	\node[draw=none,minimum size=6mm] (l1+) at (-1.5,2.75) {$\dots$};
	\node[draw=none,minimum size=6mm] (l1^-) at (1.5,2.75) {$\dots$};
	\node[draw=none,minimum size=6mm] (l0^) at (2.3,2.75) {$m\orth$};
	\node[draw=none,minimum size=6mm] (l0^+) at (3,2.75) {$\dots$};
	
	\path (l0) edge (*0);
	\path (l0^) edge (*0);
	\path (A^c) edge[-] (A^l);
	\path (A^c) edge[-] (A^r);
	\path (A^l) edge[-] (A^r);
	\path (Ac) edge[-] (Al);
	\path (Ac) edge[-] (Ar);
	\path (Al) edge[-] (Ar);
	
	\path (l^) -- ++(0,.5) -| (l0);
	\node[draw=none] () at (-4,3.5) {$a$};
	
	\node[draw=none,minimum size=6mm] (l) at (5.5,2.75) {$\color{blue}l\orth$};
	\path[draw=blue, dashed] (l) -- ++(0,.5) -| (l0^);
	\node[draw=none] () at (4,3.5) {$\color{blue}a\orth$};
\end{myscope}
\end{tikzpicture}
\caption{Illustration of \autoref{lem:iso_bi_u_link}}
\label{fig:iso_bi_u_link}
\end{figure}

\begin{lem}
\label{lem:same_choices}
Let $\lambda$ be a linking of $\psi$, and $V$ an additive vertex in its additive resolution.
Then $V\orth$ is also inside, with as premise kept the dual premise of the one kept for $V$.
\end{lem}
\begin{proof}
Assume \wolog\ that the left premise of $V$ is kept in $\lambda$.
There is a left-ancestor $l$ of $V$ in the additive resolution of $\lambda$, hence with a link $a\in\lambda$ on it.
By \autoref{lem:iso_bi_u_link}, we have $a\orth\in\lambda$, using $l\orth$.
As $l\orth$ is a right-ancestor of $V\orth$, the conclusion follows.
\end{proof}

\begin{lem}
\label{lem:desc_jump_plus}
Let $W$ and $P$ be respectively a $\with$-vertex and a $\oplus$-vertex in $\psi$, with $W$ an ancestor of $P$.
Then for any axiom link $a$ depending on $W$ in $\psi$, $a$ also depends on $P\orth$ in $\psi$.
\end{lem}
\begin{proof}
There exist linkings $\lambda,\lambda'\in\psi$ such that $W$ is the only $\with$ toggled by $\{\lambda;\lambda'\}$ and $a\in\lambda\backslash\lambda'$.
We consider a linking $\lambda_{P\orth}$ defined by taking an arbitrary $\with$-resolution of $\lambda$ where we choose the other premise for $P\orth$ (and arbitrary premises for $\with$-vertices introduced this way, meaning for $\with$-ancestors of $P\orth$): by \autoref{def:pn} (P1), there exists a unique linking on it.
By \autoref{lem:same_choices}, the additive resolutions of $\lambda$ and $\lambda_{P\orth}$ (\resp\ $\lambda$ and $\lambda'$) differ, on $A$ and $A\orth$, exactly on ancestors of $P$ and $P\orth$ (\resp\ $W$ and $W\orth$). Thus, the additive resolutions of $\lambda'$ and $\lambda_{P\orth}$ also differ, on $A$ and $A\orth$, exactly on ancestors of $P$ and $P\orth$, for $W$ is an ancestor of $P$.
In particular, $\{\lambda;\lambda_{P\orth}\}$, as well as $\{\lambda';\lambda_{P\orth}\}$, toggles only $P\orth$.
If $a\in\lambda_{P\orth}$, then $a$ depends on $P\orth$ in $\{\lambda';\lambda_{P\orth}\}$.
Otherwise, $a$ depends on $P\orth$ in $\{\lambda;\lambda_{P\orth}\}$.
\end{proof}

The key result to use distributivity is that a positive vertex ``between'' a leaf $l$ and a $\with$-vertex $W$ in the same tree prevents them from interacting, \ie\ there is no jump $l\jump W$.

\begin{lem}
\label{lem:no_jump_above_tp}
Let $l\jump W$ be a jump edge in $\psi$, with $l$ not an ancestor of $W$ and $l,W\in T(A\orth)$ (\resp\ $T(A)$).
Denoting by $N$ the first common descendant of $l$ and $W$, there is no positive vertex in the path between $N$ and $W$ in $T(A\orth)$ (\resp\ $T(A)$).
\end{lem}
\begin{proof}
Let $P$ be a vertex on the path between $N$ and $W$ in $T(A\orth)$.
By \autoref{lem:jump_par}, $N$ is a $\parr$-vertex.
We prove by contradiction that $P$ can neither be a $\oplus$ nor a $\tens$-vertex.

Suppose $P$ is a $\oplus$-vertex.
By \autoref{lem:desc_jump_plus}, $a$ depends on $P\orth$, and so does $a\orth$ through \autoref{lem:iso_bi_u_link}: there is a jump edge $l\orth\jump P\orth$.
Applying \autoref{lem:jump_par}, the first common descendant of $l\orth$ and $P\orth$, which is $N\orth$, is a $\parr$-vertex: a contradiction as it is a $\tens$-vertex.

Assume now $P$ to be a $\tens$-vertex.
As there is a jump $l\jump W$, there exist linkings $\lambda,\lambda'\in\psi$ and a leaf $m$ of $B$ such that $W$ is the only $\with$ toggled by $\{\lambda;\lambda'\}$ and $a=(l,m)\in\lambda\backslash\lambda'$.
For $P$ is a $\tens$, there is a leaf $p$ which is an ancestor of $P$ in the additive resolution of $\lambda$, from a different premise of $P$ than $W$; it is used by a link $b=(p,q)\in\lambda$ (see \autoref{fig:no_jump_above_t}).
Remark $q\neq m$, for $a$ and $b$ are two distinct links in the same linking $\lambda$.
Then the switching cycle
$l\jump W\rightarrow P\leftarrow p\edgelab{b}q\rightarrow\ast\leftarrow q\orth\edgelab{b\orth}p\orth\rightarrow P\orth\rightarrow N\orth\leftarrow l\orth\edgelab{a\orth}m\orth\rightarrow\ast\leftarrow m\edgelab{a}l$
(dashed in blue on \autoref{fig:no_jump_above_t})
belongs to $\G_{\{\lambda;\lambda'\}}$, where unlabeled arrows are paths in the syntactic forest.
Contradiction: $W$, the only $\with$ toggled by $\{\lambda;\lambda'\}$, cannot be in any switching cycle of $\G_{\{\lambda;\lambda'\}}$ by \autoref{def:pn} (P3).
\end{proof}

\begin{figure}
\begin{adjustbox}{}
\begin{tikzpicture}
\begin{myscope}
	\node[draw=none] () at (-4.75,-1) {$A\orth$};
	\node[draw=none] () at (4.75,-1) {$A$};
	
	\coordinate (A^c) at (-4.75,-.75);
	\coordinate (A^l) at (-6.5,3);
	\coordinate (A^r) at (-3,3);
	\path (A^c) edge[-] (A^l);
	\path (A^c) edge[-] (A^r);
	\path (A^l) edge[-] (A^r);
	
	\coordinate (Ac) at (4.75,-.75);
	\coordinate (Al) at (6.5,3);
	\coordinate (Ar) at (3,3);
	\path (Ac) edge[-] (Al);
	\path (Ac) edge[-] (Ar);
	\path (Al) edge[-] (Ar);
	
	\node[rectangle,draw=none,minimum size=6mm] (l) at (-5.75,2.75) {$l$};
	\node[rectangle,draw=none,minimum size=6mm] (p) at (-4.75,2.75) {$p$};
	\node[rectangle,draw=none,minimum size=6mm] (W) at (-4,2.25) {$W\with$};
	\node[rectangle,draw=none,minimum size=6mm] (P) at (-4.5,1.25) {$P\tens$};
	\node[rectangle,draw=none,minimum size=6mm] (C) at (-4.75,.25) {$N\parr$};
	\path (p) edge (P);
	\path (W) edge (P);
	\path (P) edge (C);
	\path (l) edge (C);
	
	\node[rectangle,draw=none,minimum size=6mm] (l^) at (5.75,2.75) {$l\orth$};
	\node[rectangle,draw=none,minimum size=6mm] (p^) at (4.75,2.75) {$p\orth$};
	\node[rectangle,draw=none,minimum size=6mm] (W^) at (4,2.25) {$W\orth\oplus$};
	\node[rectangle,draw=none,minimum size=6mm] (P^) at (4.5,1.25) {$P\orth\parr$};
	\node[rectangle,draw=none,minimum size=6mm] (C^) at (4.75,.25) {$N\orth\tens$};
	\path (p^) edge (P^);
	\path (W^) edge (P^);
	\path (P^) edge (C^);
	\path (l^) edge (C^);
	
	\node[draw=none] (B) at (-1.25,-1) {$B$};
	\node[draw=none] (B^) at (1.25,-1) {$B\orth$};
	
	\node[draw=none,minimum size=6mm] () at (-1.5,2.75) {$\dots$};
	\node[draw=none,minimum size=6mm] (q) at (-1,2.75) {$q$};
	\node[draw=none,minimum size=6mm] (m) at (-.5,2.75) {$m$};
	
	\node[draw=none,minimum size=6mm] (m^) at (.5,2.75) {$m\orth$};
	\node[draw=none,minimum size=6mm] (q^) at (1,2.75) {$q\orth$};
	\node[draw=none,minimum size=6mm] () at (1.5,2.75) {$\dots$};
	
	\node[draw=none,minimum size=6mm] (*m) at (0,2) {$\ast$};
	\node[draw=none,minimum size=6mm] (*q) at (0,1.5) {$\ast$};
	\node[draw=none,minimum size=6mm] () at (0,1) {$\vdots$};
	\path (m) edge (*m);
	\path (m^) edge (*m);
	\path (q) edge (*q);
	\path (q^) edge (*q);
	
	\path (l) -- ++(0,.7) -| (m);
	\node[draw=none] () at (-4,3.7) {$a$};
	\path (l^) -- ++(0,.7) -| (m^);
	\node[draw=none] () at (4,3.7) {$a\orth$};
	\path (p) -- ++(0,.5) -| (q);
	\node[draw=none] () at (-2.25,3) {$b$};
	\path (p^) -- ++(0,.5) -| (q^);
	\node[draw=none] () at (2.25,3) {$b\orth$};
	
	\path[bend right] (l) \edgelabel{$j$} (W);
\end{myscope}
\begin{myscope_sur}
	\path[bend right] (l) edge (W);
	\path (W) edge (P);
	\path (p) edge (P);
	\path (p) -- ++(0,.5) -| (q);
	\path (q) edge (*q);
	\path (q^) edge (*q);
	\path (p^) -- ++(0,.5) -| (q^);
	\path (p^) edge (P^);
	\path (P^) edge (C^);
	\path (l^) edge (C^);
	\path (l^) -- ++(0,.7) -| (m^);
	\path (m^) edge (*m);
	\path (m) edge (*m);
	\path (l) -- ++(0,.7) -| (m);
\end{myscope_sur}
\end{tikzpicture}
\end{adjustbox}
\caption{Switching cycle containing $W$ if $P$ is a $\tens$-vertex in the proof of \autoref{lem:no_jump_above_tp}}
\label{fig:no_jump_above_t}
\end{figure}

\begin{thm}
\label{th:iso_bipartite_unique}
Assuming $\isoproofs{A}{B}{\theta}{\theta'}$ with $A$ and $B$ distributed, $\theta$ and $\theta'$ are bipartite \unique.
\end{thm}
\begin{proof}
We already know that $\theta$ and $\theta'$ are bipartite full thanks to \autoref{th:iso_bipartite}.
We reason by contradiction and assume \wolog\ that $\theta$ is not \unique: there exist a leaf $l$ of $A\orth$ and two distinct leaves $l_0$ and $l_1$ of $B$ with links $a=(l,l_0)$ and $b=(l,l_1)$ in $\theta$.
We consider $\psi$ the almost reduced composition of $\theta$ and $\theta'$ over $B$, depicted on \autoref{fig:iso_bi_u_1}.
By \autoref{lem:id_match}, $a$ and $b$ are also links in $\psi$ (for the linkings they belong to in $\theta$ have matching linkings in $\theta'$, and we did not eliminate atomic cuts).
Using \autoref{lem:iso_bi_u_link}, we have in $\G_\psi$ a switching cycle $\omega=l\edgelab{a} l_0\rightarrow\ast\leftarrow l_0\orth\edgelab{a\orth} l\orth\edgelab{b\orth} l_1\orth\rightarrow\ast\leftarrow l_1\edgelab{b} l$.

Let $\Lambda$ be a set of linkings and $W$ an associated $\with$-vertex given by \autoref{lem:4.32+} applied to $\omega$.
The vertex $W$ belongs to either $T(A)$ or $T(A\orth)$: up to swapping them, $W$ is in $T(A\orth)$.
By \autoref{lem:4.32+}, $a$, $a\orth$, $b$ or $b\orth$ depends on $W$.
So $a$ or $b$ depends on $W$ by \autoref{lem:iso_bi_u_link}; \wolog\ $a$ depends on $W$.
Remark $l$ is not an ancestor of $W$: if it were, by symmetry assume it is a left-ancestor.
Whence $a$ and $b$ belong to $\Lambda^W$, so $a\orth$ and $b\orth$ too (\autoref{lem:iso_bi_u_link}); thus $\omega\subseteq\G_{\Lambda^W}$, contradicting \autoref{lem:4.32+}.
Hence, $l$ is not an ancestor of $W$, and we can apply \autoref{lem:jump_par}: the first common descendant $N$ of $l$ and $W$ in $T(A\orth)$ is a $\parr$.
Using \autoref{lem:no_jump_above_tp}, there is no $\tp$-vertex on the path between the $\parr$-vertex $N$ and its ancestor the $\with$-vertex $W$ in $T(A\orth)$.
But then, considering the first $\with$-vertex in this path, there is a sub-formula of the shape $-\parr(-\with-)$ or $(-\with-)\parr-$ in the distributed $A\orth$, a contradiction.
\end{proof}

\begin{figure}
\begin{adjustbox}{}
\begin{tikzpicture}
\begin{myscope}
	\node[draw=none] (A^) at (-5,.75) {$A\orth$};
	\node[draw=none] (*0-) at (0,1) {$\vdots$};
	\node[draw=none] (*0) at (0,1.5) {$\ast$};
	\node[draw=none] (*1) at (0,2) {$\ast$};
	\node[draw=none] (*1+) at (0,2.5) {$\vdots$};
	\node[draw=none] (A) at (5,.75) {$A$};
	
	\node[draw=none] (B) at (-2,.75) {$B$};
	\node[draw=none] (B^) at (2,.75) {$B\orth$};
	
	\node[draw=none] () at (-5,2.25) {$T(A\orth)$};
	\node[draw=none] () at (5,2.25) {$T(A)$};
	\coordinate (A^c) at (-5,1);
	\coordinate (A^l) at (-6,3);
	\coordinate (A^r) at (-4,3);
	\coordinate (Ac) at (5,1);
	\coordinate (Al) at (6,3);
	\coordinate (Ar) at (4,3);
	\node[draw=none,minimum size=6mm] (l^) at (-5.5,2.75) {$l$};
	\node[draw=none,minimum size=6mm] (l0-) at (-3.25,2.75) {$\dots$};
	\node[draw=none,minimum size=6mm] (l0) at (-2.5,2.75) {$l_0$};
	\node[draw=none,minimum size=6mm] (l1) at (-1.75,2.75) {$l_1$};
	\node[draw=none,minimum size=6mm] (l1+) at (-1,2.75) {$\dots$};
	\node[draw=none,minimum size=6mm] (l1^-) at (1,2.75) {$\dots$};
	\node[draw=none,minimum size=6mm] (l1^) at (1.75,2.75) {$l_1\orth$};
	\node[draw=none,minimum size=6mm] (l0^) at (2.5,2.75) {$l_0\orth$};
	\node[draw=none,minimum size=6mm] (l0^+) at (3.25,2.75) {$\dots$};
	
	\path (l0) edge (*0);
	\path (l0^) edge (*0);
	\path (l1) edge (*1);
	\path (l1^) edge (*1);
	\path (A^c) edge[-] (A^l);
	\path (A^c) edge[-] (A^r);
	\path (A^l) edge[-] (A^r);
	\path (Ac) edge[-] (Al);
	\path (Ac) edge[-] (Ar);
	\path (Al) edge[-] (Ar);
	
	\path (l^) -- ++(0,.5) -| (l0);
	\node[draw=none] () at (-4,3.5) {$a$};
	\path (l^) -- ++(0,1) -| (l1);
	\node[draw=none] () at (-3.75,4) {$b$};
	
	\node[draw=none,minimum size=6mm] (l) at (5.5,2.75) {$l\orth$};
	\path (l) -- ++(0,.5) -| (l0^);
	\node[draw=none] () at (4,3.5) {$a\orth$};
	\path (l) -- ++(0,1) -| (l1^);
	\node[draw=none] () at (3.75,4) {$b\orth$};
\end{myscope}
\end{tikzpicture}
\end{adjustbox}
\caption{Almost reduced composition $\psi$ of $\theta$ and $\theta'$ by cut over $B$ in the proof of \autoref{th:iso_bipartite_unique}}
\label{fig:iso_bi_u_1}
\end{figure}

%==================================================

\subsection{Non-ambiguous formulas}
\label{sec:red_nonambi}

Once our study is restricted to bipartite \emph{\unique} proof-nets, we can also restrict formulas further.
This will make apparent that isomorphisms can only be associativity and commutativity.

\begin{defi}[Non-ambiguous formula]
\label{def:non-ambiguous}
A formula $A$ is \emph{non-ambiguous} if each atom in $A$ occurs at most once (whether positively or negatively).
Otherwise, $A$ is called \emph{ambiguous}.
\end{defi}

\begin{fact}
\label{rem:non-ambiguous}
If $A$ is non-ambiguous, so is $A\orth$.
\end{fact}

For instance, $X\with Y\orth$ is non-ambiguous, whereas $(A\tens B)\oplus(A\tens C)$ is ambiguous.
The reduction to non-ambiguous formulas requires to restrict to distributed formulas first: in $(A\tens B)\oplus(A\tens C)\iso A\tens(B\oplus C)$ we need the two occurrences of $A$ to factorize.
The goal of this section is to prove that we can consider only non-ambiguous formulas (\autoref{prop:red_to_non-amb}) and that isomorphisms between these formulas correspond simply to the existence of arbitrary proof-nets (\autoref{th:bi_u_non-ambi}).
These two results are an adaptation of the work on MLL by Balat \& Di~Cosmo~\cite[Section~3]{mllisos}.\footnote{
Our definition of non-ambiguous is stronger than the one from~\cite{mllisos} (less formulas are non-ambiguous here).
In~\cite{mllisos}, a formula $A$ is non-ambiguous if each atom in $A$ occurs at most once positive and once negative.
This makes for instance $X \with X\orth$ non-ambiguous, whereas by our definition it is ambiguous.
}

\subsubsection{Reduction to non-ambiguous formulas}
\label{subsec:red_to_non-amb}

Given $A \iso B$ with $A$ or $B$ ambiguous, we replace \emph{occurrences} of atoms by new fresh atoms to obtain an isomorphism $A' \iso B'$ between non-ambiguous formulas.
This is possible thanks to \textit{ax}-unicity of the proof-nets.
An important remark is that these ``replacements of occurrences'' are more general than substitutions, because we may replace \emph{different occurrences} of the \emph{same} atom by different atoms (or negated atoms) in a formula, \ie\ we work at the level of occurrences instead of atoms.

\begin{lem}
\label{lem:pn_renamed_is_pn}
Take $\theta$ an \unique\ cut-free proof-net with $(X, X\orth)$ one of its axiom links.
For any atom (or negated atom) $Y$, the set of linkings obtained from $\theta$ by replacing the \emph{occurrences} $X$ and $X\orth$ by $Y$ and $Y\orth$ -- in both the conclusion formulas and the linkings of $\theta$ -- is a proof-net.
\end{lem}
\begin{proof}
Call $A'$ and $B'$ the formulas $A$ and $B$ where the occurrences $X$ and $X\orth$ are replaced by $Y$ and $Y\orth$.
The resulting set of linkings is indeed a proof-net on $A',B'$, because the sole use of labels of leaves is to ensure an axiom link is between dual atoms (at least without cut).
In particular, labels do not matter in the correctness criterion.
\end{proof}

\begin{lem}
\label{lem:pn_iso_bij}
Assume $\isoproofs{A}{B}{\theta}{\psi}$ with $\theta$ and $\psi$ bipartite \unique\ proof-nets.
Take $X_1\orth$ and $X_2$ \emph{occurrences} of the atoms $X\orth$ and $X$ in $A\orth$ and $B$, and call $X_1$ and $X_2\orth$ respectively their dual occurrences in $A$ and $B\orth$ (\ie\ the occurrences given by the duality function).
There is an axiom link $(X_1\orth, X_2)$ in $\theta$ if and only if $(X_1, X_2\orth)$ is an axiom link in $\psi$.
\end{lem}
\begin{proof}
By symmetry, it suffices to prove that if $(X_1\orth, X_2)$ is an axiom link in $\theta$, then $(X_1, X_2\orth)$ is an axiom link in $\psi$.
By bipartite \unique ness, $(X_1\orth, X_2)$ is the sole axiom link of $\theta$ using $X_1\orth$, and $\psi$ has a unique axiom link $(X_2\orth, X_3)$ on $X_2\orth$ -- with $X_3$ an occurrence in $A, B\orth$ of the same atom as $X_1$.
Assume \wolog\ that $X_2$ is in $B$, so that $X_1$ and $X_3$ are both in $A$.
When composing $\theta$ and $\psi$ by cut over $B$ and reducing all cuts, one gets the identity proof-net of $A$, where there is a unique link on $X_1\orth$, linking it to $X_1$ (\autoref{lem:id_ax}).
But the sole link on $X_1\orth$ that may result from this composition is $(X_1\orth, X_3)$.
Thence, $X_1 = X_3$ and $(X_1, X_2\orth)$ is a link of $\psi$.
\end{proof}

\begin{lem}[Distributed ambiguous isomorphic formulas]
\label{lem:amb_iso}
Let $A$ and $B$ be distributed formulas such that $\isoproofs{A}{B}{\theta}{\psi}$.
There exists a substitution $\sigma$ and distributed formulas $A'$ and $B'$, \emph{non-ambiguous}, such that $A=\sigma(A')$, $B=\sigma(B')$ and $\isoproofs{A'}{B'}{\theta'}{\psi'}$ for some proof-nets $\theta'$ and $\psi'$.
\end{lem}
\begin{proof}
For each axiom link $a_\theta = (X\orth, X)$ of $\theta$, there is an axiom link $a_\psi = (X, X\orth)$ in $\psi$ between the dual occurrences, and vice-versa (\autoref{lem:pn_iso_bij}).
Replace these occurrences of $X\orth$ and $X$ in $A\orth$, $B$ and $\theta$ by respectively $Y_a\orth$ and $Y_a$, for $Y_a$ a fresh atom -- with a different fresh atom $Y_a$ for each link $a_\theta$.
This operation yields a proof-net $\theta'$ on ${A'}\orth, B'$ (by repeated applications of \autoref{lem:pn_renamed_is_pn}).
Similarly, replace each link $a_\psi = (X, X\orth)$ in $\psi$ by $(Y_a, Y_a\orth)$ -- with $Y_a$ the fresh atom used for the link $a_\theta = (X\orth, X)$ of $\theta$.
Again, the resulting $\psi'$ is a proof-net (\autoref{lem:pn_renamed_is_pn}), on ${B'}\orth, A'$ as we renamed occurrences dually.
Since cut-elimination does not depend on labels, the compositions of $\theta'$ and $\psi'$ reduces to the identity proof-nets of $A'$ and $B'$: hence $\isoproofs{A'}{B'}{\theta'}{\psi'}$.
Remark that $A$ (\resp\ $B$) is obtained from $A'$ (\resp\ $B'$) by substituting each $Y_a$ by the atom $X$ of $a_\theta$, for $Y_a$ was fresh.
\end{proof}

\begin{prop}[Reduction to distributed non-ambiguous formulas]
\label{prop:red_to_non-amb}
If $\AC$ is complete for isomorphisms (of proof-nets) between unit-free MALL formulas that are distributed and non-ambiguous, then it is complete for isomorphisms (of proof-nets) between unit-free MALL formulas that are distributed.
In other words:
assume that for all distributed and non-ambiguous unit-free MALL formulas $A'$ and $B'$, $\isoproofs{A'}{B'}{\theta'}{\psi'} \implies A' \eqAC B'$;
then it holds that for all distributed unit-free MALL formulas $A$ and $B$, $\isoproofs{A}{B}{\theta}{\psi} \implies A \eqAC B$.
\end{prop}
\begin{proof}
Take $A$ and $B$ some distributed unit-free MALL formulas, and assume there exist proof-nets $\theta$ and $\psi$ such that $\isoproofs{A}{B}{\theta}{\psi}$.
Using \autoref{lem:amb_iso}, there is a substitution $\sigma$ and distributed \emph{non-ambiguous} formulas $A'$ and $B'$ such that $A=\sigma(A')$, $B=\sigma(B')$ and $\isoproofs{A'}{B'}{\theta'}{\psi'}$ for some proof-nets $\theta'$ and $\psi'$.
By completeness hypothesis on $\AC$ between distributed and non-ambiguous formulas, $A' \eqAC B'$.
As a substitution preserves the equations of an equational theory, it follows that $A = \sigma(A') \eqAC \sigma(B') = B$.
\end{proof}

\subsubsection{Simplification with non-ambiguous formulas}
\label{subsec:red_non-amb}

The goal of this part is to show isomorphisms between non-ambiguous formulas correspond simply to the existence of proof-nets, speaking no more about cut-elimination nor identity proof-nets.
This is done through \autoref{th:bi_u_non-ambi}: when $A$ is non-ambiguous, having proof-nets on $A\orth, B$ and $B\orth, A$ is enough to know their composition over $B$ composes to the identity proof-net of $A$ -- simply because this is the sole proof-net on $A\orth, A$.
This result will be deduced by proving properties close to the ones of identity proof-nets from \autoref{subsec:id}.

\begin{lem}
\label{lem:q_bi_dual_ax}
Let $\theta$ be a proof-net of conclusions $A\orth,A$, with $A$ a non-ambiguous formula.
Axiom links of $\theta$ are of the form $(l\orth,l)$ for $l$ a leaf of $A$.
\end{lem}
\begin{proof}
Let $a$ be an axiom link of $\theta$, between leaves $l$ and $m$, with the label of $l$ (\resp\ $m$) being $X$ (\resp\ $X\orth$).
By non-ambiguity of $A$, and so $A\orth$ (\autoref{rem:non-ambiguous}), $X$ (and $X\orth$) occurs exactly once in $A\orth, A$ (counting $X$ and $X\orth$ differently).
Moreover, the unique leaf labeled $X\orth$ is the dual of the leaf labeled $X$.
Thus, $m = l\orth$ and $a = (l,l\orth)$.
\end{proof}

\begin{lem}
\label{lem:q_bi_dual_add}
Let $\theta$ be a proof-net of conclusions $A\orth,A$, with $A$ a non-ambiguous formula.
Take a linking $\lambda\in\theta$ and an additive vertex $V$ in its additive resolution.
The vertex $V\orth$ is in the additive resolution of $\lambda$, and $\lambda$ keeps for $V\orth$ the dual premise it keeps for $V$.
\end{lem}
\begin{proof}
As $V$ is in the additive resolution $A\orth,A\upharpoonright\lambda$ of $\lambda$, one of its ancestor leaves, say $l$, is in $A\orth,A\upharpoonright\lambda$: there is a link $a\in\lambda$ on it.
By \autoref{lem:q_bi_dual_ax}, $a=(l,l\orth)$.
But $l\orth$ is an ancestor of $V\orth$, so $V\orth$ is in $A\orth,A\upharpoonright\lambda$, with as premise the dual premise chosen for $V$.
\end{proof}

\begin{lem}
\label{lem:bi_u_non-amb-unique}
For $A$ a non-ambiguous formula, there is exactly one proof-net of conclusions $A\orth,A$: the identity proof-net of $A$.
\end{lem}
\begin{proof}
We prove that for any proof-nets $\theta$ and $\theta'$ of conclusions $A\orth,A$, it holds that $\theta=\theta'$.
Hence the result when $\theta'$ is the identity proof-net.

Take $\lambda\in\theta$ a linking.
It is on some $\with$-resolution $R$ of $A\orth,A$.
By \autoref{def:pn} (P1), there exists a unique linking $\lambda'\in\theta'$ on $R$.
We wish to prove $\lambda=\lambda'$.
They have the same additive resolution, for their choice on a $\oplus$-vertex $P$ is determined by the premise taken for the $\with$-vertex $P\orth$, which is in $R$ (\autoref{lem:q_bi_dual_add}).
They have the same axiom links on this additive resolution, because any leaf on it is linked to its dual (\autoref{lem:q_bi_dual_ax}).
Therefore, $\lambda=\lambda'$, so $\theta\subseteq\theta'$.
By symmetry, the same reasoning yields $\theta'\subseteq\theta$, thus $\theta=\theta'$.
\end{proof}

\begin{rem}
\label{rem:id_bi_u_not_unique}
This property does not hold outside of non-ambiguous formulas, even distributed.
For instance, there are two bipartite \unique\ proof-nets of conclusions $X_1\orth\tens X_2\orth, X_3\parr X_4$ (where each $X_i$ is an occurrence of the atom $X$): the identity proof net, with axiom links $(X_1\orth,X_4)$ and $(X_2\orth,X_3)$, and the ``swap'' with axiom links $(X_1\orth,X_3)$ and $(X_2\orth,X_4)$ -- see \autoref{fig:id_bi_u_not_unique}.
\end{rem}

\begin{figure}
\begin{adjustbox}{}
\begin{tikzpicture}
\begin{myscope}
	\node (X1^) at (-3,1) {$X_1\orth$};
	\node (X0^) at (-1,1) {$X_2\orth$};
	\node (X0) at (1,1) {$X_3$};
	\node (X1) at (3,1) {$X_4$};
	\node (T) at (-2,0) {$\tens$};
	\node (P) at (2,0) {$\parr$};

	\path (X1^) edge (T);
	\path (X0^) edge (T);
	\path (X0) edge (P);
	\path (X1) edge (P);
\end{myscope}
\begin{myscopec}{red}
	\path (X0) -- ++(0,.6) -| (X0^);
	\path (X1) -- ++(0,.8) -| (X1^);
\end{myscopec}
\end{tikzpicture}
\quad\vrule\quad
\begin{tikzpicture}
\begin{myscope}
	\node (X1^) at (-3,1) {$X_1\orth$};
	\node (X0^) at (-1,1) {$X_2\orth$};
	\node (X0) at (1,1) {$X_3$};
	\node (X1) at (3,1) {$X_4$};
	\node (T) at (-2,0) {$\tens$};
	\node (P) at (2,0) {$\parr$};

	\path (X1^) edge (T);
	\path (X0^) edge (T);
	\path (X0) edge (P);
	\path (X1) edge (P);
\end{myscope}
\begin{myscopec}{red}
	\path (X0) -- ++(0,.8) -| (X1^);
	\path (X1) -- ++(0,.6) -| (X0^);
\end{myscopec}
\end{tikzpicture}
\end{adjustbox}
\caption{Identity proof-net of $X\parr X$ (left-side) and the swap on this formula (right-side)}
\label{fig:id_bi_u_not_unique}
\end{figure}

\begin{thm}[Non-ambiguous isomorphisms]
\label{th:bi_u_non-ambi}
Let $A$ and $B$ be non-ambiguous formulas.
If there exist proof-nets $\theta$ and $\psi$ of respective conclusions $A\orth,B$ and $B\orth,A$, then $\isoproofs{A}{B}{\theta}{\psi}$.
\end{thm}
\begin{proof}
Both compositions reduce to identity proof-nets.
Indeed, the composition of $\theta$ and $\psi$ by cut over $B$ reduces to a proof-net on $A\orth, A$, that by \autoref{lem:bi_u_non-amb-unique} can only be the identity proof-net of $A$ -- and similarly for a composition by cut over $A$.
\end{proof}

%==================================================

\subsection{Completeness for unit-free distributed MALL}
\label{sec:compl_unit_free}

We now prove the completeness of $\AC$ for unit-free and distributed MALL formulas by reasoning as in~\cite[Section~4]{mllisos}, with some more technicalities for we have to reorder not only $\parr$-vertices but also $\with$-vertices.

\begin{defi}[Sequentializing vertex]
\label{def:sequentializing}
A terminal (\ie\ with no descendant) non-leaf vertex $V$ in a proof-net $\theta$ is called \emph{sequentializing} if, depending on its kind:
\begin{itemize}
\item $\tens\backslash\ast$-vertex: the removal of $V$ in $\G_\theta$ results in two connected components.
\item $\oplus$-vertex: the left or right syntactic sub-tree of $V$ does not belong to $\G_\theta$ (\ie\ has no link on any of its leaves in $\G_\theta$).
\item $\pw$-vertex: a terminal $\pw$-vertex is always sequentializing.
\end{itemize}
\end{defi}

It is easy to check that removing a sequentializing vertex produces proof-net(s).
The sequentialization theorem (\autoref{th:seq}) states there is always a sequentializing vertex.

\begin{lem}
\label{lem:bi_no_pos_seq}
In a bipartite full proof-net with conclusions $A_l\odot A_r,B$, where $\odot\in\{\tens;\oplus\}$, the root of $A_l\odot A_r$ is not sequentializing.
\end{lem}
\begin{proof}
Let $l$ be a leaf of $A_l$ and $r$ one of $A_r$.
By bipartiteness and fullness, there are leaves $m$ and $s$ of $B$ with axiom links $(l,m)$ and $(r,s)$ in the proof-net, see \autoref{fig:bi_no_pos_seq}.
As there is a path in $T(B)$ between $m$ and $s$, whether $\odot=\oplus$ or $\odot=\tens$, it is not sequentializing.
\end{proof}

\begin{figure}
\centering
\begin{tikzpicture}
\begin{myscope}
	\node[draw=none] (A0) at (-5,1) {$A_l$};
	\node[draw=none] (A1) at (-2,1) {$A_r$};
	\node[draw=none,minimum size=0mm] (A) at (-3.5,.25) {$\odot$};
	\node[draw=none] (B) at (3,.25) {$B$};
	\path (A0) edge (A);
	\path (A1) edge (A);
	
	\node[draw=none] () at (-5,2.25) {$T(A_l)$};
	\node[draw=none] () at (-2,2.25) {$T(A_r)$};
	\node[draw=none] () at (3,1.25) {$T(B)$};
	\coordinate (A0c) at (-5,1.25);
	\coordinate (A0l) at (-6.25,3);
	\coordinate (A0r) at (-3.75,3);
	\coordinate (A1c) at (-2,1.25);
	\coordinate (A1l) at (-3.25,3);
	\coordinate (A1r) at (-.75,3);
	\coordinate (Bc) at (3,.5);
	\coordinate (Br) at (.75,3);
	\coordinate (Bl) at (5.25,3);
	\path (A0c) edge[-] (A0l);
	\path (A0c) edge[-] (A0r);
	\path (A0l) edge[-] (A0r);
	\path (A1c) edge[-] (A1l);
	\path (A1c) edge[-] (A1r);
	\path (A1l) edge[-] (A1r);
	\path (Bc) edge[-] (Br);
	\path (Bc) edge[-] (Bl);
	\path (Bl) edge[-] (Br);
	
	\node[draw=none,minimum size=6mm] (l) at (-5,2.75) {$l$};
	\node[draw=none,minimum size=6mm] (r) at (-2,2.75) {$r$};
	\node[draw=none,minimum size=6mm] (l^) at (1.75,2.75) {$m$};
	\node[draw=none,minimum size=6mm] (r^) at (4,2.75) {$s$};
	
	\path[draw=black, thick] (l) -- ++(0,.5) -| (l^);
	\path[draw=black, thick] (r) -- ++(0,.7) -| (r^);
	
	\path[bend right=100] (l^) edge[-] (r^);
\end{myscope}
\end{tikzpicture}
\caption{Illustration of the proof of \autoref{lem:bi_no_pos_seq}}
\label{fig:bi_no_pos_seq}
\end{figure}

\begin{lem}[Reordering $\parr$-vertices]
\label{lem:iso_complet_mult}
Let $\theta$ be a bipartite \unique\ proof-net of conclusions $A=A_l\parr A_r$ and $B=B_l\odot B_r$ with $\odot\in\{\tens;\oplus\}$ and $A$ a distributed formula.
Then $\odot=\tens$ and there exist two bipartite \unique\ proof-nets of respective conclusions $A_l',B_l$ and $A_r',B_r$ where $A_l'\parr A_r'$ is equal to $A_l\parr A_r$ up to associativity and commutativity of $\parr$.
\end{lem}
\begin{proof}
We remove all terminal (hence sequentializing) $\parr$-vertices, all in $A$, without modifying the linkings.
The resulting graph is a proof-net of conclusions $A_1,\dots,A_n,B_l\odot B_r$ (see \autoref{fig:iso_complet_mult}).
The roots of the new trees $A_i$ cannot be $\with$-vertices because $A$ is distributed: so they are $\tp$-vertices or atoms.
These $\tp$-vertices are not sequentializing, since by bipartiteness and fullness every leaf of each $A_i$ is connected to the formula $B_l\odot B_r$ (reasoning as in the proof of \autoref{lem:bi_no_pos_seq}).
Thus, the sequentializing vertex of this proof-net is necessarily $B_l\odot B_r$.
It follows $\odot=\tens$, because all leaves of $B$ are connected to leaves in $A_1,\dots,A_n$, so if $\odot=\oplus$ then $B_l\odot B_r$ cannot be sequentializing.
Removing the sequentializing $B_l\tens B_r$ gives two proof-nets, with a partition of the $A_i$ into two classes: those linked to leaves of $B_l$ and the others linked to leaves of $B_r$.
We recover from these proof-nets bipartite \unique\ ones by adding $\parr$-vertices under the $A_i$ in an arbitrary order, yielding formulas $A_l'$ (with those linked to $B_l$) and $A_r'$ (with those linked to $B_r$).
As we only removed and put back $\parr$-vertices, $A_l'\parr A_r'$ is equal to $A_l\parr A_r$ up to associativity and commutativity of $\parr$.
\end{proof}

\begin{figure}
\resizebox{\textwidth}{!}{
\begin{tikzpicture}
\begin{myscope}
	\node[draw=none] (B0) at (-5,.75) {$B_l$};
	\node[draw=none] (B1) at (-2,.75) {$B_r$};
	\node[draw=none,minimum size=0mm] (B) at (-3.5,0) {$\odot$};
	\node[draw=none] () at (-3.5,-.5) {$B$};
	\path (B0) edge (B);
	\path (B1) edge (B);
	
	\node[draw=none] () at (-5,2.25) {$T(B_l)$};
	\coordinate (B0c) at (-5,1);
	\coordinate (B0l) at (-6.25,3);
	\coordinate (B0r) at (-3.75,3);
	\path (B0c) edge[-] (B0l);
	\path (B0c) edge[-] (B0r);
	\path (B0l) edge[-] (B0r);
	\node[draw=none] () at (-2,2.25) {$T(B_r)$};
	\coordinate (B1c) at (-2,1);
	\coordinate (B1l) at (-3.25,3);
	\coordinate (B1r) at (-.75,3);
	\path (B1c) edge[-] (B1l);
	\path (B1c) edge[-] (B1r);
	\path (B1l) edge[-] (B1r);
	
	\node[draw=none] () at (1.5,2.75) {$T(A_1)$};
	\coordinate (A0c) at (1.5,2);
	\coordinate (A0l) at (.75,3);
	\coordinate (A0r) at (2.25,3);
	\path (A0c) edge[-] (A0l);
	\path (A0c) edge[-] (A0r);
	\path (A0l) edge[-] (A0r);
	\node[draw=none] () at (3,2.75) {$T(A_2)$};
	\coordinate (A1c) at (3,2);
	\coordinate (A1l) at (2.25,3);
	\coordinate (A1r) at (3.75,3);
	\path (A1c) edge[-] (A1l);
	\path (A1c) edge[-] (A1r);
	\path (A1l) edge[-] (A1r);
	\node[draw=none] () at (4.5,2.75) {$\dots$};
	\coordinate (A2c) at (4.5,2);
	\node[draw=none] () at (6,2.75) {$T(A_n)$};
	\coordinate (A3c) at (6,2);
	\coordinate (A3l) at (5.25,3);
	\coordinate (A3r) at (6.75,3);
	\path (A3c) edge[-] (A3l);
	\path (A3c) edge[-] (A3r);
	\path (A3l) edge[-] (A3r);
	\node[draw=none,minimum size=0mm] (p0) at (2.25,1.5) {$\parr$};
	\node[draw=none,minimum size=0mm] (p1) at (5.25,1.5) {$\parr$};
	\path (A0c) edge[dotted] (p0);
	\path (A1c) edge[dotted] (p0);
	\path (A2c) edge[dotted] (p1);
	\path (A3c) edge[dotted] (p1);
	\node[draw=none] () at (3.75,1.25) {$\vdots$};
	\node[draw=none,minimum size=0mm] (A0) at (5,.75) {$\parr$};
	\node[draw=none,minimum size=0mm] (A1) at (2,.75) {$\parr$};
	\node[draw=none,minimum size=0mm] (A) at (3.5,0) {$\parr$};
	\node[draw=none] () at (3.5,-.5) {$A$};
	\path (A0) edge[dotted] (A);
	\path (A1) edge[dotted] (A);
	
	\coordinate (l0) at (-5,3);
	\coordinate (l0^) at (1.5,3);
	\coordinate (l1) at (-2,3);
	\coordinate (l1^) at (4.25,3);
	\coordinate (l2) at (-1.5,3);
	\coordinate (l2^) at (2.75,3);
	\coordinate (l3) at (-5.5,3);
	\coordinate (l3^) at (6.5,3);
	\coordinate (l4) at (-4,3);
	\coordinate (l4^) at (4.5,3);
	\path[draw=black, thick] (l0) -- ++(0,.2) -| (l0^);
	\path[draw=black, thick] (l1) -- ++(0,.4) -| (l1^);
	\path[draw=black, thick] (l2) -- ++(0,.6) -| (l2^);
	\path[draw=black, thick] (l3) -- ++(0,1) -| (l3^);
	\path[draw=black, thick] (l4) -- ++(0,.8) -| (l4^);
\end{myscope}
\end{tikzpicture}
}
\caption{Proof-net of \autoref{lem:iso_complet_mult} with all terminal $\parr$-vertices removed}
\label{fig:iso_complet_mult}
\end{figure}

\begin{lem}[Reordering $\with$-vertices]
\label{lem:iso_complet_add}
Let $\theta$ be a bipartite \unique\ proof-net of conclusions $A=A_l\with A_r$ and $B=B_l\oplus B_r$ with $A$ a distributed formula.
Then there exist two bipartite \unique\ proof-nets of respective conclusions $A_l',B_l$ and $A_r',B_r$ where $A_l'\with A_r'$ is equal to $A_l\with A_r$ up to associativity and commutativity of $\with$.
\end{lem}
\begin{proof}
We remove all terminal $\with$-vertices in the proof-net, which are all in $A$.
The resulting graphs are $n$ proof-nets $\theta_i$ (for terminal negative vertices are sequentializing), with $\theta_i$ of conclusions $A_i,B_l\oplus B_r$.
An illustration is similar to \autoref{fig:iso_complet_mult}, except we have $n$ proof-nets each having a sole $T(A_i)$ and with in common $T(B)$.
We claim that the root of $B_l\oplus B_r$ is sequentializing in every $\theta_i$.
Assuming it for now, this implies that, in each $\theta_i$, $A_i$ is linked only to either $B_l$ or $B_r$ and we can remove the root of $B_l\oplus B_r$.
Once this $\oplus$-vertex removed, we put back the removed $\with$-vertices of $A$ in another order: we put together all $\theta_i$ linked to $B_l$ on one side, and all those to $B_r$ on another side, yielding two proof-nets of conclusions $B_l,A_l'$ and $B_r,A_r'$.
These proof-nets are bipartite \unique\ ones for adding a $\with$ is taking the disjoint union of linkings.
We indeed have $A_l'\with A_r'$ equal to $A$ up to associativity and commutativity of $\with$, because we only reordered $\with$-vertices.

Let us now prove our claim: the root of $B_l\oplus B_r$ is sequentializing in every $\theta_i$.
Towards a contradiction, assume it is not for some $\theta_i$.
Remove all terminal $\parr$-vertices of $\theta_i$, which are in $A$ if any.
As in the proof of \autoref{lem:iso_complet_mult}, the roots of the new trees cannot be negative vertices because the formula $A$ is distributed: hence, they are $\tp$-vertices or atoms.
These $\tp$-vertices cannot be sequentializing, since by bipartiteness and fullness every leaf of $A^i$ is connected to the formula $B_l\oplus B_r$ (reasoning as in the proof of \autoref{lem:bi_no_pos_seq}).
Since removing $\parr$-vertices preserves that $B_l\oplus B_r$ is not sequentializing, we have a proof-net with no sequentializing vertices, a contradiction.
\end{proof}

\begin{thm}[Isomorphisms completeness for unit-free MALL]
\label{th:iso_complet}
Given $A$ and $B$ two distributed unit-free MALL formulas, if $\isoproofs{A}{B}{\theta}{\psi}$ for some proof-nets $\theta$ and $\psi$, then $A\eqAC B$.
\end{thm}
\begin{proof}
We assume $A$ and $B$ to be non-ambiguous formulas by \autoref{prop:red_to_non-amb}.
We reason by induction on the size $s(A)$ of $A$, which is its number of connectives and atoms.
Remark that the size is unaffected by commutation and associativity of connectives, and that $s(A)=s(B)$ for $\theta$ is bipartite \unique\ (\autoref{th:iso_bipartite_unique}), thence $A$ and $B$ have the same number of atoms, so of connectives as they are all binary ones.

If $A$ and $B$ are atoms (\ie\ of size $1$), then the axiom link in $\theta$ between $A\orth$ and $B$ yields $A=B$ and the property holds.
Otherwise, $A\orth$ and $B$ are \emph{both} non atomic.
By \autoref{th:iso_bipartite_unique}, $\theta$ and $\psi$ are bipartite \unique; they have respective conclusions $A\orth, B$ and $B\orth,A$.
By \autoref{lem:bi_no_pos_seq}, one of the formulas $A\orth,B$ is negative, otherwise neither the root of $A\orth$ nor the one of $B$ is sequentializing in $\theta$, contradicting sequentialization (\autoref{th:seq}).
A symmetric reasoning on $\psi$ implies that the other formula is positive, so that either both $A$ and $B$ are positive or both are negative.
Assume \wolog\ that they are positive.
We have two cases: either $A$ is a $\tens$-formula, or both $A$ and $B$ are $\oplus$-formulas (if $A$ is a $\oplus$-formula and $B$ a $\tens$-formula, we switch $A$ and $B$).

Let us consider the first case: $\theta$ is a proof-net of conclusions $A\orth,B$ with $A\orth$ a $\parr$-formula.
By \autoref{lem:iso_complet_mult}, $B$ is a $\tens$-formula and there exist two bipartite \unique\ proof-nets $\theta_0$ and $\theta_1$ of respective conclusions ${A_0'}\orth,B_0$ and ${A_1'}\orth,B_1$, with ${A'}\orth={A_1'}\orth\parr{A_0'}\orth$ equal to $A\orth$ up to associativity and commutativity of $\parr$.
Similarly in the second case, \autoref{lem:iso_complet_add} yields bipartite \unique\ proof-nets $\theta_0$ and $\theta_1$ of respective conclusions ${A_0'}\orth,B_0$ and ${A_1'}\orth,B_1$, with ${A'}\orth={A_1'}\orth\with{A_0'}\orth$ equal to $A\orth$ up to associativity and commutativity of $\with$.
In both cases, remark that $A'$ is distributed and non-ambiguous for it is equal, up to associativity and commutativity, to the distributed and non-ambiguous $A$.
In particular, $A_0'$, $A_1'$, $B_0$ and $B_1$ are also distributed and non-ambiguous, ${A'}\orth\eqAC A\orth$, and $s(A_0')$ and $s(A_1')$ are both less than $s(A)$.
To conclude, we only need proof-nets $\psi_0$ and $\psi_1$ of respective conclusions $B_0\orth,A_0'$ and $B_1\orth,A_1'$.
We will then apply \autoref{th:bi_u_non-ambi} to obtain $\isoproofs{A_0'}{B_0}{\theta_0}{\psi_0}$ and $\isoproofs{A_1'}{B_1}{\theta_1}{\psi_1}$.
Thence, by induction hypothesis, $A_0'\eqAC B_0$ and $A_1'\eqAC B_1$, thus $A'\eqAC B$.
As $A\eqAC A'$, we will finally conclude $A\eqAC B$.

Thus, we look for two proof-nets of respective conclusions $B_0\orth,A_0'$ and $B_1\orth,A_1'$.
As $\isoproofs{A}{B}{\theta}{\psi}$, and $A\iso A'$ by soundness of $\Eu \supseteq \AC$ (\autoref{th:iso_sound}), it follows using \autoref{th:red_to_pn} that $\isoproofs{A'}{B}{\Theta}{\Psi}$ for some proof-nets $\Theta$ and $\Psi$.\footnote{One can easily check that isomorphisms in proof-nets form equivalence classes on formulas.}
In particular, $\Theta$ is a proof-net of conclusions $B\orth,A'$, \ie\ of conclusions $B_1\orth\odot\orth B_0\orth, A_0'\odot A_1'$ with $\odot\in\{\tens;\oplus\}$.
Remark that, by non-ambiguousness of $A'$ (\resp\ $B$), no atom $X$ occurs in both $A_0'$ and $A_1'$ (\resp\ $B_0$ and $B_1$), counting both positive and negative occurrences $X$ and $X\orth$.
We claim that no atom occurs in both $B_0$ and $A_1'$.
Indeed, as $\theta_0$ is a bipartite \unique\ proof-net of conclusions ${A_0'}\orth,B_0$, it follows the atoms of $B_0$ are exactly those of $A_0'$, which we just saw are disjoint from the atoms of $A_1'$.
Similarly with $\theta_1$, no atom occurs in both $B_1$ and $A_0'$.
These four conditions imply that axiom links in $\Theta$ must be between leaves of $B_0\orth$ and $A_0'$, and between leaves of $B_1\orth$ and $A_1'$.
Therefore, once we sequentialize the negative root $\odot\orth$ of $B\orth$ in $\Theta$, the positive root $\odot$ of $A'$ is sequentializing.
After sequentializing both, we obtain two proof-nets of respective conclusions $B_0\orth,A_0'$ and $B_1\orth,A_1'$.
\end{proof}

The above theorem, associated with preceding results, yields our main contribution.

\begin{thm}[Isomorphisms completeness]
\label{th:iso_complete}
\hfill
\begin{description}
\item[$\AC$ is complete for distributed MALL]
given $A$ and $B$ two distributed MALL formulas, if $A\iso B$, then $A\eqAC B$.
\item[$\Eu$ is complete for MALL]
given $A$ and $B$ two MALL formulas, if $A\iso B$, then $A\eqEu B$.
\item[Completeness for fragments of MALL]
for any subset $S$ of $\{\tens ; \parr ; 1 ; \bot ; \with ; \oplus ; \top ; 0\}$, if $A$ and $B$ are formulas using connectives from $S$ only and $A \iso B$, then $A = B$ in the theory $\Eu_S$ obtained by extracting from $\Eu$ the equations that use connectives in $S$ only.
\end{description}
\end{thm}
\begin{proof}
For the first point, using \autoref{th:iso_complet_mall} we only have to prove that $\AC$ is complete for isomorphisms of unit-free distributed formulas.
Take $A$ and $B$ unit-free distributed formulas such that $A\iso B$.
By \autoref{th:red_to_pn} there exist proof-nets $\theta$ and $\psi$ such that $\isoproofs{A}{B}{\theta}{\psi}$.
We conclude by \autoref{th:iso_complet} that $A\eqAC B$.

The second point follows from the first one and \autoref{prop:red_to_dist}.

For the last point, one can reason similarly as for the second one by looking closely to the proof of \autoref{prop:red_to_dist} (on \autopageref{prop:red_to_dist}).
Assume $A$ and $B$ use connectives from $S$ only, and $A_d$ (\resp\ $B_d$) is a normal form of $A$ (\resp\ $B$) for the rewriting system $\distsystem$ presented in \autoref{def:distsystem}.
Then $A\eqEu[S] A_d$ and $B\eqEu[S] B_d$ since $\distsystem$ is included in $\Eu$ and never creates a connective by reduction.
Using the first point we know $A_d\eqAC B_d$ (since $A_d \iso A \iso B\iso B_d$), so that $A_d \eqEu[S] B_d$ because, in the path $A_d\eqAC B_d$, only connectives from $S$ can be used ($\AC\subseteq\Eu$ never creates nor erases a connective).
Finally we conclude $A\eqEu[S] B$.
\end{proof}

The last point gives us for instance an equational theory for isomorphisms of unit-free MALL.
Denote by $\E$ the part of $\Eu$ not involving units, \ie\ with associativity, commutativity and distributivity equations only.
Then $\E$ is complete for unit-free MALL:
given $A$ and $B$ two unit-free MALL formulas, $A\iso B \implies A\eqE B$.
Indeed, if $A\iso B$ in unit-free MALL, then $A\iso B$ in MALL with $A$ and $B$ using only $\{\tens ; \parr ; \with ; \oplus\}$ so that $A\eqE B$ by \autoref{th:iso_complete}.
And, as usual, soundness of $\E$ is easy to check.
We also get back in the same fashion the results for MLL and unit-free MLL by Balat and Di~Cosmo~\cite{mllisos}.

On the contrary, one cannot deduce anything for intuitionistic subsystems such as IMALL in this way (see \autoref{sec:conjectures_isos_smcc} for counter examples).

%==================================================

\section{Star-autonomous categories with finite products}
\label{sec:cat}

Since MALL semantically corresponds to $\star$-autonomous categories with finite products (which also have finite coproducts), we can use our results on MALL to characterize isomorphisms valid in all such categories.
For the historical result of how linear logic can be seen as a category, see~\cite{seelycat}.
The main problem here is that formulas and objects are not strictly the same, and we have to study translations from the logic to the category and vice-versa.

We consider objects of $\star$-autonomous categories described by formulas in the language:
\begin{equation*}
  F \coloncoloneqq X \mid F\tens F \mid 1 \mid F \lolipop F \mid \bot \mid F\with F \mid \top
\end{equation*}
There is some redundancy here since one could define $1$ as $\bot\lolipop\bot$ in any $\star$-autonomous category, but we prefer to keep $1$ in the language as it is at the core of monoidal categories.
Our goal is to prove the theory $\thD$ of \autoref{tab:eqisos_cat} to be sound and complete for the isomorphisms of $\star$-autonomous categories with finite products.

\begin{table}
\resizebox{\textwidth}{!}{
$
\thD\left\{
  \begin{array}{c}
  \left.\begin{array}{rcl@{\quad}rcl@{\quad}rcl}
    F\tens(G\tens H) & = & (F\tens G)\tens H
  & F\tens G & = & G\tens F
  & F\tens 1 & = & F \\
    (F\tens G)\lolipop H & = & F\lolipop (G\lolipop H)
  & & & & 1 \lolipop F & = & F \\ \\
    F\with(G\with H) & = & (F\with G)\with H
  & F\with G & = & G\with F
  & F\with \top & = & F \\
    F\lolipop(G\with H) & = & (F\lolipop G)\with(F\lolipop H)
  & & & & F \lolipop \top & = & \top
  \end{array}\right\}\thS \\ \\
  (F\lolipop\bot)\lolipop\bot = F
\end{array}\right.
$
}
\caption{Type isomorphisms in $\star$-autonomous categories with finite products}
\label{tab:eqisos_cat}
\end{table}

We establish this result from the one on MALL, first proving that MALL (with proofs considered up to $\beta\eta$-equality) defines a $\star$-autonomous category with finite products (\autoref{subsec:cat_MALL}).
Then, we conclude using a semantic method based on this syntactic category (\autoref{subsec:cat_isos}).
In a third step we look at the more general case of symmetric monoidal closed categories -- without the requirement of a dualizing object (\autoref{subsec:smcc_isos}).

\subsection{MALL as a star-autonomous category with finite products}
\label{subsec:cat_MALL}

Before reasoning about categories, we recall here how formulas and proofs manipulated in the previous sections can be seen as objects and morphisms in a category~\cite{lambekscott,seelycat}.
We detail how the logic MALL, with proofs taken up to $\beta\eta$-equality, defines a $\star$-autonomous category with finite products, that we will call \MALL.

Objects of \MALL\ are formulas of MALL, while its morphisms from $A$ to $B$ are proofs of $\fCenter A\orth,B$, considered up to $\beta\eta$-equality.\footnote{We recall that $(\cdot)\orth$ is defined by induction, making it an involution.}
One can check that a proof of MALL is an isomorphism if and only if, when seen as a morphism, it is an isomorphism in \MALL.

We define a bifunctor $\tens$ on \MALL, associating to formulas (\ie\ objects) $A$ and $B$ the formula $A\tens B$ and to proofs (\ie\ morphisms) $\pi_0$ and $\pi_1$ respectively of $\fCenter A_0\orth, B_0$ and $\fCenter A_1\orth,B_1$ the following proof of $\fCenter (A_0\tens A_1)\orth, B_0\tens B_1$:
\begin{prooftree}
\Rsub{$\pi_0$}{$A_0\orth,B_0$}
\Rsub{$\pi_1$}{$A_1\orth,B_1$}
\Rtens{$A_1\orth,A_0\orth,B_0\tens B_1$}
\Rparr{$A_1\orth\parr A_0\orth,B_0\tens B_1$}
\end{prooftree}
One can check that $(\MALL,\tens,1,\alpha,\lambda,\rho,\gamma)$ forms a symmetric monoidal category, where $1$ is the $1$-formula, $\alpha$ are isomorphisms of MALL associated to $(A\tens B)\tens C\iso A\tens(B\tens C)$ seen as natural isomorphisms of \MALL, and similarly for $\lambda$ with $1\tens A\iso A$, $\rho$ with $A\tens 1\iso A$, and $\gamma$ with $A\tens B\iso B\tens A$.

Furthermore, define $A\lolipop B \coloneqq A\orth\parr B$ and $ev_{A,B}$ as the following morphism from $A\tens(A\lolipop B)$ to $B$ (\ie\ a proof of $\fCenter (B\orth\tens A)\parr A\orth,B$):
\begin{prooftree}
\Rax{$B$}
\Rax{$A$}
\Rtens{$B\orth\tens A,A\orth,B$}
\Rparr{$(B\orth\tens A)\parr A\orth,B$}
\end{prooftree}
It can be checked that \MALL\ is a symmetric monoidal closed category with as exponential object $(A\lolipop B,ev_{A,B})$ for objects $A$ and $B$.

Moreover, one can also check that $\bot$ is a dualizing object for this category, making \MALL\ a $\star$-autonomous category.
This relies on the following morphism from $(A\lolipop\bot)\lolipop\bot$ to $A$ (which is an inverse of the currying of $ev_{A,\bot}$):
\begin{prooftree}
\Rone{}
\Rax{$A$}
\Rbot{$A\orth,\bot,A$}
\Rparr{$A\orth\parr\bot,A$}
\Rtens{$1\tens(A\orth\parr\bot),A$}
\end{prooftree}

Finally, $\top$ is a terminal object of \MALL, and $A\with B$ is the product of objects $A$ and $B$, with, as projections $\pi_A$ and $\pi_B$, the following morphisms respectively from $A\with B$ to $A$ and from $A\with B$ to $B$:
\begin{center}
\Rax{$A$}
\Rplustwo{$B\orth\oplus A\orth,A$}
\DisplayProof
\hskip 2em and \hskip 2em
\Rax{$B$}
\Rplusone{$B\orth\oplus A\orth,B$}
\DisplayProof
\end{center}
Therefore, \MALL\ is a $\star$-autonomous category with finite products~\cite{seelycat}.

\subsection{Isomorphisms of star-autonomous categories with finite products}
\label{subsec:cat_isos}

We translate formulas in the language of $\star$-autonomous categories with finite products into MALL formulas and conversely by means of the following translations, with $\tradm{\_}$ from the category to the logic and $\tradc{\_}$ in the reverse direction:
\begin{equation*}
\begin{array}{rcl@{\qquad\quad}rcl}
  \tradm{X} &=& X & \tradc{X} &=& X \\
  & & & \tradc{X\orth} &=& X\lolipop\bot \\
  \tradm{F\tens G} &=& \tradm{F}\tens\tradm{G} & \tradc{A\tens B} &=& \tradc{A}\tens\tradc{B} \\
  \tradm{1} &=& 1 & \tradc{1} &=& 1\\
  \tradm{F\lolipop G} &=& \tradm{F}\orth\parr\tradm{G} & \tradc{A\parr B} &=& ((\tradc{A}\lolipop\bot)\tens(\tradc{B}\lolipop\bot))\lolipop\bot \\
  \tradm{\bot} &=& \bot & \tradc{\bot} &=& \bot \\
  \tradm{F\with G} &=& \tradm{F}\with\tradm{G} & \tradc{A\with B} &=& \tradc{A}\with\tradc{B} \\
  \tradm{\top} &=& \top & \tradc{\top} &=& \top \\
  & & & \tradc{A\oplus B} &=& ((\tradc{A}\lolipop\bot)\with(\tradc{B}\lolipop\bot))\lolipop\bot \\
  & & & \tradc{0} &=& \top\lolipop\bot
\end{array}
\end{equation*}
The translation $\tradm{\_}$ corresponds exactly to the interpretation of the constructions on objects of a $\star$-autonomous categories with finite products in the concrete category \MALL\ of \autoref{subsec:cat_MALL}.

\begin{lem}
\label{lem:trad_isos}
The $\tradc{\_}$ and $\tradm{\_}$ translations satisfy the following properties:
\begin{itemize}
\item $\tradc{A\orth}\isoc \tradc{A}\lolipop\bot$
\item $\tradc{\tradm{F}}\isoc F$
\item $A\isom B$ entails $\tradc{A}\isoc\tradc{B}$
\end{itemize}
\end{lem}
\begin{proof}
The second property relies on the first while the third is independent.
\begin{itemize}
\item By induction on $A$:
{\allowdisplaybreaks
\begin{align*}
\tradc{X\orth}&= X\lolipop\bot = \tradc{X}\lolipop\bot \\
\tradc{X\biorth}&= \tradc{X} = X \isoc (X \lolipop \bot)\lolipop\bot = \tradc{X\orth}\lolipop\bot \\
\tradc{(A\tens B)\orth} &= ((\tradc{B\orth}\lolipop\bot)\tens(\tradc{A\orth}\lolipop\bot))\lolipop\bot\\
  &\isoc (((\tradc{A}\lolipop\bot)\lolipop\bot)\tens((\tradc{B}\lolipop\bot)\lolipop\bot))\lolipop\bot\\
  &\isoc (\tradc{A}\tens\tradc{B})\lolipop\bot = \tradc{A\tens B}\lolipop\bot \\
\tradc{1\orth} &= \bot \isoc 1\lolipop\bot = \tradc{1}\lolipop\bot \\
\tradc{(A\parr B)\orth} &= \tradc{B\orth}\tens\tradc{A\orth} \isoc (\tradc{A}\lolipop\bot)\tens(\tradc{B}\lolipop\bot) \\
  &\isoc (((\tradc{A}\lolipop\bot)\tens(\tradc{B}\lolipop\bot))\lolipop\bot)\lolipop\bot \\
  &= \tradc{A\parr B}\lolipop\bot \\
\tradc{\bot\orth} &= 1 \isoc (1\lolipop\bot)\lolipop\bot \isoc \bot\lolipop\bot = \tradc{\bot}\lolipop\bot \\
\tradc{(A\with B)\orth} &= ((\tradc{B\orth}\lolipop\bot)\with(\tradc{A\orth}\lolipop\bot))\lolipop\bot\\
  &\isoc (((\tradc{A}\lolipop\bot)\lolipop\bot)\with((\tradc{B}\lolipop\bot)\lolipop\bot))\lolipop\bot\\
  &\isoc (\tradc{A}\with\tradc{B})\lolipop\bot = \tradc{A\with B}\lolipop\bot \\
\tradc{\top\orth} &= \top\lolipop\bot = \tradc{\top}\lolipop\bot \\
\tradc{(A\oplus B)\orth} &= \tradc{B\orth}\with\tradc{A\orth} \isoc (\tradc{A}\lolipop\bot)\with(\tradc{B}\lolipop\bot) \\
  &\isoc (((\tradc{A}\lolipop\bot)\with(\tradc{B}\lolipop\bot))\lolipop\bot)\lolipop\bot \\
  &= \tradc{A\oplus B}\lolipop\bot \\
\tradc{0\orth} &= \top \isoc (\top\lolipop\bot)\lolipop\bot = \tradc{0}\lolipop\bot \\
\end{align*}
}
\item By induction on $F$, the key case being:
\begin{align*}
% \tradc{\tradm{X}} &= \tradc{X} = X \\
% \tradc{\tradm{F\tens G}} &= \tradc{\tradm{F}\tens \tradm{G}} = \tradc{\tradm{F}}\tens\tradc{\tradm{G}} \isoc F\tens G \\
% \tradc{\tradm{1}} &= \tradc{1} = 1 \\
\tradc{\tradm{F\lolipop G}} &= \tradc{\tradm{F}\orth\parr\tradm{G}} = ((\tradc{\tradm{F}\orth}\lolipop\bot)\tens(\tradc{\tradm{G}}\lolipop\bot))\lolipop\bot \\
  &\isoc (((F\lolipop\bot)\lolipop\bot)\tens(G\lolipop\bot))\lolipop\bot \isoc (F\tens(G\lolipop\bot))\lolipop\bot\\
  &\isoc F\lolipop((G\lolipop\bot)\lolipop\bot) \isoc F\lolipop G
% \\
% \tradc{\tradm{\bot}} &= \tradc{\bot} = \bot \\
% \tradc{\tradm{F\with G}} &= \tradc{\tradm{F}\with \tradm{G}} = \tradc{\tradm{F}}\with\tradc{\tradm{G}} \isoc F\with G \\
% \tradc{\tradm{\top}} &= \tradc{\top} = \top
\end{align*}
\item We prove that the image of each equation of \autoref{tab:eqisos} through $\tradc{\_}$ is derivable with equations of \autoref{tab:eqisos_cat}.
Associativity, commutativity and unitality for $\tens$ and $\with$ are immediate.
{\allowdisplaybreaks
\begin{align*}
\tradc{A\parr(B\parr C)}
  &= ((\tradc{A} \lolipop \bot) \tens ((((\tradc{B} \lolipop \bot) \tens (\tradc{C} \lolipop \bot)) \lolipop \bot) \lolipop \bot)) \lolipop \bot\\
  &\isoc ((\tradc{A} \lolipop \bot) \tens ((\tradc{B} \lolipop \bot) \tens (\tradc{C} \lolipop \bot))) \lolipop \bot\\
  &\isoc (((\tradc{A} \lolipop \bot) \tens (\tradc{B} \lolipop \bot)) \tens (\tradc{C} \lolipop \bot)) \lolipop \bot\\
  &\isoc (((((\tradc{A} \lolipop \bot) \tens (\tradc{B} \lolipop \bot) \lolipop \bot) \lolipop \bot)) \tens (\tradc{C} \lolipop \bot)) \lolipop \bot\\
  &= \tradc{(A\parr B)\parr C} \\
\tradc{A\parr B}
  &= ((\tradc{A} \lolipop \bot) \tens (\tradc{B} \lolipop \bot)) \lolipop \bot \\
  &\isoc ((\tradc{B} \lolipop \bot) \tens (\tradc{A} \lolipop \bot)) \lolipop \bot = \tradc{B\parr A} \\
\tradc{A\parr \bot}
  &= ((\tradc{A} \lolipop \bot) \tens (\bot \lolipop \bot)) \lolipop \bot\\
  &\isoc ((\tradc{A} \lolipop \bot) \tens ((1\lolipop\bot)\lolipop\bot)) \lolipop \bot \isoc ((\tradc{A} \lolipop \bot) \tens 1) \lolipop \bot \\
  &\isoc (\tradc{A} \lolipop \bot) \lolipop \bot \isoc \tradc{A}
\end{align*}
}
Associativity, commutativity and unitality for $\oplus$ follow the same pattern as for $\parr$.

Then we have:
{\allowdisplaybreaks
\begin{align*}
\tradc{A\tens(B\oplus C)}
  &= \tradc{A}\tens(((\tradc{B}\lolipop\bot)\with(\tradc{C}\lolipop\bot))\lolipop\bot)  \\
  &\isoc ((\tradc{A}\tens(((\tradc{B}\lolipop\bot)\with(\tradc{C}\lolipop\bot))\lolipop\bot))\lolipop\bot)\lolipop\bot  \\
  &\isoc (\tradc{A}\lolipop(((\tradc{B}\lolipop\bot)\with(\tradc{C}\lolipop\bot))\lolipop\bot)\lolipop\bot)\lolipop\bot  \\
  &\isoc (\tradc{A}\lolipop((\tradc{B}\lolipop\bot)\with(\tradc{C}\lolipop\bot)))\lolipop\bot  \\
  &\isoc ((\tradc{A}\lolipop(\tradc{B}\lolipop\bot))\with(\tradc{A}\lolipop(\tradc{C}\lolipop\bot)))\lolipop\bot  \\
  &\isoc (((\tradc{A}\tens\tradc{B})\lolipop\bot)\with((\tradc{A}\tens\tradc{C})\lolipop\bot))\lolipop\bot  \\
  &= \tradc{(A\tens B)\oplus (A\tens C)} \\
\tradc{A\tens 0}
  &= \tradc{A} \tens (\top \lolipop \bot)\\
  &\isoc ((\tradc{A} \tens (\top \lolipop \bot)) \lolipop \bot) \lolipop \bot\\
  &\isoc (\tradc{A} \lolipop ((\top \lolipop \bot) \lolipop \bot)) \lolipop \bot\\
  &\isoc (\tradc{A} \lolipop \top) \lolipop \bot\\
  &\isoc \top \lolipop \bot = \tradc{0} \\
\tradc{A\parr(B\with C)}
  &= ((\tradc{A}\lolipop \bot)\tens((\tradc{B}\with\tradc{C})\lolipop\bot))\lolipop\bot \\
  &\isoc (\tradc{A}\lolipop \bot)\lolipop(((\tradc{B}\with\tradc{C})\lolipop\bot)\lolipop\bot) \\
  &\isoc (\tradc{A}\lolipop \bot)\lolipop(\tradc{B}\with\tradc{C}) \\
  &\isoc ((\tradc{A}\lolipop \bot)\lolipop\tradc{B})\with((\tradc{A}\lolipop \bot)\lolipop\tradc{C}) \\
  &\isoc ((\tradc{A}\lolipop \bot)\lolipop((\tradc{B}\lolipop\bot)\lolipop\bot))\with\\
    &\tag{$(\tradc{A}\lolipop \bot)\lolipop((\tradc{C}\lolipop\bot)\lolipop\bot))$} \\
  &\isoc (((\tradc{A}\lolipop \bot)\tens(\tradc{B}\lolipop\bot))\lolipop\bot)\with\\
    &\tag{$(((\tradc{A}\lolipop \bot)\tens(\tradc{C}\lolipop\bot))\lolipop\bot)$} \\
  &= \tradc{(A\parr B) \with (A\parr C)} \\
\tradc{A\parr\top}
  &= ((\tradc{A}\lolipop \bot)\tens(\top\lolipop\bot))\lolipop\bot \\
  &\isoc (\tradc{A}\lolipop \bot)\lolipop((\top\lolipop\bot)\lolipop\bot) \\
  &\isoc (\tradc{A}\lolipop \bot)\lolipop\top \isoc \top = \tradc{\top} \qedhere
\end{align*}
}
\end{itemize}
\end{proof}

\begin{thm}[Isomorphisms in $\star$-autonomous categories with finite products]
The equational theory $\thD$ (\autoref{tab:eqisos_cat}) is sound and complete for isomorphisms in $\star$-autonomous categories with finite products.
\end{thm}
\begin{proof}
Soundness follows by definition of $\star$-autonomous categories with finite products.
For completeness, take an isomorphism $F\iso G$.
It yields an isomorphism $\tradm{F}\iso\tradm{G}$ in \MALL.
As $A\iso B$ in \MALL\ is generated by $\Eu$ (\autoref{th:iso_complete}), we get $\tradm{F}\isom\tradm{G}$.
From \autoref{lem:trad_isos}, we deduce $F\isoc\tradc{\tradm{F}}\isoc\tradc{\tradm{G}}\isoc G$.
\end{proof}

% \begin{cor}
%   If $F\isoc G$ then $\tradm{F}\isom\tradm{G}$.
% \end{cor}

\subsection{Isomorphisms of symmetric monoidal closed categories with finite products}
\label{subsec:smcc_isos}

Isomorphisms in symmetric monoidal closed categories (SMCC) have been characterized~\cite{isossmcc} and proved to correspond to equations in the first two lines of \autoref{tab:eqisos_cat}.
We extend this result to finite products by proving the soundness and completeness of the theory $\thS$ presented in \autoref{tab:eqisos_cat}.

\begin{thm}[Isomorphisms in SMCC with finite products]\label{thm:isos_smcc}
The equational theory $\thS$ (\autoref{tab:eqisos_cat}) is sound and complete for isomorphisms in symmetric monoidal closed categories with finite products.
\end{thm}
\begin{proof}
The language of SMCC with finite products is the language of $\star$-autonomous categories with finite products in which we remove $\bot$.
In particular the translation $\tradm{\_}$ can be used to translate the associated formulas into MALL formulas.
In order to analyse the image of this restricted translation, we consider the following grammar of MALL formulas (Danos-Regnier \emph{output formulas} $o$ and \emph{input formulas} $\iform$~\cite{lamarchegames}):
\begin{equation*}
\begin{array}{ccccccccccccccc}
  o & \coloncoloneqq & X & \mid & o\tens o & \mid & o\parr\iform & \mid & \iform\parr o & \mid & 1 & \mid & o\with o & \mid & \top \\
  \iform & \coloncoloneqq & X\orth & \mid & \iform\parr\iform & \mid & \iform\tens o & \mid & o\tens\iform & \mid & \bot & \mid & \iform\oplus\iform & \mid & 0
\end{array}
\end{equation*}
The dual of an output formula is an input formula and conversely.
Moreover no MALL formula is both an input and an output formula -- let us call this the \emph{non-ambiguity property}.
One can check by induction on a formula $F$ in the language of SMCC with finite products that $\tradm{F}$ is an output formula.
We define a translation back from output formulas to SMCC formulas (which is well defined thanks to the non-ambiguity property):
\begin{equation*}
  \begin{array}{rcl@{\qquad}rcl}
    \trado{X} & = & X \\
    \bigtrado{o\tens o'} & = & \trado{o}\tens\trado{o'} &\trado{1} & = & 1 \\
     \bigtrado{o\with o'} & = & \trado{o}\with\trado{o'} & \trado{\top} & = & \top \\
   \bigtrado{o\parr\iform} & = & \bigtrado{\iform\orth}\lolipop\trado{o} &
      \bigtrado{\iform\parr o} & = & \bigtrado{\iform\orth}\lolipop\trado{o}
  \end{array}
\end{equation*}
We use the notation $\tradi{\iform}=\bigtrado{\iform\orth}$ (so that $\bigtradi{o\orth}=\trado{o}$).
% \begin{equation*}
% \begin{array}{rcl@{\qquad\qquad}rcl}
%   \trado{X} & = & X & \tradi{X\orth} & = & X \\
%   \trado{o\tens o'} & = & \trado{o}\tens\trado{o'} & \tradi{\iform\parr\iform'} & = & \tradi{\iform}\tens\tradi{\iform'} \\
%   \trado{o\parr\iform} & = & \tradi{\iform}\lolipop\trado{o} & \tradi{\iform\tens o} & = & \trado{o}\lolipop\tradi{\iform} \\
%   \trado{\iform\parr o} & = & \tradi{\iform}\lolipop\trado{o} & \tradi{o\tens\iform} & = & \trado{o}\lolipop\tradi{\iform} \\
%   \trado{1} & = & 1 & \tradi{\bot} & = & 1 \\
%   \trado{o\with o'} & = & \trado{o}\with\trado{o'} & \tradi{\iform\oplus\iform'} & = & \tradi{\iform}\with\tradi{\iform'} \\
%   \trado{\top} & = & \top & \tradi{0} & = & \top
% \end{array}
% \end{equation*}
% One can check by induction that $\tradi{o\orth}=\trado{o}$ for any $o$ (and thus $\trado{\iform\orth}=\tradi{\iform}$).
We can check, by induction on $F$, that $\trado{\tradm{F}} = F$.
% \begin{align*}
%   \trado{\tradm{X}} &= \trado{X} = X \\
%   \trado{\tradm{F\tens G}} &= \bigtrado{\tradm{F}\tens\tradm{G}} = \trado{\tradm{F}}\tens\trado{\tradm{G}} = F\tens G \\
%   \trado{\tradm{1}} &= \trado{1} = 1 \\
%   \trado{\tradm{F\lolipop G}} &= \bigtrado{\tradm{F}\orth\parr\tradm{G}} = \trado{\tradm{F}}\lolipop \trado{\tradm{G}} = F\lolipop G \\
%   \trado{\tradm{F\with G}} &= \bigtrado{\tradm{F}\with\tradm{G}} = \trado{\tradm{F}}\with\trado{\tradm{G}} = F\with G \\
%   \trado{\tradm{\top}} &= \trado{\top} = \top
% \end{align*}
We now prove that when $o\isom A$ (\resp\ $\iform\isom A$) is an equation from \autoref{tab:eqisos} (or its symmetric version) then $A$ is an output (\resp\ input) formula and $\trado{o}\isoo\trado{A}$ (\resp\ $\tradi{\iform}\isoo\tradi{A}$):
\begin{itemize}
\item If $o$ or $\iform$ is of the shape $A\tens(B\tens C)$, we have the following possibilities:
  \begin{itemize}
  \item $A$, $B$ and $C$ are output, then $(A\tens B)\tens C$ is output and:
    \begin{align*}
      \bigtrado{A\tens(B\tens C)} &= \trado{A}\tens(\trado{B}\tens\trado{C}) \\
       &\isoo (\trado{A}\tens\trado{B})\tens\trado{C} = \bigtrado{(A\tens B)\tens C}
    \end{align*}
  \item $A$ and $B$ are output and $C$ is input, then $(A\tens B)\tens C$ is input and:
    \begin{align*}
      \bigtradi{A\tens(B\tens C)} &= \trado{A}\lolipop(\trado{B}\lolipop\tradi{C}) \\
       &\isoo (\trado{A}\tens\trado{B})\lolipop\tradi{C} = \bigtradi{(A\tens B)\tens C}
    \end{align*}
  \item $A$ and $C$ are output and $B$ is input, then $(A\tens B)\tens C$ is input and:
    \begin{align*}
      \bigtradi{A\tens(B\tens C)} &= \trado{A}\lolipop(\trado{C}\lolipop\tradi{B}) \\
       &\isoo (\trado{A}\tens\trado{C})\lolipop\tradi{B} \isoo (\trado{C}\tens\trado{A})\lolipop\tradi{B} \\
       &\isoo \trado{C}\lolipop(\trado{A}\lolipop\tradi{B}) = \bigtradi{(A\tens B)\tens C}
    \end{align*}
  \item $A$ is input and $B$ and $C$ are output, then $(A\tens B)\tens C$ is input and:
    \begin{align*}
      \bigtradi{A\tens(B\tens C)} &= (\trado{B}\tens\trado{C})\lolipop\tradi{A} \\
          &\isoo (\trado{C}\tens\trado{B})\lolipop\tradi{A} \\
       &\isoo \trado{C}\lolipop(\trado{B}\lolipop\tradi{A}) = \bigtradi{(A\tens B)\tens C}
    \end{align*}
  \end{itemize}
The symmetric case follows the same pattern, as well as associativity of $\parr$.
\item If $o$ or $\iform$ is of the shape $A\tens B$, we have the following possibilities:
  \begin{itemize}
  \item $A$ and $B$ are output, then $B\tens A$ is output and $\bigtrado{A\tens B} = \trado{A}\tens\trado{B} \isoo \trado{B}\tens\trado{A} = \bigtrado{B\tens A}$
  \item $A$ is output and $B$ is input, then $B\tens A$ is input and $\bigtradi{A\tens B} = \trado{A}\lolipop\tradi{B} = \bigtradi{B\tens A}$
  \item $A$ is input and $B$ is output, then $B\tens A$ is input and $\bigtradi{A\tens B} = \trado{B}\lolipop\tradi{A} = \bigtradi{B\tens A}$
  \end{itemize}
The commutativity of $\parr$ follows the same pattern.
\item If $o$ or $\iform$ is of the shape $A\tens 1$ then either $A$ is output and $\bigtrado{A\tens 1} = \trado{A}\tens 1 \isoo \trado{A}$,
or $A$ is input and $\bigtradi{A\tens 1} = 1\lolipop\tradi{A} \isoo \tradi{A}$.
The symmetric case follows the same pattern, as well as unitality for $\parr$.
\item If $o = A\with (B\with C)$ then $A$, $B$ and $C$ are output and $(A\with B)\with C$ as well.
We have $\bigtrado{A\with (B\with C)}=\trado{A}\with(\trado{B}\with\trado{C})\isoo(\trado{A}\with\trado{B})\with\trado{C}=\bigtrado{(A\with B)\with C}$.
The symmetric case follows the same pattern, as well as associativity of $\oplus$.
\item If $o = A\with B$ then $A$ and $B$ are output and $B\with A$ as well.
We have $\bigtrado{A\with B}=\trado{A}\with\trado{B}\isoo\trado{B}\with\trado{A}=\bigtrado{B\with A}$.
The commutativity of $\oplus$ follows the same pattern.
\item If $o = A\with\top$ then $A$ is output and $\bigtrado{A\with\top} = \trado{A}\with\top \isoo \trado{A}$.
The symmetric case follows the same pattern, as well as unitality for $\oplus$.
\item If $\iform = A\tens(B\oplus C)$ then $A$ is output and $B$ and $C$ are input, and $(A\tens B)\oplus(A\tens C)$ is input. We have:
  \begin{align*}
    \bigtradi{A\tens(B\oplus C)} &= \trado{A}\lolipop(\tradi{C}\with\tradi{B}) \\
       &\isoo (\trado{A}\lolipop\tradi{C})\with(\trado{A}\lolipop\tradi{B}) = \bigtradi{(A\tens B)\oplus(A\tens C)}
  \end{align*}
The symmetric case follows the same pattern, as well as distributivity of $\parr$ over $\with$.
\item If $\iform = A\tens 0$ then $A$ is output and $\bigtradi{A\tens 0} = \trado{A}\lolipop\top \isoo \top = \tradi{0}$.
The symmetric case follows the same pattern, as well as cancellation of $\parr$ by $\top$.
\end{itemize}
Assume now that $F\iso G$ in the class of SMCC with finite products.
As \MALL\ is such a SMCC with finite products, we have $\tradm{F}\iso\tradm{G}$ in \MALL, thus $\tradm{F}\isom\tradm{G}$ by \autoref{th:iso_complete}.
As $\tradm{F}$ is an output formula, by induction on the length of the equational derivation of $\tradm{F}\isom\tradm{G}$, we get that all the intermediary steps involve output formulas and each equation is mapped to $\isoo$ by $\bigtrado{\_}$ so that $\trado{\tradm{F}}\isoo\trado{\tradm{G}}$.
We finally get $F=\trado{\tradm{F}}\isoo\trado{\tradm{G}}= G$.

Conversely soundness easily comes from the definition of SMCC and products.
\end{proof}

%==================================================

\section{Conclusion}
\label{sec:conclu}

Extending the result of Balat and Di~Cosmo in~\cite{mllisos}, we give an equational theory characterising type isomorphisms in multiplicative-additive linear logic with units as well as in $\star$-autonomous categories with finite products: the one described in \autoref{tab:eqisos} on \autopageref{tab:eqisos} (and in \autoref{tab:eqisos_cat} for $\star$-autonomous categories).
We get as a sub-result that isomorphisms for additive linear logic (\resp\ unit-free additive linear logic) are given by the equational theory $\Eu$ restricted to additive (\resp\ additive unit-free) formulas -- and more generally this applies to any fragment of MALL, namely multiplicative-additive linear logic without units, with additive units but without multiplicative units, with multiplicative units but without additive units, as well as multiplicative linear logic with and without units.
Proof-nets were a major tool to prove completeness, as notions like fullness and \unique ness are much harder to define and manipulate in sequent calculus.
However, we could not use them for taking care of the (additive) units, because there is no known appropriate notion of proof-nets.
We have thus been forced to consider rule commutations in the sequent calculus of MALL with units.

The immediate question to address is the extension of our results to the characterization of type isomorphisms for full propositional linear logic, thus including the exponential connectives.
This is clearly not immediate since the interaction between additive and exponential connectives is not well described in proof-nets.

Another perspective is to look at categories.
A $\star$-autonomous category with finite products automatically has finite coproducts given by $F\oplus G \coloneqq ((F\lolipop\bot)\with(G\lolipop\bot))\lolipop\bot$ and $0 \coloneqq \top\lolipop\bot$.
From equations of \autoref{tab:eqisos_cat}, one can derive:
\begin{equation*}
\left.\begin{array}{rcl}
  F\oplus G &=& G\oplus F \\
  F\oplus (G\oplus H) &=& (F\oplus G)\oplus H \\
  F\oplus 0 &=& F \\
  F\tens(G\oplus H) &=& (F\tens G)\oplus(F\tens H) \\
  F\tens 0 &=& 0 \\
  (F\oplus G)\lolipop H &=& (F\lolipop H)\with(G\lolipop H) \\
  0\lolipop H &=& \top
\end{array}\qquad\right\}\thCP
\end{equation*}
In the weaker setting of SMCC, finite products do not induce finite coproducts.
It justifies the possibility of considering them separately.
We solved the case of products only (\autoref{thm:isos_smcc}) and have conjectures for SMCC with both products and coproducts on one side, and with coproducts only on the other side.

We conjecture that isomorphisms in SMCC with both finite products \emph{and} finite coproducts correspond to adding the equations of theory $\thCP$ to $\thS$ (\autoref{tab:eqisos_cat}).
Our approach through $\star$-autonomous categories does not help since $\top\lolipop(\top\oplus\top)$ and $(0\with 0)\lolipop 0$ are isomorphic in $\star$-autonomous categories but not in SMCC with finite products and coproducts.

About SMCC with finite coproducts only (without products), it is important to notice that an initial object $0$ in a SMCC induces that $0\lolipop F$ is a terminal object for any $F$.
Thus, the theory of isomorphisms includes the equation $0\lolipop F\iso 0\lolipop G$ even if it does not occur in $\thCP$ (it might be the only missing equation).
Regarding a characterization through $\star$-autonomous categories, it is again not possible since, if we denote by $\top$ a terminal object, $(\top\lolipop 0)\lolipop 0$ and $\top\lolipop(\top\tens\top)$ are isomorphic in $\star$-autonomous categories but not in SMCC with finite coproducts.

For the curious reader, we prove in \autoref{sec:conjectures_isos_smcc} that our two examples of isomorphisms in $\star$-autonomous categories are not isomorphic in the considered SMCC.
This is done by bringing one of our results about patterns in isomorphisms of multiplicative-additive linear
logic to the setting of intuitionistic multiplicative-additive linear logic.

A more general problem than isomorphisms is the study of type retractions, where only one of the two compositions yields an identity.
It is much more difficult -- see for example~\cite{retractmany}.
The question is mostly open in the case of linear logic.
Even in multiplicative linear logic (where there is for example a retraction between $A$ and $(A\lolipop A)\lolipop A = (A\tens A\orth)\parr A$ which is not an isomorphism, and where the associated proof-nets are not bipartite), no characterization is known.
In the multiplicative-additive fragment, the problem looks even harder, with more retractions; for instance the one depicted on \autoref{fig:no_anticut}, but there are also retractions between $A$ and $A\oplus A$.

%==================================================

\section*{Acknowledgements}

We warmly thank the anonymous referees for their useful and detailed remarks, and in particular for bringing our attention to~\cite{calculusformall}.

This work was performed within the framework of the LABEX MILYON (ANR-10-LABX-0070) of Universit\'e de Lyon, within the program ``Investissements d'Avenir'' (ANR-11-IDEX-0007), and supported by the projects DyVerSe (ANR-19-CE48-0010), QuaReMe (ANR-20-CE48-0005) and ReCiProg (ANR-21-CE48-0019), all operated by the French National Research Agency (ANR). This work was also supported by the IRN Linear Logic.

%==================================================

\bibliographystyle{alphaurl}
\bibliography{MALL_isos_FSCD_2023}

%==================================================

\newpage
\appendix

%==================================================

\section{Rule commutation is the core of cut-elimination}
\label{subsec:proofs_add_1}

In this appendix is given a complete proof of \autoref{cor:CReqc_simpl} (\ie\ \cite[Theorem~5.1]{calculusformall}).
We had no knowledge of~\cite{calculusformall} when writing the conference paper~\cite{DiGuardiaLaurent23} on which this journal paper is based.
Our initial proof of this result (extended with the mix-rules) is given in the first author's PhD thesis~\cite[Theorem~2.47]{phddiguardia}.
After comparing our work with the one of Cockett and Pastro, the core of the two proofs happens to be very similar.
There are three main differences between them, in addition to the used formalism.
The first one is the choice of the theorems from rewriting theory which are used: Cockett and Pastro apply a result~\cite[Proposition~B.3]{calculusformall} they proved that resembles~\cite[Theorem~2.2]{aototoyama2012} while we used a simpler (but less powerful) theorem~\cite{huet80}.
In particular, we proved a stronger result when closing rewriting diagrams (known as local coherence) while their theorem allowed them to prove a weaker one, that needs less technical tools.
The second difference is a simplification by Cockett and Pastro allowing not to consider $cut-cut$ commutations~\cite[Lemma~C.2]{calculusformall}.
The last one is the norm for normalization: while they provide one that decreases by cut-elimination and is stable through $\top$-free rule commutations, we provide a norm stable by all rule commutations as well as $cut-cut$ commutations (and still decreasing for other cut-elimination steps).
Furthermore, the proof of~\cite[Theorem~5.1]{calculusformall} seems to contains a mistake due to a $\with-\tens$ commutation duplicating a cut-rule (see more explanations in \autoref{subsec:gen_eqc}), and a critical pair is lacking when closing rewriting diagrams.

We present here a fusion of these two demonstrations, in the setting of sequent calculus, with as rewriting theorem~\cite[Theorem~2.2]{aototoyama2012}, with simplifications on $cut-cut$ commutations, corrected diagrams and the lacking case from~\cite{calculusformall}.
We also detail some arguments implicit in the work of Cockett and Pastro, \eg\ why there are no more critical pairs than the ones provided (and the one they forgot).
The main idea to patch the proof of~\cite{calculusformall} is to consider only rule commutations with no $cut$-rule above, as we did in our proof, and then to normalize above the considered rules before applying a commutation between them.
Note also that, contrary to~\cite{calculusformall}, we do not assume proofs to be axiom-expanded since it almost does not impact the reasoning.

We work in sequent calculus, hence by proof we mean a sequent calculus proof, never considering proof-nets.
We use standard notations from relation algebra and rewriting theory: given a relation $\lhd$, $\lhd^*$ (\resp\ $\lhd^+$, \resp\ $\lhd^=$) is the transitive reflexive (\resp\ transitive, \resp\ reflexive) closure of $\lhd$, while $\rhd$ is the converse relation -- symmetric relations will correspond to symmetric symbols.
We denote by $\cdot$ the composition of relations.

We deduce \autoref{cor:CReqc_simpl} from a Church-Rosser property, that will be found by applying a result from Aoto and Tomaya~\cite{aototoyama2012}.

\begin{defi}
\label{def:creq}
Let $\sim$ and $\to$ be relations on a set $A$ such that $\sim$ is an equivalence relation.
The relation $\to$ is \emph{Church-Rosser modulo $\sim$} if $(\to\cup\gets\cup\sim)^*\subseteq ~\to^*\cdot\sim\cdot~\leftindex^*{\gets}$ (see \autoref{fig:def_lconeq_lcoheq}).
\end{defi}

\begin{thmC}[{\cite[Theorem~2.2]{aototoyama2012}}]
\label{lem:aototomaya}
Let $\vdashv$, $\to$ and $\leadsto$ be relations on a set $A$ such that $\vdashv$ is symmetric and $\leadsto~\subseteq~\vdashv$.
Set ${\Rightarrow~=~\to \cup \leadsto}$.
Suppose:
\begin{enumerate}[label=\textnormal{(\roman*)}]
\item\label{item:aototomaya:1}
$\to \cdot \leadsto^*$ is strongly normalizing;
\item\label{item:aototomaya:2}
$\gets \cdot \to~\subseteq~\Rightarrow^* \cdot \vdashveq \cdot \leftindex^*{\Leftarrow}$;
\item\label{item:aototomaya:3}
$\vdashv \cdot \to~\subseteq~(\vdashveq \cdot \leftindex^*{\Leftarrow}) \cup (\to \cdot \Rightarrow^* \cdot \vdashveq \cdot \leftindex^*{\Leftarrow})$.
\end{enumerate}
Then $\to$ is Church-Rosser modulo $\astvdashv$.
\end{thmC}

We denote by $\betabarto$ a $\betato$ step other than a $cut-cut$ commutation, and call $\betabar$-equality $\eqbb$ the equivalence closure of $\betabarto$ steps (as $\beta$-equality $\eqb$ is the equivalence closure of $\betato$ steps).
Also note $\eqcc$ the $cut-cut$ commutation (which is a symmetric relation).

We recall $\eqcone$ (defined in \autoref{def:rule_comm}) is one-step rule commutation \emph{of cut-free MALL}, \ie\ which is not a commutation involving a $cut$-rule nor having above the commuted rules a sub-proof with a $cut$-rule (but it may have a $cut$-rule in its external context); for instance in the $\top-\tens$ commutation creating or deleting a sub-proof $\pi$, $\pi$ is cut-free.\footnote{Actually, we only need it for the $\top-\tens$ and $\with-\tens$ commutations, but ask it for all commutations to homogenize proofs.}
As $\eqcone$ is symmetric, $\eqc$ is an equivalence relation.

We will instantiate \autoref{lem:aototomaya} as follows.
\begin{cor}
\label{lem:aototomaya_use}
Suppose:
\begin{enumerate}[label=\textnormal{(\roman*)}]
\item\label{item:aototomaya_use:1}
$\betabarto \cdot \eqc$ is strongly normalizing;
\item\label{item:aototomaya_use:2}
$\betabarfrom \cdot \betabarto~\subseteq~\betabartostar \cdot \eqceq \cdot \betabarfromstar$;
\item\label{item:aototomaya_use:3}
$\eqcone \cdot \betabarto~\subseteq~\betabarto \cdot \eqc \cdot \betabartostar \cdot \eqc \cdot \betabarfromstar$.
\end{enumerate}
Then $\betabarto$ is Church-Rosser modulo $\eqc$.
\end{cor}
\begin{proof}
Direct use of \autoref{lem:aototomaya} with $\vdashv~\coloneqq~\eqcone$, $\leadsto~\coloneqq~\eqcone$ and $\to~\coloneqq~\betabarto$.
\end{proof}
We prove the first item in \autoref{sec:sn_betabar_eqc}, and the two other hypotheses in \autoref{sec:lcon_lcoh}.
Finally, we show in \autoref{sec:no_cut-cut} that $\betabar$-equality and $\beta$-equality coincide, leading to a proof of \autoref{cor:CReqc_simpl}.

\begin{figure}
\centering
\begin{tikzpicture}
\begin{myscope}
	\node[draw=none,minimum size=3mm] (u0) at (0,1) {$\cdot$};
	\node[draw=none,minimum size=3mm] (u1) at (1,2) {$\cdot$};
	\node[draw=none,minimum size=3mm] (u2) at (2,2) {$\cdot$};
	\node[draw=none,minimum size=3mm] (u3) at (3,1) {$\cdot$};
	\node[draw=none,minimum size=3mm] (u4) at (4,1) {$\cdot$};
	\node[draw=none,minimum size=3mm] (u5) at (5,2) {$\cdot$};
	\node[draw=none,minimum size=3mm] (u6) at (6,2) {$\cdot$};
	\node[draw=none,minimum size=3mm] at (6.5,2) {$\dots$};
	\node[draw=none,minimum size=3mm] (u7) at (7,2) {$\cdot$};
	\node[draw=none,minimum size=3mm] (u8) at (8,1) {$\cdot$};

	\path (u1) edge node[draw=none, sloped, align=center]{$*$ \\} (u0);
	\node[draw=none] at (1.5,2) {$\sim$};
	\path (u2) edge node[draw=none, sloped, align=center]{$*$ \\} (u3);
	\node[draw=none] at (3.5,1) {$\sim$};
	\path (u5) edge node[draw=none, sloped, align=center]{$*$ \\} (u4);
	\node[draw=none] at (5.5,2) {$\sim$};
	\path (u7) edge node[draw=none, sloped, align=center]{$*$ \\} (u8);
\end{myscope}
\begin{myscopec}{red}
	\node[draw=none,minimum size=3mm] (bl) at (3.5,0) {$\cdot$};
	\node[draw=none,minimum size=3mm] (br) at (4.5,0) {$\cdot$};
	\node[draw=none] at (4,0) {$\sim$};

	\path (u0) edge[dashed] node[draw=none, sloped, align=center]{$*$ \\} (bl);
	\path (u8) edge[dashed] node[draw=none, sloped, align=center]{$*$ \\} (br);
\end{myscopec}
\end{tikzpicture}
\caption{Diagram of Church-Rosser modulo $\sim$ (\autoref{def:creq}), with hypotheses in solid black and conclusions in dashed red}
\label{fig:def_lconeq_lcoheq}
\end{figure}

\subsection{Strong normalization}
\label{sec:sn_betabar_eqc}

The goal of this section is proving the strong normalization of $\betabarto \cdot \eqc$, namely \autoref{lem:beta_wn}.
We do so by giving a measure which is preserved by $\eqcone$ and decreases during a $\betabarto$ step.
This measure is quite different from the one of Cockett and Pastro~\cite[Section~B.1]{calculusformall}, for ours is stable by all rule commutations and not only $\top$-free ones.
It is also quite original for it depends only on the sequents and formulas we apply the $cut$-rules to, as well as which $cut$-rules are above which other $cut$-rules, but nothing else about the surrounding proof.
Such a measure is usually not good for cut-elimination, and seems hard to come by without a linear point of view.

\begin{defi}
\label{def:mass_formula}
The \emph{mass} $\mass{A}$ of a formula $A$ is a natural number defined by induction:
\begin{itemize}
\item $\mass{X} = \mass{X\orth} = \mass{1} = \mass{\bot} = \mass{\top} = \mass{0} = 2$
\item $\mass{A\tens B} = \mass{A\parr B} = \mass{A\with B} = \mass{A\oplus B} = (\mass{A} + 1) \times (\mass{B} + 1)$
\end{itemize}
We extend this notion to sequents by defining $\mass{A_1,\dots,A_n} = \prod_{i=1}^n \mass{A_i}$ (with the usual convention that the empty product equals $1$).

The \emph{mass} $\mass{c}$ of a $cut$-rule $c$ of the shape
{
\AxiomC{$\fCenter A\orth, \Gamma$}
\AxiomC{$\fCenter A, \Delta$}
\Rcut{$\Gamma, \Delta$}
\DisplayProof}
is defined as $\mass{c} = \mass{A} \times \mass{\Gamma,\Delta}$.
\end{defi}

\begin{fact}
\label{lem:mass_formula_pos}
For any formula $A$, $\mass{A} = \mass{A\orth} > 1$.
\end{fact}

Taking the multiset of the masses of $cut$-rules as our measure is not enough when the cut-elimination step we apply is not on a top-most $cut$-rule.
The problem comes from a $\with-cut$ commutative case where the duplicated sub-proof contains a cut-rule, whose mass may be bigger than the one of the reduced $cut$-rule.
A solution is to add to a $cut$-rule the masses of all $cut$-rules below it.

\begin{defi}
\label{def:density_proof}
In a proof $\pi$, we define a partial order $\ordercut$ between $cut$-rules by $c' \ordercut c$ when $c$ is in the sub-proof of $\pi$ of root $c'$.
This relation is reflexive: $c \ordercut c$.

The \emph{density} $\density{c}$ of a $cut$-rule $c$ in a proof $\pi$ is defined as $\density{c} = \sum_{c' \ordercut c} \mass{c'}$.

The \emph{density} $\density{\pi}$ of a proof $\pi$ is the multiset of the densities of its $cut$-rules.
\end{defi}

For example, consider the following proof $\pi$, with $cut$-rules $c_1$ to $c_5$:
\begin{center}
\resizebox{\textwidth}{!}{
\Rax{$X$}
\Rax{$X$}
\Rax{$X$}
\RightLabel{$c_1$}
\BinaryInfC{$\fCenter$ $X\orth, X$}
\Rplusone{$X\orth, X\oplus 0$}
\RightLabel{$c_2$}
\BinaryInfC{$\fCenter$ $X\orth, X\oplus 0$}
\Rtop{$\top, X$}
\Rax{$X$}
\Rwith{$\top\with X\orth, X$}
\RightLabel{$c_3$}
\BinaryInfC{$\fCenter$ $X\orth,X$}
\Rax{$X$}
\Rax{$X$}
\RightLabel{$c_4$}
\BinaryInfC{$\fCenter$ $X\orth,X$}
\RightLabel{$c_5$}
\BinaryInfC{$\fCenter$ $X\orth,X$}
\DisplayProof
}
\end{center}
The masses of the $cut$-rules are
$\mass{c_1} = \mass{c_4} = \mass{c_5} = 8$ and $\mass{c_2} = \mass{c_3} = 36$.
The ordering of $cut$-rules is $c_5\ordercut c_3 \ordercut c_2 \ordercut c_1$ and $c_5\ordercut c_4$.
Thus, $\density{c_5} = \mass{c_5}$, $\density{c_4} = \mass{c_5} + \mass{c_4}$, $\density{c_3} = \mass{c_5} + \mass{c_3}$, $\density{c_2} = \mass{c_5} + \mass{c_3} + \mass{c_2}$ and $\density{c_1} = \mass{c_5} + \mass{c_3} + \mass{c_2} + \mass{c_1}$.

Proving that a cut-elimination step $\betabarto$ decreases density is a matter of checking cases, while it is easy to show a rule commutation preserves it.

\begin{lem}
\label{lem:betabar_mass}
Consider a cut-elimination step $\tau \betabarto \phi$ and call $c$ the (unique) $cut$-rule in $\tau$ involved in this step and $(c_i)_{i\in\{ 1 ; \dots ; n \}}$ the $n \in \{ 0 ; 1 ; 2 \}$ resulting $cut$-rules in $\phi$.
Then $\mass{c} > \sum_{i=1}^n \mass{c_i}$.
\end{lem}
\begin{proof}
It suffices to compute the masses before and after each cut-elimination step.
In the following case study, we use implicitly \autoref{lem:mass_formula_pos}.

\proofpar{$ax$ key case}
Here

\begin{adjustbox}{}
$
\tau = \quad
\Rax{$A$}
\Rsub{$\pi$}{$A, \Gamma$}
\RightLabel{$c$}
\BinaryInfC{$\fCenter A,\Gamma$}
\noLine
\UnaryInfC{$\rho$}
\DisplayProof
\qquad
\phi = \quad
\Rsub{$\pi$}{$A, \Gamma$}
\noLine
\UnaryInfC{$\rho$}
\DisplayProof
$
\end{adjustbox}
We have $n = 0$ and $\mass{c} = \mass{A} \times \mass{A} \times \mass{\Gamma} > 0 = \sum_{i=1}^n \mass{c_i}$.

\proofpar{$\parr-\tens$ key case}
Here

\begin{adjustbox}{}
$
\tau = \quad
\Rsub{$\pi_1$}{$A, \Gamma$}
\Rsub{$\pi_2$}{$B, \Delta$}
\Rtens{$A\tens B,\Gamma,\Delta$}
\Rsub{$\pi_3$}{$B\orth,A\orth,\Sigma$}
\Rparr{$B\orth\parr A\orth,\Sigma$}
\RightLabel{$c$}
\BinaryInfC{$\fCenter \Gamma,\Delta,\Sigma$}
\noLine
\UnaryInfC{$\rho$}
\DisplayProof
\qquad
\phi = \quad
\Rsub{$\pi_1$}{$A, \Gamma$}
\Rsub{$\pi_2$}{$B, \Delta$}
\Rsub{$\pi_3$}{$B\orth,A\orth,\Sigma$}
\RightLabel{$c_2$}
\BinaryInfC{$\fCenter A\orth,\Delta,\Sigma$}
\RightLabel{$c_1$}
\BinaryInfC{$\fCenter \Gamma,\Delta,\Sigma$}
\noLine
\UnaryInfC{$\rho$}
\DisplayProof
$
\end{adjustbox}
We have $n = 2$ and:
\begin{flalign*}
& \mass{c} - \sum_{i=1}^n \mass{c_i}\\
&= \mass{A\tens B} \times \mass{\Gamma, \Delta, \Sigma} - \mass{A} \times \mass{\Gamma, \Delta, \Sigma} - \mass{B} \times \mass{A\orth, \Delta, \Sigma}\\
&\geq (\mass{A\tens B} - \mass{A} - \mass{B} \times \mass{A}) \times \mass{\Gamma, \Delta, \Sigma}\\
&= (\mass{B} + 1) \times \mass{\Gamma, \Delta, \Sigma} > 0
\end{flalign*}

\proofpar{$\with-\oplus_1$ key case}
Here

\begin{adjustbox}{}
$
\tau = \quad
\Rsub{$\pi_1$}{$A, \Gamma$}
\Rsub{$\pi_2$}{$B, \Gamma$}
\Rwith{$A\with B,\Gamma$}
\Rsub{$\pi_3$}{$B\orth,\Delta$}
\Rplusone{$B\orth\oplus A\orth,\Delta$}
\RightLabel{$c$}
\BinaryInfC{$\fCenter \Gamma,\Delta$}
\noLine
\UnaryInfC{$\rho$}
\DisplayProof
\qquad
\phi = \quad
\Rsub{$\pi_2$}{$B, \Gamma$}
\Rsub{$\pi_3$}{$B\orth,\Delta$}
\RightLabel{$c_1$}
\BinaryInfC{$\fCenter \Gamma,\Delta$}
\noLine
\UnaryInfC{$\rho$}
\DisplayProof
$
\end{adjustbox}
We have $n = 1$ and $\mass{c} = \mass{A \with B} \times \mass{\Gamma,\Delta} > \mass{B} \times \mass{\Gamma,\Delta} = \sum_{i=1}^n \mass{c_i}$.

\proofpar{$\with-\oplus_2$ key case}
This case is very similar to the $\with-\oplus_1$ key case.

\proofpar{$\bot-1$ key case}
Here

\begin{adjustbox}{}
$
\tau = \quad
\Rone{}
\Rsub{$\pi$}{$\Gamma$}
\Rbot{$\Gamma,\bot$}
\RightLabel{$c$}
\BinaryInfC{$\fCenter \Gamma$}
\noLine
\UnaryInfC{$\rho$}
\DisplayProof
\qquad
\phi = \quad
\Rsub{$\pi$}{$\Gamma$}
\noLine
\UnaryInfC{$\rho$}
\DisplayProof
$
\end{adjustbox}
We have $n = 0$ and $\mass{c} = 2 \times \mass{\Gamma}  > 0 = \sum_{i=1}^n \mass{c_i}$.

\proofpar{$\parr-cut$ commutative case}
Here

\begin{adjustbox}{}
$
\tau = \quad
\Rsub{$\pi_1$}{$A,B,C,\Gamma$}
\Rparr{$A,B\parr C,\Gamma$}
\Rsub{$\pi_2$}{$A\orth,\Delta$}
\RightLabel{$c$}
\BinaryInfC{$\fCenter B\parr C,\Gamma,\Delta$}
\noLine
\UnaryInfC{$\rho$}
\DisplayProof
\qquad
\phi = \quad
\Rsub{$\pi_1$}{$A,B,C,\Gamma$}
\Rsub{$\pi_2$}{$A\orth,\Delta$}
\RightLabel{$c_1$}
\BinaryInfC{$\fCenter B, C,\Gamma,\Delta$}
\Rparr{$B\parr C,\Gamma,\Delta$}
\noLine
\UnaryInfC{$\rho$}
\DisplayProof
$
\end{adjustbox}
We have $n = 1$ and:
\begin{flalign*}
& \mass{c} - \sum_{i=1}^n \mass{c_i}\\
&= \mass{A} \times \mass{B\parr C} \times \mass{\Gamma, \Delta} - \mass{A} \times \mass{B} \times \mass{C} \times \mass{\Gamma, \Delta} \\
&= \mass{A} \times (\mass{B\parr C} - \mass{B} \times \mass{C}) \times \mass{\Gamma, \Delta}\\
&= \mass{A} \times (\mass{B} + \mass{C} + 1) \times \mass{\Gamma, \Delta} > 0
\end{flalign*}

\proofpar{$\tens-cut-1$, $\tens-cut-2$, $\oplus_1-cut$, $\oplus_2-cut$ and $\bot-cut$ commutative cases}
These cases are quite similar to the $\parr-cut$ commutative case.

\proofpar{$\with-cut$ commutative case}
Here

\begin{adjustbox}{}
$
\tau = \quad
\Rsub{$\pi_1$}{$A,B,\Gamma$}
\Rsub{$\pi_2$}{$A,C,\Gamma$}
\Rwith{$A,B\with C,\Gamma$}
\Rsub{$\pi_3$}{$A\orth,\Delta$}
\RightLabel{$c$}
\BinaryInfC{$\fCenter B\with C,\Gamma,\Delta$}
\noLine
\UnaryInfC{$\rho$}
\DisplayProof
\qquad
\phi = \quad
\Rsub{$\pi_1$}{$A,B,\Gamma$}
\Rsub{$\pi_3$}{$A\orth,\Delta$}
\RightLabel{$c_1$}
\BinaryInfC{$\fCenter B,\Gamma,\Delta$}
\Rsub{$\pi_2$}{$A,C,\Gamma$}
\Rsub{$\pi_3$}{$A\orth,\Delta$}
\RightLabel{$c_2$}
\BinaryInfC{$\fCenter C,\Gamma,\Delta$}
\Rwith{$B\with C,\Gamma,\Delta$}
\noLine
\UnaryInfC{$\rho$}
\DisplayProof
$
\end{adjustbox}
We have $n = 2$ and:
\begin{flalign*}
& \mass{c} - \sum_{i=1}^n \mass{c_i}\\
&= \mass{A} \times \mass{\Gamma,\Delta} \times \mass{B\with C} - \mass{A} \times \mass{\Gamma,\Delta} \times \mass{B} -\mass{A} \times \mass{\Gamma,\Delta} \times \mass{C} \\
&= \mass{A} \times \mass{\Gamma, \Delta} \times (\mass{B\with C} - \mass{B} - \mass{C})\\
&= \mass{A} \times \mass{\Gamma, \Delta} \times (\mass{B} \times \mass{C} + 1) > 0
\end{flalign*}

\proofpar{$\top-cut$ commutative case}
Here

\begin{adjustbox}{}
$
\tau = \quad
\Rtop{$A,\top,\Gamma$}
\Rsub{$\pi$}{$A\orth,\Delta$}
\RightLabel{$c$}
\BinaryInfC{$\fCenter \top,\Gamma,\Delta$}
\noLine
\UnaryInfC{$\rho$}
\DisplayProof
\qquad\qquad
\phi = \quad
\Rtop{$\top,\Gamma,\Delta$}
\noLine
\UnaryInfC{$\rho$}
\DisplayProof
$
\end{adjustbox}
We have $n = 0$ and $\mass{c} = 2 \times \mass{A} \times \mass{\Gamma,\Delta} > 0 = \sum_{i=1}^n \mass{c_i}$.
\end{proof}

\begin{lem}
\label{lem:betabar_density}
If $\tau \betabarto \phi$ then $\density{\tau} > \density{\phi}$.
\end{lem}
\begin{proof}
Using \autoref{lem:betabar_mass}, each step replaces a $cut$-rule with $cut$-rules whose sum of masses is smaller.
Other masses stay the same as the only modified sequents are not the conclusion sequent of other $cut$-rules than the one eliminated.
The only reduction where there can be more $cut$-rules below a given one, notwithstanding the replacement of the eliminated one, is when a sub-proof containing a $cut$-rule is duplicated in a $\with-cut$ commutative case.
In such a case, the only $cut$-rules not keeping their densities are $c$ the eliminated one -- yielding two $cut$-rules $c_1$ and $c_2$ of smaller masses -- and the $cut$-rules above it, by definition of $\ordercut$.
Any duplicated $cut$-rule goes from a density $\alpha + \mass{c}$ to two $cut$-rules of smaller densities $\alpha + \mass{c_1}$ and $\alpha + \mass{c_2}$, and other $cut$-rules above go from $\alpha + \mass{c}$ to either $\alpha + \mass{c_1}$ or $\alpha + \mass{c_2}$.
Thus, $\density{\tau} > \density{\phi}$.
\end{proof}

\begin{lem}
\label{lem:eqc_density}
If $\pi \eqcone \pi'$, then $\density{\pi} = \density{\pi'}$.
\end{lem}
\begin{proof}
As rule commutations act below a cut-free proof, they cannot erase nor duplicate $cut$-rules.
As a rule commutation only changes the sequents between the rules it commutes, it does not change the sequent below any $cut$-rule, and in particular does not modify the mass of any $cut$-rule.
Thence, it preserves the density of all $cut$-rules, which depends only on the masses of the $cut$-rules.
\end{proof}

\begin{prop}
\label{lem:beta_wn}
The relation $\betabarto \cdot \eqc$ is strongly normalizing.
\end{prop}
\begin{proof}
By Lemmas~\ref{lem:betabar_density} and~\ref{lem:eqc_density}, a step of $\betabarto$ decreases the density of a proof while one of $\eqcone$ preserves it.
Hence, a step of $\betabarto \cdot \eqc$ decreases the density, ensuring termination.
\end{proof}

\begin{cor}
\label{cor:beta_wn}
Cut-elimination $\betato$ and $\betabarto$ are weakly normalizing, with as normal forms cut-free proofs.
\end{cor}
\begin{proof}
Using \autoref{lem:beta_wn}, one can reach a normal form for $\betabarto~\subseteq~\betato$.
But proofs in normal form for $\betabarto$ correspond to cut-free proofs: as long as there is a $cut$-rule, a $\betabarto$ step can be applied.
Thus, no $\eqcc$ step can be applied on the reached normal form either.
\end{proof}

Having a measure preserved not only by rule commutations, but also by cut-cut commutations, is doable but much more complicated~\cite{phddiguardia}.

%\begin{rem}
%\label{rem:weight_eta}
%Axiom-expansion $\etato$ preserves the density of a proof, with a similar argument as \autoref{lem:eqc_density}.
%Thus, $\betabarto \cdot (\eqcone \cup \etato)^\ast$ is also strongly normalizing.
%\end{rem}

\subsection{Reduction diagrams}
\label{sec:lcon_lcoh}

We prove here Items~\ref{item:aototomaya_use:2} and~\ref{item:aototomaya_use:3} of \autoref{lem:aototomaya_use}.

In this section we denote graphically some proofs with the following convention.
When writing proofs as in

\begin{adjustbox}{}
\AxiomC{$\rho$}
\RightLabel{$r_2$}\UnaryInfC{\emptyforproof}
\RightLabel{$r_1$}\UnaryInfC{\emptyforproof}
\DisplayProof
$\eqcone$
\AxiomC{$\rho$}
\RightLabel{$r_1$}\UnaryInfC{\emptyforproof}
\RightLabel{$r_2$}\UnaryInfC{\emptyforproof}
\DisplayProof
\end{adjustbox}
we abuse notations in the cases where $r_1$ or $r_2$ is a $\with$- or $\top$-rule.
The meaning is that, if say $r_1$ is a $\with$-rule, then $r_2$ is duplicated, and even possibly a whole sub-proof if $r_2$ is a $\tens$-rule for instance.
Similarly, if $r_1$ is a $\top$-rule, then this schema means that on the left hand-side $r_2$ and $\rho$ are not here, and are created by the $\top$-commutation.

We begin by proving \autoref{item:aototomaya_use:2}, corresponding to reduction diagrams~0, 1 and~2 of~\cite[Section~B.2.1]{calculusformall} (including all three versions of these diagrams, namely one without units, one for additive units and one for multiplicative units).
The main difference compared to~\cite{calculusformall} is on the cases corresponding to their reduction diagram~2, which here corresponds to \autoref{fig:proof_betafrom_betato}: since rule commutations are allowed only in the absence of $cut$-rules above, we have to normalize the sub-proof(s) above the considered rules before commuting them.

\begin{lem}
\label{lem:betafrom_betato}
Let $\pi$, $\pi_1$ and $\pi_2$ be MALL proofs such that $\pi_1\betabarfrom\pi\betabarto\pi_2$.
Then there exist $\pi_1'$ and $\pi_2'$ such that $\pi_1\betabartostar\pi_1'\eqceq\pi_2'\betabarfromstar\pi_2$.
Diagrammatically:
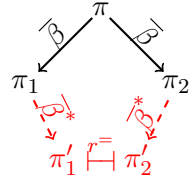
\begin{center}
\begin{tikzpicture}
\begin{myscope}
	\node[draw=none,minimum size=3mm] (pit) at (0,2) {$\pi$};
	\node[draw=none,minimum size=3mm] (pil) at (-1,1) {$\pi_1$};
	\node[draw=none,minimum size=3mm] (pir) at (1,1) {$\pi_2$};

	\path (pit) \edgelabel{$\betabar$} (pil);
	\path (pit) \edgelabel{$\betabar$} (pir);
\end{myscope}
\begin{myscopec}{red}
	\node[draw=none,minimum size=3mm] (varpil) at (-.5,0) {$\pi_1'$};
	\node[draw=none,minimum size=3mm] (varpir) at (.5,0) {$\pi_2'$};

	\path (pil) edge[dashed] node[draw=none, sloped, align=center]{$\betabar^*$ \\} (varpil);
	\node[draw=none] at (0,0) {$\eqceq$};
	\path (pir) edge[dashed] node[draw=none, sloped, align=center]{$\betabar^*$ \\} (varpir);
\end{myscopec}
\end{tikzpicture}
\end{center}
More precisely, we need a step of $\eqcone$ exactly when both $\betabarto$ steps are different commutative cases on the same $cut$-rule.
\end{lem}
\begin{proof}
If the $\pi\betabarto\pi_1$ and $\pi\betabarto\pi_2$ steps involve only distinct rules then, taking into account that rules of one may be duplicated or erased by the other step, they commute and we have a proof $\pi'$ such that $\pi_1\betabartostar\pi'\betabarfromstar\pi_2$, by applying one reduction after the other.
%With more details:
%\begin{itemize}
%\item $\pi_1\betabarto\cdot\betabarfrom\pi_2$ if neither step duplicates nor erases the other
%\item $\pi_1\betabarfrom\pi_2$ if $\pi\betabarto\pi_1$ erases a sub-proof containing the rules involved in $\pi\betabarto\pi_2$ (and symmetrically if we swap indices $1$ and $2$)
%\item $\pi_1\betabarto\cdot\betabarto\cdot\betabarfrom\pi_2$ if $\pi\betabarto\pi_1$ duplicates a sub-proof containing the rules involved in $\pi\betabarto\pi_2$ (and symmetrically if we swap indices $1$ and $2$)
%\end{itemize}
%There is no more case, as to duplicates or erases rules, a $\betabarto$ step must be below these rules in a proof. In particular, both steps cannot duplicate or erase the rules of the other step, because the rules involved are distinct.

From now on, we assume both steps involve (at least) one common rule.
If both reductions share all of their rules, then the two reductions are the same: $\pi_1=\pi_2$ and we are done (recall \autoref{rem:tens_parr_key_cut} for our convention on the $\parr-\tens$ key case).
Hence, we assume they do not share all of their rules.
We distinguish cases according to the kinds of the $\betabarto$ steps.

\proofpar{If one step is a key case other than an $ax$ one.}
Remark that on the three rules of a non-$ax$ key case, no other $\betabarto$ step can be applied (only a $cut-cut$ commutation could have been applied, but this case does not belong to $\betabarto$).
Whenceforth, this case cannot happen as it would lead to the two reductions sharing all of their rules.

\proofpar{If both steps are $ax$ key cases.}
(This corresponds to the reduction diagram~0 of~\cite[Section~B.2.1]{calculusformall}.)
As the two reductions share one rule, but not all rules, the shared rule must be the $cut$-rule, with as premises two $ax$-rules.
We can check that this critical pair leads to the same resulting proof from both choices of cut-elimination.
Thus $\pi_1=\pi_2$.

\proofpar{If one step is an $ax$ key case and the other a commutative case.}
(This corresponds to the reduction diagram~1 of~\cite[Section~B.2.1]{calculusformall}.)
By symmetry, assume $\pi\betabarto\pi_2$ is the $ax$ key case.
For the two reductions share a rule, and the $ax$-rule cannot participate in a commutative step, the shared rule must be the $cut$-rule.
We can still do this $ax$ key step after the commutation (maybe twice in case of duplication, or zero time in case of erasure), recovering $\pi_2$.
Thus:
\begin{itemize}
\item $\pi_1\betabarto\pi_2$ ($ax$-key case and not a $\with-cut$ nor $\top-cut$ commutative case) or
\item $\pi_1\betabarto\cdot\betabarto\pi_2$ ($ax$-key case and $\with-cut$ commutative case) or
\item $\pi_1=\pi_2$ ($ax$-key case and $\top-cut$ commutative case).
\end{itemize}

\proofpar{If both steps are commutative cases.}
(This corresponds to the reduction diagram~2 of~\cite[Section~B.2.1]{calculusformall}, with a modification.)
Here again, as the two reductions share a rule, it must be the $cut$-rule, because there is at most one $cut$-rule directly below a given rule.
As the reductions do not share both of their rules, in $\pi\betabarto\pi_1$ we sent a rule $r_1$ from a branch of the $cut$ below it, and in $\pi\betabarto\pi_2$ we do similarly on a rule $r_2$ in the other branch.
This case, more complex than the previous ones, is depicted schematically on \autoref{fig:proof_betafrom_betato}.
We can in $\pi_1$ commute the $cut$-rule and $r_2$ -- maybe twice in case of a duplication, or zero in case of an erasure -- obtaining $\pi_1'$, and similarly in $\pi_2$ the $cut$-rule and $r_1$, yielding $\pi_2'$.
The two resulting proofs differ exactly by a commutation of $r_1$ and $r_2$ (even if both are $\top$-rules, they differ by a $\top-\top$ commutation).
For we apply rule commutation only on cut-free sub-proofs, we first eliminate all $cut$-rules above these two rules, in the same way in all sub-proofs (and in case of duplication, in the same way in all duplicates of the sub-proofs).
This can be done thanks to weak normalization of $\betabarto$ (\autoref{cor:beta_wn}).
Finally, we commute $r_1$ and $r_2$, and thus $\pi_1\betabartostar\cdot\eqcone\cdot\betabarfromstar\pi_2$ if both steps are commutative cases.

\begin{figure}
\begin{adjustbox}{}
\begin{tikzpicture}
	\node (pil) at (0,0) {$\pi_1=${\AxiomC{$\rho_1$}\AxiomC{$\rho_2$}\RightLabel{$r_2$}\UnaryInfC{\emptyforproof}\RightLabel{$cut$}\BinaryInfC{\emptyforproof}\RightLabel{$r_1$}\UnaryInfC{$\rho_3$}\DisplayProof}};
	\node (pit) at (3,3) {$\pi=${\AxiomC{$\rho_1$}\RightLabel{$r_1$}\UnaryInfC{\emptyforproof}\AxiomC{$\rho_2$}\RightLabel{$r_2$}\UnaryInfC{\emptyforproof}\RightLabel{$cut$}\BinaryInfC{$\rho_3$}\DisplayProof}};
	\node (pir) at (6,0) {$\pi_2=${\AxiomC{$\rho_1$}\RightLabel{$r_1$}\UnaryInfC{\emptyforproof}\AxiomC{$\rho_2$}\RightLabel{$cut$}\BinaryInfC{\emptyforproof}\RightLabel{$r_2$}\UnaryInfC{$\rho_3$}\DisplayProof}};

	\node (pil1) at (0,-3) {$\pi_1'=${\AxiomC{$\rho_1$}\AxiomC{$\rho_2$}\RightLabel{$cut$}\BinaryInfC{\emptyforproof}\RightLabel{$r_2$}\UnaryInfC{\emptyforproof}\RightLabel{$r_1$}\UnaryInfC{$\rho_3$}\DisplayProof}};
	\node (pir1) at (6,-3) {$\pi_2'=${\AxiomC{$\rho_1$}\AxiomC{$\rho_2$}\RightLabel{$cut$}\BinaryInfC{\emptyforproof}\RightLabel{$r_1$}\UnaryInfC{\emptyforproof}\RightLabel{$r_2$}\UnaryInfC{$\rho_3$}\DisplayProof}};

	\node (pil2) at (0,-6) {{\AxiomC{$\tau^\text{cut-free}$}\RightLabel{$r_2$}\UnaryInfC{\emptyforproof}\RightLabel{$r_1$}\UnaryInfC{$\rho_3$}\DisplayProof}};
	\node (pir2) at (6,-6) {{\AxiomC{$\tau^\text{cut-free}$}\RightLabel{$r_1$}\UnaryInfC{\emptyforproof}\RightLabel{$r_2$}\UnaryInfC{$\rho_3$}\DisplayProof}};
\begin{myscope}
	\path (pit) \edgelabel{$\betabar$} (pil);
	\path (pit) \edgelabel{$\betabar$} (pir);
\end{myscope}
\begin{myscopec}{red}
	\path (pil) \edgelabel{$\betabar^*$} (pil1);
	\path (pir) \edgelabel{$\betabar^*$} (pir1);

	\path (pil1) \edgelabel{$\betabar^*$} (pil2);
	\path (pir1) \edgelabel{$\betabar^*$} (pir2);

	\path (pil2) \edgelabeld{$\eqcone$} (pir2);
\end{myscopec}
\end{tikzpicture}
\end{adjustbox}
\caption{Schematic representation of the last case of the proof of \autoref{lem:betafrom_betato}}
\label{fig:proof_betafrom_betato}
\end{figure}

The only difficulty here is proving that (the normal forms of) the two proofs $\pi_1'$ and $\pi_2'$ differ by a commutation of $r_1$ and $r_2$ as claimed.
This is a simple but tedious case analysis on the kind of rules $r_1$ and $r_2$ can be, and checking that in every case we indeed can apply an $\eqcone$ step.
As there are $8$ possible commutative cases for $\betabarto$ and we have two such steps, this leads to $8^2 = 64$ cases.
We give here only one of those, where $r_1$ is a $\with$-rule and $r_2$ is a $\parr$-rule.
In this case, our proofs are:
{\allowdisplaybreaks
\begin{flalign*}
\pi &=
\Rsub{$\rho_1$}{$A\orth, B, \Gamma$}
\Rsub{$\rho_2$}{$A\orth, C, \Gamma$}
\Rwith{$A\orth, B\with C, \Gamma$}
\Rsub{$\rho_3$}{$A, D, E, \Delta$}
\Rparr{$A, D\parr E, \Delta$}
\Rcut{$B\with C, D\parr E, \Gamma, \Delta$}
\DisplayProof\\
\pi_1 &=
\Rsub{$\rho_1$}{$A\orth, B, \Gamma$}
\Rsub{$\rho_3$}{$A, D, E, \Delta$}
\Rparr{$A, D\parr E, \Delta$}
\Rcut{$B, D\parr E, \Gamma, \Delta$}
\Rsub{$\rho_2$}{$A\orth, C, \Gamma$}
\Rsub{$\rho_3$}{$A, D, E, \Delta$}
\Rparr{$A, D\parr E, \Delta$}
\Rcut{$C, D\parr E, \Gamma, \Delta$}
\Rwith{$B\with C, D\parr E, \Gamma, \Delta$}
\DisplayProof\\
\pi_2 &=
\Rsub{$\rho_1$}{$A\orth, B, \Gamma$}
\Rsub{$\rho_2$}{$A\orth, C, \Gamma$}
\Rwith{$A\orth, B\with C, \Gamma$}
\Rsub{$\rho_3$}{$A, D, E, \Delta$}
\Rcut{$B\with C, D, E, \Gamma, \Delta$}
\Rparr{$B\with C, D\parr E, \Gamma, \Delta$}
\DisplayProof
\end{flalign*}
}
Following the method described above, we apply two $\parr-cut$ commutative cases in $\pi_1$ and one $\with-cut$ commutative case in $\pi_2$, obtaining:
\begin{flalign*}
\pi_1' &=
\Rsub{$\rho_1$}{$A\orth, B, \Gamma$}
\Rsub{$\rho_3$}{$A, D, E, \Delta$}
\Rcut{$B, D, E, \Gamma, \Delta$}
\Rparr{$B, D\parr E, \Gamma, \Delta$}
\Rsub{$\rho_2$}{$A\orth, C, \Gamma$}
\Rsub{$\rho_3$}{$A, D, E, \Delta$}
\Rcut{$C, D, E, \Gamma, \Delta$}
\Rparr{$C, D\parr E, \Gamma, \Delta$}
\Rwith{$B\with C, D\parr E, \Gamma, \Delta$}
\DisplayProof\\
\pi_2' &=
\Rsub{$\rho_1$}{$A\orth, B, \Gamma$}
\Rsub{$\rho_3$}{$A, D, E, \Delta$}
\Rcut{$B, D, E, \Gamma, \Delta$}
\Rsub{$\rho_2$}{$A\orth, C, \Gamma$}
\Rsub{$\rho_3$}{$A, D, E, \Delta$}
\Rcut{$C, D, E, \Gamma, \Delta$}
\Rwith{$B\with C, D, E, \Gamma, \Delta$}
\Rparr{$B\with C, D\parr E, \Gamma, \Delta$}
\DisplayProof
\end{flalign*}
Then, we eliminate all $cut$-rules above these $\with$ and $\parr$-rules, in the same way in both proofs (thanks to \autoref{cor:beta_wn}), yielding:
\begin{flalign*}
\pi_1'' &=
\Rsub{$\tau_1^{cut-free}$}{$B, D, E, \Gamma, \Delta$}
\Rparr{$B, D\parr E, \Gamma, \Delta$}
\Rsub{$\tau_2^{cut-free}$}{$C, D, E, \Gamma, \Delta$}
\Rparr{$C, D\parr E, \Gamma, \Delta$}
\Rwith{$B\with C, D\parr E, \Gamma, \Delta$}
\DisplayProof\\
\pi_2'' &=
\Rsub{$\tau_1^{cut-free}$}{$B, D, E, \Gamma, \Delta$}
\Rsub{$\tau_2^{cut-free}$}{$C, D, E, \Gamma, \Delta$}
\Rwith{$B\with C, D, E, \Gamma, \Delta$}
\Rparr{$B\with C, D\parr E, \Gamma, \Delta$}
\DisplayProof
\end{flalign*}
We observe these two last proofs are equal up to a $\with-\parr$ commutation, which gives us that $\pi_1\betabarto\cdot\betabarto\pi_1'\betabartoplus\pi_1''\eqcone\pi_2''\betabarfromplus\pi_2'\betabarfrom\pi_2$.
\end{proof}

Now, let us prove \autoref{item:aototomaya_use:3} of \autoref{lem:aototomaya_use}, corresponding to reduction diagrams~3 to~6 of~\cite[Section~B.2.1]{calculusformall} (again including all three versions of these diagrams).
As for the previous lemma, the main difference with~\cite{calculusformall} is that we normalize above the rules we wish to commute, which is a new step to integrate in reduction diagrams~3 and~4 of~\cite[Section~B.2.1]{calculusformall}.
We also consider a critical pair that was missing in~\cite[Section~B.2.1]{calculusformall}.

\begin{lem}
\label{lem:eqc_betato}
Let $\pi_1$, $\pi_2$ and $\pi_3$ be proofs such that $\pi_1\eqcone\pi_2\betabarto\pi_3$.
Then, there exist $\pi_1'$, $\pi_2'$, $\pi_3'$ and $\pi_4'$ such that $\pi_1\betabarto \pi_1' \eqc \pi_2' \betabartostar \pi_3' \eqc \pi_4' \betabarfromstar\pi_3$.
Diagrammatically:
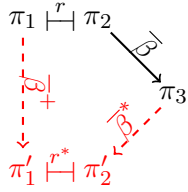
\begin{center}
\begin{tikzpicture}
\begin{myscope}
	\node[draw=none,minimum size=3mm] (pi1) at (0,2) {$\pi_1$};
	\node[draw=none,minimum size=3mm] (pi2) at (1,2) {$\pi_2$};
	\node[draw=none,minimum size=3mm] (pi3) at (2,1) {$\pi_3$};

	\node[draw=none] at (.5,2) {$\eqcone$};
	\path (pi2) \edgelabel{$\betabar$} (pi3);
\end{myscope}
\begin{myscopec}{red}
	\node[draw=none,minimum size=3mm] (c0) at (-1,1) {$\pi_1'$};
	\node[draw=none,minimum size=3mm] (c1) at (-2,1) {$\pi_2'$};
	\node[draw=none,minimum size=3mm] (c2) at (0,-.5) {$\pi_3'$};
	\node[draw=none,minimum size=3mm] (c3) at (1,-.5) {$\pi_4'$};

	\path (pi1) edge[dashed] node[draw=none, sloped, align=center]{$\betabar$ \\} (c0);
	\node[draw=none] at (-1.5,1) {$\eqc$};
	\path (c1) edge[dashed] node[draw=none, sloped, align=center]{\\ \\$\betabar^*$} (c2);
	\node[draw=none] at (.5,-.5) {$\eqc$};
	\path (pi3) edge[dashed] node[draw=none, sloped, align=center]{$\betabar^*$ \\} (c3);
\end{myscopec}
\end{tikzpicture}
\end{center}
\end{lem}
\begin{proof}
An easily handled case is when the $\pi_1\eqcone\pi_2$ and $\pi_2\betabarto\pi_3$ steps involve only distinct rules; assume for now this is the case.
A first general sub-case is when rules of one step are neither erased nor duplicated by the other.
Then these steps commute and $\pi_1\betabarto\cdot\eqcone\pi_3$, using the same steps in the other order (because these are local transformations).

Now, consider the case where the two steps still involve distinct rules, but the $\pi_2\betabarto\pi_3$ step duplicates a sub-proof containing the rules of $\pi_1\eqcone\pi_2$ (which may happen if $\pi_2\betabarto\pi_3$ is a $\with-cut$ commutative case).
We apply the corresponding $\betabarto$ step first in $\pi_1$, yielding $\pi_1\betabarto\pi_1^1$, and then the $\eqcone$ step twice, once for each occurrence, to recover $\pi_3$: we get $\pi_1\betabarto\pi_1^1\eqcone\cdot\eqcone\pi_3$.

Another general case is when the rules involved in the two steps are distinct, but the $\betabarto$ step eliminates a sub-proof containing the rules of the $\eqcone$ step (this can arise when using a $\with-\oplus_i$ key case or a $\top-cut$ commutative case).
In this case, doing in $\pi_1$ the $\betabarto$ step directly yields $\pi_3$: $\pi_1\betabarto\pi_3$.

Remark that, if these steps use distinct rules, the $\pi_1\eqcone\pi_2$ step cannot duplicate nor erase the rules involved in $\pi_2\betabarto\pi_3$.
Indeed, this may happen if the $\eqcone$ step is a $\with-\tens$ or $\top-\tens$ commutative case, but we assumed that a sub-proof corresponding to a rule commutation is cut-free, and a $\betabarto$ step involves a $cut$-rule by definition.

From now on, we suppose both steps involve at least one common rule, which cannot be a $cut$ one for no commutation of $\eqcone$ involves a $cut$-rule.
In fact, there is exactly one shared rule.
Indeed, commuting rules of an $\eqcone$ step are made of one rule on top of another rule, so both cannot be above a $cut$-rule.
We distinguish cases according to the kind of $\pi_2\betabarto\pi_3$.

\proofpar{1. If $\pi_2\betabarto\pi_3$ is an $ax$ key case.}
As an $ax$-rule never commutes, the two steps share no rule, contradiction.

\proofpar{2. If $\pi_2\betabarto\pi_3$ is another key case.}
(This corresponds to the reduction diagrams~5 and~6 of~\cite[Section~B.2.1]{calculusformall}.)
We only treat the $\parr-\tens$ key case.
The cases where $\pi_2\betabarto\pi_3$ is a $\with-\oplus_i$ or $\bot-1$ key case are similar, and even simpler as less $cut$-rules result from the reduction.
Here, $\pi_2$ and $\pi_3$ are the following proofs:
\begin{center}
\bottomAlignProof
\Rsub{$\rho_1$}{$A,\Gamma$}
\Rsub{$\rho_2$}{$B,\Delta$}
\Rtens{$A\tens B,\Gamma,\Delta$}
\Rsub{$\rho_3$}{$B\orth,A\orth,\Sigma$}
\Rparr{$B\orth\parr A\orth,\Sigma$}
\Rcut{$\Gamma,\Delta,\Sigma$}
\noLine
\UnaryInfC{$\rho_4$}
\DisplayProof
\hskip 1.25em
\bottomAlignProof
\Rsub{$\rho_1$}{$A,\Gamma$}
\Rsub{$\rho_2$}{$B,\Delta$}
\Rsub{$\rho_3$}{$B\orth,A\orth,\Sigma$}
\Rcut{$A\orth,\Delta,\Sigma$}
\Rcut{$\Gamma,\Delta,\Sigma$}
\noLine
\UnaryInfC{$\rho_4$}
\DisplayProof
\end{center}
(up to symmetry of the $cut$-rule, the case where $\pi_2$ has a $\parr$-rule on the left and a $\tens$ one on the right being symmetric and solved similarly; and up to the order of the formulas in each sequent).

By our assumption, $\pi_1\eqcone\pi_2$ was a step pushing down the $\tens$ or $\parr$-rule, and up some non $cut$-rule $r$.
On \autoref{fig:proof_eqcone_betabarto_2} is depicted the reasoning we apply here, in the case where $r$ commutes with the $\tens$-rule.
We can in $\pi_1$ commute the $cut$-rule up and $r$ down (as $r$ cannot be the rule of the main connective of the formula on which we cut, nor an $ax$-rule).
This yields a proof $\pi_1^1$ such that $\pi_1\betabarto\pi_1^1$ with this commutative step, with $\pi_1^1$ being $\pi_2$ except $r$ is below the $cut$-rule and not above the $\tens$ or $\parr$-rule (by abuse, for if $r$ is a $\top$-rule then the $cut$-rule is not here anymore).
Thus, $\pi_1^1\betabarto\pi_1^2$ using the same step as in $\pi_2\betabarto\pi_3$; unless if $r$ is a $\top$-rule, in which case there is nothing to do and we set $\pi_1^1=\pi_1^2$; or if $r$ is a $\with$-rule, where we have to apply this step in both occurrences, obtaining $\pi_1^1\betabarto\cdot\betabarto\pi_1^2$.
In any case, $\pi_1^1\betabartostar\pi_1^2$.
Observe that $\pi_1^2$ is like $\pi_3$, except that $r$ is above some $cut$-rule(s) in $\pi_3$ and below in $\pi_1^2$.
But, using $\betabarto$ in $\pi_3$, $r$ can commute down one or two of the $cut$-rules created by the key case, yielding $\pi_1^2$ (including if $r$ is a $\with$ or $\top$-rule).
Therefore, $\pi_1\betabarto\pi_1^1\betabartostar\pi_1^2\betabarfromplus\pi_3$, concluding this case.

\proofpar{3. If $\pi_2\betabarto\pi_3$ is a commutative case.}
As $\pi_1\eqcone\pi_2$ and $\pi_2\betabarto\pi_3$ have exactly one rule in common, the $\eqcone$ step involves the rule $r$ that will be commuted down in the $\betabarto$ step, and another rule that we call $s$ ($r$ and $s$ are not $cut$-rules).
The proof $\pi_1$ has from top to bottom $r$, $s$ and $cut$, $\pi_2$ has $s$, $r$ and $cut$, and $\pi_3$ has $s$, $cut$ and $r$.
We will also consider the rule $t$ on the other branch of the $cut$-rule.
Schematically (and up to symmetry):
\begin{equation*}
\pi_1=\text{\AxiomC{}\RightLabel{$r$}\UnaryInfC{\emptyforproof}\RightLabel{$s$}\UnaryInfC{\emptyforproof}\AxiomC{}\RightLabel{$t$}\UnaryInfC{\emptyforproof}\RightLabel{$cut$}\BinaryInfC{\emptyforproof}\DisplayProof}
\qquad
\pi_2=\text{\AxiomC{}\RightLabel{$s$}\UnaryInfC{\emptyforproof}\RightLabel{$r$}\UnaryInfC{\emptyforproof}\AxiomC{}\RightLabel{$t$}\UnaryInfC{\emptyforproof}\RightLabel{$cut$}\BinaryInfC{\emptyforproof}\DisplayProof}
\qquad
\pi_3=\text{\AxiomC{}\RightLabel{$s$}\UnaryInfC{\emptyforproof}\AxiomC{}\RightLabel{$t$}\UnaryInfC{\emptyforproof}\RightLabel{$cut$}\BinaryInfC{\emptyforproof}\RightLabel{$r$}\UnaryInfC{\emptyforproof}\DisplayProof}
\end{equation*}
We have here different sub-cases, according to whether the $cut$-rule commutes with $s$ and/or $t$.
More exactly, we consider the following exhaustive cases:
\begin{itemize}
\item
$s$ commutes with the $cut$-rule;
\item
$s$ does not commute with the $cut$-rule but $t$ does;
\item
$s$ does not commute with the $cut$-rule and $t$ is an $ax$-rule;
\item
neither $s$ nor $t$ commute with the $cut$-rule and $t$ is not an $ax$-rule.
\end{itemize}

\proofpar{3.a. If $s$ commutes with the $cut$-rule.}
(This corresponds to the reduction diagram~3 of~\cite[Section~B.2.1]{calculusformall}, with a modification.)
Our reasoning for this case is depicted on \autoref{fig:proof_eqcone_betabarto_3a}.
In $\pi_1$, we commute $s-cut$ then $r-cut$, yielding $\pi_1\betabarto\pi_1^1\betabartostar\pi_1^2$ (the $\betabartostar$ being of length one, except if $s$ is a $\with$-rule, in which case we apply the $r-cut$ commutation for both occurrences, or if $s$ is a $\top$-rule, in which case there is no commutation to apply).
The proof $\pi_1^2$ has from top to bottom $cut$, $r$ and $s$.
Meanwhile, in $\pi_3$ we commute $s-cut$ (twice if $r$ is a $\with$-rule, or zero time if it is a $\top$-rule), yielding $\pi_3^1$ having from top to bottom $cut$, $s$ and $r$.
Now, both $\pi_1^2$ and $\pi_3^1$ have above $r$ and $s$ a same proof (maybe duplicated or erased).
We use normalization of $\betabarto$ (\autoref{cor:beta_wn}) to eliminate all $cut$-rules in this sub-proof, in the same way for all its occurrences in $\pi_1^2$ and $\pi_3^1$, obtaining proofs $\pi_1^3$ and $\pi_3^2$ equal up to the commutation of $r$ and $s$ (the very same one that was used in $\pi_1\eqcone\pi_2$).
We thus obtain $\pi_1\betabarto\pi_1^1\betabartostar\pi_1^2\betabartostar\pi_1^3\eqcone\pi_3^2\betabarfromstar\pi_3^1\betabarfromstar\pi_3$.
To formally check that indeed we can apply a $\eqcone$ step, we would have to check all cases depending on the kind of rules $r$ and $s$ are.
We do not write this tedious case study; a motivated reader can easily check some of these cases.

\proofpar{3.b. If $t$ commutes with the $cut$-rule whereas $s$ does not.}
(This corresponds to the reduction diagram~4 of~\cite[Section~B.2.1]{calculusformall}, with a modification.)
Remark $s$ cannot be an $ax$-rule, for it commutes with $r$.
This case is represented on \autoref{fig:proof_eqcone_betabarto_3b}.
On $\pi_1$, we apply a commutative $\betabarto$ step between $t$ and the $cut$-rule, giving
$\pi_1^1=$\AxiomC{}\RightLabel{$r$}\UnaryInfC{\emptyforproof}\RightLabel{$s$}\UnaryInfC{\emptyforproof}\AxiomC{}\RightLabel{$cut$}\BinaryInfC{\emptyforproof}\RightLabel{$t$}\UnaryInfC{\emptyforproof}\DisplayProof.
We then apply the same rule commutation between $r$ and $s$ that was done in $\pi_1 \eqcone \pi_2$, obtaining
$\pi_1^2=$\AxiomC{}\RightLabel{$s$}\UnaryInfC{\emptyforproof}\RightLabel{$r$}\UnaryInfC{\emptyforproof}\AxiomC{}\RightLabel{$cut$}\BinaryInfC{\emptyforproof}\RightLabel{$t$}\UnaryInfC{\emptyforproof}\DisplayProof
(as usual, we do it twice if $t$ is a $\with$-rule and have nothing to do if $t$ is a $\top$-rule).
After that, we apply a commutative step between $r$ and the $cut$-rule, the same one as in $\pi_2\betabarto\pi_3$, yielding
$\pi_1^3=$\AxiomC{}\RightLabel{$s$}\UnaryInfC{\emptyforproof}\AxiomC{}\RightLabel{$cut$}\BinaryInfC{\emptyforproof}\RightLabel{$r$}\UnaryInfC{\emptyforproof}\RightLabel{$t$}\UnaryInfC{\emptyforproof}\DisplayProof.
On the other hand, applying in $\pi_3$ a commutative case with $t$ gives
$\pi_3^1=$\AxiomC{}\RightLabel{$s$}\UnaryInfC{\emptyforproof}\AxiomC{}\RightLabel{$cut$}\BinaryInfC{\emptyforproof}\RightLabel{$t$}\UnaryInfC{\emptyforproof}\RightLabel{$r$}\UnaryInfC{\emptyforproof}\DisplayProof.
One can check that $\pi_1^3$ and $\pi_3^1$ differ only by a commutation between $r$ and $t$ -- as usual, up to checking many cases depending on the kind of rules $r$ and $t$ are.
We normalize the sub-proof(s) above $r$ and $t$ in the same way in $\pi_1^3$ and $\pi_3^1$ (using \autoref{cor:beta_wn}), obtaining $\pi_1^4$ and $\pi_3^2$ such that $\pi_1^4 \eqcone \pi_3^2$.
This solves this case:
$\pi_1\betabarto\pi_1^1\eqc\pi_1^2\betabartostar\pi_1^3\betabartostar\pi_1^4
\eqc
\betabarfromstar\pi_3^2\betabarfromstar\pi_3^1\betabarfromstar\pi_3$.

\proofpar{3.c. If $s$ does not commute with the $cut$-rule and $t$ is an $ax$-rule.}
(This is an absent case in~\cite[Section~B.2.1]{calculusformall} for they work with axiom-expanded proofs.)
This case is represented on \autoref{fig:proof_eqcone_betabarto_3c}.
Observe that here both $s$ and $t$ introduce the cut formula.
In both $\pi_1$ and $\pi_3$, we apply an $ax$ key case using $t$, giving respectively
$\pi_1^1=$\AxiomC{}\RightLabel{$r$}\UnaryInfC{\emptyforproof}\RightLabel{$s$}\UnaryInfC{\emptyforproof}\DisplayProof
and
$\pi_3^1=$\AxiomC{}\RightLabel{$s$}\UnaryInfC{\emptyforproof}\RightLabel{$r$}\UnaryInfC{\emptyforproof}\DisplayProof
(as usual, in $\pi_3$ we do this key case twice if $r$ is a $\with$-rule and have nothing to do if $r$ is a $\top$-rule).
One can check that $\pi_1^1 \eqcone \pi_3^1$ through a commutation between $r$ and $s$, with no $cut$-rule above $r$ and $s$ as there was a commutation $\pi_1 \eqcone \pi_2$.
Thus, we have
$\pi_1\betabarto\pi_1^1
\eqcone
\pi_3^1\betabarfromstar\pi_3$.

\proofpar{3.d. If $s$ and $t$ both do not commute with the $cut$-rule and $t$ is not an $ax$-rule.}
(This is a missing case in~\cite[Section~B.2.1]{calculusformall}.)
This case is represented on \autoref{fig:proof_eqcone_betabarto_3d}.
In both $\pi_1$ and $\pi_3$ we apply the sole possible key case using $s$ and $t$, obtaining respectively
$\pi_1^1=$\AxiomC{}\RightLabel{$r$}\UnaryInfC{\emptyforproof}\RightLabel{$cut^*$}\UnaryInfC{\emptyforproof}\DisplayProof
and
$\pi_3^1=$\AxiomC{}\RightLabel{$cut^*$}\UnaryInfC{\emptyforproof}\RightLabel{$r$}\UnaryInfC{\emptyforproof}\DisplayProof
(where $cut^*$ represents 0, 1 or 2 $cut$-rules).
We remark that $\pi_3^1$ can be obtained from $\pi_1^1$ by commuting $r$ with the produced $cut$-rules -- except if the key case was a $\with-\oplus_i$ one that erased $r$, in which case $\pi_1^1 = \pi_3^1$.
Thus
$\pi_1\betabarto\pi_1^1\betabartostar
\pi_3^1\betabarfromstar\pi_3$.
\end{proof}

\begin{figure}
\begin{adjustbox}{}
\begin{tikzpicture}
	\node (pi1) at (0,0) {$\pi_1=${\Rsub{$\rho_1'$}{$A,\Gamma'$}\Rsub{$\rho_2$}{$B,\Delta$}\Rtens{$A\tens B,\Gamma',\Delta$}\RightLabel{$r$}\UnaryInfC{$\fCenter A\tens B,\Gamma,\Delta$}\Rsub{$\rho_3$}{$B\orth,A\orth,\Sigma$}\Rparr{$B\orth\parr A\orth,\Sigma$}\Rcut{$\Gamma,\Delta,\Sigma$}\noLine\UnaryInfC{$\rho_4$}\DisplayProof}};
	\node (pi2) at (10,0) {$\pi_2=${\Rsub{$\rho_1'$}{$A,\Gamma'$}\RightLabel{$r$}\UnaryInfC{$\fCenter A,\Gamma$}\Rsub{$\rho_2$}{$B,\Delta$}\Rtens{$A\tens B,\Gamma,\Delta$}\Rsub{$\rho_3$}{$B\orth,A\orth,\Sigma$}\Rparr{$B\orth\parr A\orth,\Sigma$}\Rcut{$\Gamma,\Delta,\Sigma$}\noLine\UnaryInfC{$\rho_4$}\DisplayProof}};
	\node (pi3) at (10,-4) {$\pi_3=${\Rsub{$\rho_1'$}{$A,\Gamma'$}\RightLabel{$r$}\UnaryInfC{$\fCenter A,\Gamma$}\Rsub{$\rho_2$}{$B,\Delta$}\Rsub{$\rho_3$}{$B\orth,A\orth,\Sigma$}\Rcut{$A\orth,\Delta,\Sigma$}\Rcut{$\Gamma,\Delta,\Sigma$}\noLine\UnaryInfC{$\rho_4$}\DisplayProof}};

	\node (pi11) at (0,-4) {$\pi_1^1=${\Rsub{$\rho_1'$}{$A,\Gamma'$}\Rsub{$\rho_2$}{$B,\Delta$}\Rtens{$A\tens B,\Gamma',\Delta$}\Rsub{$\rho_3$}{$B\orth,A\orth,\Sigma$}\Rparr{$B\orth\parr A\orth,\Sigma$}\Rcut{$\Gamma',\Delta,\Sigma$}\RightLabel{$r$}\UnaryInfC{$\fCenter \Gamma,\Delta,\Sigma$}\noLine\UnaryInfC{$\rho_4$}\DisplayProof}};

	\node (pi12) at (5,-8) {$\pi_1^2=${\Rsub{$\rho_1'$}{$A,\Gamma'$}\Rsub{$\rho_2$}{$B,\Delta$}\Rsub{$\rho_3$}{$B\orth,A\orth,\Sigma$}\Rcut{$A\orth,\Delta,\Sigma$}\Rcut{$\Gamma',\Delta,\Sigma$}\RightLabel{$r$}\UnaryInfC{$\fCenter \Gamma,\Delta,\Sigma$}\noLine\UnaryInfC{$\rho_4$}\DisplayProof}};
\begin{myscope}
	\path (pi1) \edgelabeld{$\eqcone$} (pi2);
	\path (pi2) \edgelabel{$\betabar$} (pi3);
\end{myscope}
\begin{myscopec}{red}
	\path (pi1) \edgelabel{$\betabar$} (pi11);
	\path (pi11) \edgelabel{$\betabar^*$} (pi12);
	\path (pi3) \edgelabel{$\betabar^+$} (pi12);
\end{myscopec}
\end{tikzpicture}
\end{adjustbox}
\caption{Schematic representation of case 2 in the proof of \autoref{lem:eqc_betato}}
\label{fig:proof_eqcone_betabarto_2}
\end{figure}

\begin{figure}
\centering
\begin{tikzpicture}
	\node (pi1) at (0,0) {$\pi_1=${\AxiomC{$\rho_1$}\RightLabel{$r$}\UnaryInfC{\emptyforproof}\RightLabel{$s$}\UnaryInfC{\emptyforproof}\AxiomC{$\rho_2$}\RightLabel{$t$}\UnaryInfC{\emptyforproof}\RightLabel{$cut$}\BinaryInfC{\emptyforproof}\DisplayProof}};
	\node (pi2) at (5,0) {$\pi_2=${\AxiomC{$\rho_1$}\RightLabel{$s$}\UnaryInfC{\emptyforproof}\RightLabel{$r$}\UnaryInfC{\emptyforproof}\AxiomC{$\rho_2$}\RightLabel{$cut$}\BinaryInfC{\emptyforproof}\DisplayProof}};
	\node (pi3) at (5,-2.5) {$\pi_3=${\AxiomC{$\rho_1$}\RightLabel{$s$}\UnaryInfC{\emptyforproof}\AxiomC{$\rho_2$}\RightLabel{$t$}\UnaryInfC{\emptyforproof}\RightLabel{$cut$}\BinaryInfC{\emptyforproof}\RightLabel{$r$}\UnaryInfC{\emptyforproof}\DisplayProof}};

	\node (pi11) at (0,-2.5) {$\pi_1^1=${\AxiomC{$\rho_1$}\RightLabel{$r$}\UnaryInfC{\emptyforproof}\AxiomC{$\rho_2$}\RightLabel{$t$}\UnaryInfC{\emptyforproof}\RightLabel{$cut$}\BinaryInfC{\emptyforproof}\RightLabel{$s$}\UnaryInfC{\emptyforproof}\DisplayProof}};

	\node (pi12) at (0,-5) {$\pi_1^2=${\AxiomC{$\rho_1$}\AxiomC{$\rho_2$}\RightLabel{$t$}\UnaryInfC{\emptyforproof}\RightLabel{$cut$}\BinaryInfC{\emptyforproof}\RightLabel{$r$}\UnaryInfC{\emptyforproof}\RightLabel{$s$}\UnaryInfC{\emptyforproof}\DisplayProof}};
	\node (pi32) at (5,-5) {$\pi_3^1=${\AxiomC{$\rho_1$}\AxiomC{$\rho_2$}\RightLabel{$t$}\UnaryInfC{\emptyforproof}\RightLabel{$cut$}\BinaryInfC{\emptyforproof}\RightLabel{$s$}\UnaryInfC{\emptyforproof}\RightLabel{$r$}\UnaryInfC{\emptyforproof}\DisplayProof}};

	\node (pi13) at (0,-7.5) {$\pi_1^3=${\AxiomC{$\tau^\text{cut-free}$}\RightLabel{$r$}\UnaryInfC{\emptyforproof}\RightLabel{$s$}\UnaryInfC{\emptyforproof}\DisplayProof}};
	\node (pi33) at (5,-7.5) {$\pi_3^2=${\AxiomC{$\tau^\text{cut-free}$}\RightLabel{$s$}\UnaryInfC{\emptyforproof}\RightLabel{$r$}\UnaryInfC{\emptyforproof}\DisplayProof}};
\begin{myscope}
	\path (pi1) \edgelabeld{$\eqcone$} (pi2);
	\path (pi2) \edgelabel{$\betabar$} (pi3);
\end{myscope}
\begin{myscopec}{red}
	\path (pi1) \edgelabel{$\betabar$} (pi11);

	\path (pi11) \edgelabel{$\betabar^*$} (pi12);
	\path (pi3) \edgelabel{$\betabar^*$} (pi32);

	\path (pi12) \edgelabel{$\betabar^*$} (pi13);
	\path (pi32) \edgelabel{$\betabar^*$} (pi33);

	\path (pi13) \edgelabeld{$\eqcone$} (pi33);
\end{myscopec}
\end{tikzpicture}
\caption{Schematic representation of case 3.a in the proof of \autoref{lem:eqc_betato}}
\label{fig:proof_eqcone_betabarto_3a}
\end{figure}

\begin{figure}
\begin{adjustbox}{}
\begin{tikzpicture}
	\node (pi1) at (0,0) {$\pi_1=${\AxiomC{$\rho_1$}\RightLabel{$r$}\UnaryInfC{\emptyforproof}\RightLabel{$s$}\UnaryInfC{\emptyforproof}\AxiomC{$\rho_2$}\RightLabel{$t$}\UnaryInfC{\emptyforproof}\RightLabel{$cut$}\BinaryInfC{$\rho_3$}\DisplayProof}};
	\node (pi2) at (5,0) {$\pi_2=${\AxiomC{$\rho_1$}\RightLabel{$s$}\UnaryInfC{\emptyforproof}\RightLabel{$r$}\UnaryInfC{\emptyforproof}\AxiomC{$\rho_2$}\RightLabel{$t$}\UnaryInfC{\emptyforproof}\RightLabel{$cut$}\BinaryInfC{$\rho_3$}\DisplayProof}};
	\node (pi3) at (5,-3) {$\pi_3=${\AxiomC{$\rho_1$}\RightLabel{$s$}\UnaryInfC{}\AxiomC{$\rho_2$}\RightLabel{$t$}\UnaryInfC{\emptyforproof}\RightLabel{$cut$}\BinaryInfC{\emptyforproof}\RightLabel{$r$}\UnaryInfC{$\rho_3$}\DisplayProof}};

	\node (pi11) at (0,-3) {$\pi_1^1=${\AxiomC{$\rho_1$}\RightLabel{$r$}\UnaryInfC{\emptyforproof}\RightLabel{$s$}\UnaryInfC{\emptyforproof}\AxiomC{$\rho_2$}\RightLabel{$cut$}\BinaryInfC{\emptyforproof}\RightLabel{$t$}\UnaryInfC{$\rho_3$}\DisplayProof}};
	\node (pi12) at (-5,-3) {$\pi_1^2=${\AxiomC{$\rho_1$}\RightLabel{$s$}\UnaryInfC{\emptyforproof}\RightLabel{$r$}\UnaryInfC{\emptyforproof}\AxiomC{$\rho_2$}\RightLabel{$cut$}\BinaryInfC{\emptyforproof}\RightLabel{$t$}\UnaryInfC{$\rho_3$}\DisplayProof}};
	\node (pi13) at (-5,-6) {$\pi_1^3=${\AxiomC{$\rho_1$}\RightLabel{$s$}\UnaryInfC{\emptyforproof}\AxiomC{$\rho_2$}\RightLabel{$cut$}\BinaryInfC{\emptyforproof}\RightLabel{$r$}\UnaryInfC{\emptyforproof}\RightLabel{$t$}\UnaryInfC{$\rho_3$}\DisplayProof}};
	\node (pi14) at (-5,-9) {$\pi_1^4=${\AxiomC{$\tau^\text{cut-free}$}\RightLabel{$r$}\UnaryInfC{\emptyforproof}\RightLabel{$t$}\UnaryInfC{$\rho_3$}\DisplayProof}};

	\node (pi31) at (5,-6) {$\pi_3^1=${\AxiomC{$\rho_1$}\RightLabel{$s$}\UnaryInfC{\emptyforproof}\AxiomC{$\rho_2$}\RightLabel{$cut$}\BinaryInfC{\emptyforproof}\RightLabel{$t$}\UnaryInfC{\emptyforproof}\RightLabel{$r$}\UnaryInfC{$\rho_3$}\DisplayProof}};
	\node (pi32) at (5,-9) {$\pi_3^2=${\AxiomC{$\tau^\text{cut-free}$}\RightLabel{$t$}\UnaryInfC{\emptyforproof}\RightLabel{$r$}\UnaryInfC{$\rho_3$}\DisplayProof}};
\begin{myscope}
	\path (pi1) \edgelabeld{$\eqcone$} (pi2);
	\path (pi2) \edgelabel{$\betabar$} (pi3);
\end{myscope}
\begin{myscopec}{red}
	\path (pi1) \edgelabel{$\betabar$} (pi11);
	\path (pi11) \edgelabeld{$\eqc$} (pi12);
	\path (pi12) \edgelabel{$\betabar^*$} (pi13);
	\path (pi13) \edgelabel{$\betabar^*$} (pi14);
	
	\path (pi3) \edgelabel{$\betabar^*$} (pi31);
	\path (pi31) \edgelabel{$\betabar^*$} (pi32);

	\path (pi14) \edgelabeld{$\eqceq$} (pi32);
\end{myscopec}
\end{tikzpicture}
\end{adjustbox}
\caption{Schematic representation of case 3.b in the proof of \autoref{lem:eqc_betato}}
\label{fig:proof_eqcone_betabarto_3b}
\end{figure}

\begin{figure}
\begin{adjustbox}{}
\begin{tikzpicture}
	\node (pi1) at (0,0) {$\pi_1=${\AxiomC{$\rho_1$}\RightLabel{$r$}\UnaryInfC{\emptyforproof}\RightLabel{$s$}\UnaryInfC{\emptyforproof}\AxiomC{}\RightLabel{$ax$}\UnaryInfC{\emptyforproof}\RightLabel{$cut$}\BinaryInfC{$\rho_3$}\DisplayProof}};
	\node (pi2) at (5,0) {$\pi_2=${\AxiomC{$\rho_1$}\RightLabel{$s$}\UnaryInfC{\emptyforproof}\RightLabel{$r$}\UnaryInfC{\emptyforproof}\AxiomC{}\RightLabel{$ax$}\UnaryInfC{\emptyforproof}\RightLabel{$cut$}\BinaryInfC{$\rho_3$}\DisplayProof}};
	\node (pi3) at (5,-3) {$\pi_3=${\AxiomC{$\rho_1$}\RightLabel{$s$}\UnaryInfC{}\AxiomC{}\RightLabel{$ax$}\UnaryInfC{\emptyforproof}\RightLabel{$cut$}\BinaryInfC{\emptyforproof}\RightLabel{$r$}\UnaryInfC{$\rho_3$}\DisplayProof}};

	\node (pi11) at (0,-6) {$\pi_1^1=${\AxiomC{$\rho_1$}\RightLabel{$r$}\UnaryInfC{\emptyforproof}\RightLabel{$s$}\UnaryInfC{$\rho_3$}\DisplayProof}};

	\node (pi31) at (5,-6) {$\pi_3^1=${\AxiomC{$\rho_1$}\RightLabel{$s$}\UnaryInfC{\emptyforproof}\RightLabel{$r$}\UnaryInfC{$\rho_3$}\DisplayProof}};
\begin{myscope}
	\path (pi1) \edgelabeld{$\eqcone$} (pi2);
	\path (pi2) \edgelabel{$\betabar$} (pi3);
\end{myscope}
\begin{myscopec}{red}
	\path (pi1) \edgelabel{$\betabar$} (pi11);

	\path (pi3) \edgelabel{$\betabar^*$} (pi31);

	\path (pi11) \edgelabeld{$\eqcone$} (pi31);
\end{myscopec}
\end{tikzpicture}
\end{adjustbox}
\caption{Schematic representation of case 3.c in the proof of \autoref{lem:eqc_betato}}
\label{fig:proof_eqcone_betabarto_3c}
\end{figure}

\begin{figure}
\begin{adjustbox}{}
\begin{tikzpicture}
	\node (pi1) at (0,0) {$\pi_1=${\AxiomC{$\rho_1$}\RightLabel{$r$}\UnaryInfC{\emptyforproof}\RightLabel{$s$}\UnaryInfC{\emptyforproof}\AxiomC{$\rho_2$}\RightLabel{$t$}\UnaryInfC{\emptyforproof}\RightLabel{$cut$}\BinaryInfC{$\rho_3$}\DisplayProof}};
	\node (pi2) at (5,0) {$\pi_2=${\AxiomC{$\rho_1$}\RightLabel{$s$}\UnaryInfC{\emptyforproof}\RightLabel{$r$}\UnaryInfC{\emptyforproof}\AxiomC{$\rho_2$}\RightLabel{$t$}\UnaryInfC{\emptyforproof}\RightLabel{$cut$}\BinaryInfC{$\rho_3$}\DisplayProof}};
	\node (pi3) at (5,-3) {$\pi_3=${\AxiomC{$\rho_1$}\RightLabel{$s$}\UnaryInfC{}\AxiomC{$\rho_2$}\RightLabel{$t$}\UnaryInfC{\emptyforproof}\RightLabel{$cut$}\BinaryInfC{\emptyforproof}\RightLabel{$r$}\UnaryInfC{$\rho_3$}\DisplayProof}};

	\node (pi11) at (0,-3) {$\pi_1^1=${\AxiomC{$\tau$}\RightLabel{$r$}\UnaryInfC{\emptyforproof}\RightLabel{$cut^*$}\UnaryInfC{$\rho_3$}\DisplayProof}};

	\node (pi31) at (5,-6) {$\pi_3^1=${\AxiomC{$\tau$}\RightLabel{$cut^*$}\UnaryInfC{\emptyforproof}\RightLabel{$r$}\UnaryInfC{$\rho_3$}\DisplayProof}};
\begin{myscope}
	\path (pi1) \edgelabeld{$\eqcone$} (pi2);
	\path (pi2) \edgelabel{$\betabar$} (pi3);
\end{myscope}
\begin{myscopec}{red}
	\path (pi1) \edgelabel{$\betabar$} (pi11);

	\path (pi3) \edgelabel{$\betabar^*$} (pi31);

	\path (pi11) \edgelabel{$\betabar^*$} (pi31);
\end{myscopec}
\end{tikzpicture}
\end{adjustbox}
\caption{Schematic representation of case 3.d in the proof of \autoref{lem:eqc_betato}}
\label{fig:proof_eqcone_betabarto_3d}
\end{figure}

We can now prove the wished result.

\begin{prop}
\label{th:CReqc}
In the sequent calculus of MALL, $\betabarto$ is Church-Rosser modulo $\eqc$.
\end{prop}
\begin{proof}
By \autoref{lem:aototomaya_use}, \autoref{lem:beta_wn} and Lemmas~\ref{lem:betafrom_betato} and~\ref{lem:eqc_betato}.
\end{proof}

\subsection{Removing cut-cut commutations}
\label{sec:no_cut-cut}

We prove here that $\beta$-equality and $\betabar$-equality coincide.
This corresponds to~\cite[Lemmas~C.2 and~C.3]{calculusformall}, proven here with two considerations that were at best implicit in the proof of~\cite{calculusformall}.
The first is that $cut-cut$ commutations and $ax$ key cases lead to proofs that are equal not directly up to $cut-cut$ commutation but up to \emph{symmetry of a $cut$-rule}.
The second is that the demonstration of~\cite[Lemmas~C.2 and~C.3]{calculusformall} proves an induction hypothesis on a single $cut-cut$ commutation, but applies it on a sequence of such commutations.

\begin{defi}
\label{rem:cut_comm_in_beta}
The \emph{symmetry of a $cut$-rule} is the identification of the following (sub-) proofs:
\begin{center}
$\pi_1=$
\Rsub{$\rho_1$}{$A,\Gamma$}
\Rsub{$\rho_2$}{$A\orth,\Delta$}
\Rcut{$\Gamma,\Delta$}
\DisplayProof
\hskip 2em and \hskip 2em
$\pi_2=$
\Rsub{$\rho_2$}{$A\orth,\Delta$}
\Rsub{$\rho_1$}{$A,\Gamma$}
\Rcut{$\Gamma,\Delta$}
\DisplayProof
\end{center}
\end{defi}

This relation is included in $\eqb$, using a $cut-cut$ commutation.
Indeed, define respectively $\pi_1'$ and $\pi_2'$ the following two proofs, where $A_1$ and $A_2$ are both the formula $A$, with indices to follow their occurrences:

\begin{adjustbox}{}
\Raxocc{$A_1\orth,A_2$}
\Rsub{$\rho_1$}{$A_1,\Gamma$}
\Rcut{$A_2,\Gamma$}
\Rsub{$\rho_2$}{$A_2\orth,\Delta$}
\Rcut{$\Gamma,\Delta$}
\DisplayProof
\hskip1em and \hskip1em
\Raxocc{$A_1\orth,A_2$}
\Rsub{$\rho_2$}{$A_2\orth,\Delta$}
\Rcut{$A_1\orth,\Delta$}
\Rsub{$\rho_1$}{$A_1,\Gamma$}
\Rcut{$\Gamma,\Delta$}
\DisplayProof
\end{adjustbox}
Then $\pi_1\betabarfrom\pi_1'\eqcc\pi_2'\betabarto\pi_2$, these $\betabarto$ steps being $ax$ key cases.

We now prove this relation is in $\betabar$-equality.

\begin{lem}
\label{lem:comm_cut_betabarto}
Let $\pi_1$ and $\pi_2$ be two proofs equal up to symmetry of $cut$-rules.
Then $\pi_1\betabartostar\cdot\betabarfromstar\pi_2$.
\end{lem}
\begin{proof}
The idea is to eliminate any $cut$-rule in $\pi_1$ and $\pi_2$ in a symmetric way, and show the resulting proofs are equal up to symmetry of $cut$-rules.
The result follows through weak normalization, as two cut-free proofs equal up to symmetry of $cut$-rules are simply equal.

We reason by induction on a sequence $\pi_1\betabartostar\rho$, with $\rho$ some cut-free proof found by weak normalization (\autoref{cor:beta_wn}).
If this sequence is empty, then $\pi_1$ is $cut$-free and so $\pi_1=\pi_2$.
Thus, take $\pi_1\betabarto\pi_1'$ its first step.
We apply the corresponding step in $\pi_2$, on the corresponding $cut$-rule which may be the symmetric version of $c$ (for if a $\betabarto$ step can be applied, then the one with switched premises can be applied on the symmetric version).
We obtain a proof $\pi_2'$ with the $cut$-rule still permuted compared to $\pi_1'$ in case of a commutative case (or the $cut$-rule erased when commuting with a $\top$, or permuted and duplicated with a $\with$) and $0$, $1$ or $2$ symmetric $cut$-rules resulting from a key case.
In all cases, $\pi_1'$ and $\pi_2'$ are equal up to symmetry of $cut$-rules, allowing us to conclude by induction hypothesis.
\end{proof}

\begin{lem}[{\cite[Lemmas~C.2 and~C.3]{calculusformall}, adapted}]
\label{lem:eqbb_eqb}
Let $\pi$ and $\tau$ be proofs such that $\pi \eqccast \tau$.
Then $\pi$ and $\tau$ are $\betabar$-equal.
\end{lem}
\begin{proof}
We reason by induction on $(cr(\pi), n)$ where $cr(\pi)$ is the sum over slices of $\pi$ of the number of non-$cut$-rules in these slices, and $n$ the number of  $\eqcc$ steps in $\pi \eqccast \tau$.
Remark the number of non-cut-rules of a slice in a proof is preserved by $cut-cut$ commutations, so that $cr(\pi) = cr(\tau)$.

If $n = 0$ we are done as $\pi = \tau$.
Otherwise, $\pi \eqccast \tau$ is a non-empty sequence $\pi \eqcc \phi \eqccast \tau$.
By induction hypothesis, $\phi \eqbb \tau$ (using $cr(\pi) = cr(\phi) = cr(\tau)$), so we have to prove $\pi \eqbb \phi$.
As $\pi \eqcc \phi$, we have:

\begin{adjustbox}{}
$\pi=$
{\Rsub{$\mu_1$}{$A\orth, B\orth, \Gamma_1$}\Rsub{$\mu_2$}{$B, \Gamma_2$}\Rcut{$A\orth, \Gamma_1, \Gamma_2$}\Rsub{$\mu_3$}{$A, \Gamma_3$}\Rcut{$\Gamma_1, \Gamma_2, \Gamma_3$}\DisplayProof}
and\hskip1em
$\phi=$
{\Rsub{$\mu_1$}{$A\orth, B\orth, \Gamma_1$}\Rsub{$\mu_3$}{$A, \Gamma_3$}\Rcut{$B\orth, \Gamma_1, \Gamma_3$}\Rsub{$\mu_2$}{$B, \Gamma_2$}\Rcut{$\Gamma_1, \Gamma_2, \Gamma_3$}\DisplayProof}
\end{adjustbox}
(or one of the three other analogous situations, up to switching branches of the two $cut$-rules).
Remark we use here that $\beta$-equality is contextual in order to assume the two commuted cut-rules are the last rules of $\pi$ and $\phi$.
By eliminating cuts in $\mu_1$, $\mu_2$ and $\mu_3$ using $\betabarto$ (\autoref{cor:beta_wn}), we can assume them cut-free.
Remark that $cr(\mu_i) \leq cr(\pi)$ for each $i$, as cut-elimination in MALL reduces the number of non-$cut$-rules in a slice.
We conclude through a case study of the last rules of $\mu_1$, $\mu_2$ and $\mu_3$.

\proofpar{1. If the last rule of one of the $\mu_i$ is an $ax$-rule.}
If the last rule of $\mu_2$ or $\mu_3$ is an $ax$-rule, then using an $ax$ key case in $\pi$ and $\phi$ we obtain $\pi \betabarto \cdot \betabarfrom \phi$.
If the last rule of $\mu_1$ is an $ax$-rule, applying an $ax$ key case in $\pi$ and $\phi$ yield two proofs equal up to symmetry of a $cut$-rule, so $\betabar$-equal (\autoref{lem:comm_cut_betabarto}).
Thus, we assume from now on that the last rule of each $\mu_i$ is not an $ax$-rule.

\proofpar{2. If the last rule of one of the $\mu_i$ commutes with both $cut$-rules.}
This case is represented on \autoref{fig:proof_eqbb_eqb_2}.
If the last rule $r$ of $\mu_1$ is not the rule introducing the formula $A\orth$ nor the formula $B\orth$, by applying in both $\pi$ and $\phi$ two commutative cases on $r$ with the two successive $cut$-rules below it, we obtain respectively

\begin{adjustbox}{}
$\pi'=$
{\Rsub{$\mu_1'$}{$A\orth, B\orth, \Gamma_1'$}\Rsub{$\mu_2$}{$B, \Gamma_2$}\Rcut{$A\orth, \Gamma_1', \Gamma_2$}\Rsub{$\mu_3$}{$A, \Gamma_3$}\Rcut{$\Gamma_1', \Gamma_2, \Gamma_3$}\RightLabel{$r$}\UnaryInfC{$\fCenter \Gamma_1, \Gamma_2, \Gamma_3$}\DisplayProof}
and
$\phi'=$
{\Rsub{$\mu_1'$}{$A\orth, B\orth, \Gamma_1'$}\Rsub{$\mu_3$}{$A, \Gamma_3$}\Rcut{$B\orth, \Gamma_1', \Gamma_3$}\Rsub{$\mu_2$}{$B, \Gamma_2$}\Rcut{$\Gamma_1', \Gamma_2, \Gamma_3$}\RightLabel{$r$}\UnaryInfC{$\fCenter \Gamma_1, \Gamma_2, \Gamma_3$}\DisplayProof}
\end{adjustbox}
with $\mu_1$ being {\bottomAlignProof\AxiomC{$\mu_1'$}\noLine\UnaryInfC{$\fCenter A\orth, B\orth, \Gamma_1'$}\RightLabel{$r$}\UnaryInfC{$\fCenter A\orth, B\orth, \Gamma_1$}\DisplayProof},
$\pi \betabartoplus \pi'$ and $\phi \betabartoplus \phi'$.
We conclude by induction hypothesis

\begin{adjustbox}{}
{\Rsub{$\mu_1'$}{$A\orth, B\orth, \Gamma_1'$}\Rsub{$\mu_2$}{$B, \Gamma_2$}\Rcut{$A\orth, \Gamma_1', \Gamma_2$}\Rsub{$\mu_3$}{$A, \Gamma_3$}\Rcut{$\Gamma_1', \Gamma_2, \Gamma_3$}	\DisplayProof}
$\eqbb$
{\Rsub{$\mu'_1$}{$A\orth, B\orth, \Gamma_1'$}\Rsub{$\mu_3$}{$A, \Gamma_3$}\Rcut{$B\orth, \Gamma_1', \Gamma_3$}\Rsub{$\mu_2$}{$B, \Gamma_2$}\Rcut{$\Gamma_1', \Gamma_2, \Gamma_3$}\DisplayProof}
\end{adjustbox}
(which have at least the non-$cut$-rule $r$ removed with respect to $\pi$), so that $\pi \betabartoplus \pi' \eqbb \phi' \betabarfromplus \phi$.
Remark that we apply the induction hypothesis two times in case $r$ is a $\with$-rule (once in each branch of $r$) and do not apply it if $r$ is a $\top$-rule.
The cases where the last rule of $\mu_2$ is not the rule introducing the formula $B\orth$, and where the last rule of $\mu_3$ is not the rule introducing the formula $A\orth$, are similar.

\proofpar{3. If the last rule of all $\mu_i$ do not commute with the $cut$-rules.}
This case is represented on \autoref{fig:proof_eqbb_eqb_3}.
If the last rule $r_1$ of $\mu_1$ is the rule introducing the formula $B\orth$, and the last rule $r_2$ of $\mu_2$ the one introducing $B$, we can apply a key case in $\pi$ and a commutative case on $r_1$ followed by a key case in $\phi$.
The proofs we obtain are

\begin{adjustbox}{}
$\pi'=$
{\AxiomC{$\mu_1'$}\noLine\UnaryInfC{$\fCenter A\orth, \Gamma'_1$}\AxiomC{$\mu_2'$}\noLine\UnaryInfC{$\fCenter \Gamma'_2$}\RightLabel{$cut^*$}\BinaryInfC{$\fCenter A\orth, \Gamma_1, \Gamma_2$}\Rsub{$\mu_3$}{$A, \Gamma_3$}\Rcut{$\Gamma_1, \Gamma_2, \Gamma_3$}\DisplayProof}
and
$\phi'=$
{\AxiomC{$\mu_1'$}\noLine\UnaryInfC{$\fCenter A\orth, \Gamma'_1$}\Rsub{$\mu_3$}{$A, \Gamma_3$}\Rcut{$\Gamma'_1, \Gamma_3$}\AxiomC{$\mu_2'$}\noLine\UnaryInfC{$\fCenter \Gamma'_2$}\RightLabel{$cut^*$}\BinaryInfC{$\fCenter \Gamma_1, \Gamma_2, \Gamma_3$}\DisplayProof}
\end{adjustbox}
with $\mu_1=${\AxiomC{$\mu_1'$}\noLine\UnaryInfC{$\fCenter A\orth, \Gamma'_1$}\RightLabel{$r_1$}\UnaryInfC{$\fCenter A\orth, B\orth, \Gamma_1$}\DisplayProof},
$\mu_2=${\AxiomC{$\mu_2'$}\noLine\UnaryInfC{$\fCenter \Gamma'_2$}\RightLabel{$r_2$}\UnaryInfC{$\fCenter B, \Gamma_2$}\DisplayProof},
$\pi \betabarto \pi'$ and $\phi \betabartoplus \phi'$.
The notation $cut^*$ here means we have a number of $cut$-rules: 0 if the key case was a $\bot-1$, 1 if it was a $\with-\oplus_1$ or $\with-\oplus_2$, 2 if it was a $\parr-\tens$.
Remark that $\pi' \eqccast \phi'$ (this sequence being of size the number of $cut$-rules represented by $cut^*$) with $cr(\pi') < cr(\pi)$ for we removed at least rules $r_1$ and $r_2$.
Hence, we conclude $\pi \betabarto \pi' \eqbb \phi' \betabarfromplus \phi$.
The case where the last rule of $\mu_1$ introduces $A\orth$, and the last rule of $\mu_2$ introduces $A$, is similar.
\end{proof}

\begin{figure}
\begin{adjustbox}{}
\begin{tikzpicture}
	\node (pi) at (0,0) {$\pi=${\AxiomC{$\mu_1'$}\noLine\UnaryInfC{$\fCenter A\orth, B\orth, \Gamma_1'$}\RightLabel{$r$}\UnaryInfC{$\fCenter A\orth, B\orth, \Gamma_1$}\Rsub{$\mu_2$}{$B, \Gamma_2$}\Rcut{$A\orth, \Gamma_1, \Gamma_2$}\Rsub{$\mu_3$}{$A, \Gamma_3$}\Rcut{$\Gamma_1, \Gamma_2, \Gamma_3$}\DisplayProof}};
	\node (phi) at (10,0) {$\phi=${\AxiomC{$\mu_1'$}\noLine\UnaryInfC{$\fCenter A\orth, B\orth, \Gamma_1'$}\RightLabel{$r$}\UnaryInfC{$\fCenter A\orth, B\orth, \Gamma_1$}\Rsub{$\mu_3$}{$A, \Gamma_3$}\Rcut{$B\orth, \Gamma_1, \Gamma_3$}\Rsub{$\mu_2$}{$B, \Gamma_2$}\Rcut{$\Gamma_1, \Gamma_2, \Gamma_3$}\DisplayProof}};

	\node (pi') at (0,-4) {$\pi'=${\Rsub{$\mu_1'$}{$A\orth, B\orth, \Gamma_1'$}\Rsub{$\mu_2$}{$B, \Gamma_2$}\Rcut{$A\orth, \Gamma_1', \Gamma_2$}\Rsub{$\mu_3$}{$A, \Gamma_3$}\Rcut{$\Gamma_1', \Gamma_2, \Gamma_3$}\RightLabel{$r$}\UnaryInfC{$\fCenter \Gamma_1, \Gamma_2, \Gamma_3$}\DisplayProof}};

	\node (phi') at (10,-4) {$\phi'=${\Rsub{$\mu_1'$}{$A\orth, B\orth, \Gamma_1'$}\Rsub{$\mu_3$}{$A, \Gamma_3$}\Rcut{$B\orth, \Gamma_1', \Gamma_3$}\Rsub{$\mu_2$}{$B, \Gamma_2$}\Rcut{$\Gamma_1', \Gamma_2, \Gamma_3$}\RightLabel{$r$}\UnaryInfC{$\fCenter \Gamma_1, \Gamma_2, \Gamma_3$}\DisplayProof}};
\begin{myscope}
	\path (pi) \edgelabeld{$\eqcc$} (phi);
\end{myscope}
\begin{myscopec}{red}
	\path (pi) \edgelabel{$\betabar^+$} (pi');
	\path (phi) \edgelabel{$\betabar^+$} (phi');
	\path (pi') \edgelabeld{$\eqcc$} (phi');
\end{myscopec}
\end{tikzpicture}
\end{adjustbox}
\caption{Schematic representation of case 2 in the proof of \autoref{lem:eqbb_eqb}}
\label{fig:proof_eqbb_eqb_2}
\end{figure}

\begin{figure}
\begin{adjustbox}{}
\begin{tikzpicture}
	\node (pi) at (0,0) {$\pi=${\AxiomC{$\mu_1'$}\noLine\UnaryInfC{$\fCenter A\orth, \Gamma'_1$}\RightLabel{$r_1$}\UnaryInfC{$\fCenter A\orth, B\orth, \Gamma_1$}\AxiomC{$\mu_2'$}\noLine\UnaryInfC{$\fCenter \Gamma'_2$}\RightLabel{$r_2$}\UnaryInfC{$\fCenter B, \Gamma_2$}\Rcut{$A\orth, \Gamma_1, \Gamma_2$}\Rsub{$\mu_3$}{$A, \Gamma_3$}\Rcut{$\Gamma_1, \Gamma_2, \Gamma_3$}\DisplayProof}};
	\node (phi) at (10,0) {$\phi=${\AxiomC{$\mu_1'$}\noLine\UnaryInfC{$\fCenter A\orth, \Gamma'_1$}\RightLabel{$r_1$}\UnaryInfC{$\fCenter A\orth, B\orth, \Gamma_1$}\Rsub{$\mu_3$}{$A, \Gamma_3$}\Rcut{$B\orth, \Gamma_1, \Gamma_3$}\AxiomC{$\mu_2'$}\noLine\UnaryInfC{$\fCenter \Gamma'_2$}\RightLabel{$r_2$}\UnaryInfC{$\fCenter B, \Gamma_2$}\Rcut{$\Gamma_1, \Gamma_2, \Gamma_3$}\DisplayProof}};

	\node (pi') at (0,-4) {$\pi'=${\AxiomC{$\mu_1'$}\noLine\UnaryInfC{$\fCenter A\orth, \Gamma'_1$}\AxiomC{$\mu_2'$}\noLine\UnaryInfC{$\fCenter \Gamma'_2$}\RightLabel{$cut^*$}\BinaryInfC{$\fCenter A\orth, \Gamma_1, \Gamma_2$}\Rsub{$\mu_3$}{$A, \Gamma_3$}\Rcut{$\Gamma_1, \Gamma_2, \Gamma_3$}\DisplayProof}};

	\node (phi') at (10,-4) {$\phi'=${\AxiomC{$\mu_1'$}\noLine\UnaryInfC{$\fCenter A\orth, \Gamma'_1$}\Rsub{$\mu_3$}{$A, \Gamma_3$}\Rcut{$\Gamma'_1, \Gamma_3$}\AxiomC{$\mu_2'$}\noLine\UnaryInfC{$\fCenter \Gamma'_2$}\RightLabel{$cut^*$}\BinaryInfC{$\fCenter \Gamma_1, \Gamma_2, \Gamma_3$}\DisplayProof}};
\begin{myscope}
	\path (pi) \edgelabeld{$\eqcc$} (phi);
\end{myscope}
\begin{myscopec}{red}
	\path (pi) \edgelabel{$\betabar$} (pi');
	\path (phi) \edgelabel{$\betabar^+$} (phi');
	\path (pi') \edgelabeld{$\eqccast$} (phi');
\end{myscopec}
\end{tikzpicture}
\end{adjustbox}
\caption{Schematic representation of case 3 in the proof of \autoref{lem:eqbb_eqb}}
\label{fig:proof_eqbb_eqb_3}
\end{figure}

The theorem claimed in the main part then follows.

\theqbeqc*
\begin{proof}
The normal forms are $\beta$-equal, so $\betabar$-equal (\autoref{lem:eqbb_eqb}).
By \autoref{th:CReqc}, they are related by $\eqc$ since no reduction $\betabarto$ can be applied to them.
\end{proof}

%==================================================

\section{Conjectures on isomorphisms of symmetric monoidal closed categories with finite (co)products}
\label{sec:conjectures_isos_smcc}

As a follow-up of \autoref{subsec:smcc_isos},
we present here some conjectures about isomorphisms in SMCC with finite products and finite coproducts on one hand, and in SMCC with finite coproducts (but not finite products) on the other hand.
We also give in these two cases formulas that are not isomorphic, while they are isomorphic in $\star$-autonomous categories with finite products.

We conjecture that isomorphisms in SMCC with both finite products \emph{and} finite coproducts correspond to adding the equations of theory $\thCP$ (\autoref{sec:conclu}) to $\thS$ (\autoref{tab:eqisos_cat}).
However, our approach through $\star$-autonomous categories does not work since for example $\top\lolipop(\top\oplus\top)$ and $(0\with 0)\lolipop 0$ are isomorphic in $\star$-autonomous categories but not in SMCC with finite products and coproducts.
Indeed, in $\star$-autonomous categories one has:
\begin{align*}
  \top\lolipop(\top\oplus\top) &\isoc \top\lolipop(((\top\lolipop\bot)\with(\top\lolipop\bot))\lolipop\bot)\\ &\isoc (\top\tens(0\with 0))\lolipop\bot \isoc ((0\with 0)\tens\top)\lolipop\bot\\ &\isoc (0\with 0)\lolipop(\top\lolipop\bot) \isoc (0\with 0)\lolipop 0
\end{align*}
or directly by interpreting these formulas into MALL, one gets:
\begin{equation*}
\tradm{\top\lolipop(\top\oplus\top)}\isom 0\parr(\top\oplus\top) \isom (\top\oplus\top)\parr 0\isom\tradm{(0\with 0)\lolipop 0}
\end{equation*}
(the analogue of the non-ambiguity property, used in \autoref{subsec:smcc_isos}, does not hold).
That this isomorphism is not valid in the SMCC setting follows from the succeeding necessary condition.

\begin{lem}
\label{lem:isoimall}
If $F\iso G$ in symmetric monoidal closed categories with finite products and coproducts with $F$ and $G$ distributed (\ie\ $\tradm{F}$ and $\tradm{G}$ are distributed) then there exist cut-free proofs of $F\vdash G$ and $G\vdash F$ in IMALL (intuitionistic multiplicative additive linear logic)~\cite{biermancatll} whose left $0$ rules introduce $0\vdash 0$ sequents only and right $\top$ rules introduce $\top\vdash\top$ sequents only.
\end{lem}
\begin{proof}
If $F$ and $G$ are isomorphic in SMCC with finite products and coproducts, the associated isomorphisms can be represented as IMALL proofs which we can assume to be cut-free.
These proofs can be interpreted as MALL proofs (corresponding to the fact that \MALL\ is an SMCC with finite products and coproducts).
By \autoref{lem:eqct_id_top}, these MALL proofs have their $\top$ rules introducing $\vdash\top,0$ sequents only which gives the required property on the IMALL proofs we started with.
\end{proof}

We deduce that $\top\lolipop(\top\oplus\top) \iso (0\with 0)\lolipop 0$ is an isomorphism of MALL but not one of IMALL, using \autoref{lem:isoimall} and since all cut-free proofs of $(0\with 0)\lolipop 0\vdash \top\lolipop(\top\oplus\top)$ in IMALL have the following shape:
\begin{prooftree}
  \AxiomC{}
  \RightLabel{$\top$R}
  \UnaryInfC{$(0\with 0)\lolipop 0, \top\vdash \top$}
  \RightLabel{$\oplus_i$R}
  \UnaryInfC{$(0\with 0)\lolipop 0, \top\vdash \top\oplus\top$}
  \RightLabel{$\lolipop$R}
  \UnaryInfC{$(0\with 0)\lolipop 0\vdash \top\lolipop(\top\oplus\top)$}
\end{prooftree}

Let us now investigate SMCC with finite coproducts only (without products).
It is important to notice that an initial object $0$ in a SMCC induces that $0\lolipop F$ is a terminal object for any $F$.
This first means that we cannot uncorrelate completely products and coproducts.
It also means that the theory of isomorphisms includes the equation $0\lolipop F\iso 0\lolipop G$ even if it does not occur in $\thCP$.\footnote{Note this equation is derivable from $\thCP$ through the use of $\top$, but $\top$ does not belong to the language here.}
It might be the only missing equation: we conjecture isomorphisms are characterized by the following equations:
\begin{equation*}
\begin{array}{rcl}
  F\tens(G\tens H) & = & (F\tens G)\tens H \\
  F\tens G & = & G\tens F \\
  F\tens 1 & = & F \\
  F\oplus (G\oplus H) &=& (F\oplus G)\oplus H \\
  F\oplus G &=& G\oplus F \\
  F\oplus 0 &=& F \\
  F\tens(G\oplus H) &=& (F\tens G)\oplus(F\tens H) \\
  F\tens 0 &=& 0 \\
  (F\tens G)\lolipop H & = & F\lolipop (G\lolipop H) \\
  1 \lolipop F & = & F \\
  0\lolipop F & = & 0\lolipop G
\end{array}
\end{equation*}

Regarding a characterization through $\star$-autonomous categories, it is again not possible since, if we denote by $\top$ a terminal object, $(\top\lolipop 0)\lolipop 0$
% ((0\lolipop\_)\lolipop 0)\lolipop 0
and $\top\lolipop(\top\tens\top)$
% (0\lolipop\_)\lolipop((0\lolipop\_)\tens(0\lolipop\_))
are isomorphic in $\star$-autonomous categories (using $0\isoc\top\lolipop\bot$):
\begin{align*}
  (\top\lolipop 0)\lolipop 0 &\isoc (\top\lolipop(\top\lolipop\bot))\lolipop(\top\lolipop\bot) \\
  &\isoc ((\top\tens\top)\lolipop\bot)\lolipop(\top\lolipop\bot) \\
  &\isoc (((\top\tens\top)\lolipop\bot)\tens\top)\lolipop\bot \\
  &\isoc (\top\tens((\top\tens\top)\lolipop\bot))\lolipop\bot \\
  &\isoc \top\lolipop(((\top\tens\top)\lolipop\bot)\lolipop\bot) \isoc \top\lolipop(\top\tens\top)
\end{align*}
(or $\tradm{(\top\lolipop 0)\lolipop 0} \isom (\top\tens\top)\parr 0 \isom 0\parr(\top\tens\top) \isom \tradm{\top\lolipop(\top\tens\top)}$ in MALL).
But they are not isomorphic in SMCC with initial and terminal objects by \autoref{lem:isoimall} since all cut-free proofs of $(\top\lolipop 0)\lolipop 0\vdash \top\lolipop(\top\tens\top)$ in IMALL have the following shape (up to permuting the premises of the ($\tens$R)-rule):
\begin{equation*}
  \AxiomC{}
  \RightLabel{$\top$R}
  \UnaryInfC{$(\top\lolipop 0)\lolipop 0\vdash \top$}
  \AxiomC{}
  \RightLabel{$\top$R}
  \UnaryInfC{$\top\vdash \top$}
  \RightLabel{$\tens$R}
  \BinaryInfC{$(\top\lolipop 0)\lolipop 0, \top\vdash \top\tens\top$}
  \RightLabel{$\lolipop$R}
  \UnaryInfC{$(\top\lolipop 0)\lolipop 0\vdash \top\lolipop(\top\tens\top)$}
  \DisplayProof
\qquad\text{or}\qquad
  \AxiomC{}
  \RightLabel{$\top$R}
  \UnaryInfC{$(\top\lolipop 0)\lolipop 0, \top\vdash \top$}
  \AxiomC{}
  \RightLabel{$\top$R}
  \UnaryInfC{$\vdash \top$}
  \RightLabel{$\tens$R}
  \BinaryInfC{$(\top\lolipop 0)\lolipop 0, \top\vdash \top\tens\top$}
  \RightLabel{$\lolipop$R}
  \UnaryInfC{$(\top\lolipop 0)\lolipop 0\vdash \top\lolipop(\top\tens\top)$}
  \DisplayProof
\end{equation*}
\end{document}